\documentclass[aps,prx,twocolumn,superscriptaddress,floatfix]{revtex4-2}

\usepackage{comment}
\usepackage{amsmath,amssymb,amsthm}
\usepackage{pgf}
\usepackage{bm}
\usepackage{microtype}
\usepackage{unicode}
\usepackage[hidelinks,colorlinks,allcolors=blue]{hyperref}
\usepackage{dsfont}
\usepackage{mathtools}
\usepackage[capitalize]{cleveref}
\usepackage{xspace}
\usepackage{graphicx}
\graphicspath{{../images/paper_plots/}}
\usepackage{enumitem}
\usepackage{tikz}
\usetikzlibrary{decorations.markings,angles}
\usepackage{tkz-euclide}
\usetikzlibrary{fadings}
\usetikzlibrary{patterns}
\usetikzlibrary{shadows.blur}
\usetikzlibrary{shapes}
\usepackage[caption=false]{subfig}
\usepackage{empheq}
\usepackage{textcomp}
\newcommand{\mytexttilde}{\raisebox{0.5ex}{\texttildelow}}
\usepackage{siunitx}
\usepackage{booktabs}
\usepackage{longtable}
\usepackage{adjustbox}
\usepackage{mdframed}
\mdfsetup{nobreak=true}
\usepackage{thmtools, thm-restate}

\newcommand*\widefbox[1]{\fbox{\hspace{0.5em}#1\hspace{0.5em}}}

\newtheorem{theorem}{Theorem}

\newtheorem{lemma}{Lemma}
\newtheorem{corollary}{Corollary}

\theoremstyle{definition}
\newtheorem{definition}{Definition}

\theoremstyle{proposition}
\newtheorem{proposition}{Proposition}

\theoremstyle{example}
\newtheorem{example}{Example}

\theoremstyle{remark}
\newtheorem{remark}{Remark}

\newcommand{\tr}[1]{\ensuremath{\mathrm{tr}\left[#1\right]}\xspace}
\newcommand{\diff}{\ensuremath{\mathrm{d}\xspace}}
\newcommand{\wt}[1]{\widetilde{#1}}
\DeclareMathOperator{\arcosh}{arcosh}
\DeclareMathOperator{\artanh}{artanh}
\DeclareMathOperator{\sech}{sech}
\newcommand{\current}{\mathcal{J}\xspace}
\newcommand{\pluseq}{\mathrel{{+}{=}}}
\renewcommand{\Re}{\operatorname{Re}}
\renewcommand{\Im}{\operatorname{Im}}
\DeclareMathOperator{\sgn}{sgn}
\DeclareMathOperator{\argmin}{argmin}
\DeclareMathOperator{\argmax}{argmax}
\DeclareMathOperator{\supp}{supp}
\DeclareMathOperator{\Ai}{Ai}
\DeclareMathOperator{\erf}{erf}

\DeclarePairedDelimiter\norm{\lVert}{\rVert}

\def\Xint#1{\mathchoice
   {\XXint\displaystyle\textstyle{#1}}%
   {\XXint\textstyle\scriptstyle{#1}}%
   {\XXint\scriptstyle\scriptscriptstyle{#1}}%
   {\XXint\scriptscriptstyle\scriptscriptstyle{#1}}%
   \!\int}
\def\XXint#1#2#3{{\setbox0=\hbox{$#1{#2#3}{\int}$}
     \vcenter{\hbox{$#2#3$}}\kern-.5\wd0}}

\def\dashint{\Xint-}

\newcommand{\subfigimg}[3][,]{%
  \setbox1=\hbox{\includegraphics[#1]{#3}}%
  \leavevmode\rlap{\usebox1}%
  \rlap{\hspace*{-3pt}\raisebox{\dimexpr\ht1-0.5\baselineskip}{\textbf{#2}}}%
  \phantom{\usebox1}%
}

\makeatletter
\def\l@subsubsection#1#2{}
\makeatother

\begin{document}

\title{Emergent random matrix universality in quantum operator dynamics}

\author{Oliver Lunt}
\affiliation{Clarendon Laboratory, Department of Physics, University of Oxford, United Kingdom}
\affiliation{Department of Physics, King's College London, United Kingdom}

\author{Thomas Kriecherbauer}
\affiliation{Department of Mathematics, Universität Bayreuth, Germany}

\author{Kenneth T-R McLaughlin}
\affiliation{Department of Mathematics, Tulane University, United States}

\author{Curt von Keyserlingk}
\affiliation{Department of Physics, King's College London, United Kingdom}

\date[]{}

\begin{abstract}

The high complexity of many-body quantum dynamics means that essentially all analytical or numerical approaches either exploit special structure or are approximate in nature.
One such approach---the memory function formalism---involves a carefully chosen split into `fast' and `slow' modes.
An approximate model for the fast modes can then be used to solve for Green's functions $G(z)$ of the slow modes, and the success of this approach depends on the accuracy of the fast space approximation. Using a formulation in operator Krylov space known as the recursion method, we prove the emergence of a universal random matrix description of the fast mode dynamics. This is captured by the `level-$n$ Green's function' $G_{n}(z)$, which we show approaches universal scaling forms in the `fast limit' $n\to\infty$. Notably, this emergent universality can occur in both chaotic and non-chaotic systems, provided their spectral functions are suﬃciently smooth.
This universality of $G_{n}(z)$ turns out to be precisely analogous to the universality of eigenvalue correlations in random matrix theory (RMT), even though there is \textit{no explicit randomness} present in the Hamiltonian.
Concretely, at finite $z$ we show that $G_{n}(z)$ approaches the Wigner semicircle law, while if $G(z)$ is the Green's function of certain hydrodynamical variables, we show that at low frequencies $G_{n}(z)$ is instead governed by the Bessel universality class from RMT. 
As an application of this universality, we give a new numerical method, the \textit{spectral bootstrap}, for approximating spectral functions, including hydrodynamic transport data, from a finite number of Lanczos coefficients.
Our proof involves a map to a Riemann-Hilbert problem which we solve using a steepest-descent-type method, rigorously controlled in the $n\to\infty$ limit.
Via the steepest-descent procedure, we are led to a related Coulomb gas optimization problem, and we discuss how a recent conjecture---the `Operator Growth Hypothesis'---implies that chaotic operator dynamics can generically be identified with the critical point of a confinement transition in this Coulomb gas.
These results elevate the recursion method from a useful numerical technique to a theoretically principled framework with universal content.

\end{abstract}

\maketitle

\tableofcontents

\section{Introduction}
Improvements in our understanding of many-body quantum systems can inform the development of new numerical algorithms to simulate them. For example, the realization that gapped 1D ground states have area law entanglement \cite{hastingsAreaLawOnedimensional2007} justifies the great success of variational tensor network algorithms~\cite{whiteDensityMatrixFormulation1992}. Moving beyond ground states, the study of many-body quantum dynamics presents a longstanding challenge.
Due to the rapid growth of entanglement for generic interacting systems~\cite{kimBallisticSpreadingEntanglement2013}, the memory requirements for faithful tensor network descriptions typically grow exponentially in time~\cite{schuchEntropyScalingSimulability2008}. Given this difficult state of affairs, it is worth asking: are there universal features of many-body quantum dynamics, and if so, how can we utilize this universality to design better algorithms?

Recently an array of algorithms have been proposed that apply new insights from quantum information and quantum chaos to the old problem of hydrodynamics~\cite{rakovszkyDissipationassistedOperatorEvolution2022,whiteQuantumDynamicsThermalizing2018,kleinkvorningTimeevolutionLocalInformation2022,artiacoEfficientLargeScaleManyBody2024,whiteEffectiveDissipationRate2023,yi-thomasComparingNumericalMethods2024,parkerUniversalOperatorGrowth2019,prosenMatrixProductSimulations2009,vonkeyserlingkOperatorBackflowClassical2022,srivatsaProbingHydrodynamicCrossovers2024,lloydBallisticDiffusiveCrossover2024}. These algorithms work by discarding information about a time-evolved quantum operator that can only be detected using very non-local probes, like high-order correlation functions. If one is primarily interested in physics described by low-order correlation functions, like linear response hydrodynamics, then such approximations can be worthwhile in order to reduce computational requirements. To develop intuition for why this could work, consider the autocorrelation function $C(t) = (O_{0}|O_{0}(t))$ of a normalized local operator $O_{0}$ evolving in the Heisenberg picture as $O_{0}(t) = e^{i H t} O_{0} e^{-i H t}$, and for simplicity let us work at infinite temperature so that $(A|B) = \tr{A^{\dag} B} / \tr{\mathds{1}}$. The correlation function gives the probability amplitude at time $t$ for the time-evolved operator to return to its starting point---that is, it measures the operator `backflow'. Decomposing $O_{0}(t)$ as a superposition over operator paths, the dominant paths that contribute to $C(t)$ are those that involve only `simple' local operators~\cite{vonkeyserlingkOperatorBackflowClassical2022,nahumRealtimeCorrelatorsChaotic2022}. The intuition is that, at least for ergodic dynamics, once an operator becomes sufficiently non-local, it is very unlikely to shrink down to a small operator again, and will therefore give a negligible contribution to $C(t)$. This suggests that one should be able to neglect such non-local operator histories while incurring only a small error in $C(t)$. Where these algorithms differ is in precisely how they choose to discard or approximate these non-local operators.

\begin{figure*}[t]

    \subfloat{\tikzset{every picture/.style={line width=0.75pt}} %

\begin{tikzpicture}[x=0.75pt,y=0.75pt,yscale=-1,xscale=1]

\draw  [fill={rgb, 255:red, 0; green, 0; blue, 0 }  ,fill opacity=1 ][line width=1.5]  (96,119.91) .. controls (96,116.06) and (99.12,112.95) .. (102.96,112.95) .. controls (106.81,112.95) and (109.93,116.06) .. (109.93,119.91) .. controls (109.93,123.76) and (106.81,126.87) .. (102.96,126.87) .. controls (99.12,126.87) and (96,123.76) .. (96,119.91) -- cycle ;
\draw  [fill={rgb, 255:red, 0; green, 0; blue, 0 }  ,fill opacity=1 ][line width=1.5]  (152.84,119.91) .. controls (152.84,116.06) and (155.96,112.95) .. (159.81,112.95) .. controls (163.65,112.95) and (166.77,116.06) .. (166.77,119.91) .. controls (166.77,123.76) and (163.65,126.87) .. (159.81,126.87) .. controls (155.96,126.87) and (152.84,123.76) .. (152.84,119.91) -- cycle ;
\draw [line width=1.5]    (109.93,119.91) -- (152.84,119.91) ;
\draw  [fill={rgb, 255:red, 0; green, 0; blue, 0 }  ,fill opacity=1 ][line width=1.5]  (210.62,120.05) .. controls (210.62,116.21) and (213.73,113.09) .. (217.58,113.09) .. controls (221.42,113.09) and (224.54,116.21) .. (224.54,120.05) .. controls (224.54,123.9) and (221.42,127.02) .. (217.58,127.02) .. controls (213.73,127.02) and (210.62,123.9) .. (210.62,120.05) -- cycle ;
\draw [line width=1.5]    (168.19,119.91) -- (179.84,120.05) ;
\draw [line width=1.5]    (200.24,119.91) -- (219,120.05) ;
\draw  [fill={rgb, 255:red, 0; green, 0; blue, 0 }  ,fill opacity=1 ][line width=1.5]  (260.49,120.19) .. controls (260.49,116.35) and (263.61,113.23) .. (267.46,113.23) .. controls (271.3,113.23) and (274.42,116.35) .. (274.42,120.19) .. controls (274.42,124.04) and (271.3,127.16) .. (267.46,127.16) .. controls (263.61,127.16) and (260.49,124.04) .. (260.49,120.19) -- cycle ;
\draw [line width=1.5]    (217.58,120.19) -- (260.49,120.19) ;
\draw [line width=1.5]    (274.14,119.91) -- (285.79,120.05) ;
\draw  [color={rgb, 255:red, 255; green, 0; blue, 2 }  ,draw opacity=1 ][dash pattern={on 4.5pt off 4.5pt}][line width=0.75]  (254,110) .. controls (254,104.48) and (258.48,100) .. (264,100) -- (350,100) .. controls (355.52,100) and (360,104.48) .. (360,110) -- (360,140) .. controls (360,145.52) and (355.52,150) .. (350,150) -- (264,150) .. controls (258.48,150) and (254,145.52) .. (254,140) -- cycle ;
\draw  [fill={rgb, 255:red, 255; green, 255; blue, 255 }  ,fill opacity=1 ] (353,126) .. controls (353,121.03) and (330.61,117) .. (303,117) -- (303,108) .. controls (330.61,108) and (353,112.03) .. (353,117) ;\draw  [fill={rgb, 255:red, 255; green, 255; blue, 255 }  ,fill opacity=1 ] (353,117) .. controls (353,121.23) and (336.82,124.77) .. (315,125.74) -- (315,122.74) -- (303,130.5) -- (315,137.74) -- (315,134.74) .. controls (336.82,133.77) and (353,130.23) .. (353,126)(353,117) -- (353,126) ;

\draw (96,130) node [anchor=north west][inner sep=0.75pt]  [font=\footnotesize]  {$O_{0}$};
\draw (153,130) node [anchor=north west][inner sep=0.75pt]  [font=\footnotesize]  {$O_{1}$};
\draw (181,116) node [anchor=north west][inner sep=0.75pt]    {$\cdots $};
\draw (210,130) node [anchor=north west][inner sep=0.75pt]  [font=\footnotesize]  {$O_{n-1}$};
\draw (260,130) node [anchor=north west][inner sep=0.75pt]  [font=\footnotesize]  {$O_{n}$};
\draw (287,116) node [anchor=north west][inner sep=0.75pt]    {$\cdots $};
\draw (287,80) node [anchor=north west][inner sep=0.75pt] [font=\normalsize]   {$G_{n}( z)$};
\draw (125,105) node [anchor=north west][inner sep=0.75pt]  [font=\footnotesize]  {$b_{1}$};
\draw (235,105) node [anchor=north west][inner sep=0.75pt]  [font=\footnotesize]  {$b_{n}$};

\draw [color={rgb, 255:red, 249; green, 56; blue, 56 }  ,draw opacity=1 ][line width=1.5]    (125,165) -- (311,165) ;
\draw [shift={(315,165)}, rotate = 180] [fill={rgb, 255:red, 249; green, 56; blue, 56 }  ,fill opacity=1 ][line width=0.08]  [draw opacity=0] (13.4,-6.43) -- (0,0) -- (13.4,6.44) -- (8.9,0) -- cycle    ;
\draw (82,156) node [anchor=north west][inner sep=0.75pt]  [font=\large,color={rgb, 255:red, 249; green, 56; blue, 56 }  ,opacity=1 ]  {$\textbf{Slow}$};
\draw (320,156) node [anchor=north west][inner sep=0.75pt]  [font=\large,color={rgb, 255:red, 249; green, 56; blue, 56 }  ,opacity=1 ]  {$\textbf{Fast}$};

\draw (60,80) node [anchor=north west][inner sep=0.75pt]   [align=left] {\textbf{{\normalsize(a)}}};

\end{tikzpicture}}\hfill%
    \subfloat{    \tikzset{every picture/.style={line width=0.75pt}} %

    \begin{tikzpicture}[x=0.75pt,y=0.75pt,yscale=-1,xscale=1]

    \draw  [fill={rgb, 255:red, 215; green, 230; blue, 255 }  ,fill opacity=1 ] (110.47,136.86) .. controls (110.47,136.86) and (110.47,136.86) .. (110.47,136.86) .. controls (122.91,136.86) and (133,146.95) .. (133,159.4) -- (110.47,159.4) -- cycle ;
    \draw  [fill={rgb, 255:red, 215; green, 230; blue, 255 }  ,fill opacity=1 ] (225,159.4) .. controls (225,134.69) and (242.73,114.58) .. (264.79,114.37) .. controls (287,114.17) and (305.19,134.21) .. (305.42,159.14) .. controls (305.42,159.22) and (305.42,159.3) .. (305.42,159.37) -- cycle ;
    \draw [color={rgb, 255:red, 150; green, 150; blue, 150 }  ,draw opacity=1 ] [dash pattern={on 4.5pt off 4.5pt}]  (264.44,159.93) -- (289.15,128.36) ;
    \draw [shift={(291,126)}, rotate = 128.06] [fill={rgb, 255:red, 150; green, 150; blue, 150 }  ,fill opacity=1 ][line width=0.08]  [draw opacity=0] (7.14,-3.43) -- (0,0) -- (7.14,3.43) -- cycle    ;
    \draw    (190,160) -- (328,160) ;
    \draw [shift={(330,160)}, rotate = 180] [color={rgb, 255:red, 0; green, 0; blue, 0 }  ][line width=0.75]    (10.93,-3.29) .. controls (6.95,-1.4) and (3.31,-0.3) .. (0,0) .. controls (3.31,0.3) and (6.95,1.4) .. (10.93,3.29)   ;
    \draw    (110,160) -- (110,102) ;
    \draw [shift={(110,100)}, rotate = 90] [color={rgb, 255:red, 0; green, 0; blue, 0 }  ][line width=0.75]    (10.93,-3.29) .. controls (6.95,-1.4) and (3.31,-0.3) .. (0,0) .. controls (3.31,0.3) and (6.95,1.4) .. (10.93,3.29)   ;
    \draw  [fill={rgb, 255:red, 0; green, 0; blue, 0 }  ,fill opacity=1 ] (261.94,159.93) .. controls (261.94,158.55) and (263.05,157.43) .. (264.44,157.43) .. controls (265.82,157.43) and (266.94,158.55) .. (266.94,159.93) .. controls (266.94,161.31) and (265.82,162.43) .. (264.44,162.43) .. controls (263.05,162.43) and (261.94,161.31) .. (261.94,159.93) -- cycle ;
    \draw  [fill={rgb, 255:red, 0; green, 0; blue, 0 }  ,fill opacity=1 ] (107.94,159.93) .. controls (107.94,158.55) and (109.05,157.43) .. (110.44,157.43) .. controls (111.82,157.43) and (112.94,158.55) .. (112.94,159.93) .. controls (112.94,161.31) and (111.82,162.43) .. (110.44,162.43) .. controls (109.05,162.43) and (107.94,161.31) .. (107.94,159.93) -- cycle ;
    \draw    (109.65,159.3) -- (184.22,159.88) ;
    \draw    (181,165) -- (187,155) ;
    \draw    (187,165) -- (193,155) ;
    \draw  [draw opacity=0] (101.81,141.09) .. controls (106.61,141.28) and (111.42,142.62) .. (115.85,145.23) .. controls (117.97,146.48) and (119.89,147.96) .. (121.58,149.6) -- (100.62,171.08) -- cycle ; \draw    (101.81,141.09) .. controls (106.61,141.28) and (111.42,142.62) .. (115.85,145.23) .. controls (117.08,145.96) and (118.24,146.76) .. (119.33,147.62) ; \draw [shift={(121.58,149.6)}, rotate = 216.73] [fill={rgb, 255:red, 0; green, 0; blue, 0 }  ][line width=0.08]  [draw opacity=0] (8.93,-4.29) -- (0,0) -- (8.93,4.29) -- cycle    ; 
    \draw  [draw opacity=0] (313.46,141.64) .. controls (309.08,146.3) and (303.17,149.61) .. (296.36,150.7) .. controls (293.92,151.09) and (291.51,151.17) .. (289.16,150.98) -- (291.62,121.08) -- cycle ; \draw    (313.46,141.64) .. controls (309.08,146.3) and (303.17,149.61) .. (296.36,150.7) .. controls (294.95,150.93) and (293.54,151.05) .. (292.15,151.08) ; \draw [shift={(289.16,150.98)}, rotate = 357.12] [fill={rgb, 255:red, 0; green, 0; blue, 0 }  ][line width=0.08]  [draw opacity=0] (8.93,-4.29) -- (0,0) -- (8.93,4.29) -- cycle    ; 

    \draw (106,164) node [anchor=north west][inner sep=0.75pt]  [font=\footnotesize]  {$0$};
    \draw (94,84) node [anchor=north west][inner sep=0.75pt]  [font=\small]  {$\mathrm{Im} \ z $};
    \draw (332,153) node [anchor=north west][inner sep=0.75pt]  [font=\small]  {$\mathrm{Re} \ z $};
    \draw (257,164) node [anchor=north west][inner sep=0.75pt]  [font=\footnotesize]  {$\beta _{n}$};
    \draw (125,163) node [anchor=north west][inner sep=0.75pt]  [font=\footnotesize]  {$\delta _{0}$};
    \draw (255,127) node [anchor=north west][inner sep=0.75pt]  [font=\footnotesize]  {$\beta _{n} \delta _{1}$};
    \draw (160,137) node [anchor=north west][inner sep=0.75pt]   [align=left] {\textbf{{\small Bulk}}};
    \draw (56,125) node [anchor=north west][inner sep=0.75pt]   [align=left] {\begin{minipage}[lt]{32.3pt}\setlength\topsep{0pt}
    \begin{flushright}
    \textbf{{\small Bessel}}\\\textbf{{\small region}}
    \end{flushright}

    \end{minipage}};
    \draw (305,112) node [anchor=north west][inner sep=0.75pt]   [align=left] {\textbf{{\small Airy}}\\\textbf{{\small region}}};

    \draw (55,80) node [anchor=north west][inner sep=0.75pt]   [align=left] {\textbf{{\normalsize(b)}}};

    \end{tikzpicture}}

    \caption{\textbf{(a)} The semi-infinite 1D operator chain produced by the Lanczos algorithm. The level-$n$ Green's function $G_{n}(z) = (O_{n}|(z-\mathcal{L}_{n})^{-1} |O_{n})$ is the effective Green's function for operator dynamics restricted to operators $O_{n}$ and above, which we call the `fast space' (indicated by the dashed box). A good model for the fast space dynamics encoded in $G_{n}(z)$ can be used to estimate the full Green's function $G(z)$ via the continued fraction in \cref{eq:G_continued_frac}. Intuitively, $G_{n}(z)$ captures the `backflow' (indicated by the curved arrow) from the fast space to the slow space. \textbf{(b)} We prove that $G_{n}(z)$ approaches \textit{universal scaling forms} as $n\to\infty$, with different limits for different regions of the complex $z$-plane. The most prominent example is that $G_{n}(z)$ approaches the Wigner semicircle law in the `bulk' of the spectrum, but there can be different behavior near the origin and the spectral edge. This emergent universality is precisely analogous to the universality of eigenvalue correlations of random matrices, even though there is \textit{no explicit randomness} here. For illustration we show only the first quadrant, with the other quadrants obtained by reflection about the axes. Here $\delta_{0}$ and $\delta_{1}$ are small $\mathcal{O}(1)$ constants in units of the microscopic couplings, and $\beta_{n} \approx 2 b_{n}$ determines the bandwidth of the bulk spectrum.
    }

    \label{fig:plancherel_schematic}

\end{figure*}
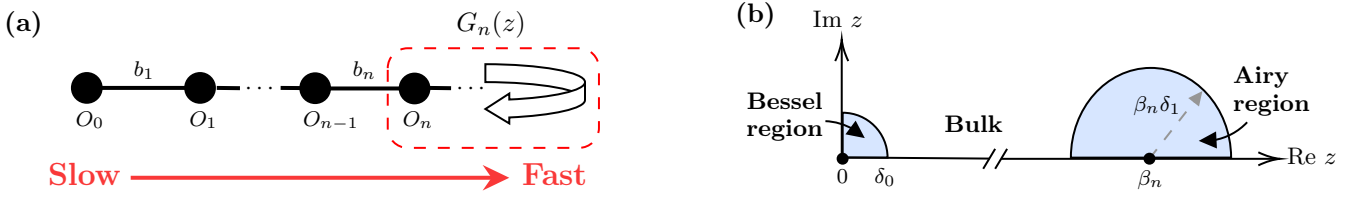

In this paper we give a new approach to systematically approximate the contributions from complicated operator paths. Our approach exploits an emergent random-matrix like universality of quantum operator dynamics. As well as having a practical application in estimating Green's functions, this therefore also yields a new connection between quantum operator dynamics and random matrix theory.
Our results can be phrased in the language of a standard tool from hydrodynamics, the Mori-Zwanzig memory function formalism~\cite{zwanzigMemoryEffectsIrreversible1961,moriContinuedFractionRepresentationTimeCorrelation1965,forsterHydrodynamicFluctuationsBroken2019}. Given a Hamiltonian $H$ and an initial operator $O_{0}$, consider the Green's function
\begin{equation}
    G(z) = \left(O_{0}\left| \dfrac{1}{z-\mathcal{L}}\right|O_{0}\right),
    \label{eq:Gz_intro}
\end{equation}
where $\mathcal{L}(\cdot) = [H, \cdot\,]$ is the Liouvillian, and we are using a vectorized notation for operators, $O_{0} \leftrightarrow  |O_{0})$. At the level of linear response, this Green's function describes the system's behavior under a perturbation generated by $O_{0}$~\cite{forsterHydrodynamicFluctuationsBroken2019}. In order to link theoretical models with experimental measurements, it is therefore of strong interest to be able to compute objects like $G(z)$. However, doing this analytically is usually intractable for generic interacting systems, and even numerically it remains challenging to compute exactly when the system is large or the frequency $z$ is small, so some approximate treatment is required.

The memory function formalism gives a framework in which to approximately calculate $G(z)$ by separating operators into `fast' and `slow' modes, according to some choice (to be discussed below). The internal dynamics of the fast modes is approximated by making some assumptions about correlation functions between fast operators, e.g., that they decay exponentially in time. Having reduced this complicated operator dynamics down to the specification of a few free parameters, like exponential decay rates, this fast mode approximation can then be used to close the equations of motion for the slow modes, allowing for an evaluation of their Green's functions. The effectiveness of this approach depends on:%
\begin{enumerate}[label=\roman*)]
    \item Making a judicious choice of fast and slow modes.
    \item The accuracy of the approximation for the fast mode dynamics.
\end{enumerate}

Given the generality of this framework, there are a number of reasonable schemes one could choose to define the fast and slow modes. One approach is `hydrodynamic', where the slow modes are taken to be those operators that overlap with any conserved quantities related to symmetries of $H$~\cite{forsterHydrodynamicFluctuationsBroken2019}; they are slow in the sense that their correlation functions typically decay algebraically in time, as opposed to superpolynomial decay for non-conserved operators. This is physically well motivated, but requires a degree of physical insight, particularly if the symmetries are not obvious, such as in integrable systems~\cite{ilievskiQuasilocalConservedOperators2015}. 

We will take a different approach, which in modern language can be thought of in terms of `operator complexity', although its origins date back to the 1970s~\cite{haydockElectronicStructureBased1972}. It can be motivated by formally expanding the resolvent $(z-\mathcal{L})^{-1}$ in \cref{eq:Gz_intro}, from which we see that $G(z)$ will contain contributions from operator paths involving operators of the form $\mathcal{L}^{n}|O_{0})$, $n=0,1,\dots$. Generically, the higher the power of $n$, the more `complicated' and nonlocal these operators become. The rough intuition is that, as these operators become more nonlocal, their internal dynamics becomes faster (this will later be made sharp). One then chooses the slow space to consist of those operators that can be generated by acting only $n< n_{\mathrm{max}}$ times with the Liouvillian $\mathcal{L}$, up to some $n_{\mathrm{max}}$ dictated by computational constraints.

To make this precise, one first forms an orthonormal basis of the operator Krylov space spanned by the operators $\{\mathcal{L}^{n}|O_{0})\}_{n=0}^{\infty}$. If $\mathcal{L}$ is self-adjoint, this orthogonalization is called the \textit{Lanczos algorithm}, but it is really just Gram-Schmidt orthogonalization on the Krylov space~\cite{saadNumericalMethodsLarge1992}. This produces two outputs: an orthonormal operator basis $\{O_{n}\}_{n=0}^{\infty}$, called the Lanczos basis, and a set of nonnegative real numbers $\{b_{n}\}_{n=1}^{\infty}$, called the Lanczos coefficients. The Lanczos coefficients have a simple origin: they are just the Gram-Schmidt normalization constants. But they play a privileged role, because it turns out that the matrix representation of the Liouvillian $\mathcal{L}$ restricted to the Lanczos basis is a tridiagonal matrix, where the nonzero matrix elements are given by the Lanczos coefficients: 
\begin{equation}
    \mathcal{L} = \begin{pmatrix}
        0 & b_{1} & 0 & 0 & \cdots \\
        b_{1} & 0 & b_{2} & 0 & \cdots \\
        0 & b_{2} & 0 & b_{3} & \cdots \\
        0 & 0 & b_{3} & 0 & \ddots \\
        \vdots & \vdots & \vdots & \ddots & \ddots
    \end{pmatrix}.
    \label{eq:tridiagonal}
\end{equation}
This matrix can be thought of as a Hamiltonian for an effective 1D hopping model, where the sites of the chain are the basis operators $O_{n}$, and the hopping strengths are given by the Lanczos coefficients (see \cref{fig:plancherel_schematic}a). The operator dynamics can then be visualized as the evolution of a single-particle wavefunction $O_{0}(t)$ on this 1D operator chain. One can associate a notion of operator complexity with this evolution through, say, its average position $\mathcal{K}(t) = \sum_{n=0}^{\infty} n \times |(O_{n}|O_{0}(t))|^{2}$ along the Lanczos chain. This is known as `operator Krylov complexity', as introduced in Ref.\ \cite{parkerUniversalOperatorGrowth2019}, and has seen a large amount of interest as a tractable measure of operator complexity~\cite{rabinoviciKrylovComplexity2025}. Ref.\ \cite{balasubramanianQuantumChaosComplexity2022} proved there is a natural sense in which the Lanczos basis is the `minimal complexity' basis, at least at early times. See Ref.\ \cite{rabinoviciKrylovComplexity2025} for a review of Krylov complexity, and Ref.\ \cite{nandyQuantumDynamicsKrylov2025} for a more general review of quantum dynamics in Krylov space.

Returning to the memory function formalism, having formed this orthonormal operator basis, we take the fast space to be the span of $\{O_{m}\}_{m\geq n}$ for some large $n$, with orthogonal fast space projector $\mathcal{Q}_{n} \coloneqq \mathds{1} - \sum_{m=0}^{n-1} |O_{m})(O_{m}|$. One useful output of this choice is a continued fraction representation for the original Green's function~\cite{viswanathRecursionMethodApplication2013}:
\begin{equation}
    G(z) = \dfrac{1}{z - \dfrac{b_{1}^{2}}{z - \cdots \underset{\cdots - \dfrac{b_{n-1}^{2}}{z - b_{n}^{2} G_{n}(z)}}{}}}
    \label{eq:G_continued_frac}
\end{equation}
The function $G_{n}(z)$ appearing at level-$n$ of the continued fraction is known as the \textit{level-$n$ Green's function}, and is defined as 
\begin{equation}
    G_{n}(z) \coloneqq \left(O_{n} \left| \dfrac{1}{z - \mathcal{L}_{n}} \right| O_{n}\right),
\end{equation}
where $\mathcal{L}_{n} = \mathcal{Q}_{n} \mathcal{L} \mathcal{Q}_{n}$ is the Liouvillian restricted to the fast space, corresponding to the lower right block of the tridiagonal matrix in \cref{eq:tridiagonal}. As such, $G_{n}(z)$ is a Green's function representing the internal dynamics of the fast space. In terms of the 1D hopping model, it describes dynamics restricted to sites $n$ and above, indicated by the boxed region in \cref{fig:plancherel_schematic}(a). The continued fraction representation shows that, if one can accurately describe the fast space dynamics, then one can leverage this to estimate the full many-body Green's function $G(z)$. Choosing such an appropriate `terminator' for $G_{n}(z)$ forms the basis of a numerical approach called the \textit{recursion method}~\cite{viswanathRecursionMethodApplication2013}. 

\subsection{Relation to previous work}

The foundations of the recursion method were developed in the 1970s--1990s~\cite{haydockElectronicStructureBased1972,haydockRecursiveSolutionSchrodinger1980,pettiforRecursionMethodIts1987,viswanathRecursionMethodApplication2013}. Making a good choice for $G_{n}(z)$ is key to the success of the method, and the nature of the continued fraction means this is more nontrivial than, say, simply setting $G_{n}(z) = 0$ like one might truncate a Taylor series (indeed, this is a very poor approximation~\cite{viswanathRecursionMethodApplication2013}). It was quickly recognized that this choice for $G_{n}(z)$ should be dictated by various `high-level' features of the underlying spectral function
\begin{equation}
    \Phi(\omega) \coloneqq \int_{\mathbb{R}} e^{-i \omega t} \big(O_{0}| O_{0}(t)\big) \diff t,
\end{equation}
such as its behavior at high- and low-frequencies, and any gaps in the spectrum. The relevance of $\Phi(\omega)$ is that it turns out to completely determine all the properties of the Lanczos basis (see \cref{sec:synopsis_orth_pols,sec:background}). To some degree, the original practice of the recursion method consisted of finding a toy spectral function with similar high-level features to those expected of $\Phi(\omega)$ on physical grounds, but having analytically computable Lanczos coefficients and Green's function, which could then be used to terminate the continued fraction~\cite{viswanathRecursionMethodApplication2013}. This procedure is not systematic, and depends heavily on the existence of an appropriate solvable toy model, which is not always available. A central goal of our work is to put the recursion method on a more systematic footing, with a particular focus on computing hydrodynamic transport coefficients by estimating $\Phi(\omega)$ at low frequencies.%

The recursion method received renewed attention in the late 2010s during an explosion of interest in operator dynamics and scrambling~\cite{swingleUnscramblingPhysicsOutoftimeorder2018}, prompted by questions related to the black hole information paradox~\cite{haydenBlackHolesMirrors2007,sekinoFastScramblers2008,almheiriEntropyHawkingRadiation2021} and the dynamics of quantum many-body systems~\cite{vonkeyserlingkOperatorHydrodynamicsOTOCs2018,nahumOperatorSpreadingRandom2018}. In particular, Ref.\ \cite{parkerUniversalOperatorGrowth2019} stated the `Operator Growth Hypothesis' (OGH). The OGH conjectures that one characteristic of quantum chaotic systems is that operators grow as fast as possible, limited only by the locality of interactions. Locality places an upper limit on operator growth by bounding the possible number of operator trajectories. The OGH can be formulated precisely as the statement that, in a generic quantum chaotic system, the Lanczos coefficients $b_{n}$ grow \textit{linearly} as $n\to\infty$ (with a logarithmic correction in one dimensional systems). This is the fastest possible growth rate in lattice systems with a local Hamiltonian~\cite{parkerUniversalOperatorGrowth2019}. Under mild assumptions on the spectral function $\Phi(\omega)$, linear growth of $b_{n}$ is implied by $\Phi(\omega)$ decaying exponentially as $|\omega|\to\infty$~\cite{lubinskyProofFreudConjecture1988}; that is, the OGH is a statement about \textit{high-frequency} behavior. Although this upper bound on $b_{n}$ does seem to be saturated in many quantum chaotic systems~\cite{parkerUniversalOperatorGrowth2019}, it turns out that some non-chaotic systems also saturate the upper bound~\cite{caoStatisticalMechanismOperator2021,dymarskyKrylovComplexityConformal2021,bhattacharjeeKrylovComplexitySaddledominated2022}, so satisfying the OGH can be seen only as a necessary but not sufficient requirement for a system to be considered quantum chaotic.

A challenging aspect of the recursion method is that, while high-frequency features are clearly reflected in the gross structure of the Lanczos coefficients, the imprints of low-frequency behavior are reflected in finer structure. Various solvable models have helped decode how different low-frequency phenomena are exhibited in subleading effects, particularly in differences between the Lanczos coefficients for odd and even $n$, referred to as `staggering'. Ref.\ \cite{viswanathRecursionMethodApplication2013} shows some examples where the spectral function $\Phi(\omega)$ has a power-law at low frequencies, $\Phi(\omega\to 0) \sim |\omega|^{\rho}$, and how the value of $\rho$ shows up in the coefficient of a subleading staggered term in the Lanczos coefficients. Other authors have studied the signatures of long-lived zero modes~\cite{yatesDynamicsAlmostStrong2020,yatesLifetimeAlmostStrong2020,yatesLonglivedPerioddoubledEdge2022,tausendpfundAlmostStrongZero2025}, many-body localization~\cite{ballartriguerosKrylovComplexityManybody2022}, and mass gaps in quantum field theories~\cite{avdoshkinKrylovComplexityQuantum2024}.

Motivated by the OGH, Ref.\ \cite{parkerUniversalOperatorGrowth2019} provides an exactly solvable spectral function with asymptotically linear Lanczos coefficients, and uses this model as a terminator for $G_{n}(z)$ in order to estimate the energy diffusion constant of the chaotic mixed field Ising model (see also Refs.\ \cite{wangDiffusionConstantsRecursion2024,fullgrafLanczosPascalApproachCorrelation2025} for similar approaches). This toy spectral function was effectively chosen only to match the high-frequency behavior of $\Phi(\omega)$ through the linear growth of Lanczos coefficients. But the relation between $G(z)$ and $G_{n}(z)$ in the continued fraction in \cref{eq:G_continued_frac} holds pointwise in the frequency $z$, so it is \textit{a priori} unclear why a choice of $G_{n}(z)$ catered to the high-frequency behavior should have been successful at reproducing a low-frequency property like a diffusion constant. We address this question in \cref{sec:gf}.

While these solvable examples are illustrative, there are relatively few general results proving a definitive link between a given low-frequency feature and its signature in Lanczos space. Such results would be valuable for laying solid foundations for new numerical methods. Most existing general results are limited to spectral functions with finite bandwidth~\cite{nevaiOrthogonalPolynomials1979,vanlessenStrongAsymptoticsRecurrence2003}, corresponding to Lanczos coefficients with bounded growth---these results are appropriate for non-interacting systems~\cite{viswanathRecursionMethodApplication2013} and finite-size systems. By contrast, interacting systems have spectral functions with infinite support in the thermodynamic limit~\cite{viswanathRecursionMethodApplication2013}, corresponding to Lanczos coefficients with unbounded growth, and rigorous results~\cite{vanlessenStrongAsymptoticsLaguerreType2007} on these systems are rare and usually disconnected from the recursion method literature. Within this class, results on systems with asymptotically linear Lanczos coefficients---the generic case according to the OGH---are even more technically challenging, because these systems turn out to be `marginal'~\cite{kriecherbauerStrongAsymptoticsPolynomials1999,levinOrthogonalPolynomialsWeights2007} in a sense we will describe in \cref{sec:synopsis_coulomb,sec:synopsis_marginality}, and this marginality makes proofs particularly delicate.

Assuming only a few physically motivated conditions on $\Phi(\omega)$, stated in \cref{sec:potential_definitions}, one of our main results is that, as $n\to\infty$, $G_{n}(z)$ approaches \textit{universal scaling forms}, to be discussed further in \cref{sec:synopsis_emergent_universality,sec:gf}. This universality can then be exploited to estimate the original Green's function $G(z)$. As we will explain in \cref{sec:synopsis_rmt_universality}, the universality of $G_{n}(z)$ turns out to be precisely analogous to the universality of eigenvalue correlations of $n\times n$ random matrices. Importantly, our assumptions on the spectral function include the generic case of linearly growing Lanczos coefficients, and much of the technical work in our proof is targeted at this case. One notable difference from previous approaches based on the recursion method~\cite{parkerLocalMatrixProduct2020,wangDiffusionConstantsRecursion2024,fullgrafLanczosPascalApproachCorrelation2025} is that we explicitly consider the role of conservation laws in determining the operator backflow encoded in $G_{n}(z)$. This is particularly relevant for spectral functions $\Phi(\omega)$ that have power-law divergences as $\omega\to 0$, which we show leads to a new `Bessel universality class' describing the behavior of $G_{n}(z)$ near $z=0$. Not accounting for this modified universality can lead to erroneous estimates of the spectral function at low frequencies, as we discuss in \cref{sec:spectral_bootstrap}. These results thus give a principled theoretical framework for numerical approaches based on the recursion method.  %

The appearance of random matrix theory (RMT) in this context may seem surprising \textit{a priori}, since we did not explicitly impose any randomness in either the Hamiltonian $H$ or the initial operator $O_{0}$. Of course, RMT is also very well studied in the context of quantum chaos~\cite{dalessioQuantumChaosEigenstate2016}, where it is conjectured that quantum chaotic systems have energy eigenvalue statistics that can be described in terms of random matrix ensembles. A natural question to ask is whether the RMT universality we discuss here is related to the appearance of RMT in quantum chaos, and we delineate the differences in \cref{sec:synopsis_eth}. Connections between the Lanczos algorithm and random matrix theory were also discussed in Ref.~\cite{balasubramanianTridiagonalizingRandomMatrices2023}, but not with a focus on universality akin to that of eigenvalue correlations, as we discuss here; instead they focused on the resulting Lanczos coefficients when $H$ itself is a random matrix.%

At a higher level, the recursion method is one of a number of approaches that seek to simplify quantum operator dynamics by somehow approximating the contributions from nonlocal operators. Some approaches, including dissipation-assisted operator evolution (DAOE)~\cite{rakovszkyDissipationassistedOperatorEvolution2022,vonkeyserlingkOperatorBackflowClassical2022,srivatsaOperatorGrowthHypothesis2024,lloydBallisticDiffusiveCrossover2024} and density matrix truncation (DMT)~\cite{yi-thomasComparingNumericalMethods2024}, suppress nonlocal operators based on a fairly explicit notion of weight or size. The recursion method is distinct from these approaches in that it seeks to approximate, rather than explicitly discard, the contributions from nonlocal operators $O_{n\gg 1}$. In this sense it has parallels with local information time evolution (LITE)~\cite{kleinkvorningTimeevolutionLocalInformation2022,artiacoEfficientLargeScaleManyBody2024}, but there are also moral similarities with DAOE and DMT since the common assumption is that the contributions from nonlocal operators should be `small' in some sense. One advantage of the Lanczos setting is that it has enough mathematical structure that one can really prove rigorous results about the `operator backflow' that is the subject of these approximations. Furthermore, we view the recursion method not as a direct competitor to those approaches but rather as complementary, and we discuss in \cref{sec:synopsis_future_directions} how one might utilize a `DAOE-enhanced' Lanczos algorithm to improve the rate of convergence with $n$ for estimates of hydrodynamic transport coefficients.

\section{Synopsis}
\label{sec:synopsis}

\subsection{Emergent universality in the $n\to \infty$ limit}
\label{sec:synopsis_emergent_universality}
When the memory function formalism was first developed it was hoped that, if $n$ is sufficiently large, then all non-universal features of the original dynamical problem would be filtered out into the coefficients $\{b_{n}\}$, and then the level-$n$ Green's function $G_{n}(z)$ could be approximated by assuming it is ultra short-range correlated in time, which amounts to setting $G_{n}(z)$ to be constant~\cite{viswanathRecursionMethodApplication2013,forsterHydrodynamicFluctuationsBroken2019}. In certain cases this turns out to be well-founded. As discussed in \cref{sec:gf}, one of our main results is that, as $n \to \infty$, for $z$ in the `bulk' of the spectrum, $G_{n}(z)$ approaches the Wigner semicircle law:
\begin{equation}
    \boxed{
    G_{n}(z) \approx \dfrac{2}{\beta_{n}^{2}} \left(z - \sqrt{z + \beta_{n}} \sqrt{z - \beta_{n}}\right).
    }
\end{equation}
Here $\beta_{n}$ is an effective $n$-dependent frequency bandwidth approximated by $\beta_{n} \approx 2 b_{n}$ as $n \to \infty$, and `bulk' means $|z| > \delta_{0}$ and $|z/\beta_{n} \pm 1| > \delta_{1}$ for some small $\mathcal{O}(1)$ constants $\delta_{0}, \delta_{1}$ in units of the microscopic couplings (see \cref{fig:plancherel_schematic}(b)). That this should be called a semicircle is perhaps clearer from considering the corresponding bulk spectral function,
\begin{equation}
    \Phi_{n}(\omega) = 2 \Im\left[ G_{n}(\omega - i 0^{+})\right] \approx \dfrac{4}{\beta_{n}} \sqrt{1 - (\omega/\beta_{n})^{2}}.
\end{equation}
Note that this semicircle form of $G_{n}(z)$ is identical to the average global resolvent $\frac{1}{n} \overline{\tr{1/(z-M)}}$---which encodes dynamics in random matrix theory---for random $n\times n$ matrices $M$ drawn from the Gaussian Unitary Ensemble~\cite{taoTopicsRandomMatrix2012}, rescaled so that the bulk eigenvalue spectrum corresponds to the interval $[-\beta_{n}, \beta_{n}]$. Thus we see the emergence of random matrix-like universality---\textit{even though there is no explicit randomness present}---in the bulk frequency profile of the `fast space' operator dynamics, i.e., restricted to $\{O_{m}\}_{m\geq n}$. Since the bandwidth scales like $\beta_{n} \approx 2 b_{n}$ to leading order in $n$, if the Lanczos coefficients grow indefinitely, $b_{n} \to \infty$, we conclude that $G_{n}(z) \approx \pm 2 i / \beta_{n} + \mathcal{O}(z/\beta_{n}^{2})$ is approximately constant to leading order in $z/\beta_{n}$, thereby justifying the assumption of short-range time correlations as $n \to \infty$. This also provides a precise sense in which the large-$n$ dynamics are `fast', since $G_{n}(z) = \mathcal{O}(1/\beta_{n})$ tends to zero as $n\to\infty$.

For spectral functions $\Phi(\omega)$ which are complex analytic at $\omega = 0$, corresponding to autocorrelation functions $C(t)$ which decay exponentially in time, we further show that this semicircle behavior persists all the way down to zero frequency. We conjecture that this remains true even if $C(t)$ decays algebraically, provided it is faster than $1/t$---but that there is a slower decay of the finite-$n$ correction to the semicircle law in this case. This conjecture is based on similar asymptotic analysis in prior work~\cite{mclaughlinSteepestDescentMethod2006}, which demonstrated how the assumed degree of differentiability affects the scaling of the error term, not the leading term.

We note that the appearance of the Wigner semicircle law in the level-$n$ Green's function $G_{n}(z)$ was recognized by Ref.~\cite{karRandomMatrixTheory2022a}, but only in the context of a \textit{finite}-dimensional system where $n$ is so large that the Lanczos coefficients have plateaued (typically this happens once $n$ is of order the system size), so they can utilize rigorous results from Ref.~\cite{vanasscheOrthogonalPolynomialsAssociated1991}. Our analysis is instead focused on interacting systems in the thermodynamic limit, where the Lanczos coefficients will generically grow indefinitely. We show that it is the $n\to\infty$ limit (i.e., not the plateau itself), which results in emergent random matrix universality. Besides the interpretational difference, this is relevant for numerical applications to interacting systems, where one will generically not see a plateau without incurring finite-size effects.

When the autocorrelation function $C(t)$ decays as a power-law slower than $1/t$, this results in power-law behavior of $G_{n}(z)$ near $z = 0$, so that the semicircle form is no longer appropriate. Such power-law decay in $C(t)$ is typical if the initial operator overlaps with a conserved quantity with sufficiently slow transport, say diffusive in 1D~\cite{forsterHydrodynamicFluctuationsBroken2019}. With hydrodynamics in mind, one of our principal goals will be to characterize how signatures of the long-time behavior of the autocorrelation function $C(t)$ imprint themselves on the Lanczos basis, and in turn how this affects the operator backflow encoded in the Green's functions $G_{n}(z)$. When the spectral function behaves like a power-law at low frequencies, we show that this modifies the behavior of $G_{n}(z)$ for $z \to 0$, such that the semicircle form breaks down, and $G_{n}(z)$ instead has an explicit expression in terms of Bessel functions (see \cref{sec:gf_bessel}). This is a reflection of `Bessel universality' that governs the low frequency behavior (see \cref{fig:plancherel_schematic}). With these fingerprints of low-frequency behavior in hand, we show how to approximate low-frequency data like hydrodynamic transport coefficients, using only a finite number of Lanczos coefficients as input (see \cref{sec:hydro}). We benchmark this approach on a range of physical models, including the mixed field Ising model and the XXZ spin chain, finding results that are competitive with tensor network methods.

Finally, for frequencies near the edge of the spectrum $z = \pm \beta_{n}$, there is a third `Airy' universality class (see \cref{fig:plancherel_schematic}b), so-called because $G_{n}(z)$ approaches a universal scaling form in terms of Airy functions (see \cref{sec:airy_gn}). This is analogous to the modified behavior of eigenvalue correlations near the edge of the spectrum of a random matrix~\cite{kuijlaarsUniversality2011}. In \cref{sec:airy_bootstrap} we also show how to utilize this Airy universality to estimate spectral functions at very high frequencies.

\subsection{Orthogonal polynomials and large-$n$ expansions}
\label{sec:synopsis_orth_pols}
Our proof makes use of some powerful machinery from complex analysis, developed in the study of orthogonal polynomials~\cite{deiftSteepestDescentMethod1993,deiftStrongAsymptoticsOrthogonal1999,deiftRiemannHilbertApproach2001,deiftUniformAsymptoticsPolynomials1999,deiftOrthogonalPolynomialsRandom2000,kuijlaarsRiemannHilbertAnalysisOrthogonal2003,kuijlaarsUniversalityEigenvalueCorrelations2003}. This is relevant because the Lanczos algorithm can be naturally phrased in the language of orthogonal polynomials (see \cref{sec:background}). Indeed, the $n$th Lanczos basis operator $|O_{n})$ is given by
\begin{equation}
    |O_{n}) = p_{n}(\mathcal{L}) |O_{0}),
    \label{eq:lanczos_poly_formulation}
\end{equation}
where the polynomials $p_{n}(\omega)$ are orthonormal with respect to the spectral function, i.e.,
\begin{equation}
    \int_{\mathbb{R}} p_{m}(\omega) p_{n}(\omega) \dfrac{\Phi(\omega)}{2\pi} \diff \omega = \delta_{m n}.
    \label{eq:poly_orthogonality}
\end{equation}
One can similarly write an expression for the level-$n$ Green's function $G_{n}(z)$ in terms of these orthogonal polynomials (see \cref{eq:gf_cauchy_ratio}). The complex analytic machinery works by devising a Riemann-Hilbert problem whose solution encodes essential data about these orthogonal polynomials with respect to $\Phi(\omega)$, which in turn gives us information about the Lanczos basis (see \cref{sec:rhp} for details). The advantage of this formulation is that it permits approximate solutions, controlled in the limit $n \to \infty$, even without knowing the explicit form of the spectral function $\Phi(\omega)$. Instead one only needs to know some `high level' features of $\Phi(\omega)$, like its rate of decay at high frequencies, and whether it has a power-law at low frequencies. Then one can use a technique similar to steepest descent to obtain a $1/n$ expansion of the polynomials $p_{n}(\omega)$ (\cref{sec:hydro,sec:spectral_bootstrap}), the recurrence coefficients $b_{n}$ (\cref{thm:recurrence_theorem}), the level-$n$ Green's function $G_{n}(z)$ (\cref{sec:gf}), and other related quantities. We remark that the $n \to \infty$ limit is quite natural, in the sense that we expect the operator Krylov space $\mathcal{K}=\mathrm{span}\{\mathcal{L}^{n}O_{0}\}_{n\geq 0}$ to be infinite dimensional in the thermodynamic limit, for generic local operators $O_{0}$ and generic Hamiltonians. This provides one with the means to perform a large-$n$ expansion of the operator dynamics, even when there is no explicit large parameter in the Hamiltonian.

Since this steepest descent approach involves deformation of contours into the complex plane, our proof technique requires us to make some assumptions about the complex analytic structure of the spectral function $\Phi(\omega)$ near the real frequency axis; this is the main technical caveat to our work, and we discuss it in more detail in \cref{sec:potential_definitions,sec:analyticity}. However, numerical evidence suggests that some of our results may remain valid, at least to leading order in $n$, assuming only weaker conditions than analyticity, such as differentiability of $\Phi(\omega)$. Indeed, a recent breakthrough result in random matrix theory proved universality assuming only a very mild local continuity condition~\cite{eichingerNecessarySufficientConditions2024}, and it would be interesting to see if these techniques could be adapted for our purposes. Finally, we emphasize that the Riemann-Hilbert formulation is used only for the proof, and readers only interested in practical applications do not need to make use of it.

\subsection{Lanczos dynamics and random matrix universality}
\label{sec:synopsis_rmt_universality}
As previously discussed, one of our main results is that, as $n \to \infty$, the level-$n$ Green's function $G_{n}(z)$ approaches universal scaling forms, the most prominent example being the Wigner semicircle law for $z$ in the bulk of the spectrum. In what follows, we will discuss how the universality of $G_{n}(z)$ and related quantities is analogous, in a precise sense, to the universality of eigenvalue correlations of random matrices. 
Mathematically, this shared universality can be understood from the fact that both Lanczos operator dynamics and random matrix theory can be formulated in the language of orthogonal polynomials. These polynomials themselves exhibit universal scaling forms in the $n\to\infty$ limit, which then implies universality in both Lanczos operator dynamics and random matrix theory. While these formulations in terms of orthogonal polynomials are well-known to practioners of both fields, it is worth spelling out how they work in some detail.

In the Lanczos context, we already saw the connection with orthogonal polynomials in \cref{eq:lanczos_poly_formulation,eq:poly_orthogonality}: the $n$th Lanczos operator $O_{n}$ can be written in terms of the $n$th orthogonal polynomial $p_{n}(\omega)$ with respect to the spectral function $\Phi(\omega) = \int_{\mathbb{R}} e^{-i \omega t} C(t) \diff t$. To get from there to the level-$n$ Green's function, one can employ the useful identity~\cite{vanasscheOrthogonalPolynomialsAssociated1991}
\begin{align}
    G_{n}(z) &= \dfrac{1}{b_{n}} \dfrac{C_{n}(z)}{C_{n-1}(z)}, \label{eq:gf_cauchy_ratio} \\[1em]
    \text{where } C_{n}(z) &\equiv \int_{\mathbb{R}} \dfrac{p_{n}(\omega)}{z-\omega} \dfrac{\Phi(\omega)}{2\pi} \diff \omega. \nonumber
\end{align}
Unfortunately, $G_{n}(z)$ does not, in general, have an interpretation as a simple low-order correlation function of a random matrix ensemble. In that sense, the result for the bulk of the spectrum is a special case---albeit one that applies almost everywhere in the complex $z$ plane---where $G_{n}(z)$ approaches the Wigner semicircle law, and can then indeed be identified with a simple RMT correlation function. But nonetheless, even away from the bulk, $G_{n}(z)$ does still approach universal scaling forms, which can be understood in terms of universality of the orthogonal polynomials $p_{n}(\omega)$ in the $n\to\infty$ limit.

In the random matrix theory context, orthogonal polynomials provide a useful tool for computing eigenvalue correlation functions~\cite{mehtaRandomMatrices2004}. Consider an ensemble of $n \times n$ random Hermitian matrices $M$ with probability density $P(M) \propto \exp(-\mathrm{tr}[Q(M)])$, for some function $Q$ called the \textit{potential}. The Gaussian unitary ensemble (GUE) corresponds to $Q(x) = x^{2}$, but it can be useful to consider more general functions $Q$ to introduce correlations between matrix elements. One can instead consider the distribution on the eigensystem of $M$ rather than its matrix elements; for general $Q$, integrating out the eigenvectors produces a distribution on the eigenvalues given by
\begin{equation}
    P_{Q}(\lambda_{1},\dots,\lambda_{n}) \propto \left(\prod_{i<j} |\lambda_{i}-\lambda_{j}|^{2}\right) \prod_{i} e^{-Q(\lambda_{i})},
    \label{eq:rmt_eigenvalue_distribution}
\end{equation}
resulting in the familiar Vandermonde determinant $\prod_{i<j} |\lambda_{i}-\lambda_{j}|^{2}$ responsible for eigenvalue repulsion. It is then natural to consider the $k$-point eigenvalue correlation function $R_{k}(\lambda_{1},\dots,\lambda_{k})$ defined by
\begin{align}
    &R_{k}(\lambda_{1}, \dots, \lambda_{k}) =\\
    &\dfrac{n!}{(n-k)!} \int_{\mathbb{R}} \cdots \int_{\mathbb{R}} P_{Q}(\lambda_{1}, \dots, \lambda_{k}, \lambda_{k+1}, \dots, \lambda_{n}) \diff \lambda_{k+1} \cdots \diff \lambda_{n}. \nonumber
\end{align}
The key fact linking RMT and orthogonal polynomials is that $R_{k}$ can be expressed as a $k\times k$ matrix determinant,
\begin{equation}
    R_{k}(\lambda_{1}, \dots, \lambda_{k}) = \mathrm{det}\left[\left(\hat{K}_{n}(\lambda_{i}, \lambda_{j})\right)_{1 \leq i,j \leq k}\right],
    \label{eq:rmt_determinant}
\end{equation}
where each matrix element is given by evaluating a 2-point correlation kernel $\hat{K}_{n}(\lambda_{i},\lambda_{j})$. The kernel is defined by
\begin{equation}
    \hat{K}_{n}(\lambda_{i}, \lambda_{j}) = \sqrt{e^{-Q(\lambda_{i})} e^{-Q(\lambda_{j})}} \sum_{m=0}^{n-1} p_{m}(\lambda_{i}) p_{m}(\lambda_{j}),
    \label{eq:rmt_kernel_polynomials}
\end{equation}
where the $p_{m}$ are orthogonal polynomials satisfying
\begin{equation}
    \int_{\mathbb{R}} p_{k}(\lambda) p_{l}(\lambda) e^{-Q(\lambda)} \diff \lambda = \delta_{kl}.
    \label{eq:rmt_poly_orthogonality}
\end{equation}
Thus we see that eigenvalue statistics of random matrices can be described in terms of orthogonal polynomials with respect to $e^{-Q(\lambda)}$, the potential defining the random matrix ensemble. One of the deepest facts about random matrices is that their eigenvalue statistics can be \textit{universal}, when probed on a local scale~\cite{kuijlaarsUniversality2011}. Here `local scale' means that we consider correlations over $\mathcal{O}(1)$ separations in units of the inverse local eigenvalue density, and `universal' means they are independent of the exact eigenvalue probability distribution. Given the formulation in terms of orthogonal polynomials in \cref{eq:rmt_poly_orthogonality}, one can derive RMT universality from universality of the orthogonal polynomials themselves.

To relate Lanczos operator dynamics to random matrix theory, let us summarize the discussion so far: the Lanczos basis operators $O_{n}$ are related to orthogonal polynomials with respect to $\Phi(\omega) / 2\pi$ (see \cref{eq:poly_orthogonality}), while RMT eigenvalue statistics are related to orthogonal polynomials with respect to $e^{-Q(\lambda)}$ (see \cref{eq:rmt_poly_orthogonality}). Therefore, under the identification of weight functions:
\begin{equation}
    \boxed{
    \dfrac{\Phi(\omega)}{2\pi} \equiv e^{-Q(\omega)},
    }
    \label{eq:Phi_Q_identification}
\end{equation}
we see that there is a connection between:
\begin{equation*}
    \boxed{
    n\textit{th Lanczos operator } O_{n} \leftrightsquigarrow n\times n \textit{ random matrices}
    }
\end{equation*}
Much as one can derive RMT universality from universality of the orthogonal polynomials defining the matrix ensemble, one can derive universality of the level-$n$ Green's function $G_{n}(z)$ and related quantities from universality of the orthogonal polynomials with respect to the spectral function $\Phi(\omega)$. 
Universal asymptotics of orthogonal polynomials is a much older phenomenon than RMT universality, going back to work in the 1920s of Plancherel and Rotach on $n\to\infty$ asymptotics for the Hermite polynomials ($Q(x) = x^{2}$)~\cite{plancherelValeursAsymptotiquesPolynomes1929}. One of our technical contributions is proving so-called `Plancherel-Rotach asymptotics' of the orthogonal polynomials for a large class of spectral functions obeying physically-motivated conditions (described in \cref{sec:potential_definitions}). From there, one is led to universality for the level-$n$ Green's function.

\subsection{Coulomb gas confinement transition, quantum chaos, and the operator growth hypothesis}
\label{sec:synopsis_coulomb}
Stemming from the foundational work of Dyson~\cite{dysonStatisticalTheoryEnergy1962}, a key concept in random matrix theory (RMT) is the Coulomb gas~\cite{forresterLogGasesRandomMatrices2010}. To give context for our result about a Coulomb gas confinement transition, let us briefly illustrate where this comes from, both in random matrix theory and in Lanczos operator dynamics.

To compute eigenvalue correlation functions in RMT, one needs to consider integrals with respect to the eigenvalue probability distribution $P_{Q}(\lambda_{1},\dots,\lambda_{n})$. For reasons that will soon become apparent, it turns out to be important to rescale by an $n$-dependent factor, $\lambda_{i} \to \beta_{n} x_{i}$, where the scale factor $\beta_{n}$ is determined by $Q$ and the $x_{i}$ are $\mathcal{O}(1)$; for $Q(\lambda) = \lambda^{2}$ we have $\beta_{n} = \sqrt{2n}$. Then, defining the rescaled potential $V_{n}(x) \equiv Q(\beta_{n}x)/n$, one can rewrite \cref{eq:rmt_eigenvalue_distribution} as
\begin{equation}
    P_{Q}(\lambda_{1},\dots,\lambda_{n}) \propto e^{-\left[\sum_{i\neq j} \log{|x_{i}-x_{j}|^{-1}} + n \sum_{i} V_{n}(x_{i})\right]},
\end{equation} 
and so one expects the leading contributions to come from tuples $\bm{x} = (x_{1},\dots,x_{n})$ for which the exponent is minimal. For any such tuple $\bm{x}$, if we introduce the normalized counting measure $\diff \psi_{\bm{x}}(y) = \frac{1}{n} \sum_{i=1}^{n} \delta(y-x_{i}) \diff y$, then we are led to consider the Coulomb gas energy functional
\begin{equation}
    E_{V_{n}}[\psi] = \int_{\mathbb{R}} \int_{\mathbb{R}} \log{|x-y|^{-1}} \diff \psi(x) \diff \psi(y) + \int_{\mathbb{R}} V_{n}(x) \diff \psi(x).
\end{equation}
In the $n\to\infty$ limit, eigenvalue correlation functions will be dominated by configurations $\psi$ which minimize this Coulomb gas energy. The rescaling $\lambda_{i} \to \beta_{n}x_{i}$ served the purpose of ensuring that the two energy terms are of comparable magnitude as $n\to\infty$. In general, this minimal configuration gives the \textit{mean eigenvalue density profile} for the random matrix ensemble defined by the potential $Q$. For $Q(\lambda) = \lambda^{2}$, the minimizing configuration is the Wigner semicircle law.

Remarkably, the same Coulomb gas ensemble also plays a central role in the $n\to\infty$ asymptotics of orthogonal polynomials~\cite{deiftOrthogonalPolynomialsRandom2000}. The easiest way to see this is through the following explicit formula for the orthogonal polynomials $p_{n}(\omega)$ with respect to the weight function $e^{-Q(\omega)}$:
\begin{equation}
    p_{n}(\omega) \propto \int \cdots \int \left( \prod_{i=1}^{n}(\omega-\omega_{i})\right) P_{Q}(\omega_{1},\dots,\omega_{n}) \diff \omega_{1} \cdots \diff \omega_{n},
\end{equation}
where the probability density $P_{Q}(\omega_{1},\dots,\omega_{n})$ was defined in \cref{eq:rmt_poly_orthogonality}. In other words, $p_{n}(\omega)$ can be expressed as an average over all degree-$n$ polynomials, weighted by the \textit{same} probability measure as that which governed the eigenvalue distribution of the random matrix ensemble defined by the potential $Q$.
In this context, the minimal energy configuration of the charge-$n$ Coulomb gas---which we call $\sigma_{n}(\omega)$---gives the \textit{limiting density distribution of zeros} of the orthogonal polynomials with respect to $e^{-Q(\omega)}$. 
In our Riemann-Hilbert steepest descent analysis (\cref{sec:rhp}), the support of this minimal energy configuration is analogous to the `oscillatory region' of a WKB approximation, and yields the dominant contribution to the steepest descent problem as $n\to\infty$.

With the relevance of the Coulomb gas established, we can now discuss its confinement transition and the connection to quantum chaos.
For systems with local interactions, one can show that the spectral function must decay at least exponentially as $|\omega|\to\infty$~\cite{abaninExponentiallySlowHeating2015,aradConnectingGlobalLocal2016,bruLiebRobinsonBoundsMultiCommutators2017}, which translates, via \cref{eq:Phi_Q_identification}, into the requirement that the potential $Q(\omega)$ grow at least linearly at large $\omega$. This is enough to conclude that the Coulomb gas is always confined, in the sense that, for any finite $n$, the equilibrium Coulomb gas density $\sigma_{n}(\omega)$ has finite support. However, within the confined phase, there can be a further phase transition between `weak confinement' and `strong confinement' (see \cref{fig:confinement_schematic})~\cite{canaliNonuniversalityRandommatrixEnsembles1995,freilikherUnitaryRandommatrixEnsemble1996,freilikherTheoryRandomMatrices1996,claeysWeakStrongConfinement2023}. These phases are distinguished by whether the Coulomb gas density is approximately uniform or not within the bulk of its support.

Importantly, the order parameter for this transition is the equilibrium density at \textit{low frequencies}, $\sigma_{n}(\omega \approx 0)$. In the weakly confined phase, $\sigma_{n}(\omega)$ has an algebraic divergence as $\omega \to 0$, while in the strongly confined phase it is approximately constant there; at the critical point, there is instead a logarithmic divergence as $\omega \to 0$ (see \cref{fig:confinement_example} for an example). What is interesting is that, despite being diagnosed by the density near $\omega =0$, the transition itself is primarily driven by the rate of \textit{high frequency} decay of the spectral function. In particular, the weakly confined phase, the critical point, and the strongly confined phase, correspond to $\Phi(\omega\to \infty)$ decaying subexponentially, exponentially, and superexponentially respectively (\cref{fig:confinement_schematic}). This critical exponential decay---the slowest possible decay consistent with locality~\cite{abaninExponentiallySlowHeating2015,aradConnectingGlobalLocal2016,bruLiebRobinsonBoundsMultiCommutators2017}---coincides with the behavior that is expected to occur \textit{generically} for chaotic systems with local interactions; this is the content of a recent conjecture in quantum chaos, the `Operator Growth Hypothesis' (OGH)~\cite{parkerUniversalOperatorGrowth2019}. Thus, assuming the OGH, there is a sense in which chaotic systems are generically marginal. Since this confinement transition is related to the density of zeros of the orthogonal polynomials $p_{n}(\omega)$, which in turn govern the Lanczos operators $O_{n}$ via \cref{eq:lanczos_poly_formulation}, this marginality has implications for the computational efficiency of numerical applications, as we will discuss. One of our contributions (\cref{lem:hn0_scaling}) is showing that this confinement phase structure is robust, in the sense that we prove it occurs for a wide class of spectral functions assuming only a fairly weak specification of how they decay at high frequencies (see \cref{sec:potential_definitions}).

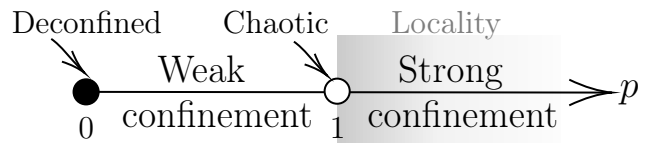
\begin{figure}[t]
    \tikzset {_vz9ade1hz/.code = {\pgfsetadditionalshadetransform{ \pgftransformshift{\pgfpoint{0 bp } { 0 bp }  }  \pgftransformrotate{0 }  \pgftransformscale{2 }  }}}
\pgfdeclarehorizontalshading{_wexlnkruz}{150bp}{rgb(0bp)=(0.8,0.8,0.8);
rgb(37.5bp)=(0.8,0.8,0.8);
rgb(48.5bp)=(0.95,0.95,0.95);
rgb(62.5bp)=(0.99,0.99,0.99);
rgb(100bp)=(0.99,0.99,0.99)}
\tikzset{every picture/.style={line width=0.75pt}} %

\begin{tikzpicture}[x=0.75pt,y=0.75pt,yscale=-1,xscale=1,scale=1.4]

\draw  [draw opacity=0][shading=_wexlnkruz,_vz9ade1hz] (170.25,229.75) -- (250.25,229.75) -- (250.25,269.75) -- (170.25,269.75) -- cycle ;
\draw    (80.25,250) -- (266.75,250.25) ;
\draw [shift={(268.75,250.25)}, rotate = 180.08] [color={rgb, 255:red, 0; green, 0; blue, 0 }  ][line width=0.75]    (15.3,-4.61) .. controls (9.73,-1.96) and (4.63,-0.42) .. (0,0) .. controls (4.63,0.42) and (9.73,1.96) .. (15.3,4.61)   ;
\draw  [fill={rgb, 255:red, 0; green, 0; blue, 0 }  ,fill opacity=1 ] (75.75,250) .. controls (75.75,247.51) and (77.76,245.5) .. (80.25,245.5) .. controls (82.74,245.5) and (84.75,247.51) .. (84.75,250) .. controls (84.75,252.49) and (82.74,254.5) .. (80.25,254.5) .. controls (77.76,254.5) and (75.75,252.49) .. (75.75,250) -- cycle ;
\draw  [fill={rgb, 255:red, 255; green, 255; blue, 255 }  ,fill opacity=1 ] (165.75,249.75) .. controls (165.75,247.26) and (167.76,245.25) .. (170.25,245.25) .. controls (172.74,245.25) and (174.75,247.26) .. (174.75,249.75) .. controls (174.75,252.24) and (172.74,254.25) .. (170.25,254.25) .. controls (167.76,254.25) and (165.75,252.24) .. (165.75,249.75) -- cycle ;
\draw    (67.75,231) .. controls (71.33,233.46) and (74.51,237.93) .. (78.36,241.88) ;
\draw [shift={(79.75,243.25)}, rotate = 223.36] [color={rgb, 255:red, 0; green, 0; blue, 0 }  ][line width=0.75]    (8.74,-2.63) .. controls (5.56,-1.12) and (2.65,-0.24) .. (0,0) .. controls (2.65,0.24) and (5.56,1.12) .. (8.74,2.63)   ;
\draw    (154.75,232) .. controls (158.33,234.46) and (161.51,238.93) .. (165.36,242.88) ;
\draw [shift={(166.75,244.25)}, rotate = 223.36] [color={rgb, 255:red, 0; green, 0; blue, 0 }  ][line width=0.75]    (8.74,-2.63) .. controls (5.56,-1.12) and (2.65,-0.24) .. (0,0) .. controls (2.65,0.24) and (5.56,1.12) .. (8.74,2.63)   ;

\draw (270.5,245) node [anchor=north west][inner sep=0.75pt]  [font=\Large]  {$p$};
\draw (76.5,258) node [anchor=north west][inner sep=0.75pt]  [font=\large]  {$0$};
\draw (166.5,258) node [anchor=north west][inner sep=0.75pt]  [font=\large]  {$1$};
\draw (52.5,220) node [anchor=north west][inner sep=0.75pt]   [align=center,font=\large] {{{Deconfined}}};
\draw (188.5,220) node [anchor=north west][inner sep=0.75pt]   [align=center,font=\large] {{{\textcolor[rgb]{0.49,0.49,0.49}{Locality}}}};
\draw (128.25,220) node [anchor=north west][inner sep=0.75pt]   [align=center,font=\large] {{{Chaotic}}};
\draw (104.75,236) node [anchor=north west][inner sep=0.75pt]   [align=left,font=\Large] {{{Weak}}};
\draw (191,236) node [anchor=north west][inner sep=0.75pt]   [align=center,font=\Large] {{{Strong}}};
\draw (92,252) node [anchor=north west][inner sep=0.75pt]   [align=center,font=\Large] {{{confinement}}};
\draw (180,252) node [anchor=north west][inner sep=0.75pt]   [align=center,font=\Large] {{{confinement}}};

\end{tikzpicture}
    \caption{Phase diagram of Coulomb gas confinement with a single-particle potential $Q(\omega)$, controlled by the large frequency decay of the spectral function $\Phi(\omega) \sim \exp[-Q(\omega)]$. Locality of the Hamiltonian forces the potential to grow polynomially at large $\omega$, $Q(|\omega|\to\infty) \sim |\omega|^{p}$, with $p \geq 1$. The operator growth hypothesis \cite{parkerUniversalOperatorGrowth2019} posits that chaotic systems generically have $p=1$. A Coulomb gas with finite charge (i.e.~finite $n$) is confined for all $p > 0$ since the two-particle repulsion is only logarithmic, while the potential is polynomial. However, there is a transition between `weak' and `strong' confinement at $p=1$, precisely at the boundary imposed by locality. This has implications for numerical applications.}
    \label{fig:confinement_schematic}
\end{figure}

\subsection{Comparison with eigenstate thermalization}
\label{sec:synopsis_eth}
The main technical caveat to our work is that we also impose some regularity and analyticity conditions on $\Phi(\omega)$ in order to prove rigorous statements; we discuss this in detail in \cref{sec:potential_definitions}. But with these conditions in mind, let us discuss the relationship between the appearance of random matrix theory (RMT) in our work, and more familiar appearances of RMT in quantum many-body physics, namely through the eigenstate thermalization hypothesis (ETH)~\cite{srednickiChaosQuantumThermalization1994,dalessioQuantumChaosEigenstate2016}. The ETH postulates that the matrix elements of a local operator $O$ in the energy eigenbasis of a chaotic local Hamiltonian take the form
\begin{equation}
    \langle E_{m} | O | E_{n} \rangle = O(\overline{E}) \delta_{mn} + e^{-S(\overline{E})/2} f_{O}(\overline{E}, \omega) R_{mn},
\end{equation}
where $\overline{E} \equiv (E_{m} + E_{n})/2$, $\omega \equiv E_{n} - E_{m}$, and $S(\overline{E})$ is the thermodynamic entropy at energy $\overline{E}$. RMT appears through the entropic factor $e^{-S(\overline{E})/2}$ and the random matrix $R$, which is postulated to be distributed according to a Gaussian random matrix ensemble~\cite{dalessioQuantumChaosEigenstate2016}.

An important part of the ETH is that the functions $O(\overline{E})$ and $f_{O}(\overline{E}, \omega)$ are postulated to be \textit{smooth} functions of their arguments. For our purposes, working in the thermodynamic limit, we do assume some smoothness properties of the spectral function $\Phi(\omega)$, which coincides with the function $|f_{O}(E_{\mathrm{av}}, \omega)|^{2}$ at the average energy $E_{\mathrm{av}} = \tr{H} / \tr{\mathds{1}}$ of the infinite temperature state~\cite{dalessioQuantumChaosEigenstate2016}. Therefore, arguably the most natural class of physical systems that one might expect to obey our assumptions about the spectral function are those that obey the ETH.
However, the RMT universality we discuss in our work appears to be distinct in nature from the appearance of random matrix theory in the ETH, since it really owes itself to the \textit{regularity} of spectral functions (and the consequence of that on the Lanczos orthogonal polynomials), which can be achieved even in non-ergodic single-particle and integrable systems.
Indeed, we also find some evidence of this RMT universality in the integrable XXZ chain, where the presence of integrability should violate the ETH (although a weak version may persist~\cite{albaEigenstateThermalizationHypothesis2015}).
Validity of the ETH may therefore be a sufficient but not necessary condition for the RMT universality we discuss here, but a fuller investigation of these connections is an important topic for future work.

We also give some evidence in \cref{sec:smoothness} that smoothness of spectral functions is indeed a necessary condition for this fast space universality. We focus on the 1D transverse field Ising model with quenched disorder, which, owing to the disorder, is provably Anderson localized and has point like spectrum even in the thermodynamic limit~\cite{aizenmanRandomOperators2015}. We perform a self-consistency check for the presence of universality in this model---which was passed by the clean version---and find that this check fails in the presence of disorder, suggesting the breakdown of fast space universality.

\subsection{Numerical application: the spectral bootstrap}
\label{sec:synopsis_spectral_bootstrap}
As described in the Introduction, the recursion method is a numerical technique based around the continued fraction representation for the full Green's function $G(z)$; traditionally it requires one to make an educated guess for the functional form of the level-$n$ Green's function $G_{n}(z)$~\cite{viswanathRecursionMethodApplication2013}. Our approach allows us to rigorously derive the functional behavior of $G_{n}(z)$, showing how this changes for $z$ in different regions of the complex plane (see \cref{fig:plancherel_schematic}). This provides a principled means of terminating the continued fraction and thereby approximating the full Green's function $G(z)$.

We show how to use this control afforded by the large-$n$ limit to produce an approximation of the spectral function $\Phi(\omega)$ using a finite number of Lanczos coefficients. Our focus is particularly on the low frequency behavior, and we show how to extract diffusion constants, as well as transport coefficients in models with nondiffusive transport, such as the isotropic Heisenberg model which has superdiffusive spin transport~\cite{gopalakrishnanSuperdiffusionNonabelianSymmetries2024}. We also show how to generalize these ideas to finite frequencies. The simplest case is when the semicircle form for the level-$n$ Green's function $G_{n}(z)$ is valid at the target frequency, because then only a single parameter, the frequency bandwidth $\beta_{n}$, needs to be specified in order to fully fix the semicircle Green's function---and $\beta_{n}$ can be simply approximated from the Lanczos coefficients via $\beta_{n} = 2 b_{n}[1 + \mathcal{O}(1/n)]$ (see \cref{thm:recurrence_theorem}). However, there are also more complicated cases, such as when $\Phi(\omega)$ has a low-frequency power-law, where an additional phase, arising from an integral over the Coulomb gas density, needs to be specified in order to approximate $G_{n}(z)$. This is the content of our `spectral bootstrap' algorithm. It works by using the appropriate $n\to\infty$ asymptotics of the orthogonal polynomials $p_{n}(\omega)$ to formulate a first-order ordinary differential equation in frequency space involving the spectral function and the desired phase factor. Iteratively solving this ODE gives a finite-$n$ approximation to the spectral function, using only the first $n$ Lanczos coefficients as input. %

One interesting aspect of this is that the convergence rate of this approximation at $\omega=0$ is determined by a combination of the high- and low-frequency behavior of the spectral function $\Phi(\omega)$, owing to the manifestation of the Coulomb gas confinement transition in the Coulomb gas density at low frequencies. We show that the rate of convergence is $1/\mathrm{poly}(n)$ when the spectral function $\Phi(\omega\to\infty)$ decays superexponentially, regardless of the low-frequency scaling of $\Phi(\omega)$. However, when $\Phi(\omega\to\infty)$ decays only (quasi-)exponentially---which is the case conjectured to be generic for chaotic quantum systems by the Operator Growth Hypothesis~\cite{parkerUniversalOperatorGrowth2019}---then the convergence rate is only $1/\mathrm{poly}(\log{n})$ if $\Phi(\omega\to 0) \sim |\omega|^{\rho}$ with $\rho \neq 0$. We recover $1/\mathrm{poly}(n)$ scaling when $\rho = 0$, but this is likely a consequence of the fact that, for $\rho=0$, our assumptions on $\Phi(\omega)$ amount to assuming analyticity at $\omega = 0$, i.e.~that the correlation function $C(t)$ decays exponentially in time. We conjecture that when $0 < \Phi(0) < \infty$ but $\Phi(\omega)$ is not analytic at $\omega = 0$ due to hydrodynamic tails, then the convergence rate again reduces to $1/\mathrm{poly}(\log{n})$. However, despite this worst-case theoretical convergence, in practice we find that the spectral bootstrap can still often give surprisingly accurate estimates using only a modest number of Lanczos coefficients (see e.g.~\cref{fig:spectral_bootstrap_rho=0}). Alternatively, since the convergence rate of the spectral bootstrap is better for non-chaotic systems than chaotic ones, one can view it as being complementary to other methods like dissipation-assisted operator evolution (DAOE)~\cite{rakovszkyDissipationassistedOperatorEvolution2022,vonkeyserlingkOperatorBackflowClassical2022,srivatsaOperatorGrowthHypothesis2024,lloydBallisticDiffusiveCrossover2024}, whose accuracy is best justified for chaotic systems~\cite{vonkeyserlingkOperatorBackflowClassical2022}.

\subsection{Marginality and quantum chaos}
\label{sec:synopsis_marginality}

One curious aspect of this work concerns the connection between `marginality' and chaotic operator dynamics. In our work the spectral function $\Phi(\omega)$ takes center stage, since it completely characterizes the Lanczos basis. In that role it functions mainly as the density of a measure $\diff \mu(\omega) = (\Phi(\omega)/2\pi) \diff \omega$ on $\mathbb{R}$, where it defines the particular set of orthogonal polynomials constructed by the Lanczos algorithm. Its origin in quantum operator dynamics is left somewhat implicit, entering only through constraints at high- and low-frequencies coming from considerations of locality and hydrodynamics.

Our numerical algorithm, the spectral bootstrap, gives one solution to the `spectral inverse problem' of recovering the measure $\Phi(\omega)/2\pi$ from its moments (which is equivalent to recovery from the Lanczos coefficients \cite{gautschiOrthogonalPolynomialsComputation2004}). Relatedly, the `Hamburger moment problem' asks whether there exists a measure $\mu$ generating a given set of moments, and if so whether this measure is unique \cite{akhiezerClassicalMomentProblem1965}. If it is unique then we say the moment problem for $\mu$ is \textit{determinate}, and if not then it is indeterminate. 
By definition it is then impossible to uniquely recover a measure from its moments if the corresponding moment problem is indeterminate.
Crucially, it is well known that the moment problem is determinate for all measures which decay at least \textit{exponentially} at infinity, meaning $\int_{\mathbb{R}} e^{c |x|} \diff \mu(x) < \infty$ for some $c>0$ \cite{levinOrthogonalPolynomialsWeights2007}. For example, the `Freud weights' $\diff \mu(x) = \exp(-|x|^{p}) \diff x$ are determinate for $p \geq 1$ and indeterminate for $p<1$~\cite{levinOrthogonalPolynomialsWeights2007}.

What seems curious is that this determinate-indeterminate boundary happens to be precisely where the Operator Growth Hypothesis (OGH) conjectures generic spectral functions of locally interacting quantum chaotic systems should fall, in the sense that they should generically have exponentially decaying spectral functions, $\Phi(\omega \to\infty) \sim \exp[-\mathcal{O}(\omega)]$ \cite{parkerUniversalOperatorGrowth2019}. It is fitting, then, that in our spectral bootstrap algorithm we find that the error bounds at $\omega=0$ decay increasingly slowly the slower the high-frequency decay of the spectral function, as discussed in the previous section.
These error bounds control how large $n$ must be to approximate $\Phi(\omega=0)$ to a given precision $\epsilon$, so we see that the inverse problem becomes increasingly difficult as the determinate-indeterminate boundary is approached, before becoming impossible (by definition) in the indeterminate regime.

Thus, according to the OGH, recovery of spectral functions of chaotic systems is \textit{exponentially hard} in $n$ as a function of the inverse precision $1/\epsilon$. From the vantage of the original operator growth problem, it is not so surprising that chaotic systems are hard to simulate; for example, it is known that the operator entanglement entropy generically grows linearly in time for chaotic systems \cite{jonayCoarsegrainedDynamicsOperator2018}, rendering tensor network descriptions challenging. But it is perhaps surprising that, in the context of the Lanczos algorithm, where entanglement entropy plays no obvious role, nonetheless the spectral theory should conspire so as to make this problem hard.

To be clear, in the context of this section, quantum chaos is a \textit{sufficient but not necessary} condition for marginality. More generally the requirement for marginality is (quasi-)linear growth of Lanczos coefficients, and there are examples of non-chaotic systems exhibiting this growth, such as certain operators in many-body localized systems~\cite{caoStatisticalMechanismOperator2021}.

\subsection{Future directions}
\label{sec:synopsis_future_directions}
\textit{Combining the spectral bootstrap with dissipation-assisted approaches}---In \cref{sec:synopsis_spectral_bootstrap} we discussed that, in the Lanczos context, it is actually the high-frequency scaling of the spectral function $\Phi(\omega)$ that determines the computational difficulty of recovering $\Phi(\omega)$ at low frequencies. But it is often the case that one is principally interested in characterizing the low-frequency behavior of a system, say for computing hydrodynamic transport coefficients, and the accuracy of estimates of high-frequency properties are of lesser concern. That motivates studying potential modifications of the system's dynamics which cause $\Phi(\omega)$ to decay superexponentially as $\omega\to\infty$, thereby improving the error bound scaling, while ideally having minimal impact on $\Phi(\omega)$ near $\omega=0$. One possibility is to use operator truncation approaches, such as in dissipaton-assisted operator evolution (DAOE)~\cite{rakovszkyDissipationassistedOperatorEvolution2022} and similar techniques. Provided the modification is such that the effective Liouvillian is still self-adjoint, then the Lanczos approach is valid; anything of the form $\mathcal{L}_{\mathrm{eff}} = \mathcal{D} \mathcal{L} \mathcal{D}$ for some self-adjoint dissipator $\mathcal{D}$ would suffice.

A concrete choice is the DAOE dissipator $\mathcal{D}_{\mathrm{DAOE}}$, which suppresses Pauli strings $\sigma^{\bm{\alpha}}\coloneqq\otimes_{i} \sigma_{i}^{\alpha_{i}}$, $\alpha_{i}\in\{0,1,2,3\}$, according to their Pauli weight $|\bm{\alpha}|\coloneqq |\{i : \alpha_{i} \neq 0\}|$; if $|\bm{\alpha}|$ is above some chosen cutoff length $\ell_{*}$ then the dissipator acts as $\mathcal{D}_{\mathrm{DAOE}}[\sigma^{\bm{\alpha}}] = e^{-\gamma(|\bm{\alpha}|-\ell_{*})} \sigma^{\bm{\alpha}}$, where $\gamma$ is a dissipation strength. Our initial experiments suggest that applying $\mathcal{D}_{\mathrm{DAOE}}$ as part of the Lanczos algorithm generically arrests the growth of the Lanczos coefficients from $b_{n} \sim \mathcal{O}(n)$ to $b_{n} \sim \mathcal{O}(1)$, which should exponentially improve the convergence of the spectral bootstrap approach to computing transport coefficients from $\epsilon\sim 1/\mathrm{poly}(\log{n})$ to $\epsilon \sim 1/\mathcal{O}(n)$. It also has the benefit of reducing the computational requirements to compute the Lanczos coefficients up to a given $n$, due to a suppression of operator trajectories. But it is worth investigating to what extent one can smoothly extrapolate the resulting transport coefficients in the limit of vanishing dissipation $\gamma \to 0$~\cite{rakovszkyDissipationassistedOperatorEvolution2022}. %

\textit{Universal dynamics of large operators}---In this work we have demonstrated that operator dynamics restricted to the large-$n$ `fast space' of Lanczos dynamics exhibits an emergent universality, akin to that of eigenvalue correlations in random matrix theory. The Lanczos setting is useful because it has enough mathematical structure to enable rigorous results, but one could speculate that this universality is more generally a statement about the universal dynamics of large operators restricted to the space of large operators. Practically, one way to investigate this in a tensor network approach would be to use the matrix product operator (MPO) formulation of the DAOE dissipator (see Ref.\ \cite{rakovszkyDissipationassistedOperatorEvolution2022} for details), but change the boundary conditions of the MPO to keep only operators with weight $|\bm{\alpha}|>\ell_{*}$, rather than keeping $|\bm{\alpha}|\leq \ell_{*}$ as in the standard DAOE approach.

\textit{Universality in unitary and open system dynamics}---All of the results in this paper concern operator dynamics generated by a time-independent and Hermitian Hamiltonian $H$. An obvious question is whether the notion of universality discussed in this paper can be generalized to other types of operator dynamics. For Floquet dynamics generated by a unitary $U$, this question is likely quite tractable, and will involve asymptotics of orthogonal polynomials on the unit circle~\cite{simonOrthogonalPolynomialsUnit2009}, corresponding to the spectra of unitary operators (for related work see \cite{suchslandKrylovComplexityTrotter2025}). Where the situation is both less clear and possibly richer is that of open quantum system dynamics~\cite{bhattacharyaOperatorGrowthKrylov2022,bhattacharjeeOperatorGrowthOpen2023,bhattacharjeeOperatorDynamicsLindbladian2024,bhattacharyaKrylovComplexityOpen2023,srivatsaOperatorGrowthHypothesis2024}. With dynamics generated by a Lindbladian $\mathcal{L}$ with nontrivial jump operators, the spectrum of $\mathcal{L}$ now generically lies in the complex plane. It is still possible to get a banded matrix representation of $\mathcal{L}$---which was necessary to get a simple recursion between $G_{n}(z)$ and $G_{n+1}(z)$---by performing the biorthogonal version of the Lanczos algorithm. However, this will now involve polynomials whose orthogonality relation is defined not on a contour but on a subset of the complex plane~\cite{saadNumericalMethodsLarge1992,bhattacharyaKrylovComplexityOpen2023,srivatsaOperatorGrowthHypothesis2024}. This complex spectrum could lead to new universality classes for open system operator dynamics, but characterizing the nature of this universality is an open question. The technical machinery of polynomials orthogonal on 2D complex domains is much less well-developed~\cite{baloghStrongAsymptoticsOrthogonal2015,kieburgOrthogonalPolynomialsNormal2024}, although in certain special cases there have been generalizations of our Riemann-Hilbert approach to this case.
 
\textit{Connection to the generalized ETH}---One of our main results is the appearance of the Wigner semicircle law in the level-$n$ Green's function $G_{n}(z)$ describing the dynamics restricted to the `fast space' $\{O_{m}\}_{m\geq n}$. One can heuristically argue for this result using the recursion relation \cref{eq:greens_func_recursion} between $G_{n}(z)$ and $G_{n+1}(z)$, where the Wigner semicircle appears as a fixed point (see \cref{sec:gf}). This recursion bears some resemblance to the $R$-transform from free probability~\cite{mingoFreeProbabilityRandom2017}, but `deformed' by the ratio $b_{n+1}/b_{n}$. Recently, the language of free probability was used to generalize the eigenstate thermalization hypothesis (ETH) to higher-order correlation functions~\cite{pappalardiEigenstateThermalizationHypothesis2022,pappalardiFullEigenstateThermalization2025}. It is an interesting topic for future work to explore potential connections between universality in Lanczos dynamics and the generalized ETH.

\textit{Krylov complexity and conservation laws}---Much of the recent interest in Krylov space dynamics has been through the lens of Krylov complexity~\cite{rabinoviciKrylovComplexity2025}. Stemming from the original work of Ref.\ \cite{parkerUniversalOperatorGrowth2019}, there is by now a reasonably good understanding of how the leading behavior of the Lanczos coefficients affects the early-time behavior of the Krylov complexity $\mathcal{K}(t)$, and in particular how the linear growth $b_{n} \sim \mathcal{O}(n)$ postulated by the OGH translates into exponential growth $\mathcal{K}(t)\sim \exp[\mathcal{O}(t)]$ of the complexity at early times. There have also been studies of Krylov complexity in finite-sized systems, examining how the finite-size plateau of the Lanczos coefficients arrests the initial complexity growth~\cite{barbonEvolutionOperatorComplexity2019}. However, what is relatively less well understood is the effect of conservation laws on Krylov complexity, which should be increasingly important for the mid-to-late time complexity of thermodynamically large systems. Our work, and particularly our discussion of Bessel universality, lays the mathematical foundations for a systematic study of these effects. Indeed, if one expands the time-evolved operator $|O_{0}(t)) = \sum_{n=0}^{\infty} c_{n}(t) |O_{n})$ in the Lanczos basis, then a quick calculation shows that the coefficient $c_{n}(t)$ is given by $c_{n}(t) = \int_{\mathbb{R}} e^{i \omega t} p_{n}(\omega) \Phi(\omega) \diff \omega/2\pi$. By exploiting the universal description of the orthogonal polynomials $p_{n}(\omega)$ we have derived for $n\gg 1$ (e.g., \cref{eq:pn_bulk_asymptotic} or expressions in \cref{sec:thm2_proof}), one could probe the large-$n$ contributions to the Krylov complexity.

\textit{New quantum algorithms for quantum dynamics}---Our results focus on the `post-processing' step of converting Lanczos coefficients into estimates of spectral functions, and as such are agnostic about how these Lanczos coefficients were actually obtained. In classical implementations of the Lanczos algorithm, a lightcone argument suggests that the memory requirement for a brute force computation of the $n$th Lanczos coefficient scales like $\exp[\mathcal{O}(n)]$, which can become prohibitively large for modest $n$. (Though see Ref.\ \cite{parkerLocalMatrixProduct2020} for a potentially better scaling approach using tensor networks.) While one approach to ameliorating these memory requirements is to use the dissipation-assisted methods discussed above, one could also implement the Lanczos algorithm on a quantum computer~\cite{bakerLanczosRecursionQuantum2021,kirbyExactEfficientLanczos2023}. One reason to consider this approach over direct real-time evolution is that it does not face the latter's tradeoff between Trotter error and circuit depth (for a review of related methods see Ref.\ \cite{mottaSubspaceMethodsElectronic2024}).

\textit{Improved treatment of hydrodynamic tails}---The main technical caveat to our proof technique is that we need to make some quite strong analyticity assumptions on the spectral function (see \cref{sec:assumptions_summary}), in order to perform the contour deformation involved in a steepest descent analysis of an associated Riemann-Hilbert problem. It would be desirable to prove our results under weaker assumptions, not just mathematically but also from a physical perspective, since it allows for a better treatment of hydrodynamic tails. In particular, for autocorrelation functions $C(t)$ that decay algebraically, $C(t\to\infty) \sim 1/t^{\alpha}$, standard Fourier arguments suggest that the spectral function $\Phi(\omega)$ should possess only a finite number of derivatives at $\omega = 0$. In a sense, our results adequately cover the case where $0<\alpha<1$, since then $\Phi(\omega)$ should have a power-law divergence at $\omega=0$, and we explicitly address this case by decomposing $\Phi(\omega)/2\pi \equiv |\omega|^{\rho} e^{-Q(\omega)}$. However, our results do not adequately cover the case where $\alpha > 1$, because in this case $\Phi(\omega)$ is finite at $\omega=0$ (so $\rho=0$) but is not infinitely differentiable, whereas we need to assume full analyticity for our proof. We have conjectured that our results should still be correct to leading order in $n$, but the lack of full analyticity will result in a more slowly decaying error term. One way to address this rigorously would be to replace our Riemann-Hilbert problem approach with a $\overline{\partial}$-problem approach~\cite{mclaughlinSteepestDescentMethod2006,mclaughlinDbarSteepestDescent2008}, which is still treatable with nonlinear steepest descent controlled in the $n\to\infty$ limit.

\subsection{Reader's guide}
We outline the contents of the paper below, and have attempted to indicate how readers can focus on those sections most relevant to their interests. There is also a \hyperref[sec:glossary]{glossary of symbols} at the end of the paper.

\cref{sec:background} provides background material, and in particular illustrates the connection between the Lanczos algorithm and orthogonal polynomials generated by the spectral function $\Phi(\omega)$. We recommend all readers to ensure they are familiar with the contents of this section, since it explains the central role that the spectral function plays in our analysis.

\cref{sec:potential_definitions} is where we state the assumptions on the spectral function under which we can prove our results. We recommend readers follow the high-level summary in \cref{sec:assumptions_summary}, and skip the precise definitions in \cref{sec:precise_assumptions} until they are interested in the proof. \cref{sec:analyticity} provides some discussion of our analyticity requirements on $\Phi(\omega)$, and can also be skipped on a first reading. For interested readers, we note that various properties of these spectral functions are derived in \cref{sec:freud_properties}, which is not too technical and is independent of the later Riemann-Hilbert analysis.

\cref{sec:coulomb} introduces the Coulomb gas problem associated with the spectral function. \cref{sec:coulomb_definition} defines the Coulomb gas and introduces some useful notation that will appear in later sections, so we recommend it for all readers. This is enough to understand the later \cref{sec:gf} on universality of the level-$n$ Green's function. \cref{sec:coulomb_rescaled} introduces some further notation related to a rescaled version of the Coulomb gas; this does appear in the statements of some theorems (e.g., \cref{thm:hydrodynamic_constants}), but is less essential to understand immediately. \cref{sec:coulomb_confinement} is where we discuss the confinement transition in the Coulomb gas. We recommend this subsection to readers interested in the connections between quantum chaos, the Operator Growth Hypothesis, and marginality.  

\cref{sec:gf} is where we discuss our results on the universality of the level-$n$ Green's function $G_{n}(z)$. We expect this section to be of interest to most readers.

In \cref{sec:recurrence_coefficients} we move focus to signatures of hydrodynamics in Lanczos space, and in \cref{thm:hydrodynamic_constants} show how a low-frequency power-law $\Phi(\omega\to 0) \sim |\omega|^{\rho}$ appears as a subleading `staggered' term in the Lanczos coefficients. However, we do not recommend using this result as a practical method to estimate $\rho$, instead preferring to focus on the zero mode, to be discussed imminently. 

Indeed, in the application-motivated \cref{sec:hydro}, we show how a low-frequency power-law affects the \textit{leading} behavior of the zero mode of the Liouvillian restricted to the Krylov space. In \cref{sec:extracting_power_law} we focus on the decay of the zero mode amplitudes with $n$, and how this can be used to extract the value of the exponent $\rho$. Then in \cref{sec:transport_coeffs} we further show how one can estimate hydrodynamic transport coefficients from the zero mode amplitudes. We benchmark this approach on physical models in \cref{sec:hydro_benchmarks}.

The previous \cref{sec:hydro} was focused on $\omega=0$ properties of the spectral function $\Phi(\omega)$, but in \cref{sec:spectral_bootstrap} we show how to estimate $\Phi(\omega)$ at nonzero frequencies in an algorithm called the \textit{spectral bootstrap}. The examples in the previous section used the $\omega\to 0$ limit of this algorithm. On a first read of \cref{sec:spectral_bootstrap}, one could just focus on \cref{sec:bulk_bootstrap}, where we illustrate the basic ideas behind the spectral bootstrap.

We expect \cref{sec:universality} to be of most interest to readers already familiar with the universality of eigenvalue correlations in random matrix theory, although we do lay out the relevant background in \cref{sec:universality_background}. The reason is that we examine quantities whose meaning is most transparent in the random matrix language. For each spectral function $\Phi(\omega)$, we define an associated random matrix ensemble. By using the spectral bootstrap to estimate the equilibrium measure, we are able to test whether the random matrix ensembles associated with various physical models possess the well-known universality of eigenvalue correlations. This can be seen as a consistency check for the universality present in the level-$n$ Green's function $G_{n}(z)$, since both types of universality have the same origin (universality of orthogonal polynomials). But at present the universality of these eigenvalue correlations does not appear to have a transparent implication for physically relevant quantities.

Finally, in \cref{sec:smoothness} we give evidence that fast space universality indeed requires a degree of smoothness of the spectral function $\Phi(\omega)$ as a function of frequency $\omega$. We demonstrate this by focusing on an Anderson localized model, which provably has a point like spectrum, and where we see the failure of a consistency check that was passed by the clean version of this model.

Moving to the Supplementary Material, we first derive in \cref{sec:freud_properties} various implications of our assumptions for the behavior of the spectral function. Then the Riemann-Hilbert analysis proceeds in \cref{sec:rhp}. By reading only \cref{sec:fundamental_rhp}, one can appreciate the relevance of the Riemann-Hilbert problem (RHP) to Lanczos dynamics, without going through the technical detail of actually solving the RHP. However, we emphasize that it is \textit{not} necessary for readers to follow the RHP analysis in order to use our proposed spectral bootstrap algorithm.

For readers familiar with RHPs for orthogonal polynomials, we note that much of our analysis in \cref{sec:rhp} is quite standard, but we have given explicit details for the benefit of a mixed readership (the same goes for \cref{sec:recurrence_proof,sec:polynomial_asymptotics}). The main novelty is in our local analysis near the origin, and this requires a detailed study of the equilibrium measure that can be found in \cref{sec:endpoint_analysis}. For the sake of completeness, in \cref{sec:airy_bootstrap} we also describe a version of the spectral bootstrap involving `Airy universality' suited for frequencies near the endpoints $\omega=\pm \beta_{n}$, but since these frequencies often become large for large $n$, we expect this to be less practically relevant than the versions of the spectral bootstrap described in \cref{sec:spectral_bootstrap}.

\subsection{Acknowledgements}
We thank Alexander Avdoshkin, Xiangyu Cao, Pieter Claeys, Fabian Essler, Michael Flynn, Jorge Kurchan, Jean-Bernard Lasserre, Doron Lubinsky, Sheehan Olver, Daniel Parker, and Gabriele Pinna for helpful discussions. Some of the numerical computations were performed using King's College London's CREATE cluster~\cite{KingsComputationalResearch2024}. The data that support the findings of this article are openly available~\cite{lunt_2025_15387574}. O.L.\ was supported by EPSRC through grant number EP/Y005058/2. K.M.\ gratefully acknowledges the support of a Royal Society Wolfson Fellowship (grant number: RSWVF/R2/212003). C.vK.\ was supported by a UKRI FLF through MR/T040947/1, MR/T040947/2 and MR/Z000297/1.

\clearpage
\section{Background}
\label{sec:background}
A time-evolved operator $A(t) \coloneqq e^{i H t} A e^{-i H t}$ can be expanded as a sum of nested commutators,
\begin{equation}
    e^{i H t} A e^{-i H t} = A + it [H,A] + \dfrac{(it)^{2}}{2!} [H,[H,A]] + \cdots.
\end{equation}
This suggests that if we want to understand the growth of $A(t)$, a natural object to study is the \textit{Krylov space} $\mathcal{K} \coloneqq \mathrm{span}\{\mathcal{L}^{k}(A)\}_{k=0}^{\infty}$, where $\mathcal{L}(\cdot) \coloneqq [H, \cdot]$ is the Liouvillian superoperator. To probe this space, it is convenient to construct an orthonormal basis for $\mathcal{K}$ using the Gram-Schmidt process. When $\mathcal{L}$ is self-adjoint, this reduces to the \textit{Lanczos algorithm}, where successive basis vectors for $\mathcal{K}$ can be defined by a recurrence relation containing only three terms, rather than all previous basis vectors. We will use a vector notation $|A)$ for operators, and for now we take the inner product to be the Hilbert-Schmidt product $(A|B) \coloneqq \tr{A^{\dag} B} / \tr{\mathds{1}}$, so we are effectively working at infinite temperature. Starting with a self-adjoint operator $|A)$, we initialize the recurrence with $|O_{-1}) \coloneqq 0$ and $|O_{0}) \coloneqq b_{0}^{-1} |A)$, where $b_{0} \coloneqq \sqrt{(A|A)} \equiv \norm{A}$. Then we recursively define  
\begin{subequations}
    \begin{align}
        |A_{n}) &\coloneqq \mathcal{L} |O_{n-1}) - b_{n-1} |O_{n-2}),\label{eq:lanczos_recurrence}\\
        b_{n} &\coloneqq \sqrt{(A_{n}| A_{n})},\label{eq:lanczos_coefficient_def}\\
        |O_{n}) &\coloneqq b_{n}^{-1} |A_{n}). \label{eq:lanczos_operator_def}
    \end{align} 
\end{subequations}
We will refer to the basis $\{|O_{n})\}_{n=0}^{\infty}$ as the \textit{Lanczos basis}, and throughout will consider the generic case where $\mathcal{K}$ is infinite-dimensional in the thermodynamic limit.

As well as the basis vectors themselves, the Lanczos algorithm produces a sequence of numbers $\{b_{n}\}_{n=1}^{\infty}$ which we will interchangeably refer to as the \textit{recurrence} or \textit{Lanczos} coefficients. It turns out that these numbers are enough to fully characterize the action of $\mathcal{L}$ within $\mathcal{K}$; indeed, the restriction of $\mathcal{L}$ to $\mathcal{K}$ can be represented in the Lanczos basis by a tridiagonal matrix
\begin{equation}
    \mathcal{L} = \begin{pmatrix}
        0 & b_{1} & 0 & 0 & \cdots \\
        b_{1} & 0 & b_{2} & 0 & \cdots \\
        0 & b_{2} & 0 & b_{3} & \cdots \\
        0 & 0 & b_{3} & 0 & \ddots \\
        \vdots & \vdots & \vdots & \ddots & \ddots
    \end{pmatrix}.
\end{equation}
That the diagonal elements are all zero is a generic consequence of \textit{Hermiticity}, since $(O | \mathcal{L} | O) = 0$ follows for any self-adjoint $O$ simply by using the definition $\mathcal{L}(\cdot) = [H, \cdot]$ with self-adjoint $H$, and inductively one can show that $i^{n} O_{n}$ is self-adjoint. %

The tridiagonal form lends itself to interpreting $\mathcal{L}$ as a tight-binding model on a semi-infinite chain, where the sites are the operators $\{O_{n}\}_{n=0}^{\infty}$, and the hopping strengths are the recurrence coefficients $\{b_{n}\}_{n=1}^{\infty}$. We can visualize the operator evolution $A(t)$ as a single-particle problem, with a wavefunction initially localized on site $n=0$ of the chain, which then spreads out along the chain over time. The autocorrelation function $C(t) \coloneqq (A | A(t))$ corresponds to the probability amplitude for the operator wavefunction to be found back at its starting point at time $t$.

The \textit{spectral function} $\Phi(\omega)$ defined as
\begin{equation}
    \Phi(\omega) \coloneqq \int_{\mathbb{R}} e^{-i \omega t} C(t) \diff t,
    \label{eq:spectral_function_fourier_def}
\end{equation}
plays a significant role in the Lanczos algorithm. In particular, the Lanczos algorithm has a natural formulation in terms of the \textit{orthogonal polynomials} with respect to the weight function $w(\omega) = \Phi(\omega) / 2\pi$. To see this, we start from
\begin{equation}
    C(t) = (A | e^{i \mathcal{L} t} | A) = \int_{\mathbb{R}} e^{i \omega t} \dfrac{\Phi(\omega)}{2\pi} \diff \omega,
\end{equation}
and evaluate the $k$th derivative with respect to $t$ at $t = 0$. Assuming $\Phi(\omega)$ decays sufficiently quickly as $|\omega| \to \infty$ (more on this later), this gives
\begin{equation}
    (A | \mathcal{L}^{k} | A) = \int_{\mathbb{R}} \omega^{k} \dfrac{\Phi(\omega)}{2\pi} \diff \omega.
    \label{eq:spectral_generating_function}
\end{equation}
Since $\mathcal{L}$ has real spectrum, $\omega_{ij} = E_{i} - E_{j}$, if we then define the inner product between real polynomials $p,q \colon \mathbb{R} \to \mathbb{R}$
\begin{equation}
    \langle p, q \rangle \coloneqq (p(\mathcal{L}) A | q(\mathcal{L}) A),
    \label{eq:inner_product_def}
\end{equation}
one can use linearity of the inner product and \cref{eq:spectral_generating_function} to write this as
\begin{equation}
    \langle p, q \rangle = \int_{\mathbb{R}} p(\omega) q(\omega) \dfrac{\Phi(\omega)}{2\pi} \diff \omega,
    \label{eq:inner_product_integral}
\end{equation}
so, as claimed, $\Phi(\omega) / 2\pi$ appears as the relevant weight function. A defining property of orthonormal polynomials with respect to an even weight function on $\mathbb{R}$ is that they obey a three-term recursion relation of the form~\cite{szegoOrthogonalPolynomials1939,gautschiOrthogonalPolynomialsComputation2004}
\begin{equation}
    b_{n} p_{n}(\omega) = \omega p_{n-1}(\omega) - b_{n-1} p_{n-2}(\omega).
    \label{eq:three_term_recursion}
\end{equation}
Comparing this to the Lanczos recursion relation, we deduce that the $n$th Lanczos vector $|O_{n})$ is given by
\begin{equation}
    |O_{n}) = p_{n}(\mathcal{L}) |A),
    \label{eq:lanczos_vector}
\end{equation}
where $p_{n}(\omega) = y_{n} \omega^{n} + \cdots, y_{n}>0$, is the $n$th order orthonormal polynomial with respect to the inner product \cref{eq:inner_product_def}~\cite{saadNumericalMethodsLarge1992,muckKrylovComplexityOrthogonal2022}. Thus, by studying the orthogonal polynomials with respect to the spectral function $\Phi(\omega)/2\pi$, we can understand properties of the Lanczos operators. %

Note that this argument is very general, so that we can also easily apply it to operator inner products other than the Hilbert-Schmidt product. In particular, at finite temperature, we can consider the following family of inner products~\cite{viswanathRecursionMethodApplication2013}
\begin{equation}
    (A|B)_{\beta}^{g} \coloneqq \dfrac{1}{\beta} \int_{0}^{\beta} g(\lambda) \langle  A^{\dag} e^{-\lambda H} B e^{\lambda H} \rangle_{\beta} \, \diff \lambda - \langle A^{\dag} \rangle_{\beta} \langle B \rangle_{\beta},
    \label{eq:thermal_prod}
\end{equation}
where $\langle B \rangle_{\beta} = \tr{e^{-\beta H}B} / \mathcal{Z}$ is a thermal expectation value, $\mathcal{Z} = \tr{e^{-\beta H}}$ is the partition function, and $g(\lambda)$ is any function defined on the thermal circle $[0,\beta]$ satisfying
\begin{equation}
    g(\lambda) \geq 0, \qquad g(\beta-\lambda) = g(\lambda), \qquad \dfrac{1}{\beta} \int_{0}^{\beta} g(\lambda) \diff \lambda  = 1.
    \label{eq:thermal_g_conditions}
\end{equation}
One can then perform a version of the Lanczos algorithm where operators are orthogonal with respect to this inner product. Examples include $g(\lambda) = [\delta(\lambda) + \delta(\lambda-\beta)]/2$ for linear response, and the Wightman product $g(\lambda) = \delta(\lambda - \beta/2)$. The autocorrelation function and spectral function are then $C_{\beta}^{g}(t) \coloneqq (A|A(t))_{\beta}^{g}$ and $\Phi_{\beta}^{g}(\omega) \coloneqq \int_{-\infty}^{\infty} e^{-i\omega t} C_{\beta}^{g}(t) \diff t$, and the polynomial inner product defined by $\langle p, q \rangle_{\beta}^{g} \coloneqq (p(L^{\dag}) A | q(L) A)_{\beta}^{g}$ can be written as
\begin{equation}
    \langle p, q \rangle_{\beta}^{g} = \int_{-\infty}^{\infty} p(\omega) q(\omega) \dfrac{\Phi_{\beta}^{g}(\omega)}{2\pi} \diff \omega,
\end{equation}
with the generalized spectral function $\Phi_{\beta}^{g}(\omega)/2\pi$ again appearing as the weight function. Since all our work is premised on assumptions made about the spectral function (\cref{sec:potential_definitions}), we expect our conclusions to hold not only at infinite temperature, but also at high enough finite temperatures, given a suitable choice of inner product.

\section{Statement of assumptions on the spectral function}
\label{sec:potential_definitions}
\subsection{Summary}
\label{sec:assumptions_summary}
We saw in \cref{sec:background} that the spectral function $\Phi(\omega) \coloneqq \int_{\mathbb{R}} e^{-i \omega t} C(t) \diff t$ determines all the properties of the Lanczos basis. However, for a given interacting many-body Hamiltonian $H$, we do not generally have a hope of exactly calculating $\Phi(\omega)$. Instead, the approach we will take is to impose a few physically motivated conditions on $\Phi(\omega)$, and study their implications. Before stating these conditions precisely, we state them at a high level:
\begin{enumerate}
    \item We fix the $|\omega| \to \infty$ behavior of $\Phi(\omega)$, taking it to decay at least exponentially in $|\omega|$ (but possibly faster). This is provably true at high temperatures for \textit{local} lattice Hamiltonians with a bounded local Hilbert space~\cite{abaninExponentiallySlowHeating2015,aradConnectingGlobalLocal2016,bruLiebRobinsonBoundsMultiCommutators2017}, and forms part of the motivation for the operator growth hypothesis~\cite{parkerUniversalOperatorGrowth2019}.
    \item Since we are interested in signatures of the $|t| \to \infty$ behavior of $C(t)$, we impose a condition on the $\omega \to 0$ behavior of its Fourier transform $\Phi(\omega)$. To model $C(t) \sim 1 / |t|^{1+\rho}$ as $|t| \to \infty$, we require $\Phi(\omega) \sim |\omega|^{\rho}$ as $\omega \to 0$. Such power-law decay in $C(t)$ is typical if the initial operator overlaps with a conserved quantity with sufficiently slow transport, say diffusive in 1D~\cite{forsterHydrodynamicFluctuationsBroken2019}. %
    \item We take $\Phi(\omega) = \Phi(-\omega)$ to be even, which is satisfied if one uses any of the thermal inner products defined in \cref{eq:thermal_prod} satisfying the conditions in \cref{eq:thermal_g_conditions}. This includes the infinite temperature Hilbert-Schmidt product $(A|B) = \tr{A^{\dag}B}/\tr{\mathds{1}}$, which we will use in all our numerical tests.
    \item We require $\Phi(\omega) > 0$ to be strictly positive, which should be generically true for interacting systems in the thermodynamic limit. Furthermore, in a sense we will make precise below, we require $\Phi(\omega)$ to be very smooth. This has both technical and physical implications. At a technical level, it allows us to use a powerful complex analytic proof technique involving a Riemann-Hilbert problem. However, we have evidence that smoothness of spectral functions is not merely an artefact of the proof technique, but rather is a physical requirement to get universality. Indeed, in \cref{sec:smoothness} we give evidence that universality fails in an Anderson localized system, where strong disorder can result in a pure point spectrum~\cite{aizenmanRandomOperators2015}.
\end{enumerate}
We remark that the analyticity requirements for our proof could likely be weakened to assuming only a few derivatives of $\Phi(\omega)$, by employing a so-called $\overline{\partial}$-problem approach~\cite{mclaughlinSteepestDescentMethod2006,mclaughlinSteepestDescentMethod2008,diengDispersiveAsymptoticsLinear2019}, rather than using a Riemann-Hilbert problem as we do here. Experience suggests that one usually finds the same leading order asymptotics, but with potentially slower decaying error terms, depending on the degree of differentiability of $\Phi(\omega)$ one assumes.

\subsection{Precise assumptions}
\label{sec:precise_assumptions}
In order to state our precise assumptions, we need to make a definition. Since $\Phi(\omega) \geq 0$ and is even, we can always write it in the form $\Phi(\omega) \equiv \exp[-Q(\omega)]$ for some even real-valued function $Q(\omega)$. However, since we are interested in hydrodynamic spectral functions, which can have an algebraic divergence $\Phi(\omega) \sim |\omega|^{\rho}$ as $\omega \to 0$, it will prove helpful to factorize out this divergence, and decompose $\Phi(\omega)$ as
\begin{equation}
    \Phi(\omega)/2\pi \equiv \left|\omega\right|^{\rho} e^{-Q(\omega)}.
    \label{eq:spectral_function_def}
\end{equation}
In principle $\Phi(\omega)$ may also have algebraic behavior near other frequencies, in which case we would also factorize those out~\cite{kuijlaarsUniversalityEigenvalueCorrelations2003}. But for simplicity we will focus on the case where there is only an algebraic divergence at $\omega = 0$. The function $Q(\omega)$ so defined is called the \textit{potential} (we will later make precise the sense in which this is a potential). %

For technical reasons, it is more convenient to make all our assumptions in terms of $Q(\omega)$ rather than $\Phi(\omega)$ directly. We will translate assumptions about the high-frequency decay of $\Phi(\omega\to\infty)$ into assumptions about the high-frequency growth of $Q(\omega\to\infty)$. In particular, we will consider a class of potentials inspired by the `very smooth Freud weights' of Refs.~\cite{lubinskyProofFreudConjecture1988,lubinskyUniformMeanApproximation1988}, which they denote by $\mathrm{VSF}(p)$, with $p$ an exponent governing the degree of the polynomial growth of $Q(\omega) \sim |\omega|^{p}$ as $|\omega| \to \infty$. We will consider a subset of $\mathrm{VSF}(p)$, where we add the requirement of analyticity, and also require a specification of the logarithmic corrections to the leading polynomial growth of $Q(\omega)$.

Our Riemann-Hilbert analysis draws heavily from Ref.~\cite{deiftStrongAsymptoticsOrthogonal1999}, where they take $Q$ to be a polynomial of even order. However, we are particularly interested in the marginal case where $Q(\omega \to \infty)$ grows linearly with $|\omega|$, since the Operator Growth Hypothesis~\cite{parkerUniversalOperatorGrowth2019} conjectures this to be generic for spectral functions in chaotic many-body quantum systems. But it is clearly not possible to simultaneously have i) $Q(\omega) \sim |\omega|$ as $|\omega| \to \infty$, and ii) $Q$ be a polynomial. This was a primary motivation for considering this generalized class $\mathrm{VSLF}(p,q)$ of `polynomial-like' weights. 
\begin{definition}[$\mathrm{VSLF}(p,q)$: log-Freud potentials of order $(p,q)$]
    Let $Q : \mathbb{R} \to \mathbb{R}$ be real-analytic, even, and satisfy
    \begin{equation}
        Q^{\prime}(\omega) > 0, \quad \text{for } \omega \text{ large enough},
        \label{eq:vsf_prop1}
    \end{equation}
    \begin{equation}
        \lim_{\omega \to \infty}\left(\dfrac{\omega Q^{\prime\prime}(\omega)}{Q^{\prime}(\omega)}\right) = p - 1,
        \label{eq:vsf_prop3}
    \end{equation}
    \begin{equation}
        \lim_{\omega\to\infty} \left(\log(\omega) \left[-p + \dfrac{\omega Q^{\prime}(\omega)}{Q(\omega)}\right]\right) = q.
        \label{eq:vslf_limit}
    \end{equation}
    for some $p > 0$ and $q \in \mathbb{R}$. Then we shall call $Q$ a \textit{log-Freud} potential of order $(p,q)$ and write $Q \in \mathrm{VSLF}(p,q)$.
\end{definition}
We characterize the behavior of these potentials in \cref{sec:freud_properties}. In particular, in \cref{eq:vslf_characterization_prop2} we show that these potentials grow as $|\omega| \to \infty$ like
\begin{equation}
    |\omega|^{p} (\log{|\omega|})^{q - \epsilon} \leq Q(\omega) \leq |\omega|^{p} (\log{|\omega|})^{q + \epsilon},
\end{equation}
where $\epsilon$ can be taken to zero as $|\omega| \to \infty$; we will informally write this as $Q(\omega) \sim |\omega|^{p} \log^{q} |\omega|$. In this sense assumptions \cref{eq:vsf_prop3,eq:vslf_limit} are similar to but slightly weaker than assuming $Q \in \Theta(|\omega|^{p} (\log{|\omega|})^{q})$. %
We make the assumption \cref{eq:vsf_prop1} for technical convenience; it amounts to assuming that the spectral function $\Phi(\omega)$ decays monotonically above some frequency scale (which is $\mathcal{O}(1)$ but potentially much larger than any characteristic frequency scale). %

In order to apply Riemann-Hilbert techniques, we also need to assume that some of these properties continue to hold in a region of the complex plane near the real axis (illustrated in \cref{fig:analyticity}).
\begin{definition}[$\mathrm{CVSLF}(p,q,\theta,\gamma)$: complex log-Freud potentials of order $(p,q)$]
    For an angle $0 < \theta \leq \pi/2$, define the `complex cone' $C_{\theta}$ by
    \begin{equation}
        C_{\theta} \coloneqq \left\{z : |\arg{z}| < \theta \right\} \cup \left\{ z : |\arg{z}| > \pi - \theta\right\},
        \label{eq:complex_cone}
    \end{equation}
    using the convention $-\pi < \arg{z} \leq \pi$. We consider the open cone, so $z=0$ is not included in $C_{\theta}$. Now suppose there is some $0 < \theta \leq \pi/2$ and $\gamma > 0$ such that $Q \in \mathrm{VSLF}(p,q)$ can be analytically continued to $C_{\theta} \cup \left\{z : |z| < \gamma\right\}$, the union of $C_{\theta}$ and the disk of radius $\gamma$ centered at the origin (see \cref{fig:analyticity}).
    Also assume that \cref{eq:vsf_prop3,eq:vslf_limit} generalize to this region, in the sense that for $z$ restricted to $C_{\theta}$ we have
    \begin{equation}
        \lim_{|z|\to\infty} \dfrac{z Q^{\prime\prime}(z)}{Q^{\prime}(z)} = p-1,
        \label{eq:complex_assumption}
    \end{equation}
    \begin{equation}
              \lim_{|z|\to\infty} \left(\log(z) \left[-p + \dfrac{z Q^{\prime}(z)}{Q(z)}\right]\right) = q.
        \label{eq:complex_assumption_T}
    \end{equation}
    Given these properties, we say that $Q$ is a complex log-Freud potential of order $(p,q)$, and write $Q \in \mathrm{CVSLF}(p,q,\theta,\gamma)$.
\end{definition}

\fbox{\begin{minipage}{\columnwidth}Our proofs will apply for $Q \in \mathrm{CVSLF}(p,q,\theta,\gamma)$ in the cases $p>1,q\in \mathbb{R}$, and $p=1,q>-1$, with the latter case conjectured to be generic for quantum chaotic systems, according to the Operator Growth Hypothesis~\cite{parkerUniversalOperatorGrowth2019}. Note that local interactions guarantee $p\geq 1$~\cite{abaninExponentiallySlowHeating2015,aradConnectingGlobalLocal2016,bruLiebRobinsonBoundsMultiCommutators2017}.\end{minipage}}

\begin{example}
    All even polynomials with positive leading coefficient, $Q(x) = q_{2m} x^{2m} + \cdots$, $q_{2m} > 0$, lie in $\mathrm{CVSLF}(p,q,\theta,\gamma)$ with $p=2m$, $q=0$, $\theta=\pi/2$, $\gamma=\infty$.
\end{example}
\begin{example}
    Certain fractional powers of polynomials satisfy our assumptions, e.g. $Q(x) = \sqrt{1 + x^{2}}$ has $p=1$, $q=0$, $0 < \theta < \pi/2$, and $0 < \gamma < 1$. A similar example is $Q(x) = \sqrt{1+x^{2}} \log(1 + x^{2})$, which has $p=1,q=1$.
\end{example}
\begin{example}
    The symmetric Meixner-Pollaczek weights have $w(x) = \exp[-Q(x)] = \Gamma(\lambda + i x)\Gamma(\lambda -i x)$. ($Q$ can then be defined as in \cref{eq:Q_integral_def} below.) With $\lambda > 0$, we have $p=1$, $q=0$, $0 < \theta < \pi/2$, and $0 < \gamma < \lambda$. These weights appear in a rescaled form as the 2-point Wightman spectral function in the SYK model~\cite{maldacenaRemarksSachdevYeKitaevModel2016,dodelsonThermalProductFormula2024}, and were utilized in Ref.~\cite{parkerUniversalOperatorGrowth2019} to give an exactly solvable spectral function with exponential decay for use with the recursion method~\cite{viswanathRecursionMethodApplication2013}.
\end{example}
\begin{example}
    Taking $Q^{\prime}(z) = \erf(z) = (2/\sqrt{\pi}) \int_{0}^{z} e^{-t^{2}} \diff t$ gives an example of a $p=1$ potential which is also entire, unlike the previous $p=1$ examples. The conditions \cref{eq:complex_assumption,eq:complex_assumption_T} hold for $0<\theta < \pi/4$.
\end{example}

\begin{figure}[t]
    \centering
      
\tikzset {_9jbjxx02h/.code = {\pgfsetadditionalshadetransform{ \pgftransformshift{\pgfpoint{0 bp } { 0 bp }  }  \pgftransformrotate{-180 }  \pgftransformscale{2 }  }}}
\pgfdeclarehorizontalshading{_sitsp0slu}{150bp}{rgb(0bp)=(0.91,0.95,1);
rgb(37.5bp)=(0.91,0.95,1);
rgb(62.5bp)=(0.7,0.84,1);
rgb(100bp)=(0.7,0.84,1)}

\tikzset {_0xvn51kyz/.code = {\pgfsetadditionalshadetransform{ \pgftransformshift{\pgfpoint{0 bp } { 0 bp }  }  \pgftransformrotate{0 }  \pgftransformscale{2 }  }}}
\pgfdeclarehorizontalshading{_gifwjbc9r}{150bp}{rgb(0bp)=(0.91,0.95,1);
rgb(37.5bp)=(0.91,0.95,1);
rgb(62.5bp)=(0.7,0.84,1);
rgb(100bp)=(0.7,0.84,1)}

\tikzset {_ipvoxy33n/.code = {\pgfsetadditionalshadetransform{ \pgftransformshift{\pgfpoint{0 bp } { 0 bp }  }  \pgftransformrotate{-180 }  \pgftransformscale{2 }  }}}
\pgfdeclarehorizontalshading{_xrk0l28i5}{150bp}{rgb(0bp)=(0.91,0.95,1);
rgb(37.5bp)=(0.91,0.95,1);
rgb(62.5bp)=(0.7,0.84,1);
rgb(100bp)=(0.7,0.84,1)}

\tikzset {_yfoxe3hk5/.code = {\pgfsetadditionalshadetransform{ \pgftransformshift{\pgfpoint{0 bp } { 0 bp }  }  \pgftransformrotate{0 }  \pgftransformscale{2 }  }}}
\pgfdeclarehorizontalshading{_4vor5h302}{150bp}{rgb(0bp)=(0.91,0.95,1);
rgb(37.5bp)=(0.91,0.95,1);
rgb(62.5bp)=(0.7,0.84,1);
rgb(100bp)=(0.7,0.84,1)}
\tikzset{every picture/.style={line width=0.75pt}} %

\begin{tikzpicture}[x=0.75pt,y=0.75pt,yscale=-1,xscale=1]

\draw  [draw opacity=0][fill={rgb, 255:red, 179; green, 215; blue, 255 }  ,fill opacity=1 ] (224.79,141.16) .. controls (224.79,127.35) and (235.99,116.16) .. (249.79,116.16) .. controls (263.6,116.16) and (274.79,127.35) .. (274.79,141.16) .. controls (274.79,154.96) and (263.6,166.16) .. (249.79,166.16) .. controls (235.99,166.16) and (224.79,154.96) .. (224.79,141.16) -- cycle ;
\draw  [draw opacity=0][shading=_sitsp0slu,_9jbjxx02h] (350.8,101.34) -- (248.79,141.16) -- (350.8,141.16) -- cycle ;
\draw  [draw opacity=0][shading=_gifwjbc9r,_0xvn51kyz] (147.79,101.34) -- (248.79,141.16) -- (147.79,141.16) -- cycle ;
\draw  [draw opacity=0][shading=_xrk0l28i5,_ipvoxy33n] (350.8,180.97) -- (248.79,141.16) -- (350.8,141.16) -- cycle ;
\draw  [draw opacity=0][shading=_4vor5h302,_yfoxe3hk5] (147.79,180.97) -- (248.79,141.16) -- (147.79,141.16) -- cycle ;
\draw  (131.19,141.16) -- (371.19,141.16)(249.79,97.39) -- (249.79,182.54) (364.19,136.16) -- (371.19,141.16) -- (364.19,146.16) (244.79,104.39) -- (249.79,97.39) -- (254.79,104.39)  ;

\draw  [draw opacity=0] (286.36,127.21) .. controls (286.38,127.27) and (286.4,127.32) .. (286.42,127.38) .. controls (287.91,131.79) and (288.32,136.31) .. (287.79,140.64) -- (258,137) -- cycle ; \draw   (286.36,127.21) .. controls (286.38,127.27) and (286.4,127.32) .. (286.42,127.38) .. controls (287.91,131.79) and (288.32,136.31) .. (287.79,140.64) ;  

\draw (291,127) node [anchor=north west][inner sep=0.75pt]    {$\theta $};
\draw (235,83) node [anchor=north west][inner sep=0.75pt]    {$\Im{z}$};
\draw (355,122) node [anchor=north west][inner sep=0.75pt]    {$\Re{z}$};

\end{tikzpicture}
    \caption{We require the potential $Q(z)$ defined in \cref{eq:spectral_function_def} to have an analytic continuation to the shaded region of the complex plane, where $0 < \theta \leq \pi/2$ is any positive angle. This region is the union of the `complex cone' $C_{\theta}$ (see \cref{eq:complex_cone}) and a disk of constant radius around $z=0$.}
    \label{fig:analyticity}
\end{figure}
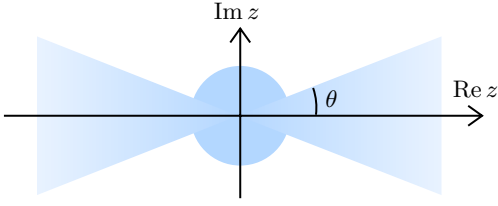

\subsection{Comments on analyticity requirements}
\label{sec:analyticity}
Regarding the definition of the class $\mathrm{CVSLF}$, the reason we need analyticity in a cone of constant argument rather than, say, a strip of constant width, is that we need this region of analyticity to be invariant under rescaling $z \mapsto z / \beta$ for $\beta>0$. Note that $\theta$ could be very small; we just require that it is a positive constant. We also require $Q(z)$ to be analytic in a disk of constant radius around $z = 0$. Handling the fact that this disk is \textit{not} invariant under rescaling is one of the technical achievements of our work. Note that this analyticity requirement on $Q$ is not in general the same as requiring $\Phi(z)$ to be analytic at $z=0$, since we are explicitly factorizing out a low-frequency power law, c.f.~\cref{eq:spectral_function_def}, which is supposed to capture the main non-analyticity of $\Phi(z)$.

Typically, we are more used to talking about analytic continuations of Green's functions than of spectral functions. Suppose we have a retarded Green's function $G_{R}(z)$, defined in the lower-half plane by
\begin{equation}
    G_{R}(z) \coloneqq \left(O_{0}\left|\dfrac{1}{z-\mathcal{L}}\right|O_{0}\right) = \int_{\mathbb{R}} \dfrac{\Phi(\omega)}{z-\omega} \dfrac{\diff \omega}{2\pi}, \quad z \in \mathbb{C}_{-}.
\end{equation}
By construction $G_{R}(z)$ is analytic in the lower-half plane, and $\Phi(\omega) = 2 \Im{\lim_{\epsilon \to 0^{+}} G_{R}(\omega - i \epsilon)}$. Now suppose that $G_{R}(z)$ has an analytic continuation across a section of the real line to some subset $\Omega \subseteq \mathbb{C}_{+} \cup \mathbb{R}$ of the upper-half plane. We can use this to define an analytic continuation of the spectral function by
\begin{equation}
    \Phi(z) \coloneqq \dfrac{1}{i} \left(G_{R}(z) + G_{R}(-z)\right).
\end{equation}
$\Phi(z)$ is then analytic in $\Omega \cup (-\Omega)$, and satisfies $\Phi(-z) = \Phi(z)$. So, for example, if $G_{R}(z)$ has a diffusive pole at $z = i D k^{2}$, then $\Phi(z)$ would have poles at $z = \pm i D k^{2}$. Having analytically continued $\Phi(z)$ in this way, $Q(z)$ would then be defined such that \cref{eq:spectral_function_def} is satisfied. In the case $\rho=0$, this can be achieved by taking
\begin{equation}
    Q(z) = -\log\left(\dfrac{\Phi(0)}{2\pi}\right) - \int_{0}^{z} \dfrac{\Phi^{\prime}(s)}{\Phi(s)} \diff s,
    \label{eq:Q_integral_def}
\end{equation}
where the integral is along any path where $\Phi(s) > 0$ and avoiding any poles. Then $\exp[-Q(z)] = \Phi(z)/2\pi$ by construction. A similar construction is possible for $\rho \neq 0$ given an analytic continuation of $|\omega|^{\rho}$ to the relevant quadrant of the complex plane.

Note that we do \textit{not} show that the violation of these analyticity conditions on $Q$ necessarily leads to a modification of our results. Versions of some of our results have been proven, though not always with explicit constants, only using assumptions about the behavior of the spectral function on the real line, see e.g.~\cite{levinOrthogonalPolynomialsExponential2012,kasugaOrthonormalPolynomialsGeneralized2003}. Rather than using Riemann-Hilbert techniques, which rely quite crucially on analyticity, it may also be possible to prove similar statements using a $\overline{\partial}$ steepest descent method~\cite{mclaughlinSteepestDescentMethod2008}, which does not require such strong analyticity assumptions on $Q$, at the expense of getting weaker error bounds.

One potential violation of our analyticity conditions could come from systems with a `diffuson cascade'~\cite{delacretazHeavyOperatorsHydrodynamic2020,rajDiffusionCascadeModel2024}, which can lead to an infinite series of poles in $k$-space retarded Green's functions $G_{R}(z,k)$ accumulating all the way down to the real $z$ axis, as opposed to the conventional expectation of just a simple pole at $z = i D k^{2}$~\cite{forsterHydrodynamicFluctuationsBroken2019}. Handling such singular points would require a modification of our analysis which goes beyond the scope of this work (though see \cite{kriecherbauerStrongAsymptoticsPolynomials1999} for some work in this direction). However, despite their analyticity conditions, our calculations give certain predictions that can be tested numerically. We find that they are obeyed well even in some interacting non-integrable systems that might be expected to have a diffuson cascade. This suggests that, even if we weaken our analyticity conditions, some of our conclusions may continue to hold.

\section{Coulomb gas, universality, and confinement}
\label{sec:coulomb}
To state our later results about hydrodynamics and Lanczos dynamics, we first need to introduce an associated Coulomb gas problem. As outlined in \cref{sec:synopsis_coulomb}, this Coulomb gas arises naturally both in the context of random matrix theory and orthogonal polynomials, and finding the minimum energy Coulomb gas configuration will give us a handle on the regions which dominantly contribute to the Lanczos coefficients. As we alluded to in \cref{sec:synopsis_coulomb}, one interesting aspect of this Coulomb gas is that it undergoes a \textit{confinement transition}, and that the Operator Growth Hypothesis implies that quantum chaotic systems are generically at the critical point of this transition. We describe this in more detail in \cref{sec:coulomb_confinement}.

\subsection{Coulomb gas definition}
\label{sec:coulomb_definition}
Given a potential $Q(\omega)$ determined from the spectral function by \cref{eq:spectral_function_def}, we define a Coulomb gas energy functional $E_{Q}$ as follows. We take as input a charge density function $\sigma : \mathbb{R} \to \mathbb{R}_{\geq 0}$ (defined on frequency space), and give it an energy consisting of two terms,
\begin{align}
    E_{Q}[\sigma] \coloneqq &- \int_{\mathbb{R}} \int_{\mathbb{R}} \log{|\omega_{1}-\omega_{2}|} \, \sigma(\omega_{1}) \sigma(\omega_{2}) \diff \omega_{1} \diff \omega_{2}\\
    &+ \int_{\mathbb{R}} Q(\omega) \sigma(\omega)\diff \omega,\nonumber
\end{align}
namely a logarithmic repulsion between charges, and a single-particle energy determined by the potential $Q$. For a given Lanczos index $n$, we imagine distributions with total charge $n$, so that $\int_{\mathbb{R}} \sigma(\omega) \diff \omega = n$. We then want to consider the distribution $\sigma_{n}$ with charge $n$ that minimizes the energy $E_{Q}[\sigma]$, so that
\begin{equation}
    \sigma_{n} \coloneqq \underset{\sigma}{\argmin}\left\{E_{Q}[\sigma] : \int_{\mathbb{R}} \sigma(\omega) \diff \omega = n \right\}.
    \label{eq:coulomb_gas_def}
\end{equation}
We will refer to $\sigma_{n}$ as the \textit{equilibrium density}. For large enough $n$, one can show that $\sigma_{n}$ is indeed uniquely defined and continuous under reasonable assumptions on $Q$~\cite{saffLogarithmicPotentialsExternal1997}. Let us make a few observations about this equilibrium density. Since the potential $Q(x)$ must grow at least linearly for $|x| \to \infty$ due to constraints from locality, while the two-particle repulsion is only logarithmic, we conclude that for any finite $n$, the charge will remain at a finite distance from the origin. This means that $\sigma_{n}(\omega)$ has support within some finite interval we denote by
\begin{equation}
    \supp{\sigma_{n}} \equiv (-\beta_{n}, \beta_{n}),
    \label{eq:sigman_support}
\end{equation}
which we refer to as the `bulk'. This distribution is even, $\sigma_{n}(-\omega) = \sigma_{n}(\omega)$, because $Q$ is even. As $n \to \infty$ the density $\sigma_{n}(\omega)$ gives the density of zeros of the orthogonal polynomials with respect to the weight $e^{-Q(\omega)}$~\cite{saffLogarithmicPotentialsExternal1997}, and the dominant contribution to the recurrence coefficients will come from this bulk frequency range $(-\beta_{n}, \beta_{n})$. This interval is analogous to the oscillatory region in a WKB approximation~\cite{deiftStrongAsymptoticsOrthogonal1999,hallQuantumTheoryMathematicians2013}. Note that, for technical convenience, we do not include the power law $|\omega|^{\rho}$ in the Coulomb gas potential, instead handling it by other means (Szeg\H{o} functions).

The distance $\beta_{n}$ is referred to as the $n$th \textit{Mhaskhar-Rakhmanov-Saff} (MRS) number~\cite{mhaskarExtremalProblemsPolynomials1984,rakhmanovAsymptoticPropertiesPolynomials1984,saffLogarithmicPotentialsExternal1997}, and is defined for even $Q$ as the positive solution to the integral equation
\begin{equation}
    \dfrac{1}{2\pi} \int_{-1}^{1} \dfrac{\beta_{n} s Q^{\prime}(\beta_{n} s)}{\sqrt{1-s^{2}}} \diff s = n.
    \label{eq:MRS_def}
\end{equation}
This solution is unique for large enough $n$. How does $\beta_{n}$ scale with $n$? Clearly, the slower the growth of the single-particle potential $Q$, the more the charge will spread out. It turns out that the growth of $Q$ and $\beta_{n}$ are related by
\begin{align}
    Q(\omega) &\sim |\omega|^{p} \log^{q}{|\omega|} \text{ as } |\omega| \to \infty \nonumber \\
    &{} \hspace{3em} \Updownarrow \label{eq:beta_n_asymptotics}\\
    \beta_{n} &\sim \left(\dfrac{n}{\log^{q}{n}}\right)^{1/p} \text{ as } n \to \infty. \nonumber
\end{align}
It is straightforward to verify this scaling when $Q$ is literally a polynomial, e.g.~$Q(\omega) = \omega^{2m}$, but it continues to hold when $Q(\omega)$ is sufficiently `polynomial-like' as $|\omega| \to \infty$ (see \cref{lem:vslf_characterization}). The locality constraint on $Q$ translates into $p \geq 1$, so that $\beta_{n}$ grows at most linearly in local systems. It is well known, as we will see in \cref{thm:recurrence_theorem}, that the scaling of $\beta_{n}$ fixes the leading behavior of the recurrence coefficients $b_{n}$, which are given by $b_{n} \approx \beta_{n} / 2$ to leading order in $n$~\cite{lubinskyProofFreudConjecture1988}. Note that by construction $\beta_{n}$ depends only on $Q^{\prime}$ and so is independent of the hydrodynamic exponent $\rho$, so the hydrodynamics will only show up at \textit{subleading} orders in the recurrence coefficients.
\begin{figure}[t]
    \centering
    \includegraphics[width=\columnwidth]{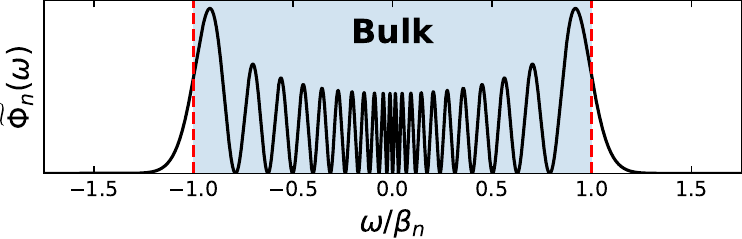}
    \caption{The $n$th spectral function $\wt{\Phi}_{n}(\omega) = p_{n}(\omega)^{2} \Phi(\omega)$ for the toy spectral function $\Phi(\omega) / 2\pi = \sech(\pi \omega)$ and $n=20$. This governs the dynamics of $O_{n}$ with respect to the \textit{full} Liouvillian $\mathcal{L}$, not the restricted Liouvillian $\mathcal{L}_{n}$ (see \cref{fig:level_n_gf_example} for the latter). In this case $\beta_{n} \approx \sqrt{n(n-1)}$~\cite{chenAsymptoticsExtremeZeros1997}, with the linear scaling $\beta_{n} \sim \mathcal{O}(n)$ reflecting the exponential decay of $\Phi(\omega \to \infty)$. One can see that $\wt{\Phi}_{n}(\omega)$ is peaked near $|\omega| = \beta_{n}$, and as $n \to \infty$ this becomes increasingly sharp. The region $|\omega| \leq \beta_{n}$ is referred to as the `bulk', and is analogous to the oscillatory region in a WKB approximation. For $|\omega| \gg \beta_{n}$, $\wt{\Phi}_{n}(\omega)$ is exponentially small. When $\Phi(\omega \to 0) \sim |\omega|^{\rho}$ has a power-law at $\omega = 0$, the orthogonal polynomials $p_{n}(\omega)$ behave differently near the origin (see e.g.~\cref{lem:pn0_scaling}).}
    \label{fig:toy_weighted_poly}
\end{figure}

One way to understand the MRS number $\beta_{n}$ is as an $n$-dependent frequency bandwidth. From \cref{eq:MRS_def} it is easy to see that if we introduce a frequency scale $\omega_{0}$ by mapping $Q(\omega) \mapsto Q(\omega / \omega_{0})$, then $\beta_{n}$ transforms as $\beta_{n} \mapsto \omega_{0} \beta_{n}$, so $\beta_{n}$ indeed has units of frequency. Furthermore, if we consider the spectral function $\wt{\Phi}_{n}(\omega) \coloneqq \int_{\mathbb{R}} e^{-i \omega t} (O_{n} | e^{i \mathcal{L}t} | O_{n}) \diff t$ for the dynamics of the $n$th Lanczos operator $O_{n}$ under the full Liouvillian $\mathcal{L}$, then from \cref{eq:lanczos_vector} we can see that this is related to the original spectral function by
\begin{equation}
    \wt{\Phi}_{n}(\omega) = p_{n}(\omega)^{2} \Phi(\omega).
\end{equation}
(Note $\wt{\Phi}_{n}(\omega)$ is not to be confused with the spectral function $\Phi_{n}(\omega)$ related to $G_{n}(z)$; they correspond to dynamics generated by $\mathcal{L}$ and $\mathcal{L}_{n}$ respectively). A remarkable result due to Mhaskar and Saff~\cite{mhaskarWhereDoesSup1985} shows that, under certain conditions on $Q$, the MRS number $\beta_{n}$ provides an asymptotically sharp characterization of where this weighted polynomial is peaked, so $\argmax_{\omega\geq 0} \wt{\Phi}_{n}(\omega) \xrightarrow{n \to \infty} \beta_{n}$. One can understand the scaling in \cref{eq:beta_n_asymptotics} as coming from the tradeoff between the growth of the degree $2n$ polynomial $p_{n}(\omega)^{2}$ and the decay of the spectral function $\Phi(\omega \to \infty) \sim \exp[-\mathcal{O}(\omega^{p} \log^{q}{\omega})]$. Physically, $\beta_{n}$ therefore represents the frequency bandwidth over which the $n$th spectral function $\wt{\Phi}_{n}(\omega)$ is non-negligible. For frequencies $|\omega| \gg \beta_{n}$, $\wt{\Phi}_{n}(\omega)$ will be exponentially small (see \cref{fig:toy_weighted_poly} for an example); this region is analogous to the `classically forbidden' region in a WKB analysis, and will give only exponentially small contributions to the Lanczos coefficients. %

\subsection{Rescaled equilibrium measure}
\label{sec:coulomb_rescaled}
Given a minimizing density distribution $\sigma_{n}$, it is helpful to define a related distribution
\begin{equation}
    \psi_{n}(x) \coloneqq \dfrac{\beta_{n}}{n} \sigma_{n}(\beta_{n}x),
    \label{eq:psin_rescaling}
\end{equation}
which is normalized to 1 and has support $[-1,1]$. We refer to $\psi_{n}$ as the rescaled equilibrium density. It is the energy minimizing distribution with charge 1 for the `rescaled potential'
\begin{equation}
    V_{n}(x) \coloneqq \dfrac{1}{n} Q(\beta_{n} x).
\end{equation}
The rescaling by $\beta_{n}$ in the definition of $V_{n}$ washes out the non-universal details of $Q$ when $n$ is large. For example, for the class of potentials we consider, if $Q(x) \sim x^{p}$ as $x \to \infty$, then for any finite $x$ we have
\begin{equation}
    \lim_{n \to \infty} V_{n}(x) = \kappa_{p} |x|^{p}, 
    \label{eq:limit_to_freud}
\end{equation}
where $\kappa_{p}=\Gamma[\frac{1}{2}]\Gamma[\frac{p}{2}]/\Gamma[\frac{p+1}{2}]$~\cite[Lemma 3.2]{lubinskyUniformMeanApproximation1988}. This emergent dominance of the high frequency scaling means that the equilibrium density $\psi_{n}(x)$ has several displays of \textit{universality} in the $n \to \infty$ limit, with the same properties as if the potential were the corresponding `Freud potential' $Q^{(p)}(x) \equiv \kappa_{p} |x|^{p}$. These potentials have $\beta_{n} = n^{1/p}$, and an equilibrium measure given by the `Ullman distribution'~\cite{saffLogarithmicPotentialsExternal1997}
\begin{equation}
    \psi^{(p)}(x) \coloneqq \dfrac{1}{\pi} \int_{|x|}^{1} \dfrac{p u^{p-1}}{\sqrt{u^{2}-x^{2}}} \diff u.
    \label{eq:freud_eq_measure}
\end{equation}
One simple manifestation of this is that, quite generally, the density $\psi_{n}(x)$ vanishes like $\sqrt{1-x^{2}}$ at the endpoints $x=\pm 1$; this has a manifestation in the famous Wigner semicircle law of random matrix theory~\cite{mehtaRandomMatrices2004}, but holds near the edge more generally beyond Gaussian ensembles. Other universal properties of the equilibrium Coulomb gas distribution are more easily stated in terms of a function $h_{n}(x)$ defined by
\begin{equation}
    \psi_{n}(x) \eqqcolon \dfrac{1}{2\pi} h_{n}(x) \sqrt{1-x^{2}},
    \label{eq:psin_hn}
\end{equation}
for $x \in [-1,1]$. The values of $h_{n}(x)$ at $x=0$ and $x=1$ appeared in \cref{thm:recurrence_theorem} describing hydrodynamic corrections to the Lanczos coefficients.

For large enough $n$ (but $\mathcal{O}(1)$ in terms of the microscopic couplings), we show that the minimum energy configuration of $E_{Q}$ is obtained when $h_{n}(x)$ is given by
\begin{equation}
    h_{n}(x) = \dfrac{1}{\pi} \int_{-1}^{1} \dfrac{V_{n}^{\prime}(s) - V_{n}^{\prime}(x)}{s-x} \dfrac{\diff s}{\sqrt{1 - s^{2}}}.
    \label{eq:hn_integral}
\end{equation}
We then show that $h_{n}(x)$ displays universal behavior at the origin $x = 0$ and the endpoints $x = \pm 1$, similar to that of the Ullman distribution \cref{eq:freud_eq_measure}. This behavior depends only on the $|\omega| \to \infty$ behavior of $Q$, provided it is sufficiently smooth near the origin, essentially because of the large-$n$ scaling of $V_{n}(x)$ given by \cref{eq:limit_to_freud}.

As an example, for our class of spectral functions, at the edge of the bulk we have:
\begin{lemma}[Informal]
    If $Q(\omega) \sim |\omega|^{p}\log^{q}{|\omega|}$ as $|\omega| \to \infty$ for some $p > 0$ and $q \in \mathbb{R}$, then
    \begin{equation}
        \lim_{n \to \infty} h_{n}(1) = 2p.
    \end{equation}
    \label{lem:hn1_scaling}
\end{lemma}
We test this result numerically for the mixed field Ising model in \cref{sec:airy_bootstrap}. Note that $h_{n}(x)$ is even if $Q(x)$ is even, so the same conclusion holds for $h_{n}(-1)$. 

\subsection{Confinement transition}
\label{sec:coulomb_confinement}
The behavior of $h_{n}(0)$ is more interesting, since it is sensitive to a \textit{confinement transition} in the Coulomb gas, between `strong confinement' for $p > 1$ to `weak confinement' for $p < 1$ (see \cref{fig:confinement_schematic})~\cite{canaliNonuniversalityRandommatrixEnsembles1995,freilikherUnitaryRandommatrixEnsemble1996,freilikherTheoryRandomMatrices1996,claeysWeakStrongConfinement2023}. Despite this confinement transition being driven primarily by the \textit{high frequency} behavior of the spectral function, it turns out to manifest in a divergent density at \textit{low frequencies}, as codified in \cref{lem:hn0_scaling}. This means it will have an imprint on the signatures of hydrodynamics on the recurrence coefficients $b_{n}$, as we saw in \cref{thm:recurrence_theorem}. The following result is stated in terms of $h_{n}(0)$, but can be translated back to a confinement transition in the original equilibrium density $\sigma_{n}(0)$ using \cref{eq:psin_rescaling,eq:psin_hn}. 
\begin{lemma}[Informal]
    Suppose $Q(\omega) \sim |\omega|^{p}\log^{q}{|\omega|}$ as $|\omega| \to \infty$ for some $p \geq 1$ and $q \in \mathbb{R}$.

    If $p > 1$, then for all $q \in \mathbb{R}$ we have
    \begin{equation}
        \lim_{n \to \infty} h_{n}(0) = \dfrac{2p}{p-1},
    \end{equation}
    while in the marginal case $p = 1$, as $n \to \infty$ we have
    \begin{equation}
        h_{n}(0) = (\log{n})^{1 + o(1)}
    \end{equation}
    provided $q > -1$. (Local interactions enforce $q \geq 0$.)
    \label{lem:hn0_scaling}
\end{lemma}
We prove \cref{lem:hn1_scaling,lem:hn0_scaling} in \cref{sec:endpoint_analysis}. It is interesting to now recall the \textit{operator growth hypothesis} (OGH)~\cite{parkerUniversalOperatorGrowth2019}, which posits that operators undergoing chaotic many-body dynamics generically grow as fast as possible under locality constraints. These constraints stipulate that the spectral function must decay at least exponentially with $\omega$, so that $Q(\omega) \sim |\omega|^{p}$ with $p \geq 1$, and hence $b_{n}$ grows no faster than $\mathcal{O}(n)$. The OGH therefore amounts to the statement that chaotic many-body systems generically have $p=1$; here we see that this implies they can be identified with the critical point of a Coulomb gas confinement transition (see \cref{fig:confinement_schematic}). Systems with spectral functions decaying superexponentially ($p>1$) will instead lie in the strongly confined phase of the Coulomb gas. 

\begin{figure}[t]
    \centering
    \includegraphics[width=0.9\columnwidth]{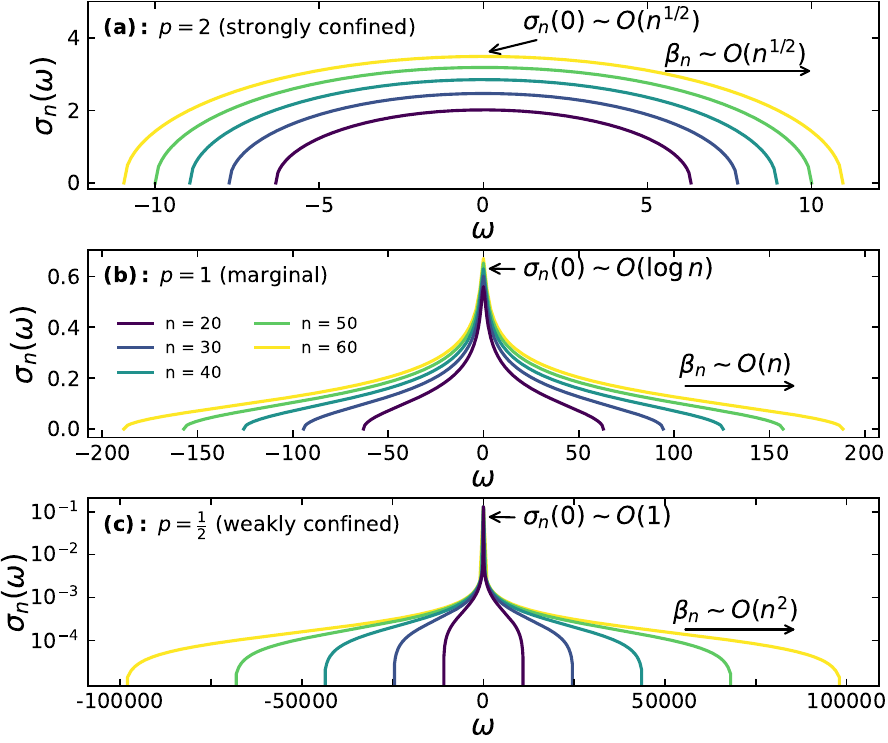}
    \caption{Illustration of the confinement transition using the potential $Q(\omega) = (1 + \omega^{2})^{p/2}$ for different growth exponents $p$. The values $p=\frac{1}{2}$ and $p=2$ lie in the weakly and strongly confined phases respectively, while $p=1$ is marginal. All systems with local interactions should have $p \geq 1$. The confinement transition can be diagnosed via the equilibrium density $\sigma_{n}(0)$ at zero frequency. In the strongly confined phase $\sigma_{n}(0) \sim \mathcal{O}(n/\beta_{n})$ grows algebraically with $n$, reducing to logarithmic growth $\sigma_{n}(0) \sim \mathcal{O}(\log{n})$ at the critical point (or $\sigma_{n}(0) \sim \mathcal{O}(\log^{2}{n})$ in one spatial dimension). In the weakly confined phase, $\sigma_{n}(0) \sim \mathcal{O}(1)$ does not grow with $n$.}
    \label{fig:confinement_example}
\end{figure}
The confinement transition can be heuristically understood as follows. We have already seen that the large frequency scaling of the potential $Q(\omega) \sim |\omega|^{p} \log^{q}{|\omega|}$ determines the width $\beta_{n}$ of the equilibrium charge density distribution $\sigma_{n}(\omega)$ via \cref{eq:beta_n_asymptotics}. The confinement transition is related to the relative scaling of the bandwidth $\beta_{n}$ and the total charge $n$. For $p>1$, $\beta_{n}$ grows slower than $n$, so the average charge density $n/(2\beta_{n})$ increases with $n$; heuristically, the charge is increasingly `packed in'. A useful diagnostic is the charge density at zero frequency, $\sigma_{n}(\omega=0)$, which is related to $h_{n}(0)$ via
\begin{equation}
    \sigma_{n}(0) = \dfrac{1}{2\pi} \frac{n}{\beta_{n}} h_{n}(0).
\end{equation}
\cref{lem:hn0_scaling} tells us that $h_{n}(0)$ is $\mathcal{O}(1)$ for $p>1$, where the $\mathcal{O}(1)$ refers to scaling with $n$. Combining this with the scaling of $\beta_{n}$ given in \cref{eq:beta_n_asymptotics}, we conclude that $\sigma_{n}(0) \sim \mathcal{O}(n/\beta_{n})$ increases algebraically with $n$, with the same $n$-scaling as the average charge density. In this sense the charge density is approximately uniform (in $\omega$) in the bulk. An example is shown in \cref{fig:confinement_example}(a), where we show the equilibrium density $\sigma_{n}(\omega)$ for the toy potential $Q(\omega) = (1 + \omega^{2})^{p/2}$ with $p=2$.

By contrast, in the (locality forbidden) weakly confined phase ($p<1$), the bandwidth $\beta_{n}$ grows faster than $n$, so the average charge density $n/(2\beta_{n})$ decreases with $n$. However, the charge is now distributed differently at low frequencies compared with the rest of the bulk; the density profile is peaked at $\omega=0$, where it becomes $\sigma_{n}(0) \sim \mathcal{O}(1)$. An example is shown in \cref{fig:confinement_example}(c). One can understand this divergence from the behavior of the derivative $Q^{\prime}(\omega) \sim \omega^{p-1}$, which for $p<1$ grows algebraically as $\omega$ becomes smaller. For our class of spectral functions we only assume that this behavior holds at large frequencies, but when $n$ is large this scaling still governs the behavior of the rescaled potential $V_{n}^{\prime}(s) = (\beta_{n}/n) Q^{\prime}(\beta_{n}s)$ for $s \sim \mathcal{O}(1)$ (c.f.~\cref{eq:limit_to_freud}). As such, for $x=0$ the integrand of \cref{eq:hn_integral} diverges algebraically as $s$ approaches zero, until the divergence is eventually cut off by our differentiability assumption on $Q$. This gives rise to a contribution scaling like $h_{n}(0) \sim \mathcal{O}(n^{(1-p)/p})$, so that $\sigma_{n}(0) \sim \mathcal{O}(1)$.

Finally, in the marginal case $p=1$ relevant for chaotic systems, the bandwidth $\beta_{n}$ grows at the same rate as the total charge $n$ (up to logarithmic factors depending on $q$). By a similar argument to the $p<1$ case, the derivative $V_{n}^{\prime}(s)$ is now approximately constant for $s \sim \mathcal{O}(1)$, so the integrand of \cref{eq:hn_integral} has a logarithmic divergence as $s \to 0$. This is again cut off at very small $s\sim\mathcal{O}(1/\beta_{n})$ by our differentiability assumption on $Q$, with the end result that $h_{n}(0)$ diverges like $\mathcal{O}(\log{n})$. Any logarithmic corrections to $Q(\omega) \sim |\omega| \log^{q}{|\omega|}$ do not affect the logarithmic scaling of $h_{n}(0)$ provided $q > -1$. The case $q=-1$ corresponds to the transition to weak confinement; this is also where the Hamburger moment problem for $\Phi(\omega)$ becomes indeterminate~\cite{shohatProblemMoments1943}, which follows from Carleman's condition~\cite{akhiezerClassicalMomentProblem1965}. Note that local interactions enforce $q \geq 0$ if $p=1$~\cite{abaninExponentiallySlowHeating2015}. The scaling of $\sigma_{n}(0)$ is then $\sigma_{n}(0) \sim \mathcal{O}(n h_{n}(0)/\beta_{n})$, giving $\sigma_{n}(0) \sim \mathcal{O}(\log{n})$ for $\beta_{n} \sim n$ and $\sigma_{n}(0) \sim \mathcal{O}(\log^{2}{n})$ for $\beta_{n} \sim n / \log{n}$, the latter case relevant in one spatial dimension~\cite{parkerUniversalOperatorGrowth2019}.

Before proceeding, we note that the convergence with $n$ of the equilibrium Coulomb gas density to the asymptotic forms indicated in \cref{lem:hn1_scaling,lem:hn0_scaling} can be extremely slow, depending on the functional form for $Q$. When $Q$ is an even-order polynomial, $Q(x) \sim \mathcal{O}(x^{2m})$, then it is known that the finite $n$ corrections to $h_{n}(0)$ and $h_{n}(1)$ scale like $\mathcal{O}(n^{-1/2m})$~\cite{deiftStrongAsymptoticsOrthogonal1999}. When $Q$ is not necessarily an actual polynomial, but merely of polynomial-growth as we treat here, then in general the rate of convergence depends on how quickly $Q$ approaches its asymptotic scaling $Q(x) \sim x^{p} \log^{q}{x}$. In toy examples we have also found that convergence of $h_{n}(0)$ tends to be slower than that of $h_{n}(1)$, particularly in the quasi-linear case $p=1$ where $h_{n}(0)$ has a logarithmic divergence. Having said that, this slow convergence will not necessarily be a problem for estimating e.g.~diffusion constants: in \cref{sec:transport_coeffs} we show how to eliminate the dependence on these slowly converging Coulomb gas densities. After developing some technical tools, we perform some quantitative checks of the convergence of the equilibrium measure to the Ullman distribution \cref{eq:freud_eq_measure} in \cref{sec:eq_measure_scaling}. We find, at the values of $n$ typically available for numerical simulations, that the Ullman distribution successfully describes the large-scale qualitative shape of the true equilibrium measure, but at finite $n$ there remain significant fluctuations which must be accounted for in order to get quantitative accuracy in estimates of the spectral function.

\section{Universality of the level-$n$ Green's function}
\label{sec:gf}
In this section we will explore the emergent universality in Green's functions outlined in \cref{sec:synopsis_emergent_universality}. 
Let $\mathcal{L}_{n} = (\mathcal{L}_{ij})_{i, j \geq n}$ denote the tridiagonal Liouvillian matrix restricted to sites $n$ and above, and let 
\begin{equation}
    G_{n}(z) \coloneqq \left(O_{n}\left|\dfrac{1}{z-\mathcal{L}_{n}}\right|O_{n}\right),
    \label{eq:greens_func_def}
\end{equation}
which we refer to as the `level-$n$ Green's function'. In the language of the memory function formalism, one should think of this as the Green's function for dynamics restricted to the `fast space' spanned by $\{O_{m}\}_{m\geq n}$.

Due to the tridiagonal structure of the Liouvillian, $G_{n}(z)$ obeys the recursion relation~\cite{viswanathRecursionMethodApplication2013}
\begin{equation}
    G_{n}(z) = \dfrac{1}{z - b_{n+1}^{2} G_{n+1}(z)}.
    \label{eq:greens_func_recursion}
\end{equation}
Recursing this gives a continued fraction expansion for the original Green's function $G(z) \equiv G_{0}(z)$ (taking $\norm{A} = 1$ for notational simplicity)
\begin{equation}
    G(z) = \dfrac{1}{z - \dfrac{b_{1}^{2}}{z - \cdots \underset{\cdots - \dfrac{b_{n-1}^{2}}{z - b_{n}^{2} G_{n}(z)}}{}}}
    \label{eq:greens_func_continued_fraction}
\end{equation}
In practice one can only compute some finite number of coefficients $\{b_{k}\}_{k=1}^{n}$, so one must somehow terminate the continued fraction expansion at level-$n$. Simply setting $G_{n}(z) = 0$ is a bad idea because it amounts to terminating the Lanczos chain after site $n$, which gives unphysical reflections of the operator wavefunction off the hard boundary, leading to very slow convergence with $n$. A better approach is to choose a model for $G_{n}(z)$ which is designed to accurately capture the operator backflow. The `recursion method'~\cite{viswanathRecursionMethodApplication2013} is a popular numerical technique where one chooses such a `terminator' level-$n$ Green's function based on some high level features of the spectral function, like its high-frequency decay and any algebraic singularities. The existence of appropriate terminators is something of a lottery because one needs exact expressions for all three of the Lanczos coefficients, the spectral function, and the Green's function. These solutions are usually expressed in terms of special functions, and there is no guarantee an exactly solvable model will exist with all the appropriate characteristics. With our Riemann-Hilbert approach, we can analyze the $n \to \infty$ asymptotic behavior of $G_{n}(z)$ directly, bypassing the need for these exactly solvable models.

We prove that, as $n \to \infty$, $G_{n}(z)$ approaches different universal scaling forms in different regions of the complex plane. The simplest case is when $z$ is sufficiently far from the special points $z=0$ and $z=\pm \beta_{n}$. In this case we show that $G_{n}(z)$ approaches the Wigner semicircle law, the same as the average global resolvent $\frac{1}{n} \overline{\tr{1/(z-M)}}$ for random $n\times n$ matrices $M$ drawn from the Gaussian Unitary Ensemble~\cite{taoTopicsRandomMatrix2012} with a suitably rescaled bandwidth. If the spectral function $\Phi(\omega)$ is complex analytic at $\omega=0$, this semicircle behavior persists all the way down to $z=0$. On the other hand, if $\Phi(\omega \to 0) \sim |\omega|^{\rho}$ has a low-frequency power law, then the semicircle law breaks down near $\omega=0$, and $G_{n}(z)$ instead is described by a form dictated by the Bessel universality class. Finally, near the edge of the spectrum, $z \approx \pm \beta_{n}$, the behavior of $G_{n}(z)$ is described by the Airy universality class. We show an example of this in \cref{fig:level_n_gf_example} for the toy spectral function $\Phi(\omega) = |\omega|^{-1/2} \sech(\pi \omega)$, where we focus on the behavior along the real line, as encoded in the corresponding spectral function $\Phi_{n}(\omega) = 2 \Im[G_{n}(\omega -i 0^{+})]$. Since in this case we know the spectral function $\Phi(\omega)$, we can numerically compute $G_{n}(z)$ using a formula expressing it as a ratio of two Cauchy-Stieltjes transforms (see \cref{eq:Gn_from_cauchy}).

\begin{figure}[t]
    \centering
    \includegraphics[width=\columnwidth]{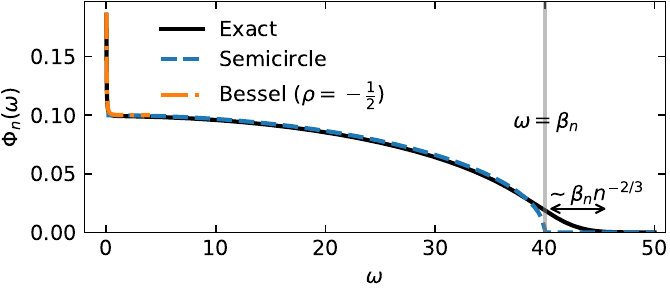}
    \caption{The spectral function $\Phi_{n}(\omega) = 2 \Im[ G_{n}(\omega - i 0^{+})]$ corresponding to the level-$n$ Green's function $G_{n}(z)$, shown for the toy spectral function $\Phi(\omega) = |\omega|^{-1/2} \sech(\pi \omega)$ and $n=40$. The Wigner semicircle law \cref{eq:Phi_n_semicircle} describes the behavior very well in the bulk, but breaks down near $\omega=0$ due to the low frequency power-law in the spectral function. Instead the behavior near $\omega=0$ is described by the universal Bessel form corresponding to \cref{eq:n_gf_asymptotic}. Near the endpoint $\omega=\beta_{n}$ the behavior is described by Airy universality.}
    \label{fig:level_n_gf_example}
\end{figure}

\subsection{Frequencies in the bulk}
First let us consider frequencies sufficiently far from the special points $z=0$ and $z = \pm \beta_{n}$, where the level-$n$ Green's function $G_{n}(z)$ takes on a particularly simple form as $n \to \infty$. In \cref{sec:greens_func_derivation} we show that, as $n \to \infty$, $\beta_{n} G_{n}(\beta_{n} z)$ approaches a universal scaling form
\begin{equation}
    \boxed{\beta_{n} G_{n}(\beta_{n} z) \approx 2 \left(z -  \sqrt{z+1}\sqrt{z-1}\right),}
    \label{eq:greens_func_bulk}
\end{equation}
for all $z$ satisfying $|z| > \delta_{0}$ and $|z/\beta_{n} \pm 1| > \delta_{1}$, where $\delta_{0}$ and $\delta_{1}$ are small but $\mathcal{O}(1)$ constants (see \cref{fig:plancherel_schematic}), and we use the principal branch of the square root, so the RHS has a branch cut along $[-1,1]$. Thus, away from the special points $z=0$ and $z=\pm \beta_{n}$, the only free parameter determining $G_{n}(z)$ is the frequency bandwidth $\beta_{n}$, which can be approximated up to $\mathcal{O}(1/n)$ relative error by $\beta_{n} \approx 2 b_{n}$ (c.f.~\cref{thm:recurrence_theorem}). When the spectral function $\Phi(\omega)$ is complex analytic at $\omega=0$, we prove this semicircle behavior extends all the way to $z = 0$. One can then understand the special case \cref{eq:Phi0bn} of our formula for recovering $\Phi(\omega=0)$ from the Lanczos coefficients as being equivalent to substituting $G_{n}(\pm i 0^{+}) \approx \mp 2i/\beta_{n} \approx \mp i/b_{n}$ into the continued fraction expansion for $G(z)$. Although our proof technique in this case requires analyticity at $\omega=0$, our numerical tests in \cref{sec:transport_coeffs} suggest that this scaling $G_{n}(\pm i 0^{+}) \approx \mp 2i/\beta_{n}$ may still be accurate to leading order in $n$, even without such strong analyticity requirements.

Within the `bulk' frequency range $\delta_{0} < |\omega| < (1 - \delta_{1})\beta_{n}$, the spectral function corresponding to \cref{eq:greens_func_bulk} is
\begin{align}
    \Phi_{n}(\omega) &\coloneqq \int_{-\infty}^{\infty} e^{-i \omega t} (O_{n} | e^{i \mathcal{L}_{n} t} | O_{n}) \diff t \nonumber \\
    &= i\left[G_{n}(\omega + i 0^{+}) - G_{n}(\omega - i 0^{+})\right] \nonumber \\
    &\approx \dfrac{4}{\beta_{n}} \sqrt{1 - (\omega/\beta_{n})^{2}},
    \label{eq:Phi_n_semicircle}
\end{align}
namely the \textit{semicircle law} famous from random matrix theory (RMT)~\cite{taoTopicsRandomMatrix2012}. Thus we see the emergence of RMT universality in the bulk frequency response of the operators $O_{n}$ restricted to the `fast space'. To leading order in $\omega/\beta_{n} \ll 1$, this spectral function is constant, providing some justification for approximating the large-$n$ dynamics by white noise, as was the original intuition for the Mori-Zwanzig memory function formalism~\cite{viswanathRecursionMethodApplication2013,forsterHydrodynamicFluctuationsBroken2019}. In \cref{sec:gf_bessel} we will show how this can fail near $\omega=0$.

\subsubsection*{Semicircle law as a fixed point}
Although this is not how we prove it (see \cref{sec:greens_func_derivation} for the actual proof), a quick way to arrive at \cref{eq:greens_func_bulk} is to \textit{assume} the existence of the limit $\beta_{n} G_{n}(\beta_{n} z) \xrightarrow{n \to \infty} f(z)$. This is a nontrivial assumption. One way it can fail is that this limit is different for even and odd $n$---this is what happens near $z=0$ when $\rho\neq 0$. These concerns aside, if we assume the limit exists, then using the properties $\lim_{n \to \infty} b_{n} / \beta_{n} = 1/2$ (c.f.~\cref{thm:recurrence_theorem}) and $\lim_{n \to \infty} b_{n+1}/b_{n} = 1$, one concludes from \cref{eq:greens_func_recursion} that $f(z)$ must solve the fixed point equation
\begin{equation}
    f(z) = \dfrac{1}{z - \frac{1}{4}f(z)}.
\end{equation}
The RHS of \cref{eq:greens_func_bulk} is then the unique solution of this fixed point, subject to the requirement that $f(z)$ is separately analytic in $\mathbb{C}_{\pm}$ and approaches zero like $\mathcal{O}(1/z)$ as $z \to \infty$ (these properties follow from the resolvent representation of $G_{n}(z)$). We remark that this fixed point argument is analogous to a method of proving the free central limit theorem in random matrix theory, where the Wigner semicircle law emerges as the fixed point of repeated free convolutions~\cite{mingoFreeProbabilityRandom2017,erdosMatrixDysonEquation2019}. (Similar to the argument for the standard central limit theorem that the Gaussian distribution is the fixed point of repeated convolutions.)

Without going into the details of the actual proof, let us note that the origin of the semicircle law behavior for $G_{n}(z)$ seems somewhat different to its origin for Gaussian random matrix ensembles. For the latter, the semicircle law arises because the equilibrium measure of the Gaussian potential $Q(\omega) = -\log[\exp(-\omega^{2})] = \omega^{2}$ is the semicircle law (c.f.~\cref{eq:freud_eq_measure} for $p=2$). In that context, the equilibrium measure governs the average eigenvalue density as $n \to \infty$. However, here the semicircle law for $G_{n}(z)$ arises for a much wider class of potentials, even when the corresponding equilibrium measure does \textit{not} approach the semicircle law. The main ingredient to get the semicircle law for $G_{n}(z)$ seems to be that the equilibrium density $\sigma_{n}(\omega)$ is nonzero throughout the whole interval $(-\beta_{n}, \beta_{n})$, and that it vanishes like a square root as $\omega$ approaches the endpoints $\pm \beta_{n}$ (this is called `regular' behavior). We prove this is the case for large enough $n$ for our class of spectral functions (see \cref{lem:hn_positive_x_g1}), subject to the regularity condition that $Q^{\prime}(\omega) > 0$ for large enough $\omega$. In this sense the universal form \cref{eq:greens_func_bulk} seems quite generic, but it would be possible to violate it locally if e.g.~the spectral function contained spectral gaps.%

In the language of the recursion method~\cite{haydockRecursiveSolutionSchrodinger1980,viswanathRecursionMethodApplication2013}, the Green's function in \cref{eq:greens_func_bulk} is known as the `square root terminator', and is one of the simplest ways to terminate the infinite continued fraction for the true Green's function. Indeed, within the approximation $\beta_{n} \approx 2 b_{n}$, it amounts to setting all $b_{n+k} \mapsto b_{n}$, $k=0,1,2,\dots$, to be equal, so that $\mathcal{L}_{n}$ is just a constant tridiagonal matrix with $b_{n}$ on the off-diagonals. 
Our results show that, remarkably, this simple square root terminator is sufficient to accurately compute diffusion constants ($\rho=0$). In fact, we have even more freedom: any terminator satisfying $G_{n}(\omega \pm i 0^{+}) \approx \mp 2i / \beta_{n}$ as $\omega \to 0$ would work just as well. For example, this means that it is \textit{not} necessary to match the high frequency tail of the true spectral function (this is not generally true if $\rho \neq 0$---see \cref{sec:bessel_bootstrap}). In related work, we explore this `stitching freedom' in more detail~\cite{pinnaApproximationTheoryGreens2025}.

We discuss the error term of \cref{eq:greens_func_bulk} in \cref{sec:greens_func_derivation}. In brief, for $z=\mathcal{O}(1)$ there is a multiplicative error of $\mathcal{O}(1/n)$. For $\rho = 0$ this is true for all $z$, but if $\rho \neq 0$ the error can become larger as $z \to 0$, reaching $\mathcal{O}(1/\sigma_{n}(0))$ at $z \sim 1/\beta_{n}$ (corresponding to a physical frequency $\omega = \beta_{n}z$ of $\mathcal{O}(n^{0})$). This is $\wt{\mathcal{O}}(1/n^{(p-1)/p})$ for $p>1$, but only $\mathcal{O}(1/(\log{n})^{1+q+o(1)})$ in the marginal case $p=1$. (We remind the reader that, for $p=1$, $q\geq 0$ is guaranteed by locality, with $q\geq 1$ in $d=1$ spatial dimensions~\cite{abaninExponentiallySlowHeating2015,parkerUniversalOperatorGrowth2019}.)

\subsection{Frequencies near $\omega=0$}
\label{sec:gf_bessel}
In the limit $\omega \to 0$, the distinction between even and odd $n$ is important for $\rho \neq 0$; we will discuss the even case. In the order of limits where $\omega \to 0$ before $n \to \infty$, we can show
\begin{equation}
    \boxed{
    G_{2n}(\omega \pm i 0^{+}) \approx \dfrac{-2}{\beta_{2n}} \left(\dfrac{J_{\frac{1}{2}(\rho-1)} \pm i Y_{\frac{1}{2}(\rho-1)}}{J_{\frac{1}{2}(\rho+1)} \pm i Y_{\frac{1}{2}(\rho+1)}}\right)[\pi \sigma_{2n}(0) \omega],}
    \label{eq:n_gf_asymptotic}
\end{equation}
where $J_{\nu}$ and $Y_{\nu}$ are Bessel functions of the first and second kind respectively, and our shorthand indicates all Bessel functions should be evaluated with the argument $\pi \sigma_{2n}(0) \omega$ (this is a linear approximation of $\pi I_{2n}(\omega) = \pi \int_{0}^{\omega} \sigma_{2n}(s) \diff s \approx \pi \sigma_{2n}(0) \omega$). See \cref{sec:greens_func_derivation} for a derivation. The appearance of Bessel functions here is an indication of the Bessel universality governing the behavior near the origin due to the spectral function scaling $\Phi(\omega \to 0) \sim |\omega|^{\rho}$. In principle one can use this Bessel universality of $G_{n}(z)$ to estimate spectral functions near $\omega = 0$. However, in order to eliminate the \textit{a priori} unknown Coulomb gas density $\sigma_{2n}(0)$, we find it more convenient to follow an equivalent approach and work directly with the zero mode in Lanczos space. We will describe this approach in \cref{sec:hydro}.

\subsection{Frequencies near the edges $\omega=\pm \beta_{n}$}
\label{sec:airy_gn}
Near the edge of the spectrum, meaning $|z/\beta_{n} \pm 1| < \delta_{1}$ for a small $\mathcal{O}(1)$ constant $\delta_{1}$, the Wigner semicircle law for $G_{n}(z)$ again breaks down, and instead the behavior is governed by the Airy universality class. This is analogous to the behavior near the turning points of a WKB approximation~\cite{hallQuantumTheoryMathematicians2013}, or the maximal eigenvalue distribution of a random matrix ensemble~\cite{tracyLevelspacingDistributionsAiry1994}. We refer the reader to \cref{sec:greens_func_derivation} for the full expression for $G_{n}(z)$ in terms of Airy functions, and simply state here the leading behavior precisely at the endpoint $\beta_{n}$, where we have
\begin{equation}
    \boxed{
    G_{n,\pm}(\beta_{n}) \approx \dfrac{2}{\beta_{n}} \left(\dfrac{2\Ai(0) \sqrt{f_{n}^{\prime}(1)} + e^{\pm \frac{i \pi}{3}} \sqrt{2} \Ai^{\prime}(0) (\rho+1)}{2\Ai(0) \sqrt{f_{n}^{\prime}(1)} + e^{\pm \frac{i \pi}{3}} \sqrt{2} \Ai^{\prime}(0) (\rho-1)}\right),}
\end{equation}
with $f_{n}^{\prime}(1) = (n h_{n}(1)/\sqrt{2})^{2/3} \sim \mathcal{O}(n^{2/3})$, and where $\Ai(0) = (3^{2/3} \Gamma[\frac{2}{3}])^{-1}$ and $\Ai^{\prime}(0) = -(3^{1/3} \Gamma[\frac{1}{3}])^{-1}$. The appearance of the Airy function $\Ai$ is a reflection of the Airy universality near the spectral edge. From this expression we see that the corresponding spectral function $\wt{\Phi}_{n}(\omega) = \mp 2 \Im[G_{n,\pm}(\omega)]$ does not vanish precisely at $\omega=\beta_{n}$, as the semicircle law would predict, but instead scales to zero like $\mathcal{O}(\beta_{n}^{-1} n^{-1/3})$. As discussed in \cref{sec:greens_func_derivation}, this expression for $G_{n,\pm}(\beta_{n})$ has a multiplicative error of $\mathcal{O}(1/n)$.

\section{Hydrodynamic contributions to the Lanczos coefficients}
\label{sec:recurrence_coefficients}

Our next result, \cref{thm:recurrence_theorem}, is a statement about the recurrence coefficients $\{b_{n}\}$ associated with spectral functions of the form in \cref{eq:spectral_function_def}. It is an asymptotic statement, controlled in the limit $n \to \infty$ of large Lanczos number. This limit is analogous to the large matrix dimension limit in random matrix theory~\cite{deiftOrthogonalPolynomialsRandom2000}, and similarly will lead to universal behavior.

It is well established that the scaling of $\Phi(\omega)$ as $|\omega| \to \infty$ determines the \textit{leading} behavior of the $b_{n}$ as $n \to \infty$~\cite{lubinskyProofFreudConjecture1988}. What we show in \cref{thm:recurrence_theorem} is that the hydrodynamic power-law imprints itself on the \textit{subleading} behavior of $b_{n}$ in a universal manner. Some of the quantities appearing in the theorem, namely $\beta_{n}$ and $h_{n}(x)$, are defined in terms of a classical Coulomb gas problem described in \cref{sec:coulomb}.

\begin{mdframed}
\begin{theorem}[Informal]
    Consider $\Phi(\omega) / 2\pi \equiv |\omega|^{\rho} \exp[-Q(\omega)]$, where $\rho > -1$, and $Q$ is an even function scaling like $Q(\omega) \sim |\omega|^{p} \log^{q}{|\omega|}$ as $|\omega| \to \infty$, for some $p \geq 1$ and $q \in \mathbb{R}$. Then the recurrence coefficients $b_{n}$ associated with $\Phi(\omega)$ have a large-$n$ expansion given by
    \begin{equation}
        b_{n} = \dfrac{\beta_{n}}{2} \left[1 + \rho \left(\dfrac{1}{h_{n}(1)} - (-1)^{n} \dfrac{1}{h_{n}(0)}\right) \dfrac{1}{n} + \cdots \right],
    \end{equation}
    where the scaling of $\beta_{n}$ is given in \cref{eq:beta_n_asymptotics}, the scaling of $h_{n}(1)$ and $h_{n}(0)$ are given in \cref{lem:hn1_scaling,lem:hn0_scaling}, and the dots indicate terms subleading in $n$.

    The case $p=1$ of (quasi-)exponential decay---most relevant for generic chaotic systems according to the operator growth hypothesis~\cite{parkerUniversalOperatorGrowth2019}---turns out to be marginal, with the above expression simplifying to
    \begin{equation}
        b_{n} = \dfrac{\beta_{n}}{2} \left[ 1 + \dfrac{\rho}{2}\left(1 - (-1)^{n} \dfrac{1}{(\log{n})^{1+o(1)}} \right) \dfrac{1}{n} + \cdots\right],
        \label{eq:recurrence_thm_linear}
    \end{equation}
    provided $q > -1$. (Local interactions enforce $q\geq 0$.)

    If $p > 1$, then this expression instead simplifies to
    \begin{equation}
        b_{n} = \dfrac{\beta_{n}}{2} \left[ 1 + \dfrac{\rho}{2p} \Bigg(1 - (-1)^{n}(p-1)\Bigg)\dfrac{1}{n} + \cdots\right],
        \label{eq:recurrence_thm_sublinear}
    \end{equation}
    for all $q \in \mathbb{{R}}$.%
    \label{thm:recurrence_theorem}
\end{theorem}
\end{mdframed}
This theorem is proved in \cref{sec:recurrence_proof}, with the scaling of the $\cdots$ error term discussed in \cref{sec:bn_combining}. It shows that the hydrodynamic power-law, $\Phi(\omega) \sim |\omega|^{\rho}$ as $\omega \to 0$, shows up as a \textit{multiplicative correction} to the leading scaling governed by $\beta_{n}$, whose scaling depends only on the $|\omega| \to \infty$ asymptotics via \cref{eq:beta_n_asymptotics}. The subleading hydrodynamic term displays \textit{staggering} due to the factor $(-1)^{n}$, but the magnitude of the staggering decays to zero as $n \to \infty$. Many authors have noted examples where such staggering is sufficient to give singular $\omega \to 0$ behavior~\cite{viswanathRecursionMethodApplication2013,yatesDynamicsAlmostStrong2020,yatesLifetimeAlmostStrong2020,dymarskyKrylovComplexityConformal2021,ballartriguerosKrylovComplexityManybody2022,bhattacharjeeKrylovComplexitySaddledominated2022,rabinoviciKrylovLocalizationSuppression2022,rabinoviciKrylovComplexityIntegrability2022}. Our result proves in a wide class of models that this staggering \textit{necessarily} arises from the low-frequency behavior. Hermiticity is an important ingredient in producing the specific $(-1)^{n}$ form: it enters in the proof via the identity
\begin{equation}
    (-1)^{n} = \exp\left(2\pi i \int_{0}^{\beta_{n}} \sigma_{n}(\omega) \diff \omega\right),
    \label{eq:staggering_identity}
\end{equation}
where $\sigma_{n}(\omega)$ is the charge density of the Coulomb gas distribution that we defined in \cref{sec:coulomb}. This identity follows very generically by symmetry, since the distribution has total charge $n$ by construction, $\int_{-\beta_{n}}^{\beta_{n}}\sigma_{n}(\omega) \diff \omega = n$, and $\sigma_{n}(\omega)$ is \textit{even}, the latter coming from the evenness of the spectral function, which is a property of Hermitian systems. Thus we have $\int_{0}^{\beta_{n}} \sigma_{n}(\omega) \diff \omega = n/2$, giving the staggering factor $(-1)^{n}$. It is also worth emphasizing that this is a contribution from the $\omega = 0$ behavior of the spectral function; that is why the lower limit of the integral in \cref{eq:staggering_identity} is zero. This $(-1)^{n}$ staggering factor also has the same origins as the sign of the polynomials at zero, $\sgn[p_{2n}(0)] = (-1)^{n}$ (c.f.~\cref{eq:zero_mode_recursion}), which again is a generic property of even weight functions~\cite{szegoOrthogonalPolynomials1939}.

Notice that the case $p = 1$ is marginal, in the sense that the staggered multiplicative correction associated with the hydrodynamics scales like $1 / n \log{n}$, rather than $1 / n$ for $p > 1$. This is a signature of the Coulomb gas confinement transition discussed in \cref{sec:coulomb}. This transition is primarily a consequence of the \textit{high frequency} scaling of the spectral function, but it affects the hydrodynamic signature because of its manifestation in a logarithmically divergent equilibrium charge density at \textit{low frequencies}.

We emphasize that the operator growth hypothesis \cite{parkerUniversalOperatorGrowth2019} conjectures that chaotic many-body quantum systems generically have $p = 1$, and so are marginal in this sense. The possibility of a log-correction for $p=1$ was noted in \cite{bhattacharjeeKrylovComplexitySaddledominated2022}. Note that $p=1$ scaling is not restricted to chaotic models---certain many-body localized models have been shown to also exhibit this scaling (with a $q=1$ log-correction)~\cite{caoStatisticalMechanismOperator2021}. There is numerical evidence that the case $p = 2$, corresponding to $b_{n} \sim \sqrt{n}$, holds for interacting integrable systems at infinite temperature~\cite{parkerUniversalOperatorGrowth2019,viswanathRecursionMethodApplication2013,leeErgodicTheoryInfinite2001,hevelingNumericallyProbingUniversal2022}. We can also consider non-interacting systems as the limit $p \to \infty$, such that $b_{n} \sim n^{0}$ and the spectral function has compact support (see \cite[Theorem 7.4]{nevaiOrthogonalPolynomials1979}).

It is tempting to extract the value of $\rho$ by fitting the Lanczos coefficients to the relevant asymptotic form given in \cref{thm:recurrence_theorem}. In practice, we have found that, while this method gives qualitatively correct answers for $\rho$, it tends to give rather large error bars, since one is attempting to fit the coefficient of a small subleading correction. Instead, we recommend following the procedure outlined in \cref{sec:extracting_power_law} for extracting $\rho$ from the leading scaling of the zero mode in the Krylov space, which amplifies the effects of the staggered subleading term coming from the low-frequency power-law. 

\begin{example}
   As a check of \cref{thm:recurrence_theorem}, we can consider the generalized Hermite polynomials, which have weight function $\Phi(\omega)/2\pi = |\omega|^{\rho} \exp[-\omega^{2}]$. With $Q(\omega)= \omega^{2}$, we can determine $\beta_{n} = \sqrt{2n}$ from \cref{eq:MRS_def}, and $h_{n}(0) = h_{n}(1) = 4$ from \cref{eq:hn_integral} (in agreement with \cref{lem:hn0_scaling,lem:hn1_scaling} with $p=2$). Substituting into \cref{thm:recurrence_theorem}, we get agreement to $\mathcal{O}(1/n)$ with the exact recurrence coefficients, which are known to be $b_{n} = \frac{1}{\sqrt{2}} \sqrt{n + \frac{1}{2}[1 - (-1)^{n}]\rho}$~\cite{mastroianniInterpolationProcessesBasic2008}. 
\end{example}

\section{Hydrodynamics from the zero mode}
\label{sec:hydro}
So far we have given a large-$n$ expansion of the Lanczos coefficients $b_{n}$, showing how a low-frequency power-law in the spectral function, $\Phi(\omega\to 0) \sim |\omega|^{\rho}$, shows up as a subleading correction to the leading scaling of the Lanczos coefficients. In this section we explain how by studying the zero mode $|\omega = 0)_{\mathcal{K}}$ of the Liouvillian $\mathcal{L}$ within the Krylov space, satisfying $\mathcal{L}|\omega = 0)_{\mathcal{K}} = 0$, one can extract the value of the low frequency power-law exponent $\rho$, as well as hydrodynamic transport coefficients encoded in $\lim_{\omega \to 0} \Phi(\omega) / |\omega|^{\rho}$. The advantage of studying the zero mode over directly analyzing the Lanczos coefficients is that the zero mode amplifies the effects of the small subleading terms in the Lanczos coefficients, such that the hydrodynamics is manifest in the \textit{leading} behavior of the zero mode as $n \to \infty$. This approach to computing transport coefficients can be summarized in three steps:
\begin{enumerate}
    \item By fitting the leading behavior of the Lanczos coefficients $b_{n} \sim (n / \log^{q}{n})^{1/p}$, determine the growth exponents $p$ and $q$. Generic models usually have $p=1$ and $q=0$ (or $q=1$ in 1D).
    \item By fitting the leading behavior of the zero mode amplitudes $p_{2n}(0)^{2}$ (defined below) to the scalings given in \cref{cor:pn0_scaling,tab:zero_mode}, determine the power-law exponent $\rho$ governing $\Phi(\omega\to 0) \sim |\omega|^{\rho}$.
    \item Having fixed $\rho$, extract hydrodynamic transport coefficients by computing $\lim_{\omega\to 0} \Phi(\omega) / |\omega|^{\rho}$ using \cref{thm:hydrodynamic_constants}.
\end{enumerate}
In the following subsections we will describe the second and third steps in more detail. Throughout this section we will focus on $\omega=0$ properties of the spectral function. In \cref{sec:spectral_bootstrap} we will generalize these ideas to $\omega \neq 0$ through an algorithm we call the \textit{spectral bootstrap}; the approach described in this section is the $\omega\to 0$ limit of the spectral bootstrap.
\subsection{Extracting the low frequency power-law}
\label{sec:extracting_power_law}
If we expand the zero mode ${|\omega = 0)_{\mathcal{K}}} = \sum_{n=0}^{\infty} c_{n} |O_{n})$ in the Lanczos basis, then \cref{eq:lanczos_vector} tells us that
\begin{equation}
    \dfrac{c_{n}}{c_{0}} = \dfrac{p_{n}(0)}{p_{0}(0)},
    \label{eq:zero_mode_amplitudes}
\end{equation}
where $p_{0}(0) = \norm{A}^{-1}$ for an initial operator $|A)$. It is a general property of orthogonal polynomials with respect to even weight functions that they have definite parity~\cite{szegoOrthogonalPolynomials1939}, $p_{n}(-x) = (-1)^{n} p_{n}(x)$, and hence $c_{n} = 0$ for odd $n$. Furthermore, by using the recurrence relation in \cref{eq:three_term_recursion}, we find the recursive formula $p_{2n}(0) = -(b_{2n-1}/b_{2n}) p_{2n-2}(0)$, and hence
\begin{equation}
    \dfrac{p_{2n}(0)}{p_{0}(0)} = (-1)^{n} \dfrac{b_{2n-1}}{b_{2n}} \dfrac{b_{2n-3}}{b_{2n-2}} \cdots \dfrac{b_{1}}{b_{2}}.
    \label{eq:zero_mode_recursion}
\end{equation}
Because this quantity involves the ratio of even and odd Lanczos coefficients, it amplifies the effects of the even/odd staggering in the subleading corrections to the Lanczos coefficients indicated in \cref{thm:recurrence_theorem}. Previous authors have noted that this is analogous to the origin of the localized zero mode in the Su-Schrieffer-Heeger chain~\cite{suSolitonsPolyacetylene1979,yatesLifetimeAlmostStrong2020,yatesDynamicsAlmostStrong2020}. Our next result gives us a precise handle on how the low-frequency power-law imprints itself on the scaling of this zero mode, which is a reflection of the Bessel universality class discussed in \cref{sec:gf_bessel}.
\begin{mdframed}
\begin{theorem}[Informal]
    Given the spectral function $\Phi(\omega)/2\pi \equiv |\omega|^{\rho} \exp[-Q(\omega)]$, so that ${\Phi(\omega\to 0) \sim |\omega|^{\rho}}$ for some $\rho > -1$, then as $n \to \infty$ we have
    \begin{equation}
        p_{2n}(0)^{2} = \dfrac{C_{\rho}}{e^{-Q(0)}} \dfrac{[\pi \sigma_{2n}(0)]^{\rho}}{\beta_{2n}}[1 + o(1)],
        \label{eq:exact_pn0_scaling}
    \end{equation}
    where $\sigma_{2n}$ and $\beta_{2n}$ are respectively the density and support of the equilibrium Coulomb gas distribution defined in \cref{sec:coulomb}, and the constant $C_{\rho}$ is given by
    \begin{equation}
        C_{\rho} = \dfrac{2^{1-\rho}}{\Gamma\left[\frac{1}{2}(1+\rho)\right]^{2}}.
        \label{eq:C_rho_def}
    \end{equation}
    \label{lem:pn0_scaling}
\end{theorem}
\end{mdframed}
See \cref{sec:thm2_proof} for a proof; the $o(1)$ error term is given by the RHS of \cref{eq:R_error_bound}. For the purposes of numerically extracting the value of $\rho$, it is worth isolating the $n$-dependence of $p_{2n}(0)^{2}$. The scaling of $\beta_{2n}$ is given in \cref{eq:beta_n_asymptotics}, and the scaling of $\sigma_{2n}(0) = (2n / \beta_{2n}) h_{2n}(0) / 2\pi$ can be deduced from \cref{lem:hn0_scaling}, leading to the following.
\begin{corollary}
    With $\Phi(\omega) \equiv |\omega|^{\rho} \exp[-Q(\omega)]$, for $n \to \infty$ we have
    \begin{equation}
        p_{2n}(0)^{2} \sim \mathcal{O}\left(\dfrac{[n h_{2n}(0)]^{\rho}}{\beta_{2n}^{1+\rho}}\right).
    \end{equation}
    With $Q(x) \sim |x|^{p} \log^{q}{|x|}$ as $|x| \to \infty$, this leads to
    \begin{equation}
        p_{2n}(0)^{2} \sim \begin{dcases}
            n^{-\frac{1}{p} + \rho\left(1 - \frac{1}{p}\right)} (\log{n})^{\frac{q}{p}(1+\rho)}, & \text{for } p > 1, q \in \mathbb{R},\\
            n^{-1} (\log{n})^{\rho + q(1+\rho)}, & \text{for } p = 1, q > -1,
        \end{dcases}
    \end{equation}
    where the $\sim$ indicates only the overall $n$-dependence, and for $p=1$ we have dropped the $o(1)$ subleading exponent in $h_{n}(0) = (\log{n})^{1+o(1)}$ from \cref{lem:hn0_scaling}.
    \label{cor:pn0_scaling}
\end{corollary}
In \cref{tab:zero_mode} we summarize some of the most physically relevant cases for the scaling of $p_{2n}(0)^{2}$. The statements about chaotic dynamics are made according to the operator growth hypothesis~\cite{parkerUniversalOperatorGrowth2019}, while the statement about integrable systems is based on numerical evidence only.

\begin{table}[t]
\begin{tabular}{c@{\hskip 0.5em}cc@{\hskip 0.5em}cc}
    \toprule
    \textbf{Dynamical class} & $p$ & $q$ & $\mathcal{O}(p_{2n}(0)^{2})$ & $\mathcal{O}\left(\sum_{m=0}^{n} p_{2m}(0)^{2}\right)$ \\
    \midrule
    Chaotic ($d>1$) & 1 & 0 & $(\log{n})^{\rho}/n$ & $(\log{n})^{1+\rho}$ \\
    Chaotic ($d=1$) & 1 & 1 & $(\log{n})^{1+2\rho}/n$ & $(\log{n})^{2(1+\rho)}$ \\
    Interacting integrable & 2 & 0 & $n^{\frac{1}{2}(-1+\rho)}$ & $n^{\frac{1}{2}(1+\rho)}$ \\
    Non-interacting & $\infty$ & 0 & $n^{\rho}$ & $n^{1+\rho}$ \\
    \bottomrule
\end{tabular}
\caption{Zero mode amplitude $p_{2n}(0)^{2}$ scaling as a function of the low-frequency power-law exponent $\Phi(\omega\to 0) \sim |\omega|^{\rho}$ and the high-frequency behavior $\Phi(\omega\to\infty) \sim \omega^{p} \log^{q}{\omega}$.}
\label{tab:zero_mode}
\end{table}

Since we must have $1+\rho > 0$ in order for the spectral function to be Lebesgue integrable across $\omega=0$, we see that $\sum_{n} p_{2n}(0)^{2}$ diverges algebraically for $p > 1$, but polylogarithmically for $p = 1$, so the zero mode is \textit{marginally delocalized} for $p = 1$. Unlike the original Su-Schrieffer-Heeger chain, here the zero mode is not a true localized eigenstate; although its support does decay further into the bulk of the Lanczos chain, it does so too slowly to be localized. The only case in which the zero mode is normalizable is when the spectral function has a delta-function peak at zero frequency. In \cref{fig:zero_mode_toy_examples} we test this scaling numerically for some toy spectral functions, finding good agreement with the predictions of \cref{lem:pn0_scaling}.

Having fixed the values of $p$ and $q$ from the leading behavior of the Lanczos coefficients $b_{n} \sim (n / \log^{q}{n})^{1/p}$, one can extract the value of the power-law exponent $\rho$ by computing $p_{2n}(0)^{2}$ and fitting to the relevant asymptotic form indicated in \cref{cor:pn0_scaling}. In practice, we have found this to be more robust than attempting to extract $\rho$ by fitting the subleading terms in \cref{thm:recurrence_theorem}, since here one is fitting to the leading behavior of $p_{2n}(0)^{2}$. Generally this procedure works well for $p>1$, i.e.~sublinear Lanczos coefficients. However, for $p=1$, where the Lanczos coefficients are (quasi-)linear, we have found that $n$ has to be extremely large before this procedure gives the correct value of $\rho$, and so in this case it might be best to assume the value of $\rho$ on phenomenological grounds. We attribute this to the potentially very slow convergence with $n$ of the Coulomb gas density $h_{n}(0)$ to its asymptotic form given in \cref{lem:hn0_scaling}. While this slow convergence seems troubling, our algorithms for extracting hydrodynamic coefficients like diffusion constants, discussed in the next section, do \textit{not} rely on the asymptotics for the Coulomb gas because we ultimately eliminate $h_{n}(0)$ from the equations that will determine the hydrodynamic coefficients.

\begin{figure}[t]
    \centering
    \includegraphics[width=\columnwidth]{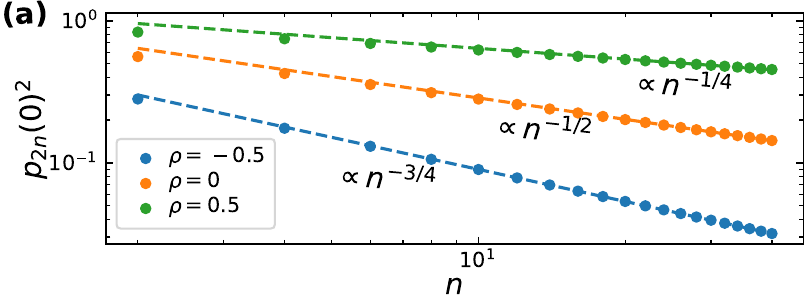}
    \includegraphics[width=\columnwidth]{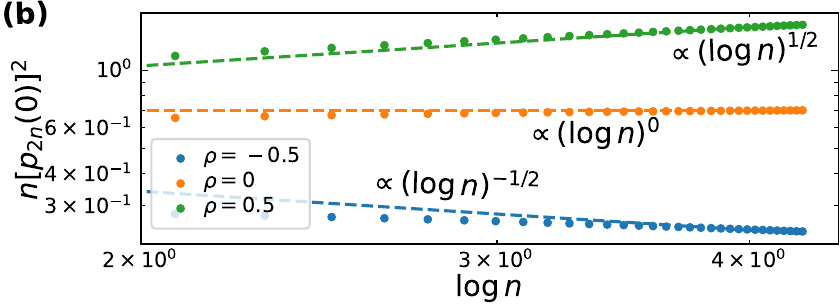}
    \caption{Scaling of the zero mode amplitudes $p_{2n}(0)^{2}$ as a function of the power-law exponent $\rho$ governing the low-frequency spectral function $\Phi(\omega\to 0) \sim |\omega|^{\rho}$. \textbf{(a)} Toy model $\Phi(\omega)/2\pi = |\omega|^{\rho} \exp[-Q(\omega)]$ with $Q(\omega) = (1 + \omega^{2} + \omega^{4})^{1/2}$. This model has $p=2$, so $p_{2n}(0)^{2}$ scales algebraically as $p_{2n}(0)^{2} \sim n^{\frac{1}{2}(\rho-1)}$. \textbf{(b)} Toy model $\Phi(\omega)/2\pi = |\omega|^{\rho} \exp[-Q(\omega)]$ with $Q(\omega) = (1 + \omega^{2} + \omega^{4})^{1/4}$. This model has $p=1$, so $n p_{2n}(0)^{2}$ scales polylogarithmically as $n p_{2n}(0)^{2} \sim (\log{n})^{\rho}$. (Note the axes are log-scale, so for case \textbf{(b)} the x-axis is linear in $\log{\log{n}}$.)}
    \label{fig:zero_mode_toy_examples}
\end{figure}

\subsection{Extracting hydrodynamic transport coefficients}
\label{sec:transport_coeffs}
Now let us apply this to develop a numerical algorithm for computing diffusion constants and other transport coefficients. To warm up, let us consider diffusion. The infinite temperature dc conductivity $\sigma_{\mathrm{dc}}$, rescaled by the temperature~\cite{bertiniFinitetemperatureTransportOnedimensional2021}, is given by the Kubo formula~\cite{kuboStatisticalMechanicalTheoryIrreversible1957}
\begin{equation}
    \sigma_{\mathrm{dc}} = \lim_{\tau \to \infty} \lim_{L \to \infty} \int_{0}^{\tau} \dfrac{( \current | \current(t) )}{L} \diff t,
\end{equation}
where $\current$ is the total current for the conserved charge, and $L$ is the system size. Now we set our initial Lanczos operator to be $A \coloneqq \current / \sqrt{L}$, the zero wavevector component of the Fourier transform of the current density, and take the thermodynamic limit $L \to \infty$. For a translationally invariant system, the Lanczos algorithm can be run in Fourier space with minimal modification compared with real space~\cite[Appendix C]{parkerUniversalOperatorGrowth2019}. Since $(\current | \current(t)) = (\current | \current(-t))$ at infinite temperature, we have
\begin{equation}
    \sigma_{\mathrm{dc}} = \dfrac{1}{2} \int_{-\infty}^{\infty} (A | A(t)) \diff t = \dfrac{1}{2} \Phi(0).
\end{equation}
Given our ansatz $\Phi(\omega)/2\pi = |\omega|^{\rho} \exp[-Q(\omega)]$, we have
\begin{equation}
    \lim_{\omega \to 0} \dfrac{\Phi(\omega)}{|\omega|^{\rho}} = 2\pi e^{-Q(0)}.
\end{equation}
For regular diffusion we expect the dc conductivity $\sigma_{\mathrm{dc}}$ to be finite, which corresponds to $\rho = 0$. To obtain the diffusion constant, one can use the Einstein relation $D = \sigma_{\mathrm{dc}} / \chi$, where $\chi$ is the static susceptibility.

The first thing to notice is that \cref{eq:exact_pn0_scaling} contains this factor $\exp[-Q(0)]$ which gives us the value of $\lim_{\omega \to 0} \Phi(\omega) / |\omega|^{\rho}$. However, despite being able to exactly compute $p_{2n}(0)$ via \cref{eq:zero_mode_recursion}, the problem with using \cref{eq:exact_pn0_scaling} on its own to extract $\exp[-Q(0)]$ is that it also contains the factor $\sigma_{n}(0)$ related to the density of the equilibrium measure at the origin, and in general computing $\sigma_{n}(0)$ exactly requires already knowing the spectral function (since one needs to know $Q$). While \cref{lem:hn0_scaling} indicates that $\sigma_{n}(0) = (1/2\pi)(n h_{n}(0) / \beta_{n})$ exhibits some universal scaling as $n\to\infty$, this asymptotic scaling can take a long time to set in, so given only a modest number of Lanczos coefficients, it would be preferable to rely only on quantities we can compute exactly or approximately using a finite number of Lanczos coefficients. Fortunately, it turns out to be possible to eliminate the leading-order dependence on $\sigma_{n}(0)$ using the precise asymptotics of the partial sums of the amplitudes $p_{2n}(0)^{2}$. We define
\begin{equation}
    K_{n}(x,y) \coloneqq \sum_{m=0}^{n-1} p_{m}(x) p_{m}(y),\label{eq:kernel_def}
\end{equation}
which is known as the `Christoffel-Darboux kernel' in the language of orthogonal polynomials~\cite{simonChristoffelDarbouxKernel2008}. Under the same assumptions as \cref{lem:pn0_scaling}, we can derive the asymptotic expression
\begin{equation}
    K_{n}(0,0) = \dfrac{\hat{c}_{\rho}}{e^{-Q(0)}} \left[\pi \sigma_{n}(0) \right]^{1+\rho}[1 + o(1)],
    \label{eq:exact_kn_formula}
\end{equation}
where the $o(1)$ error term is given in the RHS of \cref{eq:R_error_bound}, and the constant $\hat{c}_{\rho}$ is given by
\begin{equation}
    \hat{c}_{\rho} \coloneqq \dfrac{1}{2^{1+\rho} \Gamma\left[\frac{1}{2}(1+\rho)\right] \Gamma\left[\frac{1}{2}(3+\rho)\right]}.
\end{equation}
Combining \cref{eq:exact_kn_formula,lem:pn0_scaling} gives
\begin{equation}
     \pi \sigma_{2n}(0) = 2(1+\rho) \dfrac{K_{2n}(0,0)}{\beta_{2n} p_{2n}(0)^{2}}[1 + o(1)],
     \label{eq:sigma_practical_identity}
\end{equation}
which allows us to eliminate the equilibrium measure $\sigma_{2n}(0)$, leading to the following result:
\begin{mdframed}
\begin{theorem}
    With $\Phi(\omega) \sim |\omega|^{\rho}$ as $\omega \to 0$, we have
    \begin{equation}
        \lim_{\omega \to 0} \dfrac{\Phi(\omega)}{|\omega|^{\rho}} = \lim_{n \to \infty}  \dfrac{c_{\rho} [K_{2n}(0,0)]^{\rho}}{\left[ \beta_{2n} p_{2n}(0)^{2}\right]^{1+\rho}},
        \label{eq:hydrodynamic_constants_kernel}
    \end{equation}
    where the constant $c_{\rho}$ is given by
    \begin{equation}
        c_{\rho} = \dfrac{4\pi (1+\rho)^{\rho}}{\Gamma\left[\frac{1}{2}(1+\rho)\right]^{2}}.
    \end{equation}
    For the reader's convenience, we can use \cref{eq:zero_mode_recursion,eq:kernel_def} to give the following more explicit expression in terms of Lanczos coefficients:
    \begin{equation}
        \hspace{-0.9em}\lim_{\omega \to 0} \dfrac{\Phi(\omega)}{|\omega|^{\rho}} = \lim_{n \to \infty}  c_{\rho} \norm{A}^{2} \dfrac{\left[1 + \sum_{k=1}^{n-1} \left(\dfrac{b_{2k-1}}{b_{2k}} \cdots \dfrac{b_{1}}{b_{2}}\right)^{2}\right]^{\rho}}{\left[\beta_{2n} \left(\dfrac{b_{2n-1}}{b_{2n}} \cdots \dfrac{b_{1}}{b_{2}}\right)^{2}\right]^{1+\rho}}
        \label{eq:hydrodynamic_constants_Lanczos}
    \end{equation}
    where $A$ is the initial Lanczos operator, $O_{0} = A / \norm{A}$. 
    \label{thm:hydrodynamic_constants}
\end{theorem}
\end{mdframed}
By \cref{thm:recurrence_theorem} we know that $\beta_{2n} = 2 b_{2n} [1 + \mathcal{O}(1/n)]$, so from \cref{eq:hydrodynamic_constants_Lanczos} we can see that simply approximating $\beta_{2n} \approx 2 b_{2n}$ also results in a small $\mathcal{O}(1/n)$ relative error in $\lim_{\omega \to \infty} \Phi(\omega) / |\omega|^{\rho}$. All of the subtle effects from subleading staggered terms in the Lanczos coefficients are captured in $p_{2n}(0)^{2}$ and $K_{2n}(0,0)$, which we can compute exactly using the first $2n$ Lanczos coefficients.

\subsection{Numerical benchmarks}
\label{sec:hydro_benchmarks}
To start our benchmarks, we begin with regular diffusion. In this case $\Phi(0)$ for the current operator is nonzero, corresponding to $\rho = 0$.  In this special case, \cref{thm:hydrodynamic_constants} becomes
\begin{align}
    \Phi(0) &= \lim_{n \to \infty} \dfrac{4}{\beta_{2n} p_{2n}(0)^{2}},\label{eq:Phi0}\\
    &= \lim_{n \to \infty} \dfrac{4 \norm{A}^{2} }{\beta_{2n}} \left(\dfrac{b_{2n}}{b_{2n-1}} \dfrac{b_{2n-2}}{b_{2n-3}} \cdots \dfrac{b_{2}}{b_{1}}\right)^{2},\label{eq:Phi0bn}
\end{align}
where in the second line we used \cref{eq:zero_mode_recursion} to replace $p_{2n}(0)^{2}$. A similar formula appears in Ref.~\cite{wangDiffusionConstantsRecursion2024}, based on considerations motivated by the operator growth hypothesis (OGH)~\cite{parkerUniversalOperatorGrowth2019} (see also Ref.~\cite{fullgrafLanczosPascalApproachCorrelation2025}). Our derivation shows that this behavior does not necessarily require the system to obey the OGH in order for it to be valid. 

\begin{figure}[t]
    \centering
    \subfigimg[width=\columnwidth]{(a)}{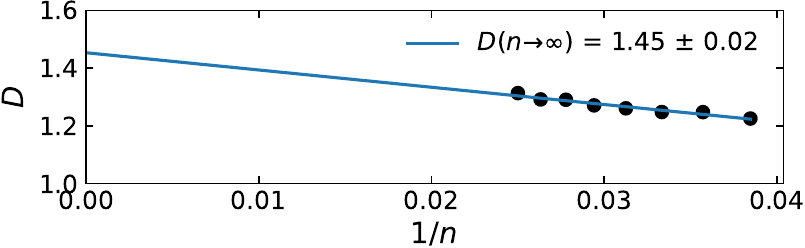}
    \subfigimg[width=\columnwidth]{(b)}{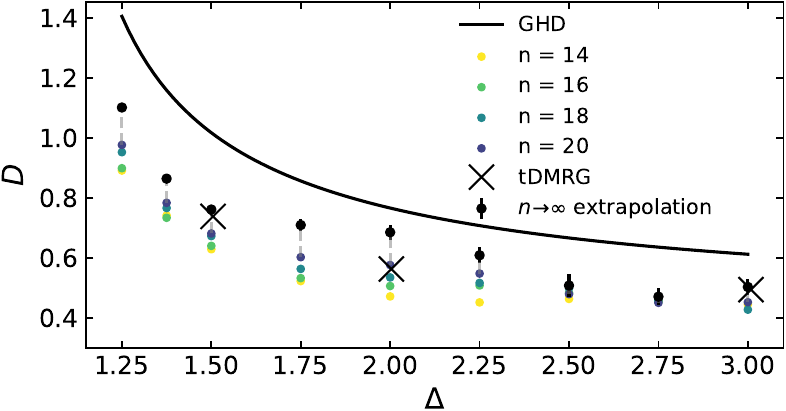}
    \subfigimg[width=\columnwidth]{(c)}{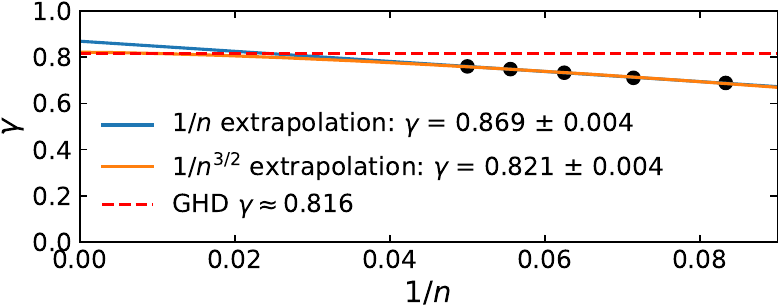}
    \caption{\textbf{(a)} Energy diffusion constant of the mixed field Ising chain using \cref{thm:hydrodynamic_constants} with $\rho=0$ and up to $n=40$ Lanczos coefficients. \textbf{(b)} Spin diffusion constant of the XXZ spin chain as a function of the anisotropy $\Delta$ using \cref{thm:hydrodynamic_constants} with $\rho=0$ and up to $n=20$ Lanczos coefficients, compared with the prediction \cref{eq:xxz_ghd} from generalized hydrodynamics (GHD)~\cite{gopalakrishnanKineticTheorySpin2019,denardisDiffusionGeneralizedHydrodynamics2019}. The black dots show the result of a linear extrapolation in $1/n$ to $n \to \infty$, and the crosses show tDMRG data from Ref.~\cite{karraschRealtimeRealspaceSpin2014}. \textbf{(c)} Coefficient $\gamma$ of the time-dependent spin diffusion constant $D(t) \sim \gamma t^{1/3}$ for the isotropic Heisenberg chain using \cref{thm:hydrodynamic_constants} with $\rho = -1/3$ and up to $n=20$ Lanczos coefficients. Extrapolations in $1/n$ and $1/n^{3/2}$ give slightly different predictions, with the $1/n^{3/2}$ result agreeing better with the GHD prediction $\gamma \approx 0.816$~\cite{denardisSuperdiffusionEmergentClassical2020,denardisStabilitySuperdiffusionNearly2021}.}
    \label{fig:spin_diffusion}
\end{figure}

We start with the mixed field Ising model,
\begin{equation}
    H_{\mathrm{MFIM}} = \sum_{i=1}^{L} \left(Z_{i} Z_{i+1} + g_{x} X_{i} + g_{z} Z_{i}\right),
\end{equation}
which has energy diffusion for generic values of $g_{x}$ and $g_{z}$~\cite{kimBallisticSpreadingEntanglement2013}. The energy current operator is defined by a continuity equation to be $\current = \sum_{i} g_{x} \left(Y_{i} Z_{i+1} - Z_{i} Y_{i+1}\right)$, so given the initial Lanczos operator $A \coloneqq \current / \sqrt{L}$, in the thermodynamic limit we have $\norm{A}^{2} = \lim_{L \to \infty} (\current | \current) / L = 2 g_{x}^{2}$, which is needed to compute $p_{2n}(0)$ via \cref{eq:zero_mode_recursion}. The infinite temperature static energy susceptibility is $\chi = \lim_{L \to \infty} [\langle H_{\mathrm{MFIM}}^{2} \rangle - \langle H_{\mathrm{MFIM}} \rangle^{2}] / L = 1 + g_{x}^{2} + g_{z}^{2}$. Fixing $g_{x} = 1.4$ and $g_{z} = 0.9045$, \cref{fig:spin_diffusion}(a) shows the result of using \cref{thm:hydrodynamic_constants} with $\rho = 0$ to extract the energy diffusion constant $D$. A linear extrapolation in $1/n$ on the predictions from up to $n=40$ Lanczos coefficients gives an estimate of $D(n \to \infty) = 1.45 \pm 0.02$, close to the prediction $D = 1.4$--$1.45$ from tensor network methods~\cite{rakovszkyDissipationassistedOperatorEvolution2022,yi-thomasComparingNumericalMethods2024}.

For a more challenging example, we study spin diffusion in the XXZ chain,
\begin{equation}
    H_{\mathrm{XXZ}} = \frac{1}{4} \sum_{i} \left( X_{i} X_{i+1} + Y_{i} Y_{i+1} + \Delta Z_{i} Z_{i+1}\right),
\end{equation}
which has diffusive spin transport for $\Delta > 1$~\cite{bertiniFinitetemperatureTransportOnedimensional2021}. The spin current operator is defined by a continuity equation to be $\current = \frac{1}{4} \sum_{i} \left(X_{i} Y_{i+1} - Y_{i} X_{i+1}\right)$, so $\norm{A}^{2} = 1 / 8$. The infinite temperature static spin susceptibility is $\chi = \lim_{L \to \infty} [\langle (S^{z}_{\mathrm{tot}})^{2}\rangle - \langle S^{z}_{\mathrm{tot}}\rangle^{2}]/L = 1/4$. Performing a linear extrapolation in $1/n$ on the predictions from up to $n=20$ Lanczos coefficients, \cref{fig:spin_diffusion}(b) shows that the resulting diffusion constants come reasonably close to the prediction from generalized hydrodynamics~\cite{gopalakrishnanKineticTheorySpin2019,denardisDiffusionGeneralizedHydrodynamics2019},
\begin{equation}
    D_{\mathrm{XXZ}} = \dfrac{2 \sinh{\eta}}{9 \pi} \sum_{s=1}^{\infty} (1+s) \left[\dfrac{s+2}{\sinh{\eta s}} - \dfrac{s}{\sinh{\eta(s+2)}}\right],
    \label{eq:xxz_ghd}
\end{equation}
where $\eta \coloneqq \arcosh{\Delta}$. We note that $n=20$ is still relatively small, but even at these small values the predictions are close to those obtained using tensor network methods like tDMRG~\cite{karraschRealtimeRealspaceSpin2014}.

Finally, to test the method for $\rho \neq 0$, we study spin transport in the isotropic Heisenberg chain (i.e.~$\Delta = 1$ for the XXZ spin chain). This has recently been discovered to have superdiffusive spin transport at infinite temperature~\cite{gopalakrishnanSuperdiffusionNonabelianSymmetries2024}, with a current-current correlator decaying like $( \mathcal{J} | \mathcal{J}(t)) \sim (\gamma \chi/3) t^{-2/3}$ as $t \to \infty$, leading to a time-dependent diffusion constant growing like $D(t) \sim \gamma t^{1/3}$. The value of $\gamma$ can be extracted from the Lanczos coefficients of the spin current operator using \cref{thm:hydrodynamic_constants} with $\rho = -1/3$ via the relation
\begin{equation}
\gamma = \dfrac{\chi \sqrt{3}}{\Gamma\left[\frac{1}{3}\right]} \lim_{\omega \to 0} \dfrac{\Phi(\omega)}{|\omega|^{-1/3}},
\end{equation}
where we used the fact that $\int_{\mathbb{R}} e^{-i \omega t} |t|^{-2/3} \diff t = \sqrt{3} \Gamma[1/3] |\omega|^{-1/3}$. Generalized hydrodynamics gives the prediction $\gamma =  (2/3) (10 \pi / 27)^{4/3} \approx 0.816$~\cite{denardisSuperdiffusionEmergentClassical2020,denardisStabilitySuperdiffusionNearly2021}. Using only $n=20$ Lanczos coefficients, we obtain predictions close to this result (\cref{fig:spin_diffusion}(c)). Interestingly, we find extrapolating in $n^{-3/2}$ gives better agreement with the GHD prediction than extrapolating in $n^{-1}$. Finally, we emphasize that, since this example has $\rho \neq 0$, not accounting for the Bessel universality governing the $\omega\approx 0$ behavior can lead to erroneous results, as we will illustrate in more detail in \cref{sec:bessel_bootstrap} and \cref{fig:bessel_bootstrap}. This is where our approach differs the most from other recent approaches based on the recursion method~\cite{parkerLocalMatrixProduct2020,wangDiffusionConstantsRecursion2024,fullgrafLanczosPascalApproachCorrelation2025}.

\section{The spectral bootstrap: approximating the spectral function at finite frequencies}
\label{sec:spectral_bootstrap}
In this section we explain how to extend ideas from the previous section to extract the spectral function at finite frequencies. The high-level summary is that, using the appropriate $n \to \infty$ asymptotics of the orthogonal polynomials, we formulate a first-order differential equation involving the spectral function and the equilibrium density of the Coulomb gas defined in \cref{sec:coulomb}. Together with an initial condition at $\omega = 0$ provided to us by Hermiticity (i.e.~$\Phi(-\omega) = \Phi(\omega)$), this differential equation can be iteratively solved to give a finite $n$ approximation to the spectral function at finite frequencies. The convergence of this approximation is determined by the convergence of the polynomial asymptotics. We refer to this general procedure of using orthogonal polynomial asymptotics to give a differential equation for the spectral function as the \textit{spectral bootstrap}. We emphasize that the approach to hydrodynamic transport coefficients in \cref{sec:transport_coeffs} used the $\omega\to 0$ limit of these asymptotics, and so can be viewed as the initial $\omega=0$ step of the spectral bootstrap.

In order for this approach to work well, it is important to derive asymptotics taking into account all the relevant `high-level' features of the spectral function, such as algebraic singularities, spectral gaps, etc. We will focus on a particular instance of this idea, where the spectral function $\Phi(\omega) / 2\pi \equiv |\omega|^{\rho} \exp[-Q(\omega)]$ has a power-law at $\omega = 0$ but is otherwise smooth. To warm up we will start by taking $\rho = 0$, so that the spectral function is also smooth at $\omega = 0$. The procedure for $\rho \neq 0$ is conceptually similar, but the equations involved are a little more complicated. For $\rho = 0$, the whole frequency range $\omega \in (-\beta_{n}, \beta_{n})$ is referred to as the `bulk' (see \cref{fig:toy_weighted_poly}), and the following procedure is controlled for all frequencies inside the bulk and away from the edges $\omega \approx \pm \beta_{n}$. Since $\beta_{n}$ increases with $n$, for computationally accessible values of $n$ the size of the bulk is large enough to capture the frequency range where the spectral function is non-negligible.

We remark that if one is only interested in recovering the spectral function $\Phi(\omega)$ in the bulk, there is a much simpler procedure than what we are about to describe:
\begin{enumerate}
    \item Approximate the level-$n$ Green's function by the universal `semicircle form'
    \begin{equation*}G_{n}(z) \approx \frac{2}{\beta_{n}^{2}}(z - \sqrt{z+\beta_{n}} \sqrt{z - \beta_{n}}),\end{equation*}
    from \cref{eq:greens_func_bulk}, where we approximate $\beta_{n} \approx 2 b_{n}$.
    \item Substitute $G_{n}(z)$ into the continued fraction in \cref{eq:greens_func_continued_fraction} to obtain the full Green's function $G(z)$.
    \item Compute $\Phi(\omega) = 2 \Im G(\omega - i 0^{+})$.
\end{enumerate}
For $\rho=0$, this simple procedure gives numerically identical results to the spectral bootstrap algorithm we will describe in the next section. However, for $\rho \neq 0$, the behavior of $\Phi(\omega)$ near $\omega=0$ is strongly influenced by the behavior of the Coulomb gas density $\sigma_{n}(\omega)$, which can have strong fluctuations as a function of $\omega$. Unfortunately, even for finite $n$, computing $\sigma_{n}(\omega)$ exactly requires already knowing the full spectral function $\Phi(\omega)$, which would defeat the point. What we show is how to iteratively compute an approximation to $\sigma_{n}(\omega)$, controlled in the $n \to \infty$ limit, and in turn how this can be used to obtain a finite $n$ approximation of the spectral function $\Phi(\omega)$. As a byproduct, having access to $\sigma_{n}(\omega)$ is also useful in checking the emergence of random matrix universality in this quantum operator growth problem---we will discuss this in more detail in \cref{sec:universality}.

As well as the `bulk' and `Bessel' versions of the spectral bootstrap we will describe imminently, one can also formulate a version suited for frequencies near the spectral edge $|\omega| \approx \beta_{n}$. In this region there is Airy universality, as in \cref{sec:airy_gn}, and one must formulate polynomial asymptotics which are relevant for this universality class. We defer the exposition of this `Airy bootstrap' to \cref{sec:airy_bootstrap}.

\subsection{The bulk bootstrap: $\rho=0$}
\label{sec:bulk_bootstrap}
In terms of the Coulomb gas density $\sigma_{n}(x)$ discussed in \cref{sec:coulomb}, we define
\begin{equation}
    \theta_{n}(\omega) \coloneqq -\pi \int_{\omega}^{\beta_{n}} \sigma_{n}(\omega^{\prime}) \diff \omega^{\prime} - \dfrac{\pi}{4}.
    \label{eq:theta_def}
\end{equation}
We will see shortly that $\theta_{n}(\omega)$ acts as a phase factor in a WKB-like asymptotic for the orthogonal polynomials. However, let us first briefly discuss some consequences of Hermiticity. For an arbitrary spectral function, all one knows about $\sigma_{n}$ is that it integrates to $n$ over some interval whose right endpoint is $\beta_{n}$. But for an \textit{even} spectral function, which for us is guaranteed by Hermiticity, we know that $\sigma_{n}(\omega)$ is even and supported on $[-\beta_{n}, \beta_{n}]$, so that $\int_{0}^{\beta_{n}} \sigma_{n}(x) \diff x = n/2$ (the same integral gave the staggering factor $(-1)^{n}$ in the recurrence coefficients, c.f.~\cref{thm:recurrence_theorem}). Thus we can rewrite $\theta_{n}(\omega)$ as
\begin{equation}
    \theta_{n}(\omega) = \pi I_{n}(\omega) - \dfrac{n\pi}{2} - \dfrac{\pi}{4},
    \label{eq:theta_identity}
\end{equation}
where $I_{n}(\omega)$ is defined as
\begin{align}
    I_{n}(\omega) &\coloneqq \int_{0}^{\omega} \sigma_{n}(\omega^{\prime}) \diff \omega^{\prime}.\label{eq:In_def}
\end{align}
Note that $I_{n}(0) = 0$, independent of the density distribution $\sigma_{n}$. This is important because it sets the initial condition for the first-order differential equation we will derive below. Although this is manifest from its definition, this discussion shows that it is really Hermiticity that gives us an initial condition at $\omega = 0$. For a non-even weight function, one would only have an initial condition for $\theta_{n}(\omega)$ at $\omega = \beta_{n}$. %

The Riemann-Hilbert analysis gives us access to detailed asymptotic formulae for the orthogonal polynomials, controlled as $n \to \infty$. In the orthogonal polynomials literature, these are known as Plancherel-Rotach asymptotics~\cite{plancherelValeursAsymptotiquesPolynomes1929}.

We have the following generalization of \cref{lem:pn0_scaling} for $\rho=0$ to finite frequencies (see \cref{sec:polynomial_asymptotics} for a proof)
\begin{widetext}
    \begin{align}
        p_{n-1}(\omega) \approx \dfrac{-1}{\sqrt{\Phi(\omega)}} \sqrt{\dfrac{2}{\beta_{n}}} &\left[ \left(\dfrac{\beta_{n}-\omega}{\beta_{n}+\omega}\right)^{1/4} \cos{\theta_{n}(\omega)} + \left(\dfrac{\beta_{n}+\omega}{\beta_{n}-\omega}\right)^{1/4} \sin{\theta_{n}(\omega)}\right],\label{eq:pnm1_bulk_asymptotic}\\
        p_{n}(\omega) \approx \dfrac{1}{\sqrt{\Phi(\omega)}} \sqrt{\dfrac{2}{\beta_{n}}} &\left[ \left(\dfrac{\beta_{n}-\omega}{\beta_{n}+\omega}\right)^{1/4} \cos{\theta_{n}(\omega)} - \left(\dfrac{\beta_{n}+\omega}{\beta_{n}-\omega}\right)^{1/4} \sin{\theta_{n}(\omega)}\right],\label{eq:pn_bulk_asymptotic}
    \end{align}
where at this level of approximation we are setting $\beta_{n} \approx 2 b_{n}$ (see \cref{thm:recurrence_theorem} for $\rho=0$). The first thing to notice is that $p_{n-1}(\omega)$ and $p_{n}(\omega)$ are both pointwise proportional to $1/\sqrt{\Phi(\omega)}$; this important feature is quite generic for orthogonal polynomials, and we will see that it continues to hold in other frequency regimes. We can combine these equations and use \cref{eq:theta_identity} to get the first of our bootstrap equations:
\begin{equation}
    \boxed{
     \Phi(\omega) \approx \dfrac{4}{\beta_{n}} \dfrac{1}{p_{n-1}(\omega)^{2} + p_{n}(\omega)^{2}}  \dfrac{\beta_{n} - (-1)^{n} \omega \sin[2 \pi I_{n}(\omega)]}{\sqrt{\beta_{n}^{2} - \omega^{2}}}.\label{eq:Phi_bulk}}
\end{equation}
Note that we used \textit{two} polynomials, $p_{n-1}$ and $p_{n}$, in order to avoid dividing by zero: it is a general property of orthogonal polynomials that the zeros of $p_{n-1}$ and $p_{n}$ interlace~\cite{freudOrthogonalPolynomials1971}, which guarantees $p_{n-1}(\omega)^{2} + p_{n}(\omega)^{2} > 0$.
\end{widetext}

Given $n$ Lanczos coefficients $\{b_{k}\}_{k=1}^{n}$, the orthogonal polynomials $p_{k}(\omega)$ can be computed exactly up to order $n$ using the recursion relation \cref{eq:three_term_recursion}, so if we knew the phase function $I_{n}(\omega)$, then we would be able to recover the spectral function $\Phi(\omega)$ from \cref{eq:Phi_bulk}. However, $I_{n}(\omega) = \int_{0}^{\omega} \sigma_{n}(\omega^{\prime}) \diff \omega^{\prime}$ is defined in terms of the equilibrium measure $\sigma_{n}$, and computing that \textit{exactly} via \cref{eq:psin_rescaling,eq:psin_hn,eq:hn_integral} requires prior knowledge of $\Phi$, precisely the function we are trying to estimate in the first place.

Our solution is to derive a large-$n$ approximation for $\sigma_{n}(\omega)$ using the polynomial asymptotics \cref{eq:pnm1_bulk_asymptotic,eq:pn_bulk_asymptotic}. If we differentiate either of these equations with respect to $\omega$, we will get terms involving the derivative $I_{n}^{\prime}(\omega) = \sigma_{n}(\omega)$, which is what we want. However, we will also get terms involving the derivative $\Phi^{\prime}(\omega)$, which is a new unknown. The trick is to consider a certain `determinantal' combination of derivatives. In particular, the Christoffel-Darboux formula \cite{szegoOrthogonalPolynomials1939} shows that the diagonal Christoffel-Darboux kernel $K_{n}(\omega,\omega) = \sum_{k=0}^{n-1} p_{k}(\omega)^{2}$ can be expressed as
\begin{align}
    K_{n}(\omega,\omega) &= b_{n} \det\begin{pmatrix} p_{n-1}(\omega) & p_{n-1}^{\prime}(\omega) \\ p_{n}(\omega) & p_{n}^{\prime}(\omega) \end{pmatrix}.\label{eq:diagonal_CD}
\end{align}
This determinantal structure ensures exact cancellation of the terms involving the unwanted derivative $\Phi^{\prime}(\omega)$. More explicitly, if we write the asymptotics in \cref{eq:pnm1_bulk_asymptotic,eq:pn_bulk_asymptotic} as $p_{n-1}(\omega) \equiv \Phi(\omega)^{-1/2} \alpha_{n-1}(\omega)$ and $p_{n}(\omega) \equiv \Phi(\omega)^{-1/2} \alpha_{n}(\omega)$, then we have (suppressing $\omega$ arguments for ease of notation)
\begin{align*}
    &K_{n} = b_{n} \det\begin{pmatrix} \Phi^{-1/2} \alpha_{n-1} & (\Phi^{-1/2} \alpha_{n-1})^{\prime} \\ \Phi^{-1/2} \alpha_{n} & (\Phi^{-1/2} \alpha_{n})^{\prime}\end{pmatrix},\\
    &= b_{n} \Phi^{-1} \det\begin{pmatrix} \alpha_{n-1} & \alpha_{n-1}^{\prime} \\ \alpha_{n} & \alpha_{n}^{\prime} \end{pmatrix} - \frac{b_{n}}{2} \Phi^{\prime} \Phi^{-2} \det\begin{pmatrix} \alpha_{n-1} & \alpha_{n-1} \\ \alpha_{n} & \alpha_{n} \end{pmatrix},\\
    &= b_{n} \Phi^{-1} \left[\alpha_{n-1}\alpha_{n}^{\prime} - \alpha_{n-1}^{\prime} \alpha_{n}\right].
\end{align*}
Thus the term proportional to $\Phi^{\prime}(\omega)$ vanishes as claimed due to rank deficiency. This argument does not rely on the detailed form of the polynomial asymptotics, other than the fact that $p_{n-1}(\omega)$ and $p_{n}(\omega)$ are both pointwise proportional to $\Phi(\omega)^{-1/2}$, and so will generalize to other frequency regimes discussed in subsequent sections.

Now, carrying out the relevant derivatives of \cref{eq:pnm1_bulk_asymptotic,eq:pn_bulk_asymptotic} and rearranging, we get the second of our bootstrap equations:
\begin{equation}
    \boxed{
     \sigma_{n}(\omega) \approx \dfrac{\Phi(\omega)}{2\pi} K_{n}(\omega,\omega) + (-1)^{n} \dfrac{\beta_{n} \cos[2\pi I_{n}(\omega)]}{2\pi(\beta_{n}^{2} - \omega^{2})}.
     \label{eq:sigma_bulk}}
\end{equation}
The kernel $K_{n}(\omega,\omega)$ can be calculated in terms of the orthogonal polynomials using \cref{eq:diagonal_CD}, which we remind the reader can be computed exactly in terms of the Lanczos coefficients $\{b_{k}\}_{k=1}^{n}$ using the three-term recurrence~\cref{eq:three_term_recursion}. \cref{eq:sigma_bulk} was proven by different means for a similar class of spectral functions in \cite[Theorem 9.5]{levinOrthogonalPolynomialsExponential2012}, but only with the first term on the RHS, which dominates the second term by a factor of $\mathcal{O}(\beta_{n})$ when $\omega$ is $\mathcal{O}(1)$.

Since $\partial_{\omega} I_{n}(\omega) = \sigma_{n}(\omega)$, \cref{eq:Phi_bulk,eq:sigma_bulk} constitute a first-order ordinary differential equation for $I_{n}(\omega)$ which is closed at leading order in $n$, and whose solution yields $\Phi(\omega)$ as a byproduct. Crucially, we also have the initial condition $I_{n}(0) = 0$, allowing us to iteratively solve this differential equation and thereby obtain a finite-$n$ approximation to the spectral function. Given a choice of frequency spacing $\delta \omega \ll 1$ and a maximum frequency $\omega_{\mathrm{max}}$ satisfying $0 < \omega_{\mathrm{max}} \ll \beta_{n}$, the algorithm works as follows:
\begin{enumerate}
    \setcounter{enumi}{-1}
    \item Set $\omega = 0$ and $I_{n}(0) = 0$.
    \item Compute $\Phi(\omega)$ using \cref{eq:Phi_bulk}. 
    \item Compute $\sigma_{n}(\omega)$ using \cref{eq:sigma_bulk}.
    \item Set $I_{n}(\omega+\delta \omega) = I_{n}(\omega) + \sigma_{n}(\omega) \times \delta \omega$.
    \item Increment $\omega \mapsto \omega + \delta \omega$.
    \item Repeat steps 1-4 until $\omega = \omega_{\mathrm{max}}$, then terminate.
\end{enumerate}

As an example we will again consider the mixed field Ising model (MFIM). We will compare the results to those obtained through the `simple procedure' described at the start of the section where we replace the level-$n$ Green's function $G_{n}(z)$ by the universal `semicircle' form in \cref{eq:greens_func_bulk}, as well as to the spectral function resulting from Fourier transforming the real-time autocorrelation function $(\mathcal{J}|\mathcal{J}(t))/\sqrt{L}$ obtained using time-evolving block decimation (TEBD). For the spectral bootstrap we compute $n=40$ Lanczos coefficients, going up to $\omega_{\mathrm{max}} = 0.99 \beta_{n}$, with $\beta_{n} \approx 2 b_{n} \approx 46.1$ for $n=40$ in the MFIM. For the real-time evolution, we go up to $t_{\mathrm{max}} = 10$ using a timestep of $\diff t = 0.01$, with a maximum bond dimension of $\chi_{\mathrm{max}} = 512$ and a system size $L=201$ large enough that finite-size effects are negligible over these timescales. We can see from the main panel of \cref{fig:spectral_bootstrap_rho=0}(a) that the spectral bootstrap and semicircle approximation give numerically identical results throughout this bulk frequency range, providing an important self-consistency check. The spectral bootstrap also generally agrees well with the Fourier transform data at low to moderate frequencies. There is some deviation around $\omega = 0$, but it is hard to do a fair comparison because it is not clear how to systematically convert between $n$ and $t_{\mathrm{max}}$. As a rough indication of computational effort, our Julia implementation of the Lanczos algorithm took \mytexttilde\qty{30}{\minute} and \mytexttilde\qty{60}{\gibi\byte} RAM on a single CPU core to compute $n=40$ Lanczos coefficients for the MFIM, while our TEBD simulation performed using TeNPy~\cite{hauschildEfficientNumericalSimulations2018} took \mytexttilde\qty{30}{\hour} on 8 cores using \mytexttilde\qty{2}{\gibi\byte} RAM. However, the TEBD runtime is highly parameter dependent, and could be reduced by using a larger $\diff t$ or a smaller $\chi_{\mathrm{max}}$. At any rate, the initial $\omega=0$ step of the spectral bootstrap corresponds to the $\rho=0$ case of the procedure for extracting diffusion constants described in \cref{sec:transport_coeffs}, and in \cref{fig:spin_diffusion}(a) we already demonstrated quantitative accuracy in the $n \to \infty$ limit.

\begin{figure}[t]
    \centering
    \includegraphics[width=\columnwidth]{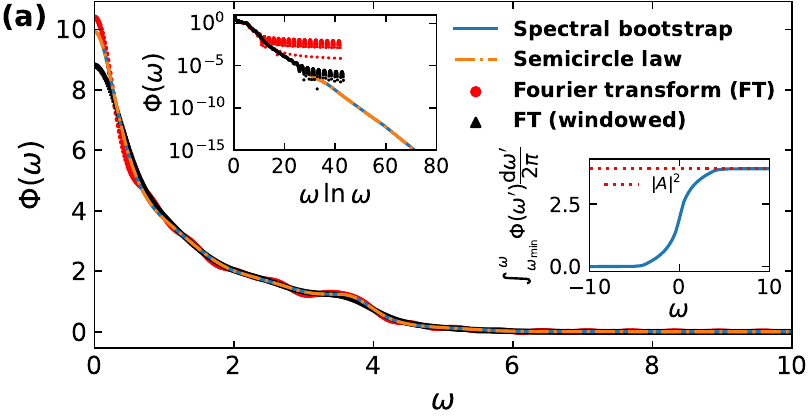}
    \includegraphics[width=0.495\columnwidth]{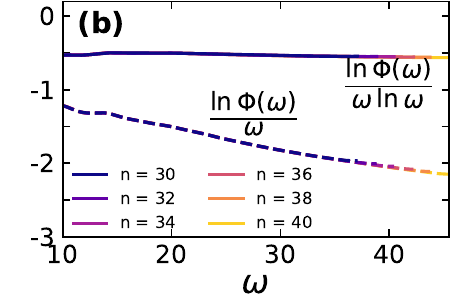}
    \includegraphics[width=0.485\columnwidth]{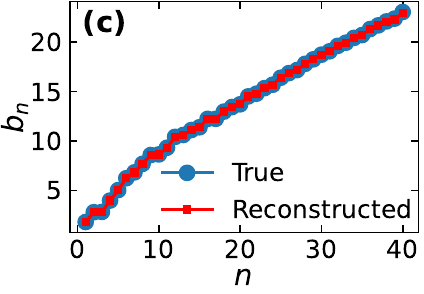}
    \caption{\textbf{(a)} The $n=40$ approximation to the spectral function of the energy current operator of the mixed field Ising model, with a comparison between the spectral bootstrap and approximating the level-$n$ Green's function by the `semicircle form' in \cref{eq:greens_func_bulk}, and the Fourier transform (FT) of the real-time correlator $C(t) = (\mathcal{J} | \mathcal{J}(t))/L$ up to $t_{\mathrm{max}}=10$. We also show the Fourier transform of $C(t)$ multiplied by the cosine window $\frac{1}{2}[1 + \cos(\pi t/t_{\mathrm{max}})]$ in order to reduce spectral leakage due to the finite $t_{\mathrm{max}}$. The left inset demonstrates the high-frequency scaling $\Phi(\omega) \sim \exp[-\mathcal{O}(\omega \log{\omega})]$ expected in one spatial dimension~\cite{parkerUniversalOperatorGrowth2019}. The right inset verifies the sum rule for the approximate spectral function obtained using the spectral bootstrap. \textbf{(b)} The spectral function $\Phi(\omega)$ obtained using the spectral bootstrap, rescaled in two ways, as evidence that $\Phi(\omega)$ decays like $\Phi(\omega \to \infty) \sim \exp[-\mathcal{O}(\omega \log{\omega})]$ rather than $\Phi(\omega \to \infty) \sim \exp[-\mathcal{O}(\omega)]$.
    \textbf{(c)} Reconstructed Lanczos coefficients $b_{n}$, calculated using the estimate of the spectral function obtained using the spectral bootstrap, and compared with the exact Lanczos coefficients. The close agreement is strong evidence that the spectral bootstrap is producing a good approximation to the true spectral function.
    }
    \label{fig:spectral_bootstrap_rho=0}
\end{figure}

More noticeable differences between the spectral bootstrap and the Fourier transformed correlator can be found at high frequencies, as shown in the left inset of \cref{fig:spectral_bootstrap_rho=0}(a). What is striking about the spectral bootstrap data is that we are able to resolve the log-correction due to geometric constraints on operator growth in 1D~\cite{parkerUniversalOperatorGrowth2019}: Lanczos coefficients growing like $b_{n} \sim n / \log{n}$ correspond to a spectral function decaying like $\Phi(\omega \to \infty) \sim \exp[-\mathcal{O}(\omega \log{\omega})]$ (c.f.~\cref{eq:beta_n_asymptotics}), rather than the generic exponential bound $\Phi(\omega \to \infty) \sim \exp[-\mathcal{O}(\omega)]$ expected in higher dimensions. Further evidence for this is shown in \cref{fig:spectral_bootstrap_rho=0}(b), where we find that $\log[\Phi(\omega)] / \omega \log{\omega}$ plateaus at large frequencies, while without the log correction $\log[\Phi(\omega)] / \omega$ continues to decay. The estimates appear to be well converged with $n$ provided we restrict to $|\omega| < \beta_{n}$ (which is why the smaller $n$ curves stop at smaller $\omega$). Convincingly resolving this log-correction had previously required going to much larger $n$ using a Monte-Carlo algorithm that exploits the lack of a sign-problem for operator growth in the MFIM~\cite{caoStatisticalMechanismOperator2021,deStochasticSamplingOperator2024}.

Getting the correct high-frequency tail from the Fourier transformed real-time correlation function can be challenging due to spectral leakage arising from cutting off the Fourier transform integral at a finite $t_{\mathrm{max}}$. Indeed, we can see in the left inset to \cref{fig:spectral_bootstrap_rho=0}(a) that the spectral function from the Fourier transform does not have the expected quasi-exponential high-frequency decay. The standard fix to this problem of spectral leakage is to multiply the real-time correlator by a windowing function before taking the Fourier transform~\cite{whiteSpectralFunction$S1$2008}. We also show in \cref{fig:spectral_bootstrap_rho=0}(a) the results from multiplying by a simple finite cosine window $\frac{1}{2} \left[1 + \cos(\pi t / t_{\mathrm{max}})\right]$. This extends the maximum frequency at which the Fourier transform data agrees with the high frequency tail of the spectral bootstrap, but still eventually leads to an unphysical plateau at very high frequencies. This high frequency tail could likely be improved further by both reducing $\diff t$ and increasing $t_{\mathrm{max}}$, at the cost of increased computational resources. \cref{fig:spectral_bootstrap_rho=0}(b) shows that the high-frequency tail of the spectral bootstrap is converged with $n$ at least up to $\omega \approx 40$, so this Lanczos approach appears to be particularly well suited to reliably extracting high-frequency information.

To further verify the accuracy of the spectral function approximation $\Phi_{\mathrm{est}}(\omega)$ obtained using the spectral bootstrap, we attempt to reconstruct the exact sequence of Lanczos coefficients by carrying out the three-term recurrence relation \cref{eq:three_term_recursion} for the orthogonal polynomials, but computing the norm
\begin{equation}
    b_{n}^{2} = \int_{\mathbb{R}} \big[\omega p_{n-1}(\omega) - b_{n-1} p_{n-2}(\omega)\big]^{2} \dfrac{\Phi(\omega)}{2\pi} \diff \omega
    \label{eq:bn_reconstruction}
\end{equation}
by numerically evaluating the integral with $\Phi(\omega)$ approximated by $\Phi_{\mathrm{est}}(\omega)$. We can only compute this integral up to the maximum frequency $\omega_{\mathrm{max}} = 0.99\beta_{n}$ for which we performed the spectral bootstrap, but this is sufficient to capture almost all of the norm contributing to $b_{n}$, for reasons discussed in \cref{sec:coulomb}. The result is shown in \cref{fig:spectral_bootstrap_rho=0}(c), where we find very close agreement between the exact and the reconstructed Lanczos sequence, thus providing strong evidence that the spectral bootstrap yields a good approximation to the true spectral function (see \cref{fig:bulk_universality}(c) for more examples).  As a final nontrivial check, the right inset to \cref{fig:spectral_bootstrap_rho=0}(a) shows that the sum rule $\int_{\mathbb{R}} \Phi(\omega) \mathrm{d}\omega/2\pi = (A|A)$ is verified very accurately. Using the spectral function extracted in the range $\omega \in [-0.99 \beta_{n}, 0.99 \beta_{n}]$, we find that the sum rule is satisfied up to a relative error of \mytexttilde$10^{-6}$.

Before proceeding, we remark that it is also possible to start this bulk bootstrap at some nonzero frequency $\omega_{0}$ as follows. Use the semicircle approximation (\cref{eq:greens_func_bulk}) to the level-$n$ Green's function $G_{n}(\omega_{0} \pm i 0^{+})$ to approximate the spectral function $\Phi(\omega_{0})$, as described at the start of \cref{sec:spectral_bootstrap}. Then numerically solve \cref{eq:Phi_bulk} to determine $I_{n}(\omega_{0}) \pmod{1}$ and use \cref{eq:sigma_bulk} to compute $\sigma_{n}(\omega_{0})$. Then one can proceed with the main routine of the bulk spectral bootstrap. This ability to start away from the special points $\omega = 0, \pm \beta_{n}$ is a particular feature of the bulk bootstrap, owing to the fact that the bulk Green's function in \cref{eq:greens_func_bulk} contains only a single free parameter, $\beta_{n}$, which can be directly estimated from the Lanczos coefficients.

\subsection{Generalization to arbitrary $\rho$: the Bessel bootstrap}
\label{sec:bessel_bootstrap}
The algorithm outlined in the previous section made use of the $n \to \infty$ asymptotics of the orthogonal polynomials. When $\rho \neq 0$ these asymptotics are modified, particularly near $\omega = 0$, the location of the power-law. Once we have the correct asymptotics, however, then the differential equation can be solved iteratively as before. We will aim to capture the `envelope function' $\Phi(\omega) / |\omega|^{\rho} \equiv 2\pi \exp[-Q(\omega)]$, since this is what is needed to determine the values of hydrodynamic transport coefficients.

For the initial step at $\omega = 0$, we use the result of \cref{thm:hydrodynamic_constants}, which we restate here, together with a result relating this to the equilibrium density at the origin:
\begin{widetext}
\begin{gather}
    e^{-Q(0)} \approx \dfrac{2 (1+\rho)^{\rho}}{\Gamma\left[\frac{1}{2}(1+\rho)\right]^{2}} \dfrac{[K_{n}(0,0)]^{\rho}}{\left[ \beta_{n} (p_{n-1}(0)^{2} + p_{n}(0)^{2})\right]^{1+\rho}},\label{eq:expmQ0}\\
    \sigma_{n}(0) \approx \dfrac{2}{\pi} \left(\Gamma\left[\frac{1}{2}(1+\rho)\right] \Gamma\left[\frac{1}{2}(3+\rho)\right] e^{-Q(0)} K_{n}(0,0)\right)^{\frac{1}{1+\rho}}.\label{eq:sigman0}
\end{gather}
These are the $\omega \to 0^{+}$ limits of the following equations which we will use for $\omega > 0$. For the envelope function we have
\begin{align}
    e^{-Q(\omega)} \approx \dfrac{1}{p_{n-1}(\omega)^{2} + p_{n}(\omega)^{2}}  \dfrac{1}{\sqrt{\beta_{n}^{2} - \omega^{2}}} \dfrac{\pi I_{n}(\omega)}{\omega^{\rho}} \Bigg[ &\left(J_{\frac{1}{2}(\rho-1)}^{2} + J_{\frac{1}{2}(\rho+1)}^{2}\right)(\pi I_{n}(\omega))\label{eq:expmQ_finite}\\
    - &(-1)^{n} \frac{\omega}{\beta_{n}}  \left(2 J_{\frac{1}{2}(\rho-1)} J_{\frac{1}{2}(\rho+1)}\right)(\pi I_{n}(\omega)) \cos\left\{\rho \, \mathrm{arcsin}\left(\frac{\omega}{\beta_{n}}\right)\right\}\Bigg],\nonumber
\end{align}
where $I_{n}(\omega)$ is defined as before in \cref{eq:In_def}, $J_{\alpha}$ is a Bessel function of the first kind, and we are using the shorthand $(J_{\alpha}^{2} + J_{\beta}^{2})(x) \equiv J_{\alpha}^{2}(x) + J_{\beta}^{2}(x)$, $(J_{\alpha} J_{\beta})(x) \equiv J_{\alpha}(x) J_{\beta}(x)$, etc. For the equilibrium density we have
\begin{equation}
    \sigma_{n}(\omega) \approx \dfrac{4}{\pi} \dfrac{K_{n}(\omega,\omega) \omega^{\rho} e^{-Q(\omega)}}{\pi I_{n}(\omega)} \dfrac{1}{\left[J_{\frac{1}{2}(\rho-1)}^{2} + J_{\frac{1}{2}(\rho+1)}^{2} - J_{\frac{1}{2}(\rho-3)} J_{\frac{1}{2}(\rho+1)} - J_{\frac{1}{2}(\rho-1)} J_{\frac{1}{2}(\rho+3)}\right](\pi I_{n}(\omega))}.\label{eq:sigman_finite}
\end{equation}
\end{widetext}
We derive these expressions in \cref{sec:spectral_bootstrap_derivation}. One can check that they reduce to \cref{eq:Phi_bulk,eq:sigma_bulk} upon sending $\rho \to 0$. Unlike for $\rho = 0$, where the asymptotics were valid out to $\omega = \mathcal{O}(\beta_{n})$, these asymptotics are in principle valid only for $\omega$ within some $\mathcal{O}(1)$ window around $\omega = 0$, essentially the `hydrodynamic' frequency regime where we expect the power-law to dominate the behavior of the spectral function. Note that, taking $\omega = \mathcal{O}(1)$, \cref{eq:expmQ_finite,eq:sigman_finite} have been slightly simplified by dropping some terms which are subleading as $n \to \infty$; see \cref{eq:expmQ_finite_full,eq:sigman_finite_full} for the full expressions. In practice we have found that these simplified expressions give numerical results very close to those of the full expressions, provided $n$ is large enough that $\omega/\beta_{n} \ll 1$.

The appearence of the Bessel functions in these equations is a signature of `Bessel universality', akin to that found in random matrix ensembles whose probability measures behave as a power-law near the origin~\cite{akemannUniversalityRandomMatrices1997,kuijlaarsUniversalityEigenvalueCorrelations2003}. For larger $\omega$, one could then switch to the asymptotics in the `bulk', which should then be valid up until $|\omega| \lesssim \beta_{n}$ like those stated for $\rho = 0$, but in practice we have found that just using the Bessel asymptotics gives good results for the frequency regime where the spectral function is non-negligible.

With these asymptotics in hand, we now have another first-order differential equation which we can iteratively solve as follows to obtain a finite $n$ approximation to the spectral function. We initialize by setting $\omega = 0$ and $I_{n}(0) = 0$, computing $e^{-Q(0)}$ via \cref{eq:expmQ0} and $\sigma_{n}(0)$ via \cref{eq:sigman0}, then setting $I_{n}(\delta \omega) = I_{n}(0) + \sigma_{n}(0) \times \delta \omega$ and incrementing $\omega \pluseq \delta \omega$. Then the main routine is conceptually identical to the procedure outlined for $\rho = 0$:
\begin{enumerate}
    \item Compute $e^{-Q(\omega)}$ using \cref{eq:expmQ_finite}. 
    \item Compute $\sigma_{n}(\omega)$ using \cref{eq:sigman_finite}.
    \item Set $I_{n}(\omega+\delta \omega) = I_{n}(\omega) + \sigma_{n}(\omega) \times \delta \omega$.
    \item Increment $\omega \mapsto \omega + \delta \omega$.
    \item Repeat steps 1-4 until $\omega = \omega_{\mathrm{max}}$, then terminate.
\end{enumerate}

To illustrate this generalized spectral bootstrap (SB), we try to reconstruct the envelope functions of some toy spectral functions of the form $\Phi(\omega) / 2\pi = |\omega|^{\rho} \exp[-Q(\omega)]$ for $\rho = -\frac{1}{2}$, again comparing with the recursion method (RM)~\cite{viswanathRecursionMethodApplication2013}. \cref{fig:bessel_bootstrap}(a) shows the comparison for $Q(x) = \frac{1}{2}(x^{6} - 5 x^{4} + 20 x^{2} +1)^{1/3}$. This function has no special significance other than the fact that $Q(x) \sim \mathcal{O}(x^{2})$, so we have $b_{n} \sim \mathcal{O}(\sqrt{n})$. Here there is an exactly solvable model available for the recursion method which has the relevant features of Gaussian decay and a zero frequency power-law~\cite{viswanathRecursionMethodApplication2013}. The model spectral function is $\wt{\Phi}(\omega) = \frac{2\pi / \omega_{0}}{\Gamma[\frac{1}{2}(1+\rho)]} |\omega / \omega_{0}|^{\rho} \exp[-(\omega/\omega_{0})^{2}]$, where $\omega_{0}$ and $\rho$ are tunable parameters. Crucially, its Lanczos coefficients are known to be $b_{2k-1} = \omega_{0} \sqrt{(2k-1+\rho)/2}$ and $b_{2k} = \omega_{0} \sqrt{2k/2}$, and its Green's function is $\wt{G}(z) = (i z / \omega_{0}^{2}) \exp[-(z/\omega_{0})^{2}] E_{\frac{1}{2}(1+\rho)}[-(z/\omega_{0})^{2}]$, where $E_{\alpha}(x)$ is the generalized exponential integral. We recover the envelope function from the RM using $e^{-Q(\omega)} = |\omega|^{-\rho} \Phi(\omega) / 2\pi$ with $\Phi(\omega) = 2 \Im{G(\omega - i \epsilon)}$, $\epsilon = 10^{-14}$. In this case, the SB and the RM with $\rho = -1/2$ give essentially identical results. We also show the result of $e^{-Q(\omega)} = |\omega|^{1/2} \Phi(\omega) / 2\pi$ where $\Phi(\omega)$ is computed using the SB with $\rho = 0$, showing that one can get large errors near $\omega = 0$ if one does not take into account the power-law, although at larger $\omega$ the agreement does become better.
\begin{figure}[t]
    \centering
    \subfigimg[width=\columnwidth]{(a)}{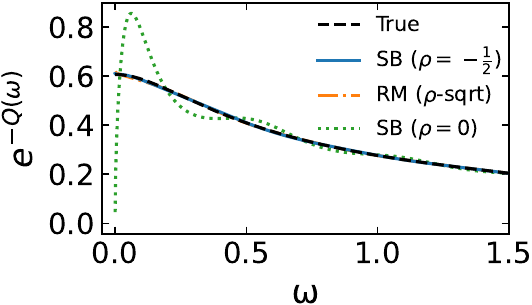}
    \subfigimg[width=\columnwidth]{(b)}{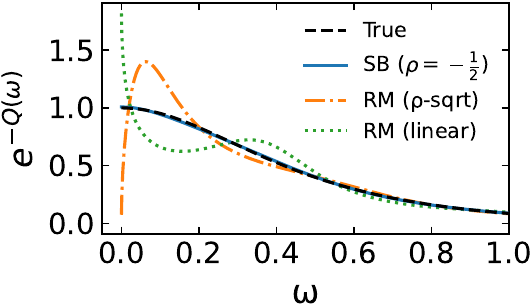}
    \caption{Reconstructing the spectral envelope function $\exp[-Q(\omega)]$ for a toy spectral function $\Phi(\omega)/2\pi = |\omega|^{-1/2} \exp[-Q(\omega)]$, using the spectral bootstrap (SB) and the recursion method (RM). \textbf{(a)} $Q(\omega) = \frac{1}{2}(\omega^{6} - 5 \omega^{4} + 20 \omega^{2} +1)^{1/3}$, which has $b_{n} \sim \mathcal{O}(\sqrt{n})$. Here there is an appropriate exact `stitching function' for the RM (see main text), so the RM and the $\rho = -1/2$ SB work equally well. The $\rho = 0$ SB curve illustrates the errors one can make by not accounting for the power-law. \textbf{(b)} Same as (a), except now $Q(\omega) = -\log{\sech{\pi \omega}}$, which has $b_{n} \sim \mathcal{O}(n)$.  Only the $\rho = -\frac{1}{2}$ SB accurately recovers the spectral function. Now there is no available stitching function with the appropriate power-law behavior, so the RM performs poorly. For the RM one has exact solutions with the right low frequency power-law but the wrong high frequency tail ($\rho$-sqrt), or no low frequency power-law but the right high frequency tail (linear), neither of which gives the correct result at low frequencies.}
    \label{fig:bessel_bootstrap}
\end{figure}

This first example showed that, when we have available an exactly solvable spectral function that captures the relevant features of the high-frequency decay and the low-frequency power-law, then the recursion method can work very well. One problem with the recursion method, though, is that the existence of these special solutions is not guaranteed for all problems of interest (see Ref.~\cite{viswanathRecursionMethodApplication2013} for a review). In order for the recursion method to work well, one needs to have an exact model solution for all three of the spectral function, the Green's function, and the Lanczos coefficients. Unfortunately, while there is a power-law solution available with Gaussian high-frequency decay, to our knowledge there is no such solution available which has both a low-frequency power-law and \textit{exponential} decay. This is potentially a problem because the operator growth hypothesis posits that such exponential decay is generic for chaotic quantum systems~\cite{parkerUniversalOperatorGrowth2019}. \cref{fig:bessel_bootstrap}(b) shows that when $\rho \neq 0$ it is important to get both of these ingredients correct to accurately capture the $\omega \to 0$ behavior of the spectral function. We consider the toy spectral function $\Phi(\omega)/2\pi = |\omega|^{-1/2} \sech(\pi \omega)$, which decays exponentially as $\omega \to \infty$. The $\rho = -\frac{1}{2}$ spectral bootstrap (SB) captures the true spectral function very well. However, since there is now no available exact solution, the RM performs poorly. One could try using the Meixner-Pollaczek solution from Ref.~\cite{parkerUniversalOperatorGrowth2019} to capture the high-frequency tail while ignoring the low frequency power-law, which is labeled `RM (linear)'. Or one could try using the exact solution from the previous example to get the low frequency power-law correct, at the cost of getting the high-frequency tail wrong, which is labeled `RM ($\rho$-sqrt)'. However, both of these approaches result in large errors at low frequencies, and so only the $\rho = -\frac{1}{2}$ spectral bootstrap correctly captures the true low frequency behavior of the spectral function. This highlights one of the advantages of the spectral bootstrap over the recursion method: it does not require any exact solutions, and so can be equally employed for a wide range of growth behaviors of the Lanczos coefficients.

From the discussion in \cref{sec:gf}, we know that for $\rho = 0$ there is a great deal of freedom in choosing the stitching function for the recursion method, and in particular it is \textit{not} necessary for the stitching function to have a high frequency decay matching that of the true spectral function; one only needs that the level-$n$ Green's function satisfies $G_{n}(\omega \pm i 0^{+}) \approx \mp i/b_{n}$ as $\omega \to 0$ and $n \to \infty$, which is much more generic. One then might wonder why for $\rho \neq 0$ we no longer have such freedom, since it appears to be necessary to also match the high frequency behavior (the relevance of matching the low frequency behavior is perhaps more intuitive). This is because the nontrivial low-frequency behavior of $\Phi(\omega \to 0) \sim |\omega|^{\rho}$ couples to the equilibrium density $\sigma_{n}(0)$ when $\rho \neq 0$ (c.f.~\cref{thm:recurrence_theorem,eq:exact_pn0_scaling}), and $\sigma_{n}(0)$ depends nontrivially on the high frequency decay of the spectral function (c.f.~\cref{lem:hn0_scaling}), being sensitive to a Coulomb gas confinement transition. This coupling is seen most explicitly in \cref{eq:n_gf_asymptotic}, where we show that the level $2n$ Green's function $G_{2n}(\omega \pm i 0^{+})$ is primarily a function of $\pi \sigma_{2n}(0) \omega$ as $\omega \to 0$ and $n \to \infty$.

\section{Random matrix universality}
\label{sec:universality}
While we developed the spectral bootstrap with the practical aim of recovering the spectral function $\Phi(\omega)$ from its Lanczos coefficients, as a byproduct we also get an estimate of the equilibrium measure $\sigma_{n}(\omega)$. We can use this to perform various self-consistency checks. In particular, while we cannot usually verify that a given many-body spectral function obeys all the assumptions under which we derive our results, we can at least check that our estimates of $\Phi(\omega)$ and $\sigma_{n}(\omega)$ are consistent with the behavior expected given those assumptions. One notable example is the emergence of `universality' akin to that governing eigenvalue correlations in random matrix theory (RMT)~\cite{kuijlaarsUniversality2011}. Just as universality appears in the $n \to \infty$ of $n \times n$ random matrices, here it should appear in the limit $n \to \infty$ of large Lanczos index.

From a physical point of view, this RMT-like universality is interesting because it seems to be `superuniversal', in the sense that it can appear not only for chaotic models, which might be expected from studies of eigenstate thermalization~\cite{dalessioQuantumChaosEigenstate2016}, but also in integrable and non-interacting models, provided they have sufficiently smooth spectral functions in the thermodynamic limit. What changes between these classes of operator dynamics is the $n$-dependence of the frequency scale on which this RMT universality appears, due to the different scaling of $b_{n}$.

We expect this section to be of most interest to readers already familiar with universality of eigenvalue correlations in random matrix theory. Nonetheless, we first provide the necessary background before illustrating our results.

\subsection{Background}
\label{sec:universality_background}
Given a spectral function $\Phi(\omega)$, we can define a unitarily-invariant ensemble of $n \times n$ Hermitian random matrices whose eigenvalues are distributed according to
\begin{equation}
    P(\bm{\lambda}) \diff^{n} \bm{\lambda} = \dfrac{1}{Z_{n}} \left( \prod_{i<j} |\lambda_{i} - \lambda_{j}|^{2}\right) \prod_{i} \dfrac{\Phi(\lambda_{i})}{2\pi} \diff \lambda_{i},
\end{equation}
where $\bm{\lambda} = (\lambda_{1},\dots,\lambda_{n})$ with the ordering $\lambda_{1} \leq \lambda_{2} \leq \cdots \leq \lambda_{n}$, and $Z_{n}$ is a normalization constant. The bracketed prefactor is the Vandermonde determinant which gives rise to level repulsion. We consider this random matrix ensemble because its eigenvalue distribution will be governed by the same Coulomb gas as that controlling the Lanczos operators via the orthogonal polynomials $p_{n}(\omega)$.

For any $1 \leq m \leq n-1$, let
\begin{align}
    &R_{m}(\lambda_{1}, \dots, \lambda_{m}) =\\
    &\dfrac{n!}{(n-m)!} \int_{\mathbb{R}} \cdots \int_{\mathbb{R}} P(\lambda_{1}, \dots, \lambda_{m}, \lambda_{m+1}, \dots, \lambda_{n}) \diff \lambda_{m+1} \cdots \diff \lambda_{n} \nonumber
\end{align}
denote the $m$-point eigenvalue correlation function. A remarkable property of these random matrix ensembles is that all of the correlation functions can be written in terms of a determinant of a 2-point correlation kernel~\cite{mehtaRandomMatrices2004},
\begin{equation}
    R_{m}(\lambda_{1}, \dots, \lambda_{m}) = \mathrm{det}\left[\left(\hat{K}_{n}(\lambda_{i}, \lambda_{j})\right)_{1 \leq i,j \leq m}\right].
\end{equation}
The correlation kernel $\hat{K}_{n}$ is a weighted version of the Christoffel-Darboux kernel $K_{n}$ defined in \cref{eq:kernel_def},
\begin{align}
    \hat{K}_{n}(\omega_{a},\omega_{b}) &= \sqrt{w(\omega_{a}) w(\omega_{b})} K_{n}(\omega_{a},\omega_{b}),
\end{align}
where $w(\omega) = \Phi(\omega) / 2\pi$ is the weight function defining the random matrix ensemble. In the Lanczos language, from \cref{eq:lanczos_vector} we can see that $\hat{K}_{n}(\omega_{a},\omega_{b})$ is the integral kernel (in frequency space) of the projection superoperator $\mathcal{P}_{n}$ on to the first $n$ Lanczos operators~\cite{deiftOrthogonalPolynomialsRandom2000}
\begin{align}
    \mathcal{P}_{n} &= \sum_{m=0}^{n-1} |O_{m})(O_{m}|,
\end{align}
which is the complement of the projector $\mathcal{Q}_{n} = \mathds{1} - \mathcal{P}_{n}$ used to construct the level-$n$ Green's function $G_{n}(z)$ (\cref{sec:gf}). By studying $\hat{K}_{n}(\omega_{a},\omega_{b})$, we can characterize this slow/fast operator projection in frequency space.

A further remarkable property of random matrix ensembles is that for large-$n$ their eigenvalue correlations often display \textit{universality}~\cite{kuijlaarsUniversality2011}. This manifests in $\hat{K}_{n}(\omega_{a},\omega_{b})$ approaching a universal form in different sections of the spectrum when probed on the appropriate scale, independent of the precise form of the weight $w$. For example, for $\omega$ in the `bulk' of the spectrum, $|\omega| \ll \beta_{n}$, we get 
\begin{equation}
    \dfrac{1}{\sigma_{n}(\omega)} \hat{K}_{n}\left(\omega + \dfrac{u}{\sigma_{n}(\omega)}, \omega + \dfrac{v}{\sigma_{n}(\omega)}\right) \xrightarrow{n \to \infty} \mathbb{S}(u,v),
    \label{eq:bulk_universality}
\end{equation}
for $u,v \in \mathbb{R}$, so for large-$n$ the local eigenvalue correlations are described by the \textit{sine kernel}
\begin{equation}
    \mathbb{S}(u,v) = \dfrac{\sin[\pi (u-v)]}{\pi (u-v)}.
    \label{eq:sine_kernel}
\end{equation}
Note that the RHS of \cref{eq:bulk_universality} is translation invariant, provided $\omega$ is within the bulk of the spectrum.

If the weight function has a power-law at the origin, $w(x) = |x|^{\rho} e^{-Q(x)}$, then this affects the local eigenvalue correlations near the origin. Rather than the sine kernel, at $\omega = 0$ we instead have \textit{Bessel universality}:
\begin{equation}
    \dfrac{1}{\sigma_{n}(0)} \hat{K}_{n}\left(\dfrac{u}{\sigma_{n}(0)}, \dfrac{v}{\sigma_{n}(0)}\right) \xrightarrow{n \to \infty} e^{-\frac{i \rho}{2}(\arg{u} + \arg{v})} \mathbb{J}_{\rho/2}(u,v),
    \label{eq:bessel_universality}
\end{equation}
where $\mathbb{J}_{\rho/2}$ is the Bessel kernel given by
\begin{equation}
    \mathbb{J}_{\rho/2}(u,v) = \pi \sqrt{u} \sqrt{v} \dfrac{J_{\frac{\rho+1}{2}}(\pi u) J_{\frac{\rho-1}{2}}(\pi v) - J_{\frac{\rho-1}{2}}(\pi u) J_{\frac{\rho+1}{2}}(\pi v)}{2(u-v)}
    \label{eq:bessel_kernel_def}
\end{equation}
with $J_{\alpha}$ a Bessel function of the first kind.

Finally, for eigenvalues near the edge of the spectrum, $\omega \approx \beta_{n}$, the bulk sine universality gives way to \textit{Airy universality}, such that
\begin{equation}
    \dfrac{\beta_{n}}{c_{n} n^{2/3}} \hat{K}_{n}\left(\beta_{n} + \dfrac{\beta_{n}}{c_{n} n^{2/3}} u, \beta_{n} + \dfrac{\beta_{n}}{c_{n} n^{2/3}} v\right) \xrightarrow{n \to \infty} \mathbb{A}(u,v),
    \label{eq:edge_universality}
\end{equation}
where $c_{n} \coloneqq (h_{n}(1) / \sqrt{2})^{2/3}$ is $\mathcal{O}(1)$ (note $c_{n} n^{2/3} = f_{n}^{\prime}(1)$ by~\cref{eq:fnp1}), and the Airy kernel $\mathbb{A}$ is given by
\begin{equation}
    \mathbb{A}(u,v) = \dfrac{\Ai(u)\Ai^{\prime}(v) - \Ai(v) \Ai^{\prime}(u)}{u-v},
    \label{eq:airy_kernel_def}
\end{equation}
where $\Ai$ is the Airy function of the first kind.

\begin{mdframed}
\begin{theorem}
    For our class of spectral functions (see \cref{sec:potential_definitions} for a definition), the corresponding correlation kernels $\hat{K}_{n}$ satisfy the universal scalings \cref{eq:bulk_universality,eq:bessel_universality,eq:edge_universality} in the relevant sections of the spectrum.
\end{theorem}
\end{mdframed}
This theorem follows as a direct consequence of the $n\to\infty$ asymptotics we established for the orthogonal polynomials, such as \cref{eq:pnm1_bulk_asymptotic,eq:pn_bulk_asymptotic} in the bulk of the spectrum. See \cref{sec:polynomial_asymptotics} for the full asymptotic expressions. %

\subsection{Universality in quantum operator dynamics}
While we can prove the emergence of universality for spectral functions obeying our assumptions, it would be gratifying to see this also emerge in more familiar many-body quantum models, even though we cannot explicitly verify that our assumptions hold. Indeed, universality is a very generic phenomenon~\cite{erdosUniversalityWignerRandom2011,lubinskyBulkUniversalityHolds2012,eichingerNecessarySufficientConditions2024}, particularly in the bulk, and it may well hold even if some of our assumptions are violated (such as our analyticity requirements). Using the spectral bootstrap algorithm described in \cref{sec:spectral_bootstrap}, we can estimate both the weight function $w(\omega) = \Phi(\omega)/2\pi$ and the equilibrium measure $\sigma_{n}(\omega)$ for a given many-body model, using only a finite number of its Lanczos coefficients. We can then perform a self-consistent check that these approximations reproduce the universal scalings \cref{eq:bulk_universality,eq:bessel_universality,eq:edge_universality}. 

The evidence we will present for universality in many-body models comes with a qualification. To test for universality requires explicitly knowing the spectral function $\Phi(\omega)$ and the equilibrium measure $\sigma_{n}(\omega)$, but for most interacting many-body models this is not analytically tractable. For models where we do not know $\Phi(\omega)$, we use our spectral bootstrap algorithm to first obtain a finite-$n$ approximation to $\Phi(\omega)$ and $\sigma_{n}(\omega)$. The qualification is that this algorithm \textit{assumes} the validity of Plancherel-Rotach asymptotics for the orthogonal polynomials (e.g.~\cref{eq:pn_bulk_asymptotic,eq:pnm1_bulk_asymptotic}), from which universality can be derived as a consequence~\cite{deiftStrongAsymptoticsOrthogonal1999}. (But we do prove these asymptotics hold under reasonable assumptions on the spectral function.) To avoid circularity, we perform independent cross-checks of these estimates of $\Phi(\omega)$ and $\sigma_{n}(\omega)$, such as the ability to reconstruct the numerically exact Lanczos coefficients (\cref{fig:bulk_universality}(c)). If these checks are passed convincingly, which turns out to be the case for the physical models we consider, then this gives strong evidence that we have obtained accurate estimates of the true spectral function and equilibrium measure. With those in hand, it is then reasonable to test for universality.

Before proceeding, we make a practical consideration. The statements of universality in \cref{eq:bulk_universality,eq:bessel_universality,eq:edge_universality} are made under the approximation that the equilibrium density $\sigma_{n}(\omega)$ is locally constant. This approximation becomes increasingly well justified as $n\to\infty$, but for the finite $n$ we have available to us numerically, this approximation is not necessarily so good, particularly near the origin for exponentially decaying spectral functions which are at the Coulomb gas confinement transition (see \cref{fig:bulk_universality}(b) for some physical examples). We will adopt a practice that is standard in the random matrix theory literature, namely `unfolding' the spectrum~\cite{brodyRandommatrixPhysicsSpectrum1981}, which involves transforming to new coordinates that account for variation of the local density. Fixing a frequency $\omega$ at which we want to test for universality, we define
\begin{align}
    F_{n,\omega}(x) &\coloneqq \int_{\omega}^{x} \sigma_{n}(s) \diff s, \label{eq:Fn_def}\\
    &= I_{n}(x) - I_{n}(\omega), \nonumber
\end{align}
where $I_{n}(\omega) = \int_{0}^{\omega} \sigma_{n}(s) \diff s$ was defined in \cref{eq:In_def}. Now note that, if we approximate $\sigma_{n}(s) \approx \sigma_{n}(\omega)$ as constant in \cref{eq:Fn_def}, then each of the arguments of the kernel $\hat{K}_{n}$ in \cref{eq:bulk_universality} can be expressed in terms of the inverse function $F_{n,\omega}^{-1}(u) \approx \omega + u/\sigma_{n}(\omega) $. Then the statement \cref{eq:bulk_universality} of bulk universality at frequency $\omega$ can be rewritten as
\begin{equation}
    \left(\dfrac{F_{n,\omega}^{-1}(u) - F_{n,\omega}^{-1}(v)}{u-v}\right) \hat{K}_{n}\Big(F_{n,\omega}^{-1}(u), F_{n,\omega}^{-1}(v)\Big) \xrightarrow{n \to \infty} \mathbb{S}(u,v).
    \label{eq:unfolded_bulk_universality}
\end{equation}
The point is that, even in the situation where the equilibrium density $\sigma_{n}(\omega)$ is not constant and $F_{n,\omega}$ must be inverted numerically, this formulation of bulk universality will account for that variation, and so one can still expect to observe the sine kernel. (The inverse $F_{n,\omega}^{-1}$ is well-defined because $F_{n,\omega}$ is the integral of a positive function $\sigma_{n}(\omega) > 0$, so is increasing.) The statement \cref{eq:bessel_universality} of Bessel universality can be expressed in the same way by setting $\omega=0$, with only the RHS changing to the Bessel kernel instead of the sine kernel. The statement \cref{eq:edge_universality} of Airy universality around $\omega=\beta_{n}$ can also be made with the same LHS, but with $F_{n,\beta_{n}}(x)$ redefined as

\begin{equation}
    F_{n,\beta_{n}}(x) \coloneqq f_{n}(x/\beta_{n}), \mathclap{\hspace{-31.5em} \text{(Airy)}}
    \label{eq:airy_unfolding_def}
\end{equation}
where $f_{n}$ is defined in \cref{eq:fn_In_relation}.

\subsubsection{Sine universality in the bulk}
\label{sec:bulk_universality}
\begin{figure}[t]
    \centering
    \includegraphics[width=0.52\columnwidth]{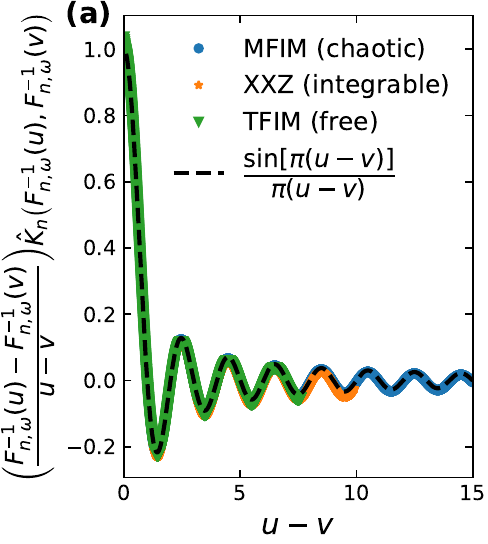}
    \includegraphics[width=0.45\columnwidth]{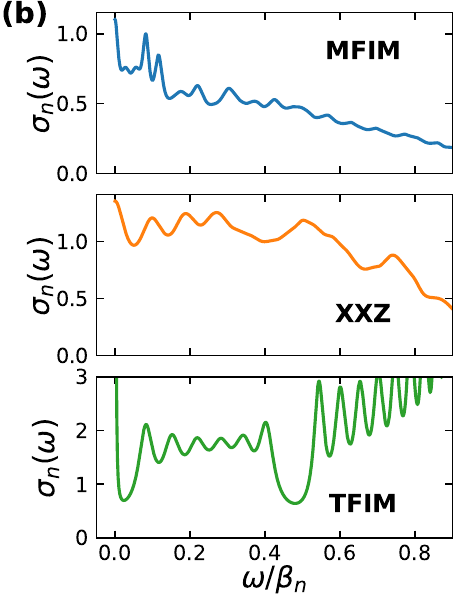}
    \includegraphics[width=\columnwidth]{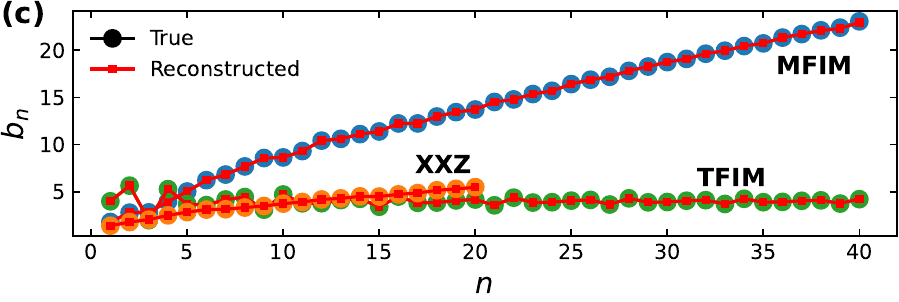}
    \includegraphics[width=\columnwidth]{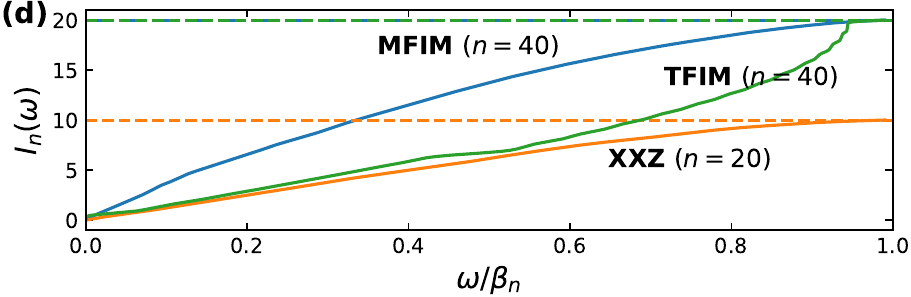}
    \caption{\textbf{(a)} Sine kernel universality [c.f.~\cref{eq:unfolded_bulk_universality,eq:sine_kernel}] in the bulk spectra of the chaotic mixed field Ising model (MFIM), the interacting integrable XXZ chain, and the non-interacting transverse field Ising model (TFIM). There are no free fitting parameters; everything is extracted from the Lanczos coefficients using the spectral bootstrap (\cref{sec:bulk_bootstrap}).
    \textbf{(b)} Equilibrium densities $\sigma_{n}(\omega)$ estimated using the spectral bootstrap algorithm. Note the significant variation of $\sigma_{n}(\omega)$ with $\omega$, which necessitates the unfolding procedure.
    \textbf{(c)} Reconstructed Lanczos coefficients $b_{n}$, calculated using the estimate of the spectral function obtained using the spectral bootstrap, and compared with the exact Lanczos coefficients. \textbf{(d)} Cumulative integrals $I_{n}(\omega) = \int_{0}^{\omega} \sigma_{n}(\omega^{\prime}) \diff \omega^{\prime}$ of the equilibrium measures from panel \textbf{(b)}, showing that they integrate to $I_{n}(\omega=\beta_{n}) = n/2$, as expected from \cref{eq:coulomb_gas_def,eq:sigman_support}.}
    \label{fig:bulk_universality}
\end{figure}
First we check for sine universality in the bulk of the spectrum, taking as examples the chaotic mixed field Ising model (MFIM), the interacting integrable XXZ chain, and the non-interacting transverse field Ising model (TFIM). We take the initial Lanczos operator to be the total energy current for the MFIM, the total spin current for the $\Delta=2$ XXZ model, and $\frac{1}{\sqrt{L}}\sum_{x} Y_{x} Y_{x+1}$ for the TFIM (with $XX$ interactions), and compute $n=40,20,40$ Lanczos coefficients respectively. Since in each case we expect the autocorrelation function $C(t)$ to decay faster than $1/t$, such that $0< |\Phi(\omega=0)| < \infty$, we set $\rho=0$ and start testing for bulk universality at $\omega=0$. Note that, for $\rho=0$, our assumptions on the spectral function (\cref{sec:potential_definitions}) required for our Riemann-Hilbert analysis amount to assuming analyticity of $\Phi(\omega)$ at $\omega=0$, which does not explicitly account for possible effects of power-law decay of $C(t)$ faster than $1/t$  on the analytic structure of $\Phi(\omega)$. The test of bulk universality we are about to perform thereby gives us a probe to see whether this non-analyticity affects the universality class near $\omega=0$.

We estimate the weight function $w(\omega) = \Phi(\omega)/2\pi$ and the equilibrium density $\sigma_{n}(\omega)$ throughout the bulk frequency spectrum by carrying out the `bulk bootstrap' algorithm, described in \cref{sec:spectral_bootstrap}, in the frequency range $\omega \in [0, 0.99\beta_{n}]$, with $\beta_{n=40} \approx 46.1$ for MFIM, $\beta_{n=20} \approx 11.2$ for XXZ, and $\beta_{n=40} \approx 8.5$ for TFIM. As a check of the accuracy of our estimate of $\Phi(\omega)$, like in \cref{fig:spectral_bootstrap_rho=0}(c) we again use it to reconstruct the Lanczos coefficients for these models by evaluating the three-term recurrence \cref{eq:three_term_recursion}, with the results shown in \cref{fig:bulk_universality}(c), where we find very good agreement. For the MFIM energy current, in \cref{fig:spectral_bootstrap_rho=0}(a) we also found good agreement with the spectral function obtained by Fourier transforming the real-time autocorrelation function obtained by TEBD, providing another independent accuracy test. To check the accuracy of the equilibrium measure $\sigma_{n}(\omega)$, as shown in \cref{fig:bulk_universality}(b), we compute the cumulative integral $I_{n}(\omega) =\int_{0}^{\omega} \sigma_{n}(\omega^{\prime}) \diff \omega^{\prime}$. By the definitions \cref{eq:coulomb_gas_def,eq:sigman_support} and the even symmetry $\sigma_{n}(-\omega) = \sigma_{n}(\omega)$, we should have the sum rule $I_{n}(\omega = \beta_{n}) = n/2$. As shown in \cref{fig:bulk_universality}(d), we find this is well satisfied in all three models. This sum rule is not enforced by the spectral bootstrap, so this is a nontrivial check of the accuracy of our estimates of $\sigma_{n}(\omega)$.
\begin{figure}[t]
    \centering
    \includegraphics[width=\columnwidth]{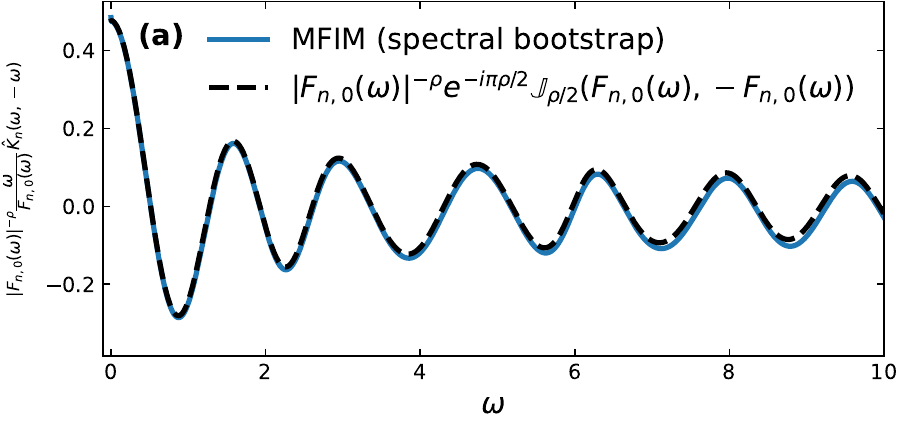}
    \includegraphics[width=\columnwidth]{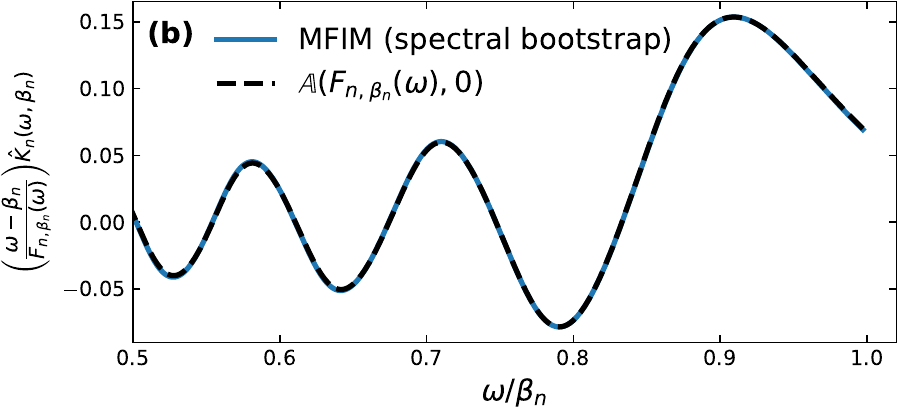}
    \caption{\textbf{(a)} Bessel universality near $\omega=0$ in the mixed field Ising model (MFIM). The initial Lanczos operator is the local energy density; since energy transport is diffusive in this model, the spectral function should scale like $\Phi(\omega\to 0) \sim |\omega|^{\rho}$ with $\rho = -\frac{1}{2}$. There is good agreement with the corresponding Bessel kernel $\mathbb{J}_{\rho/2}$ (\cref{eq:bessel_kernel_def}) in an $\mathcal{O}(1)$ region around the origin. \textbf{(b)} Airy universality near the spectral edge $\omega=\beta_{n}$ in the MFIM. We compare the results from the spectral bootstrap to the Airy kernel $\mathbb{A}$ (\cref{eq:airy_kernel_def}), finding very close agreement.}
    \label{fig:bessel_and_edge_universality}
\end{figure}

With this assurance, we then use these estimates of $\Phi(\omega)$ and $\sigma_{n}(\omega)$ to numerically compute the inverse function $F_{n,\omega}^{-1}$, and then evaluate the LHS of \cref{eq:unfolded_bulk_universality}, with the kernel $K_{n}(x,y) = \sum_{m=0}^{n-1} p_{m}(x) p_{m}(y)$ computed from the exact recurrence coefficients $\{b_{n}\}$ using the three-term recurrence \cref{eq:three_term_recursion}.
The results are shown in \cref{fig:bulk_universality}(a), where we see that all three models collapse on to the sine kernel $\mathbb{S}(u,v) = \sin[\pi(u-v)] / \pi(u-v)$, indicating the emergence of sine universality in the bulk. We emphasize that there are no free fitting parameters---everything is computed from the Lanczos coefficients---so getting such close agreement is nontrivial.
Note that the curve continues for larger separations $u-v$ for the MFIM than the XXZ model and then the TFIM simply because the frequency bandwidth $\beta_{n}$ is larger, allowing us to probe larger separations before hitting the edge of the spectrum.
This verification of bulk universality can be seen as an \textit{a posteriori} explanation of the effectiveness of our spectral bootstrap algorithm to estimate the spectral functions of these physical models.
\subsubsection{Bessel universality at the origin}
Next we check for Bessel universality near $\omega = 0$, which is expected when the spectral function has a power-law at the origin, $\Phi(\omega \to 0) \sim |\omega|^{\rho}$. We will consider the energy density $h_{0} = \frac{1}{2}(Z_{-1}Z_{0} + Z_{0}Z_{1}) + g_{x} X_{0} + g_{z} Z_{0}$ at site zero of the mixed field Ising model (MFIM), which should have $\Phi(\omega \to 0) \sim |\omega|^{-1/2}$ (i.e.~$\rho=-\frac{1}{2}$) due to energy diffusion~\cite{kimBallisticSpreadingEntanglement2013}. We compute $n=40$ Lanczos coefficients for this model, and then use the `Bessel bootstrap' described in \cref{sec:bessel_bootstrap} to extract the spectral function and the equilibrium measure. After checking again that these estimates successfully reproduce the true Lanczos coefficients (not shown), we then evaluate the LHS of \cref{eq:unfolded_bulk_universality} for $v=-u$ and using $F_{n,\omega=0}$ (note the RHS of \cref{eq:unfolded_bulk_universality} should now be replaced by $e^{-i \pi \rho/2} \mathbb{J}_{\rho/2}(u,-u)$). For clarity, we opt to plot as a function of the physical frequency $\omega$, with $u = F_{n,0}(\omega)$. Note that, due to the even symmetry of the spectral function, we have $v=-u = F_{n,0}(-\omega)$. To account for the $\omega\to 0$ divergent power-law of both the spectral function and the $\rho=-\frac{1}{2}$ Bessel kernel, we rescale both sides by $|u|^{-\rho} = |F_{n,0}(\omega)|^{-\rho}$ when plotting. The results are shown in \cref{fig:bessel_and_edge_universality}(a). We find good agreement with the Bessel kernel for $\mathcal{O}(1)$ frequencies, providing evidence for Bessel universality at $\omega=0$ for this model. There is gradually increasing disagreement for larger frequencies, but this is not surprising since Bessel universality is only expected to hold in an $\mathcal{O}(1)$-sized region around the origin (c.f.~\cref{fig:plancherel_schematic}(b) and \cref{sec:bessel_bootstrap}).

\subsubsection{Airy universality at the edge}
Finally we check for Airy universality at the edge of the spectrum, $\omega \approx \beta_{n}$. We consider again the mixed field Ising model, and take the energy current as our initial Lanczos operator, computing $n=40$ Lanczos coefficients. We use the `Airy bootstrap' of \cref{sec:airy_bootstrap} to approximate the spectral function $\Phi(\omega)$ and the equilibrium measure $\sigma_{n}(\omega)$ for frequencies in a range $\omega \in [\omega_{\mathrm{min}},\beta_{n}]$ close to the edge $\omega = \beta_{n}$. With these estimates, we evaluate the LHS of \cref{eq:unfolded_bulk_universality}, with the RHS now replaced by the Airy kernel $\mathbb{A}(u,v)$ (\cref{eq:airy_kernel_def}). Note that, for Airy universality, the unfolding map $F_{n}$ must now be defined according to \cref{eq:airy_unfolding_def}. Since our spectral bootstrap algorithm estimates the weight only up to the spectral edge $\omega = \beta_{n}$, we will fix $v=0$, where $F_{n,\omega=\beta_{n}}^{-1}(v=0) = \beta_{n}$, and take the other argument $u$ to be negative. Again for clarity we plot in terms of the physical frequency $\omega$, so negative $u=F_{n,\beta_{n}}(\omega)$ corresponds to frequencies below $\omega=\beta_{n}$. The results are shown in \cref{fig:bessel_and_edge_universality}(b). We find very good agreement with the Airy kernel, indicating the emergence of Airy universality near the edge of the spectrum in the mixed field Ising model.

\subsection{Scaling of the equilibrium measure}
\label{sec:eq_measure_scaling}
\begin{figure}[t]
    \centering
    \includegraphics[width=\columnwidth]{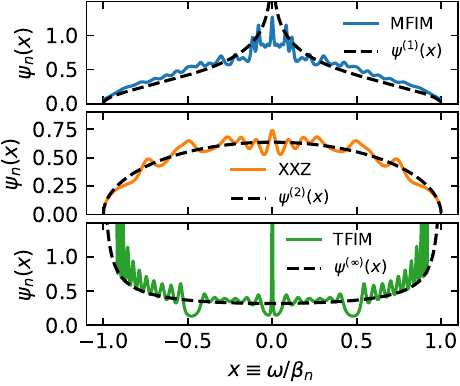}
    \caption{Rescaled equilibrium measures $\psi_{n}(x) = (\beta_{n}/n) \sigma_{n}(\beta_{n}x)$ for the mixed field Ising model (MFIM), the XXZ chain, and the transverse field Ising model (TFIM), as extracted using the spectral bootstrap. Given that they respectively have $b_{n} \sim \mathcal{O}(n / \log{n})$, $b_{n} \sim \mathcal{O}(\sqrt{n})$, and $b_{n} \sim \mathcal{O}(1)$, we compare with the equilibrium measures $\psi^{(p)}(x)$ for the Freud weights with equivalent growth rates, as indicated in \cref{eq:freud_p1,eq:freud_p2,eq:freud_pinf}. In each case, $\psi^{(p)}$ captures the qualitative shape of the equilibrium measure for the physical model, but there are significant fluctuations, which have implications for quantitative extraction of the spectral function.}
    \label{fig:rescaled_eq_measure}
\end{figure}

Since our class of spectral functions behaves like the Freud weights $w^{(p)}(x) = \exp(-\kappa_{p} |x|^{p})$ at large frequency scales (c.f.~\cref{eq:limit_to_freud}), it is also interesting to compare our estimates of the rescaled equilibrium measure $\psi_{n}(x) = (\beta_{n} / n) \sigma_{n}(\beta_{n} x)$ with the corresponding function $\psi^{(p)}(x)$ for the Freud weights, which is given by the Ullman distribution defined in \cref{eq:freud_eq_measure}~\cite{saffLogarithmicPotentialsExternal1997}.
For our physical test systems, we will use the same models as in \cref{sec:bulk_universality}---the mixed field Ising model (MFIM), the XXZ chain, and the transverse field Ising model (TFIM)---take the same initial Lanczos operators, and extract the equilibrium measure in the same way. Since these models have Lanczos coefficients growing like $b_{n} \sim \mathcal{O}(n/\log{n})$, $b_{n} \sim \mathcal{O}(n^{1/2})$, and $b_{n} \sim \mathcal{O}(1)$ respectively, we will compare to the following cases of the Ullman distribution:
\begin{align}
    \psi^{(p=1)}(x) &= \dfrac{1}{\pi} \artanh\left(\sqrt{1-x^{2}}\right),\label{eq:freud_p1}\\
    \psi^{(p=2)}(x) &= \dfrac{2}{\pi} \sqrt{1-x^{2}},\label{eq:freud_p2}\\
    \psi^{(p\to\infty)}(x) &= \dfrac{1}{\pi} \dfrac{1}{\sqrt{1-x^{2}}}.\label{eq:freud_pinf}
\end{align}
Notice that $\psi^{(p=2)}(0)$ and $\psi^{(p\to\infty)}(0)$ are finite, while $\psi^{(p=1)}(x)$ has a logarithmic divergence as $x \to 0$. This is an instance of the Coulomb gas confinement transition discussed in \cref{sec:coulomb}, which occurs at $p=1$. The divergence of $\psi^{(p\to \infty)}(x)$ as $|x| \to 1$ is consistent with the scaling $h_{n}(1) \xrightarrow{n \to \infty} 2p$ in \cref{lem:hn1_scaling}. For $p \to \infty$, the rescaled potential $V_{n}(x)=Q(\beta_{n}x)/n$ approaches a box potential which is zero for $|x|\leq 1$ and infinite outside, and the logarithmic repulsion between charges causes a build up of charge at the boundaries $x=\pm 1$ of the box.

The comparisons between the rescaled equilibrium measures $\psi_{n}(x)$ for the physical models and the corresponding Ullman distributions are shown in \cref{fig:rescaled_eq_measure}. Again we note there are no free fitting parameters. In each case the Ullman distribution $\psi^{(p)}(x)$ gives a good qualitative description of the large-scale profile of the equilibrium measure, but the physical equilibrium measures have significant fluctuations on smaller scales. These are a reflection of the fact that the nonuniversal structure of these spectral functions has not yet been washed out by the rescaling by $\beta_{n}$. Given these noticeable fluctuations, if one is aiming for quantitative accuracy in recovering a spectral function from Lanczos coefficients, it is not usually enough to simply approximate the equilibrium measure $\sigma_{n}(\omega)$ by a rescaled Ullman distribution $(n/\beta_{n}) \psi^{(p)}(\omega/\beta_{n})$. Rather, one has to account for these fluctuations by using something like the spectral bootstrap to more accurately extract the equilibrium measure.

\begin{figure}[t]
    \centering
    \includegraphics[width=\columnwidth]{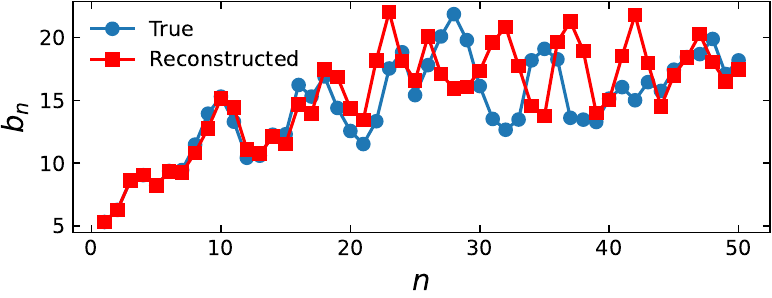}
    \caption{Reconstructed Lanczos coefficients $b_{n}$ for the disordered transverse-field Ising model, which is Anderson localized. The reconstructed Lanczos coefficients are calculated using the estimate of the spectral function obtained using the spectral bootstrap, and compared with the exact Lanczos coefficients. The large discrepancies between the exact Lanczos coefficients and the reconstructed coefficients suggests the failure of fast space universality in this disordered model, owing to the lack of smoothness of the spectral function.}
    \label{fig:bn_reconstruction_anderson_localized}
\end{figure}

We emphasize that we considered this class of `Freud-like' weights for our proofs because they are sufficiently regular to be amenable to rigorous proofs, but we expect the expressions for the polynomial asymptotics to be more universal, at least to leading order in $n \to \infty$. Thus we do not think there is a sense in which we have inadvertently assumed Freud-like behavior, so any agreement should be interpreted as a statement about the physical behavior of these models.

\section{The role of smoothness: failure of universality in a localized model}
\label{sec:smoothness}
In order to prove our results about emergent universality in the $n\to\infty$ limit, we needed to assume that the spectral function $\Phi(\omega)$ was very \textit{smooth} as a function of $\omega$ (see assumptions in \cref{sec:assumptions_summary}). In this section we give some evidence that this is not merely a requirement for our proof technique, but rather has physical relevance. In particular, we will demonstrate how a non-smooth spectral function can result in a failure of large-$n$ universality. To be concrete, we consider the 1D transverse-field Ising model (TFIM) with quenched disorder, defined by
\begin{equation}
    H = \sum_{x\in\mathbb{Z}} \left( X_{x} X_{x+1} + h_{x} Z_{x} \right),
\end{equation}
where the transverse fields $\{h_{x}\}$ are now independent random variables sampled from a Gaussian distribution $\mathcal{N}(0,\sigma^{2})$. In 1D, for any nonzero variance $\sigma^{2} \neq 0$, this model is provably Anderson localized~\cite{aizenmanRandomOperators2015}. As such, the spectrum of $H$ is pure point like even in the thermodynamic limit~\cite{goldshteinPurePointSpectrum1977}, as opposed to the continuous spectrum that one usually obtains in the absence of disorder. We fix $\sigma=4$ for this discussion, and consider a single disorder realization rather than averaging over disorder. We have checked that the following results do not depend significantly on the particular disorder realization.

We check for the effects of a point like spectrum by performing the same consistency check that we performed in \cref{fig:spectral_bootstrap_rho=0}c and \cref{fig:bulk_universality}c. That is, we \textit{assume} the validity of the universal scaling forms underlying the spectral bootstrap, and use that to produce an estimate $\Phi_{\mathrm{est}}(\omega)$ of the spectral function using the first $n=50$ Lanczos coefficients. With this estimate, we can then attempt to reconstruct the Lanczos coefficients solely by performing numerical integration weighted by $\Phi_{\mathrm{est}}$, as described in the text surrounding \cref{eq:bn_reconstruction}. If universality is valid in this model, then the reconstructed Lanczos coefficients should match the true ones used to estimate $\Phi_{\mathrm{est}}$. This reconstruction worked well for the clean version of the TFIM (\cref{fig:bulk_universality}c). However, as shown in \cref{fig:bn_reconstruction_anderson_localized}, for the disordered version the reconstruction fails badly, suggesting the breakdown of fast space universality in this model, owing to the lack of a smooth spectrum.

\clearpage
\bibliography{bibliography}

\clearpage
\onecolumngrid
\begin{center}
    \textbf{\large Supplemental Material}
\end{center}
\setcounter{equation}{0}
\setcounter{theorem}{0}
\setcounter{lemma}{0}
\setcounter{corollary}{0}
\setcounter{definition}{0}
\setcounter{example}{0}
\setcounter{figure}{0}
\setcounter{table}{0}
\setcounter{section}{0}
\counterwithin*{equation}{section}
\makeatletter
\renewcommand\theequation{\thesection.\arabic{equation}}
\renewcommand{\thefigure}{S\arabic{figure}}
\renewcommand{\bibnumfmt}[1]{[S#1]}
\renewcommand{\thetheorem}{S\arabic{theorem}}
\renewcommand{\thelemma}{S\arabic{lemma}}
\renewcommand{\thedefinition}{S\arabic{definition}}
\renewcommand{\theproposition}{S\arabic{proposition}}
\renewcommand{\theexample}{S\arabic{example}}
\renewcommand{\thesection}{S\arabic{section}}
\renewcommand{\thecorollary}{S\arabic{corollary}}

\renewcommand{\theHsection}{S\arabic{section}}
\renewcommand{\theHfigure}{S\arabic{figure}}
\renewcommand{\theHlemma}{S\arabic{lemma}}
\renewcommand{\theHequation}{\thesection.\arabic{equation}}

\def\TK{\textcolor{red}}

\section{Class of spectral functions (weights)}
For ease of reference, we first restate here our precise assumptions on the spectral function.
\subsection{Definition of potentials}
We want to consider weights $w(\omega) \equiv \Phi(\omega)/2\pi$ which decay at least exponentially at large $\omega$. We will decompose them as $w(\omega) \equiv |\omega|^{\rho} \exp[-Q(\omega)]$, where $Q$ is called the \textit{potential}. Given the exponential decay of $w$, $Q(\omega)$ should grow at least linearly as $|\omega| \to \infty$. To model this behavior, we will consider a class of weights inspired by the `very smooth Freud weights' of Refs.~\cite{lubinskyProofFreudConjecture1988,lubinskyUniformMeanApproximation1988}, which they denote by $\mathrm{VSF}(p)$, with $p$ an exponent governing the degree of the polynomial growth of $Q(\omega) \sim |\omega|^{p}$ as $|\omega| \to \infty$. Our weights will be a subset of $\mathrm{VSF}(p)$, where we add the requirement of analyticity, and also require a specification of the logarithmic corrections to the leading polynomial growth of $Q(\omega)$.

Our Riemann-Hilbert analysis draws heavily from Ref.~\cite{deiftStrongAsymptoticsOrthogonal1999}, where they take $Q$ to be a polynomial of even order. However, we are particularly interested in the marginal case where $Q(\omega \to \infty)$ grows linearly with $|\omega|$, since the Operator Growth Hypothesis~\cite{parkerUniversalOperatorGrowth2019} conjectures this to be generic for spectral functions in chaotic many-body quantum systems. But it is clearly not possible to simultaneously have i) $Q(\omega) \sim |\omega|$ as $|\omega| \to \infty$, and ii) $Q$ be a polynomial. This was a primary motivation for considering this generalized class $\mathrm{VSLF}(p,q)$ of `polynomial-like' weights. 
\begin{definition}[$\mathrm{VSLF}(p,q)$: log-Freud potentials of order $(p,q)$]
    Let $Q : \mathbb{R} \to \mathbb{R}$ be real-analytic, even, and satisfy
    \begin{equation}
        Q^{\prime}(\omega) > 0, \quad \text{for } \omega \text{ large enough},
        \label{eq:vsf_prop1}
    \end{equation}
    \begin{equation}
        \lim_{\omega \to \infty}\left(\dfrac{\omega Q^{\prime\prime}(\omega)}{Q^{\prime}(\omega)}\right) = p - 1,
        \label{eq:vsf_prop3}
    \end{equation}
    \begin{equation}
        \lim_{\omega\to\infty} \left(\log(\omega) \left[-p + \dfrac{\omega Q^{\prime}(\omega)}{Q(\omega)}\right]\right) = q.
        \label{eq:vslf_limit}
    \end{equation}
    for some $p > 0$ and $q \in \mathbb{R}$. Then we shall call $Q$ a \textit{log-Freud} potential of order $(p,q)$ and write $Q \in \mathrm{VSLF}(p,q)$.
\end{definition}
We will see in \cref{eq:vslf_characterization_prop2} that these potentials grow as $Q(\omega) \sim |\omega|^{p} (\log{|\omega|})^{q + o(1)}$ as $|\omega| \to \infty$, where the $o(1)$ in the exponent refers to scaling with $\omega$. In this sense assumptions \cref{eq:vsf_prop3,eq:vslf_limit} are similar to but slightly weaker than assuming $Q \in \Theta(|\omega|^{p} (\log{|\omega|})^{q})$. %
We make the assumption \cref{eq:vsf_prop1} for technical convenience; it allows us to prove that for large enough $n$ we need consider only the simplest case, where the support of the equilibrium measure consists of a single interval (see \cref{sec:equilibrium_measure}). Similar analyses have been performed in the more complicated case where the support consists of multiple disjoint intervals, but mostly for so-called `varying weights' of the form $w(x) = \exp[-n Q(x)]$ where the weight depends on $n$~\cite{deiftUniformAsymptoticsPolynomials1999,kuijlaarsUniversalityEigenvalueCorrelations2003}.

In order to apply Riemann-Hilbert techniques, we also need to assume that some of these properties continue to hold in a region of the complex plane near the real axis.
\begin{definition}[$\mathrm{CVSLF}(p,q,\theta,\gamma)$: complex log-Freud potentials of order $(p,q)$]
    For an angle $0 < \theta \leq \pi/2$, define the `complex cone' $C_{\theta}$ by
    \begin{equation}
        C_{\theta} \coloneqq \left\{z : |\arg{z}| < \theta \right\} \cup \left\{ z : |\arg{z}| > \pi - \theta\right\},
        \label{eq:complex_cone}
    \end{equation}
    using the convention $-\pi < \arg{z} \leq \pi$. We consider the open cone, so $z=0$ is not included in $C_{\theta}$. Now suppose there is some $0 < \theta \leq \pi/2$ and $\gamma > 0$ such that $Q \in \mathrm{VSLF}(p,q)$ can be analytically continued to $C_{\theta} \cup \left\{z : |z| < \gamma\right\}$, the union of $C_{\theta}$ and the disk of radius $\gamma$ centered at the origin (see \cref{fig:analyticity}).
    Also assume that \cref{eq:vsf_prop3,eq:vslf_limit} generalize to this region, in the sense that for $z$ restricted to $C_{\theta}$ we have
    \begin{equation}
        \lim_{|z|\to\infty} \dfrac{z Q^{\prime\prime}(z)}{Q^{\prime}(z)} = p-1,
        \label{eq:complex_assumption}
    \end{equation}
    \begin{equation}
              \lim_{|z|\to\infty} \left(\log(z) \left[-p + \dfrac{z Q^{\prime}(z)}{Q(z)}\right]\right) = q.
        \label{eq:complex_assumption_T}
    \end{equation}
    Given these properties, we say that $Q$ is a complex log-Freud potential of order $(p,q)$, and write $Q \in \mathrm{CVSLF}(p,q,\theta,\gamma)$.
\end{definition}

\fbox{\begin{minipage}{\textwidth}Our proofs will apply for $Q \in \mathrm{CVSLF}(p,q,\theta,\gamma)$ in all cases where the corresponding Hamburger moment problem is determined. This encompasses $p>1,q\in \mathbb{R}$, and $p=1,q>-1$.\end{minipage}}

\begin{example}
    All even polynomials with positive leading coefficient, $Q(x) = q_{2m} x^{2m} + \cdots$, $q_{2m} > 0$, lie in $\mathrm{CVSLF}(p,q,\theta,\gamma)$ with $p=2m$, $q=0$, $\theta=\pi/2$, $\gamma=\infty$.
\end{example}
\begin{example}
    Certain fractional powers of polynomials satisfy our assumptions, e.g. $Q(x) = \sqrt{1 + x^{2}}$ has $p=1$, $q=0$, $0 < \theta < \pi/2$, and $0 < \gamma < 1$. A similar example is $Q(x) = \sqrt{1+x^{2}} \log(1 + x^{2})$, which has $p=1,q=1$.
\end{example}
\begin{example}
    The symmetric Meixner-Pollaczek weights have $w(x) = \exp[-Q(x)] = \Gamma(\lambda + i x)\Gamma(\lambda -i x)$. ($Q$ can then be defined as in \cref{eq:Q_integral_def} below.) With $\lambda > 0$, we have $p=1$, $q=0$, $0 < \theta < \pi/2$, and $0 < \gamma < \lambda$. These weights appear in a rescaled form as the 2-point Wightman spectral function in the SYK model~\cite{maldacenaRemarksSachdevYeKitaevModel2016,dodelsonThermalProductFormula2024}, and were utilized in Ref.~\cite{parkerUniversalOperatorGrowth2019} to give an exactly solvable spectral function with exponential decay for use with the recursion method~\cite{viswanathRecursionMethodApplication2013}.
\end{example}
\begin{example}
    Taking $Q^{\prime}(z) = \erf(z) = (2/\sqrt{\pi}) \int_{0}^{z} e^{-t^{2}} \diff t$ gives an example of a $p=1$ potential which is also entire, unlike the previous $p=1$ examples. The conditions \cref{eq:complex_assumption,eq:complex_assumption_T} hold for $0<\theta < \pi/4$.
\end{example}

\subsection{Properties of log-Freud weights}
\label{sec:freud_properties}
First we adapt some results from Refs.~\cite{lubinskyProofFreudConjecture1988,lubinskyUniformMeanApproximation1988} to characterize the asymptotic behavior of these weights on the real line. Note that those references include a requirement on the 3rd derivative of $Q$ in the definition of their $\mathrm{VSF}(p)$ class, which we do not include in our $\mathrm{VSLF}(p,q)$ definition, since it is not necessary to prove any of the properties relevant for our purposes. When quoting results from Refs.~\cite{lubinskyProofFreudConjecture1988,lubinskyUniformMeanApproximation1988}, we will sometimes use the shorthand $Q \in \mathrm{VSF}(p)$, with the implicit understanding that the relevant proof uses only properties of the $\mathrm{VSF}(p)$ class that we do include in our definition of the $\mathrm{VSLF}(p,q)$ class.

\begin{definition}[$n$th Mhaskar-Rakhmanov-Saff (MRS) number $\beta_{n}$]
    Given an even potential $Q(x)$ and a positive integer $n$, define $\beta_{n}$ as the positive solution to the integral equation
    \begin{equation}
        \dfrac{1}{2\pi} \int_{-1}^{1} \dfrac{\beta_{n} s Q^{\prime}(\beta_{n} s)}{\sqrt{1-s^{2}}} \diff s = n.
        \label{eq:MRS_def_sup}
    \end{equation}
\end{definition}
As discussed in \cref{sec:coulomb} of the main text, the MRS number $\beta_{n}$ is important because it sets the bandwidth of the region in which the weighted orthogonal polynomials $p_{n}(x) w(x)^{1/2}$ are non-negligible. This also turns out to set the location of the dominant contribution to the Riemann-Hilbert problem for the orthogonal polynomials, and the leading behavior of the recurrence coefficients is given by $b_{n} \approx \beta_{n} / 2$.

\begin{lemma}
    Suppose $Q \in \mathrm{VSLF}(p,q)$ for some $p>0$, $q \in \mathbb{R}$. Then the following hold.
    \begin{enumerate}[label=(\roman*)]
        \item $x Q^{\prime}(x)$ is increasing for large enough $x$.
        \item For any $\epsilon>0$, there exists $x_{0}$ such that for $|x| \geq x_{0}$ we have
            \begin{equation}
                |x|^{p} \left(\log{|x|}\right)^{q-\epsilon/2} \leq |Q(x)| \leq |x|^{p} \left(\log{|x|}\right)^{q+\epsilon/2}.
                \label{eq:vslf_characterization_prop2}
            \end{equation}
        \item For any $\epsilon>0$, there exists $x_{0}$ such that for $|x| \geq x_{0}$ we have
            \begin{equation}
                |x|^{p-1} \left(\log{|x|}\right)^{q-\epsilon/2} \leq |Q^{\prime}(x)| \leq |x|^{p-1} \left(\log{|x|}\right)^{q+\epsilon/2}.
                \label{eq:vslf_characterization_prop3}
            \end{equation}
        \item For large enough $n$, the $n$th MRS number $\beta_{n}$ for $Q$ is uniquely defined, monotonically increasing, and furthermore
            \begin{equation}
                \lim_{n \to \infty} \dfrac{\beta_{n} Q^{\prime}(\beta_{n})}{n} = \dfrac{p}{\lambda_{p}},
                \label{eq:vslf_characterization_prop4}
            \end{equation}
            where $\lambda_{p}$ is defined by
            \begin{equation}
                \lambda_{p} = \dfrac{\Gamma\left(\frac{p+1}{2}\right)}{\Gamma\left(\frac{1}{2}\right)\Gamma\left(\frac{p}{2}\right)}.
                \label{eq:lambda_def}
            \end{equation}
        \item Uniformly for $s$ in any compact subinterval of $(0,\infty)$, we have
            \begin{equation}
                \lim_{n \to \infty} \dfrac{1}{n} \beta_{n}s Q^{\prime}(\beta_{n} s) = \dfrac{p}{\lambda_{p}} s^{p}.
                \label{eq:vslf_characterization_prop6}
            \end{equation}
        \item For any $\epsilon >0$, if $n$ is large enough then we have
            \begin{equation}
                \left(\dfrac{n}{\left(\log{n}\right)^{q + \epsilon/2}}\right)^{1/p} \leq \beta_{n} \leq \left(\dfrac{n}{\left(\log{n}\right)^{q - \epsilon/2}}\right)^{1/p}.
                \label{eq:vslf_characterization_prop5}
            \end{equation}
    \end{enumerate}
    \label{lem:vslf_characterization}
\end{lemma}
\begin{proof}
    \begin{enumerate}[label=(\roman*)]
        \item This is \cite[Lemma 3.1(i)]{lubinskyUniformMeanApproximation1988}, whose proof we repeat here. \cref{eq:vsf_prop3,eq:vsf_prop1} imply that for large enough $x$ we have
            \begin{align}
                \dfrac{\diff}{\diff x}\left[x Q^{\prime}(x)\right] &= Q^{\prime}(x)\left[1 + \dfrac{x Q^{\prime\prime}(x)}{Q^{\prime}(x)}\right],\\
                &\geq Q^{\prime}(x) (p/2) > 0.
            \end{align}
        \item 
            Fixing $\epsilon>0$, for large enough $x$ \cref{eq:vslf_limit} yields
            \begin{equation}
                \dfrac{p}{x} + \dfrac{q - \epsilon/2}{x \log{x}} \leq \dfrac{Q^{\prime}(x)}{Q(x)} \leq \dfrac{p}{x} + \dfrac{q + \epsilon/2}{x \log{x}}.
                \label{eq:relation_for_proof_prop2_T}
            \end{equation}
            Integrating w.r.t.~$x$ and exponentiating then gives the result (the integration constant can be absorbed by slightly increasing $\epsilon$).
        \item The condition in \cref{eq:vslf_limit} implies
            \begin{equation}
                \lim_{x \to \infty}\left[\dfrac{x Q^{\prime}(x)}{Q(x)}\right] = p.
            \end{equation}
            Fixing $0 < \epsilon^{\prime} < p$, for $x$ large enough this limit implies
            \begin{equation}
                (p - \epsilon^{\prime}/2) \dfrac{Q(x)}{x} \leq Q^{\prime}(x) \leq (p + \epsilon^{\prime}/2) \dfrac{Q(x)}{x}.
            \end{equation}
            Then we apply \cref{eq:vslf_characterization_prop2} to get
            \begin{equation}
                (p - \epsilon^{\prime}/2) x^{p-1} \left(\log{x}\right)^{q-\epsilon/2} \leq Q^{\prime}(x) \leq (p + \epsilon^{\prime}/2) x^{p-1} \left(\log{x}\right)^{q+\epsilon/2}.
            \end{equation}
            We can get rid of the constants $(p \pm \epsilon^{\prime}/2)$ by slightly increasing $\epsilon$, and so we obtain \cref{eq:vslf_characterization_prop3}.
        \item This is \cite[Lemma 3.2(i) and (ii)]{lubinskyUniformMeanApproximation1988}, the proofs of which directly carry over because $Q \in \mathrm{VSF}(p)$.
        \item This follows after using \cite[Lemma 3.1(ii)]{lubinskyUniformMeanApproximation1988}, which states that for $Q \in \mathrm{VSF}(p)$, uniformly for $s$ in any compact subinterval of $(0,\infty)$ we have
            \begin{equation}
                \lim_{n \to \infty} \dfrac{Q^{\prime}(\beta_{n}s)}{Q^{\prime}(\beta_{n})} = s^{p-1}.
            \end{equation}
            The statement then follows by combining this with \cref{eq:vslf_characterization_prop4}, since
            \begin{equation*}
                \lim_{n \to \infty} \dfrac{1}{n} \beta_{n} s Q^{\prime}(\beta_{n}s) = \lim_{n \to \infty} \dfrac{\beta_{n} Q^{\prime}(\beta_{n})}{n} \dfrac{Q^{\prime}(\beta_{n}s)}{Q^{\prime}(\beta_{n})} s = \dfrac{p}{\lambda_{p}} s^{p-1} s = \dfrac{p}{\lambda_{p}} s^{p}.
            \end{equation*}
        \item We will prove the upper bound on $\beta_{n}$, with the lower bound proceeding analogously. From \cref{eq:vslf_characterization_prop4} we have $\lim_{n\to\infty} \beta_{n} Q^{\prime}(\beta_{n})/n = p / \lambda_{p}$, and hence, given $0 < \epsilon^{\prime} < p / \lambda_{p}$, for sufficiently large $n$ we have
            \begin{equation*}
                \dfrac{\beta_{n} Q^{\prime}(\beta_{n})}{n} \leq \dfrac{p}{\lambda_{p}} + \dfrac{\epsilon^{\prime}}{2}.
            \end{equation*}
            Then lower bounding $Q^{\prime}(\beta_{n})$ using \cref{eq:vslf_characterization_prop3} gives
            \begin{equation*}
                \beta_{n} \leq \left(\dfrac{p}{\lambda_{p}} + \dfrac{\epsilon^{\prime}}{2}\right)^{1/p} \left(\dfrac{n}{\left(\log{\beta_{n}}\right)^{q - \epsilon/2}}\right)^{1/p}.
            \end{equation*}
            Since $Q \in \mathrm{VSF}(p)$, we can apply \cite[Lemma 3.3(iii)]{lubinskyProofFreudConjecture1988}, which gives the bound $\beta_{n} \geq n^{1/(p + \epsilon^{\prime\prime})}$ for $0 < \epsilon^{\prime\prime} < p$, and hence $\log{\beta_{n}} \geq (p+\epsilon^{\prime\prime})^{-1} \log{n}$. Using this to bound the $\log{\beta_{n}}$ factor on the RHS, and redefining $\epsilon$ to absorb the constant prefactor, we get the upper bound in \cref{eq:vslf_characterization_prop5}.
    \end{enumerate}
\end{proof}

Next we generalize some of these properties to the region of the complex plane where $Q$ is assumed to be complex analytic. 
First we generalize the upper bound of \cref{eq:vslf_characterization_prop3} to this complex region.
\begin{lemma}
    Assume $Q \in \mathrm{CVSLF}(p,q,\theta,\gamma)$. For any $\epsilon>0$, there exists $A$ such that for $z \in C_{\theta}$ with $|z| \geq A$ we have
            \begin{equation}
                 |Q^{\prime}(z)| \leq |z|^{p-1} \left(\log{|z|}\right)^{q+\epsilon/2}.
                \label{eq:complex_bound_T}
            \end{equation}
    \label{lem:Qprime_complex_bound}
\end{lemma}
\begin{proof}
One may proceed as in the proofs of properties (ii) and (iii) of \cref{lem:vslf_characterization} where the starting point \cref{eq:relation_for_proof_prop2_T} is replaced by the estimate
  \begin{equation}
                \left| \dfrac{Q^{\prime}(z)}{Q(z)} - \dfrac{p}{z} - \dfrac{q }{z \log{z}} \right| \leq   \dfrac{\epsilon/2}{\left| z \log{z} \right|} \,.
\end{equation}                
\end{proof}
Second we generalize \cite[Lemma 3.1(ii)]{lubinskyUniformMeanApproximation1988} to the complex region.
\begin{lemma}
    Assume $Q \in \mathrm{CVSLF}(p,q,\theta,\gamma)$. Then uniformly for $z$ in any compact subset of $C_{\theta}$ with $\Re{z} > 0$, we have
    \begin{equation}
        \lim_{\beta \to \infty} \dfrac{Q^{\prime}(\beta z)}{Q^{\prime}(\beta)} = z^{p-1},
        \label{eq:complex_limit}
    \end{equation}
    with the behavior for $\Re{z} < 0$ given by symmetry, since $Q^{\prime}(z)$ is assumed odd.
\end{lemma}
\begin{proof}
    Throughout we consider $z \in C_{\theta}$ with $\Re{z} > 0$. First note that $C_{\theta}$ is invariant under rescaling, so $z \in C_{\theta} \Rightarrow \beta z \in C_{\theta}$ for $\beta > 0$. Then we use the observation that
    \begin{equation}
        \log{Q^{\prime}(\beta z)} - \log{Q^{\prime}(\beta)} - (p-1)\log{z} = \int_{\beta}^{\beta z} \left[\dfrac{u Q^{\prime \prime}(u)}{Q^{\prime}(u)} - (p-1)\right] \dfrac{\diff u}{u},
    \end{equation}
    where the integral is along the straight line contour $u(t) = \beta[1 + t(z-1)]$, $t \in [0,1]$, connecting $\beta$ and $\beta z$. This contour remains within $C_{\theta}$ because the subset of $C_{\theta}$ with $\Re{z} > 0$ is convex. By the assumption \cref{eq:complex_assumption}, given $\epsilon > 0$, for sufficiently large $\beta$ we have
    \begin{equation}
        \left| \dfrac{u Q^{\prime \prime}(u)}{Q^{\prime}(u)} - (p-1) \right| \leq \epsilon
        \label{eq:inequality_in_proof_T}
    \end{equation}
    for all $u$ along the contour. Then changing variables from $u$ to $t$ and applying the bound gives
    \begin{equation}
        \left|\log{Q^{\prime}(\beta z)} - \log{Q^{\prime}(\beta)} - (p-1)\log{z} \right| \leq \epsilon \int_{0}^{1} \left| \dfrac{z-1}{1 + t(z-1)}\right| \diff t,
    \end{equation}
    where we note that the $\beta$ dependence has disappeared from the RHS. Since $\Re{z} > 0$, the integrand denominator is always nonzero, so the integral converges, and hence there exists finite $c_{z} > 0$, independent of $\beta$, such that
    \begin{equation}
        \left|\log{Q^{\prime}(\beta z)} - \log{Q^{\prime}(\beta)} - (p-1)\log{z} \right| \leq \epsilon \, c_{z}.
    \end{equation}
    Since $\epsilon$ can be taken arbitrarily small by increasing $\beta$, so that $\epsilon \, c_{z} \to 0$ uniformly for $z$ in any compact subset of $C_{\theta}$, we have
    \begin{equation}
        \lim_{\beta \to \infty} \left[ \log{Q^{\prime}(\beta z)} - \log{Q^{\prime}(\beta)} \right] = (p-1) \log{z}.
    \end{equation}
    Exponentiating both sides then gives \cref{eq:complex_limit}.
\end{proof}

Combining this lemma with \cref{lem:vslf_characterization}(iv), we get the complex analogue of \cref{lem:vslf_characterization}(v) (the proof is identical):
\begin{lemma}
    Uniformly for $z$ in any compact subset of $C_{\theta}$,
    \begin{equation}
        \lim_{n \to \infty} \dfrac{1}{n} \beta_{n} z Q^{\prime}(\beta_{n} z) = \dfrac{p}{\lambda_{p}} z^{p}.
    \end{equation}
    \label{lem:complex_VSF_to_freud}
\end{lemma}

Next we mention a useful corollary of the proof of \cref{eq:complex_limit}. The first part follows from the integral inequality
\begin{equation}
 \int_{\beta s}^{\beta} \dfrac{Q^{\prime \prime}(u)}{Q^{\prime}(u)} \diff u  \geq \int_{\beta s}^{\beta} \dfrac{p-1-\epsilon}{u} \diff u,
\end{equation}
for $0<s \leq 1$ with $\beta s$ large enough so that the estimate~\eqref{eq:inequality_in_proof_T} can be obtained from condition~\eqref{eq:vsf_prop3}. The second part follows after using \cref{eq:vslf_characterization_prop4}.
\begin{corollary}
Assume $\epsilon >0$ and $\beta > 0$. Then there exists $A>0$ such that 
\begin{equation}
Q^{\prime}(\beta s) \le Q^{\prime}(\beta) s^{p-1-\epsilon}
	\end{equation}
for all $A/\beta \le s \le 1$. Thus, for sufficiently large $n$ and all $A/\beta_{n} \le s \le 1$, we have
\begin{equation}
0 < V_{n}^{\prime}(s) = \dfrac{\beta_{n}}{n} Q^{\prime}(\beta_{n} s) \le \dfrac{2p}{\lambda_{p}}   s^{p-1-\epsilon} .
\end{equation}  
\label{cor:useful_estimate_T}
\end{corollary}

Finally, since it will be relevant later, we consider the Lagrange multiplier $l_{n} \in \mathbb{R}$, which appears in the Euler-Lagrange equation
\begin{equation}
    g_{n,+}(x) + g_{n,-}(x) - V_{n}(x) - l_{n} = 0, \quad \text{for } x \in [-1,1],
\end{equation}
for the Coulomb gas energy minimization problem discussed in \cref{sec:coulomb}, which will be relevant in \cref{sec:euler_lagrange}. Here $V_{n}(x) = Q(\beta_{n}x)/n$ is the rescaled potential, and $g_{n}$ is the logarithmic transform of the equilibrium measure, which we will define later in \cref{eq:log_transform_def}. The Lagrange multiplier $l_{n}$ can be determined using the expression
\begin{equation}
    l_{n} = -2 \log{2} - \dfrac{2}{\pi} \int_{0}^{1} \dfrac{V_{n}(s)}{\sqrt{1-s^{2}}} \diff s.
    \label{eq:lagrange_at_zero}
\end{equation}
(This follows from \cite[Theorem IV.3.1]{saffLogarithmicPotentialsExternal1997}, where in their language we have $l_{n} = -2 C_{\frac{1}{2}V_{n}}$.) For our class of weights, the limiting value is the same as that for the simple Freud weight $\kappa_{p} |x|^{p}$, where $\kappa_{p} = 1/\lambda_{p}$ with $\lambda_{p}$ defined in \cref{eq:lambda_def}.
\begin{lemma}
    For $Q \in \mathrm{VSLF}(p,q)$ with $p>0$, we have
    \begin{equation}
        \lim_{n \to \infty} l_{n} = -2 \log{2} - \dfrac{2}{p}.
    \end{equation}
    \label{lem:lagrange_multiplier}
\end{lemma}
Thus the Lagrange multiplier should be $\mathcal{O}(1)$ and negative for sufficiently large $n$.
\begin{proof}
    By \cite[Lemma 3.2(iii)]{lubinskyUniformMeanApproximation1988} (our weights satisfy all the relevant conditions), we have $\lim_{n\to\infty} V_{n}(s) = \kappa_{p} |s|^{p}$ uniformly for $s$ in compact subsets of $\mathbb{R}$. This implies pointwise convergence of the integrand for $s \in (0,1)$. In order to apply the dominated convergence theorem we appeal to \cref{cor:useful_estimate_T}, choosing $\epsilon$ so that $p-1-\epsilon>-1$. However, this upper bound is only available for $s \in [A/\beta_{n},1]$. Observe that the remaining part of the integral can be bounded by
    \begin{equation}
        \int_{0}^{A/\beta_{n}}\dfrac{|V_{n}(s)|}{\sqrt{1 - s^{2}}} \diff s \leq \dfrac{1}{\sqrt{1 - (A/\beta_{n})^{2}}} \dfrac{1}{n \beta_{n}} \int_{0}^{A} |Q(u)| \diff u \leq \mathcal{O}\left(\dfrac{1}{n \beta_{n}}\right),
    \end{equation}
    which goes to zero as $n\to\infty$. Then using the dominated convergence theorem gives
    \begin{align}
        \lim_{n\to\infty} l_{n} &= -2\log{2} - \dfrac{2}{\pi} \int_{0}^{1}\dfrac{\kappa_{p} s^{p}}{\sqrt{1-s^{2}}}\diff s,\\
        &= -2\log{2} - \dfrac{2}{p},
    \end{align}
    where the final step comes from explicitly evaluating the integral and simplifying using \cref{eq:lambda_def} with $\kappa_{p} = 1/\lambda_{p}$.
\end{proof}

\clearpage
\section{Riemann-Hilbert steepest descent analysis}
\label{sec:rhp}

To obtain $n\to\infty$ asymptotics of the orthogonal polynomials with respect to $w(x)\equiv \Phi(x)/2\pi$, we make use of a `steepest descent'-inspired analysis of a Riemann-Hilbert problem (RHP) associated with $w$~\cite{deiftSteepestDescentMethod1993,deiftStrongAsymptoticsOrthogonal1999,deiftRiemannHilbertApproach2001,deiftUniformAsymptoticsPolynomials1999,deiftOrthogonalPolynomialsRandom2000,kuijlaarsRiemannHilbertAnalysisOrthogonal2003,kuijlaarsUniversalityEigenvalueCorrelations2003}. This fundamental RHP looks for a $2 \times 2$ matrix-valued function $Y$ which is analytic in $\mathbb{C} \setminus \mathbb{R}$, with a specified jump across $\mathbb{R}$ that is related to $w(x)$. This RHP is formulated in such a way to have a unique solution which encodes much valuable information about the orthogonal polynomials with respect to $w$.

\subsection{Fundamental Riemann-Hilbert problem for $Y$}
\label{sec:fundamental_rhp}
Let us decompose $w(x) \equiv |x|^{\rho} e^{-Q(x)}$, and note that we must have $\rho > -1$ in order for $w$ to be integrable. Throughout we will assume $w(-x)=w(x)$ is even, and to account for the power-law at $x=0$, we will need to add an extra condition to ensure uniqueness of the solution to the RHP. Then we look for a $2 \times 2$ matrix-valued function $Y(z)$ satisfying the following conditions.
\begin{enumerate}[label=(\alph*)]
    \item[$(Y_{a})$] $Y(z)$ is analytic in $\mathbb{C} \setminus \mathbb{R}$.
    \item[$(Y_{b})$] $Y(z)$ takes continuous boundary values $Y_{\pm} (x) \coloneqq \lim_{y \to 0^{\pm}} Y(x + i y)$ such that
        \begin{equation}
            Y_{+}(x) = Y_{-}(x) \begin{pmatrix}
                1 & w(x) \\
                0 & 1
            \end{pmatrix}, \mathclap{\hspace{9em} \text{for } x \in \mathbb{R} \setminus \{0\}.}
        \end{equation}
    \item[$(Y_{c})$] $Y(z)$ has the following asymptotic behavior at infinity:
        \begin{equation}
            Y(z) = \left[ \mathds{1} + \mathcal{O}(1/z) \right] \begin{pmatrix}
                z^{n} & 0 \\
                0 & z^{-n}
            \end{pmatrix}, \mathclap{\hspace{14em} \text{as } |z| \to \infty \text{ for } z \in \mathbb{C} \setminus \mathbb{R}.}
            \label{eq:RH_infty_condition}
        \end{equation}
        
    \item[$(Y_{d})$] $Y(z)$ has the following behavior near $z = 0$:
        \begin{equation}
            Y(z) = \begin{dcases}
                \mathcal{O} \begin{pmatrix}
                    1 & |z|^{\rho} \\
                    1 & |z|^{\rho}
                \end{pmatrix},
                 & \rho < 0;\\
                \mathcal{O} \begin{pmatrix}
                    1 & 1\\
                    1 & 1
                \end{pmatrix},
                 & \rho \geq 0,\\
            \end{dcases}
            \label{eq:RH_zero_condition}
        \end{equation}
        as $z \to 0$ for $z \in \mathbb{C} \setminus \mathbb{R}$, where the $\mathcal{O}$ notation is taken elementwise. (For noneven weights with $\rho=0$, the second column of $Y(z)$ would diverge as $\mathcal{O}(\log|z|)$ as $z\to 0$, but this logarithmic divergence is exactly canceled for an even weight.)
\end{enumerate}
Then by the result of Fokas, Its and Kitaev~\cite{fokasIsomonodromyApproachMatrix1992} (see also \cite{deiftOrthogonalPolynomialsRandom2000,kuijlaarsRiemannHilbertApproach2004,vanlessenStrongAsymptoticsRecurrence2003,kuijlaarsUniversalityEigenvalueCorrelations2003}), we have
\begin{theorem}
    The unique solution to the above RHP is given by
    \begin{equation}
        \boxed{Y(z) = \begin{pmatrix}
            P_{n}(z) & C[P_{n} w](z) \\
            c_{n} P_{n-1}(z) & c_{n} C[P_{n-1}w](z)
        \end{pmatrix},}
        \label{eq:fki_sol}
    \end{equation}
    where the $P_{n}(x)$ are the \textit{monic} orthogonal polynomials associated with $w$. Here $c_{n} = -2 \pi i y_{n-1}^{2}$, with $y_{n-1}>0$ the coefficient of the leading term in the corresponding \textit{normalized} orthogonal polynomial $p_{n-1}(x) = y_{n-1} x^{n-1} + \cdots$, and
    \begin{equation*}
        C[f](z) \coloneqq \dfrac{1}{2\pi i} \int_{-\infty}^{\infty} \dfrac{f(s)}{s - z} \diff s, \quad z \in \mathbb{C} \setminus \mathbb{R}
    \end{equation*}
    is the Cauchy-Stieltjes transform of $f \in L^{2}(\mathbb{R})$. Furthermore, there exists $Y_{1} \in \mathbb{C}^{2\times 2}$ such that
    \begin{equation*}
        Y(z) \begin{pmatrix}
            z^{-n} & 0 \\
            0 & z^{n}
        \end{pmatrix} = \mathds{1} + \dfrac{Y_{1}}{z} + \mathcal{O}\left(\dfrac{1}{|z|^{2}}\right) \quad \text{as } |z| \to \infty.
    \end{equation*}
    The recurrence coefficient $b_{n}$ and the leading coefficient $y_{n}$ are then given by
    \begin{empheq}[box=\widefbox]{align}
        b_{n} &= \sqrt{(Y_{1})_{12} (Y_{1})_{21}},\label{eq:Y1_to_bn}\\
        y_{n} &= \dfrac{1}{\sqrt{-2\pi i (Y_{1})_{12}}}. \label{eq:Y1_to_yn}
    \end{empheq}
\end{theorem}
Note that many authors in the orthogonal polynomials literature label what we call $b_{n}$ as $b_{n-1}$, but we will stick with this convention to be consistent with the physics literature, which typically starts at $b_{1}$ rather than $b_{0}$.

For a proof that \cref{eq:fki_sol} solves this Riemann-Hilbert problem, see Sec 3.2 of Ref.~\cite{deiftOrthogonalPolynomialsRandom2000}; the crucial point is that the Cauchy-Stieltjes transform satisfies the operator identity $C_{+} - C_{-} = \mathrm{Id}$, where $C_{\pm}[f](x) = \lim_{y \to 0^{+}} C[f](x \pm i y)$ denote its limits on the $\pm$ sides of the contour (c.f.~the Sokhotski–Plemelj theorem). For a proof that this is the unique solution for $Y(z)$, see Lemma 2.3 of Ref.~\cite{kuijlaarsRiemannHilbertApproach2004}.

\subsection{Proof overview}
The initial Riemann-Hilbert problem (RHP) is too complicated to solve immediately, so we perform a sequence of transformations which gradually simplify the problem until it is solvable by standard techniques. We can then reverse the transformations to obtain the desired asymptotics of $Y(z)$. We denote the sequence of transformations by
\begin{equation*}
    Y \mapsto U \mapsto T \mapsto S \mapsto R
\end{equation*}
\begin{itemize}
    \item $Y \mapsto U$ is a simple rescaling by the Mhaskar-Rakhmanov-Saff number $\beta_{n}$, so that the dominant contributions will come from near the interval $[-1,1]$ in the rescaled coordinates. This finite interval is the analogue of a point of stationary phase when performing a saddle point approximation of a contour integral. 
    \item $U \mapsto T$ involves the function $g_{n}(z) = \int_{-1}^{1} \log(z-t) \psi_{n}(t) \diff t$, the logarithmic transform of the equilibrium measure. This step normalizes the RHP at infinity, since $\exp[n g_{n}(z)] \approx z^{n}$ as $z \to \infty$. Furthermore, as $n \to \infty$, $\psi_{n}(t)$ approximates the density of zeros $\{x_{j,n}\}$ of the orthogonal polynomials with respect to $\exp[-Q(x)]$ when appropriately rescaled from $[-\beta_{n}, \beta_{n}]$ to $[-1,1]$~\cite{saffLogarithmicPotentialsExternal1997,deiftStrongAsymptoticsOrthogonal1999}. Morally we have \begin{equation*}\exp[n g_{n}(z)] \approx \exp\left[n \int_{-1}^{1} \log(z-t) \left( \frac{1}{n} \sum_{j=1}^{n} \delta(t - x_{j,n}/\beta_{n})\right) \diff t\right] = \prod_{j=1}^{n} (z - x_{j,n}/\beta_{n}) = \beta_{n}^{-n} P_{n}(\beta_{n} z),\end{equation*}where $P_{n}$ are the monic orthogonal polynomials. For technical convenience we will always define the equilibrium measure with respect to $Q$; this means that for $\rho \neq 0$ the density of zeros of the orthogonal polynomials with respect to the full weight $|x|^{\rho} \exp[-Q(x)]$ is slightly different to that with respect to $\exp[-Q(x)]$ alone, with an enhancement or suppression near $x=0$ depending on whether $\rho$ is negative or positive. In the end this still provides a good enough approximation away from $x=0$ to be useful, and we will handle the effects of the power-law more explicitly using a Szeg\H{o} function (see \cref{sec:outside_parametrix}).
    \item $T \mapsto S$ involves a factorization of the jump matrix, and a subsequent deformation of the jump contours into the complex plane. Under this deformation, oscillatory terms transform into terms which decay with $n$ and are usually subleading. In most cases these terms decay exponentially with $n$, but in the marginal case where $w(\omega)$ decays only exponentially in $\omega$, these terms decay more slowly with $n$.
    \item $S \mapsto R$ is the most technically involved step. We explicitly construct a parametrix $S_{\mathrm{par}}$ which approximately solves the RHP for $S$. This construction is most delicate near the endpoints of the jump contour, and at the location of the power-law; here we construct approximate local solutions out of appropriate special functions: Airy functions near the endpoints $\pm 1$, and Bessel functions near the origin. We then set $R = S S_{\mathrm{par}}^{-1}$, so that $R$ has a jump matrix which is uniformly close enough to the identity that the RHP for $R$ can be solved using standard techniques. One can think of $R$ as a `residual', and $R(z)\approx\mathds{1}$ to leading order in $n$.
\end{itemize}

\subsection{Riemann-Hilbert problem for $U$}
In terms of the Mhaskar-Rakhmanov-Saff number $\beta_{n}$ defined in \cref{eq:MRS_def_sup}, we define
\begin{equation}
    U(z) = \beta_{n}^{-(n + \rho/2) \sigma_{3}} Y(\beta_{n} z) \beta_{n}^{(\rho/2) \sigma_{3}}, \quad \text{for } z \in \mathbb{C} \setminus \mathbb{R},
    \label{eq:U_def}
\end{equation}
where $\sigma_{3} = \begin{pmatrix} 1 & 0 \\ 0 & -1 \end{pmatrix}$ is the third Pauli matrix. Given that $Y$ is the unique solution to the RHP stated above, one can verify that $U$ is the unique solution to an equivalent RHP, stated as follows.
\begin{itemize}
    \item[$(U_{a})$] $U(z)$ is analytic in $\mathbb{C} \setminus \mathbb{R}$;
    \item[$(U_{b})$] $U(z)$ takes continuous boundary values $U_{\pm} (x) \coloneqq \lim_{y \to 0^{\pm}} U(x + i y)$ such that
        \begin{equation}
            U_{+}(x) = U_{-}(x) \begin{pmatrix}
                1 & |x|^{\rho} e^{-n V_{n}(x)} \\
                0 & 1
            \end{pmatrix}
        \end{equation}
        for $ x \in \mathbb{R} \setminus \{0\}$, where $V_{n}(x) \coloneqq Q(\beta_{n} x) / n$ is the rescaled potential.
    \item[$(U_{c})$] $U$ satisfies the same behaviour as $z \to \infty$ as $Y$ does, given by \cref{eq:RH_infty_condition}.
    \item[$(U_{d})$] $U$ satisfies the same behaviour near $z = 0$ as $Y$ does, given by \cref{eq:RH_zero_condition}.
\end{itemize}

\subsection{Equilibrium measures for complex log-Freud weights}
\label{sec:equilibrium_measure}
In this section we wish to construct the equilibrium measure for the rescaled potential $V_{n}(x) \coloneqq Q(\beta_{n}x) / n$. To that end, for $z$ in the domain of analyticity of $V_{n}(z)$, define the function
\begin{equation}
    h_{n}(z) \coloneqq \dfrac{1}{2\pi i} \oint_{\Gamma_{z}} \dfrac{V_{n}^{\prime}(s)}{s-z} \dfrac{\diff s}{r(s)},
    \label{eq:hn_def}
\end{equation}
where $r(s) = (s+1)^{1/2} (s-1)^{1/2}$, taking principal branches of powers such that $r(s) \sim +s$ as $s \to \infty$, and $\Gamma_{z}$ is any simple closed anticlockwise oriented contour within the domain of analyticity of $V_{n}$ and with $[-1,1] \cup \{z\}$ in its interior. Now note that
\begin{equation}
    \oint_{\Gamma_{z}} \dfrac{1}{r(s)}\dfrac{\diff s}{s-z} = 0,
\end{equation}
because the integrand is analytic outside $\Gamma_{z}$, so we can deform the contour to infinity. Then, in order to deal with any potential singularities, we are free to rewrite
\begin{equation}
    h_{n}(z) = \dfrac{1}{2\pi i} \oint_{\Gamma_{z}} \dfrac{V_{n}^{\prime}(s) - V_{n}^{\prime}(z)}{s-z} \dfrac{\diff s}{r(s)}
    \label{eq:hn_representation_useful_T}
\end{equation}
This is helpful because the integrand now has zero residue at the pole $s=z$. In order to obtain an expression for $h_{n}(x)$ on the real line, we pull the contour through $z$ (at zero cost because the residue is zero), and then deform it around $[-1,1]$, obtaining
\begin{equation}
    h_{n}(x) = \dfrac{1}{\pi} \int_{-1}^{1} \dfrac{V_{n}^{\prime}(s) - V_{n}^{\prime}(x)}{s-x} \dfrac{\diff s}{\sqrt{1-s^{2}}}, \quad x \in \mathbb{R},
\end{equation}
where we used $r_{+}(s) = i \sqrt{1-s^{2}}$. Since $V_{n}$ is even, we can also write this as
\begin{equation}
    h_{n}(x) = \dfrac{2}{\pi} \int_{0}^{1} \dfrac{x V_{n}^{\prime}(x) - s V_{n}^{\prime}(s)}{x^{2} - s^{2}} \dfrac{\diff s}{\sqrt{1-s^{2}}}.
    \label{eq:hn_even_integral}
\end{equation}
\subsubsection{Definition of the equilibrium measure}
Having defined $h_{n}(z)$ in \cref{eq:hn_def} (reducing to \cref{eq:hn_even_integral} on the real line), we now define the `candidate' equilibrium measure $\psi_{n}(x)$ for $x \in [-1,1]$ by
\begin{equation}
    \psi_{n}(x) \coloneqq \dfrac{1}{2\pi} \sqrt{1-x^{2}} h_{n}(x), \qquad x \in [-1,1].
    \label{eq:psi_def}
\end{equation}
We will show that this candidate function $\psi_{n}(x)$ is indeed the true equilibrium measure for $V_{n}(x)$ by verifying that the Euler-Lagrange variational conditions for the energy minimization problem are satisfied (for large enough $n$). We will furthermore show that $\psi_{n}(x)$ is `regular' in the sense that it is positive in $(-1,1)$ and vanishes like a square root as $|x| \to 1$. This means we will get Airy universality near the endpoints $x=\pm 1$.

\subsubsection{Support of the equilibrium measure}
\begin{restatable}{lemma}{firstposeqmeasure}
    Suppose $Q \in \mathrm{VSLF}(p,q)$ with $p\geq 1$, and $q$ arbitrary for $p > 1$ but constrained to $q>-1$ for $p=1$. 
    Then there exists $n_{0} \in \mathbb{N}$ such that, for every $M>0$, there exists a constant $C>0$ such that for all $n \geq n_{0}$ we have $h_{n}(x) > C$ for all $|x|<M$.
    \label{lem:hn_positive_x_g1}
\end{restatable}
We defer the proof of this statement to \cref{sec:poseqmeasureproof}. The fact that for sufficiently large $n$ we have $h_{n}(x) > 0$ for $|x| \leq 1$ implies via \cref{eq:psi_def} that the equilibrium measure $\psi_{n}(x) > 0$ for $|x| < 1$ and vanishes like a square root as $|x| \to 1$, as claimed.

\subsubsection{Check of the Euler-Lagrange variational conditions}
\label{sec:euler_lagrange}
\begin{lemma}
    Define $g_{n}(z) \coloneqq \int_{-1}^{1} \log(z-t) \psi_{n}(t) \diff t$ for $z \in \mathbb{C} \setminus [-1,1]$. Then for large enough $n$ the following Euler-Lagrange variational conditions are satisfied (c.f.~\cite{deiftStrongAsymptoticsOrthogonal1999}).
    \begin{enumerate}[label=(\roman*)]
        \item $g_{n,+}(x) + g_{n,-}(x) - V_{n}(x) - l_{n} = 0$, for $|x| \leq 1$.
        \item $g_{n,+}(x) + g_{n,-}(x) - V_{n}(x) - l_{n} < 0$, for $|x| > 1$.
    \end{enumerate}
    The strict inequality in $(ii)$ means that $\psi_{n}$ is `regular' (for large enough $n$).
    \label{lem:variational_cond}
\end{lemma}
\begin{proof}
    Part (i) follows from the discussion in \cite[Theorem IV.3.1]{saffLogarithmicPotentialsExternal1997}, and does not require a large $n$ limit. Part (ii) follows by combining \cref{lem:hn_positive_x_g1} with the relation 
    \begin{equation}
        g_{n,+}(x) + g_{n,-}(x) - V_{n}(x) - l_{n} = - \int_{1}^{x} h_{n}(t) \sqrt{t^{2}-1} \diff t, \qquad \text{for } |x| > 1,
    \end{equation}
    which we will now prove. We essentially rehash the relevant part of the proof of \cite[Lemma 3.2]{deiftUniformAsymptoticsPolynomials1999}. We start from the Hilbert transform
    \begin{equation}
        H \psi_{n}(t) \coloneqq \dfrac{1}{\pi} \dashint_{-1}^{1} \dfrac{\psi_{n}(s)}{t-s} \diff s,
    \end{equation}
    where $\dashint$ denotes a principal value integral, and integrate $t$ from 1 to $x$ to give
    \begin{equation}
        \int_{1}^{x} H \psi_{n}(t) \diff t = \dfrac{1}{\pi} \int_{-1}^{1} \psi_{n}(s) \left(\log{|x-s|} - \log{|1-s|}\right) \diff s.
    \end{equation}
    (This integration is justified for sufficiently smooth $\psi_{n}$; see proof of \cite[Theorem IV.3.1]{saffLogarithmicPotentialsExternal1997}.) Since $g_{n,+}(x) + g_{n,-}(x) = 2 \int_{-1}^{1} \log{|x-s|} \psi_{n}(s) \diff s$ for $x \in \mathbb{R}$, this means
    \begin{equation}
        g_{n,+}(x) + g_{n,-}(x) - \left[g_{n,+}(1) + g_{n,-}(1)\right] = 2 \pi \int_{1}^{x} H \psi_{n}(t) \diff t.
    \end{equation}
    The variational condition gives $g_{n,+}(1) + g_{n,-}(1) = V_{n}(1) + l_{n}$, and so
    \begin{equation}
        g_{n,+}(x) + g_{n,-}(x) - V_{n}(x) - l_{n} = 2\pi \int_{1}^{x} \left(H \psi_{n}(t) - \dfrac{V_{n}^{\prime}(t)}{2\pi}\right) \diff t.
    \end{equation}
    It remains to show that
    \begin{equation}
        H \psi_{n}(t) - \dfrac{V_{n}^{\prime}(t)}{2\pi} = -\dfrac{1}{2\pi} h_{n}(t) \sqrt{t^{2}-1}, \quad \text{for } t > 1.
        \label{eq:hilbert_identity}
    \end{equation}
    This will follow the proof of \cite[Theorem 3.1]{deiftUniformAsymptoticsPolynomials1999}. Define the function
    \begin{equation}
        F_{n}(z) \coloneqq \dfrac{1}{\pi i} \int_{-1}^{1} \dfrac{\psi_{n}(s)}{s-z} \diff s, \quad z \in \mathbb{C} \setminus [-1,1],
    \end{equation}
    which on $[-1,1]$ has the jump
    \begin{equation}
        F_{n,\pm}(x) = \pm \psi_{n}(x) + i H\psi_{n}(x), \quad x \in [-1,1].
    \end{equation}
    However, differentiating the variational condition on $[-1,1]$ gives
    \begin{equation}
        H \psi_{n}(x) = \dfrac{V_{n}^{\prime}(x)}{2\pi}, \quad x \in [-1,1].
    \end{equation}
    It follows that $F_{n}$ satisfies the following scalar Riemann-Hilbert problem,
    \begin{align}
        F_{n,+}(x) + F_{n,-}(x) &= \dfrac{i V_{n}^{\prime}(x)}{\pi}, \quad x \in [-1,1],\\
        F_{n,+}(x) - F_{n,-}(x) &= 0, \quad x \in \mathbb{R} \setminus [-1,1],\\
        F_{n}(z) &= \dfrac{-1}{i \pi z} + \mathcal{O}(z^{-2}), \quad z \to \infty,
    \end{align}
    which has the standard solution
    \begin{equation}
        F_{n}(z) = r(z) \dfrac{1}{\pi i} \int_{-1}^{1} \dfrac{i V_{n}^{\prime}(s) / (2\pi)}{r_{+}(s)} \dfrac{\diff s}{s-z}, \quad z \in \mathbb{C} \setminus \mathbb{R},
    \end{equation}
    where $r(z) = (z-1)^{1/2} (z+1)^{1/2}$. For nonreal $z$ in the domain of analyticity of $V_{n}$, we can rewrite
    \begin{equation}
        F_{n}(z) = \dfrac{i V_{n}^{\prime}(z)}{2\pi} + \dfrac{1}{2\pi i} h_{n}(z) r(z),
    \end{equation}
    where $h_{n}(z)$ is given by the contour integral
    \begin{equation}
        h_{n}(z) = \dfrac{1}{2\pi i} \oint_{\Gamma_{z}} \dfrac{V_{n}^{\prime}(s)}{r(s)} \dfrac{\diff s}{s-z},
    \end{equation}
    which for $z \to x \in [-1,1]$ gives the familiar expression for $h_{n}(x)$ from the definition of the equilibrium measure. For $t > 1$ we have
    \begin{equation}
        F_{n}(t) = \dfrac{i V_{n}^{\prime}(t)}{2\pi} + \dfrac{1}{2\pi i} h_{n}(t) \sqrt{t^{2}-1},
    \end{equation}
    and the expression \cref{eq:hilbert_identity} comes from noticing that $H \psi_{n}(t) = -i F_{n}(t)$ for $t \in \mathbb{R} \setminus [-1,1]$.
\end{proof}

\subsection{Logarithmic transform of the equilibrium measure, and the Riemann-Hilbert problem for $T$}
Given the rescaled equilibrium measure $\psi_{n}(x)$ for $V_{n}(x) = Q(\beta_{n} x)/n$ (see \cref{sec:equilibrium_measure}), we define its logarithmic transform $g_{n}(z)$ by
\begin{equation}
    g_{n}(z) \coloneqq \int_{-1}^{1} \log(z - t) \psi_{n}(t) \diff t, \quad \text{for } z \in \mathbb{C} \setminus (-\infty,1],
    \label{eq:log_transform_def}
\end{equation}
where we take the principal branch of the logarithm. For $z$ approaching the real axis, we have the limiting values $g_{n,\pm}(x) \coloneqq \lim_{y \to 0^{+}} g_{n}(x \pm i y)$ given by
\begin{equation}
    g_{n,\pm}(x) = \int_{-1}^{1} \log{|x-t|} \psi_{n}(t) \diff t \pm \begin{dcases}
        i \pi, & x \leq 1, \\
        i \pi \int_{x}^{1} \psi_{n}(t) \diff t, & -1<x<1,\\
        0, & x \geq 1.
    \end{dcases}
    \label{eq:log_potential_limits}
\end{equation}
The following proposition is immediate from \cref{eq:log_transform_def,eq:log_potential_limits}.
\begin{proposition} For all $n \in \mathbb{N}$, the following holds. 
    \begin{enumerate}[label=(\alph*)]
        \item $g_{n}$ is analytic and $g_{n} |_{\mathbb{C}_{\pm}}$ have continuous extensions to $\mathbb{C}_{\pm}$.
        \item The map $z \mapsto e^{n g_{n}(z)}$ possesses an analytic continuation to $\mathbb{C} \setminus [-1,1]$, and
            \begin{equation}
                e^{n g_{n}(z)} z^{-n} = 1 + \mathcal{O}\left(\dfrac{1}{|z|}\right), \quad \text{as } z \to \infty.
                \label{eq:log_potential_inf}
            \end{equation}
    \end{enumerate}
\end{proposition}

In terms of the logarithmic transform $g_{n}(z)$, we define the matrix-valued function $T$ by
\begin{equation}
    T(z) = e^{-\frac{n l_{n}}{2} \sigma_{3}} U(z) e^{\frac{n l_{n}}{2} \sigma_{3}} e^{-n g_{n}(z) \sigma_{3}}.
    \label{eq:T_def}
\end{equation}
Given that $U(z)$ satisfies the Riemann-Hilbert problem $(U_{a})$-$(U_{d})$, one can combine \cref{eq:log_potential_limits,eq:log_potential_inf} to verify that $T(z)$ satisfies the following equivalent Riemann-Hilbert problem.
\begin{itemize}
    \item[$(T_{a})$] $T(z)$ is analytic in $\mathbb{C} \setminus \mathbb{R}$.
    \item[$(T_{b})$] $T$ satisfies the following jump relations on $\mathbb{R}$:
        \begin{align}
            T_{+}(x) &= T_{-}(x) \begin{pmatrix} e^{-n(g_{n,+}(x) - g_{n,-}(x))} & |x|^{\rho} \\ 0 & e^{n(g_{n,+}(x) - g_{n,-}(x))} \end{pmatrix}, \quad \text{for } x \in [-1,1] \setminus \{0\},\\
            T_{+}(x) &= T_{-}(x) \begin{pmatrix} 1 & |x|^{\rho} e^{n(g_{n,+}(x) + g_{n,-}(x) - V_{n}(x) - l_{n})} \\ 0 & 1 \end{pmatrix}, \quad \text{for } |x| > 1.
        \end{align}
    \item[$(T_{c})$] $T(z) = \mathds{1} + \mathcal{O}(1/z)$ as $z \to \infty$.
    \item[$(T_{d})$] $T(z)$ has the same behavior as $Y(z)$ as $z \to 0$, given by \cref{eq:RH_zero_condition}.
\end{itemize}
Notice that, by $(T_{c})$, the Riemann-Hilbert problem is now normalized at infinity. Furthermore, by \cref{lem:hn_positive_x_g1,lem:variational_cond}, the jump matrix is exponentially close to the identity for $|x| > 1$. For $|x|<1$, we see from \cref{eq:log_potential_limits} that the diagonal elements of the jump matrix are complex phases which rapidly oscillate for large $n$; these can be viewed as the analogue of the `fast phase' arising in the analysis of linear differential equations in the WKB limit~\cite{deiftStrongAsymptoticsOrthogonal1999}.

\subsection{Analytic continuation of the equilibrium measure and its logarithmic transform}
In terms of $h_{n}(z)$, which is defined by \cref{eq:hn_def} within the domain of analyticity of $V_{n}(z)$, we can define an analytic continuation $\hat{\psi}_{n}(z)$ of the equilibrium measure by
\begin{equation}
    \hat{\psi}_{n}(z) = \dfrac{1}{2\pi i} r(z) h_{n}(z),
    \label{eq:psi_hat_def}
\end{equation}
with $r(z) = (z+1)^{1/2}(z-1)^{1/2}$. Since $h_{n}(z)$ is analytic, while $r(z)$ flips sign across $[-1,1]$, we have
\begin{align}
    \hat{\psi}_{n,+}(x) = - \hat{\psi}_{n,-}(x) = \psi_{n}(x), \quad \text{for } x \in (-1,1),\label{eq:psi_cont_jump}\\
    \hat{\psi}_{n,+}(x) = \hat{\psi}_{n,-}(x), \quad \text{for } x \in \mathbb{R} \setminus (-1,1),\label{eq:psi_cont_smooth}
\end{align}
so $\hat{\psi}_{n}(z)$ is an analytic continuation of $\psi_{n}(x)$ to $\mathbb{C}_{+}$, and similarly $-\hat{\psi}_{n}(z)$ provides an analytic continuation to $\mathbb{C}_{-}$. Finally, having defined these continuations, we define
\begin{equation}
    \phi_{n}(z) \coloneqq -\pi i \int_{1}^{z} \hat{\psi}_{n}(y) \diff y, \quad \text{for } z \in \mathbb{C} \setminus \mathbb{R},
    \label{eq:phi_n_def}
\end{equation}
where the path of integration does not cross the real axis. By using the analyticity of $\hat{\psi}_{n}(z)$ and the fact that $\int_{0}^{1} \psi_{n}(x) \diff x = 1/2$ by symmetry, one can show that $\phi_{n}(z)$ can equivalently be written as
\begin{equation}
    \phi_{n}(z) = \pm \dfrac{i \pi}{2} - \pi i \int_{0}^{z} \hat{\psi}_{n}(y) \diff y, \quad \text{for } z \in \mathbb{C}_{\pm},
    \label{eq:phi_n_origin}
\end{equation}
where again the path of integration does not cross the real axis. Note that \cref{eq:psi_cont_smooth} implies that $\phi_{n}(z)$ is analytic across $\mathbb{R} \setminus (-1,1)$. Also, \cref{eq:phi_n_origin} shows that $2 \phi_{n,+}(0) = 2\pi i \int_{0}^{1} \psi_{n}(t) \diff t = \pi i$. This is a special property of \textit{even} weight functions, but should be true generically given that symmetry. This factor (exponentiated) is what ends up giving rise to the $(-1)^{n}$ staggering factor in the recurrence coefficients, so here we see that this is a direct consequence of the $x \to -x$ symmetry, which itself is a consequence of unitarity.

For the logarithmic transform, note that by \cref{eq:log_potential_limits} we have
\begin{equation}
    g_{n,+}(x) - g_{n,-}(x) = \begin{cases}
        2\pi i, & x \leq -1,\\
        2\pi i \int_{x}^{1} \psi_{n}(t) \diff t, & |x| < 1,\\
        0, & x \geq 1.
    \end{cases}
    \label{eq:g_diff}
\end{equation}
From \cref{eq:psi_cont_jump,eq:phi_n_def} we deduce that the functions $2\phi_{n}(z)$ and $-2\phi_{n}(z)$ provide analytic continuations of $(g_{n,+} - g_{n,-})(z)$ into the upper and lower half planes respectively.

\subsection{Riemann-Hilbert problem for $S$: contour deformation and `opening the lens'}
Next we notice that the jump matrix for $T(x)$ on $x \in [-1,1]$ can be factorized as
\begin{align}
    \begin{pmatrix} e^{-n(g_{n,+}(x) - g_{n,-}(x))} & |x|^{\rho} \\ 0 & e^{n(g_{n,+}(x) - g_{n,-}(x))} \end{pmatrix} &= \\
    \begin{pmatrix} 1 & 0 \\ |x|^{-\rho} e^{n(g_{n,+}(x) - g_{n,-}(x))} & 1 \end{pmatrix} &\begin{pmatrix} 0 & |x|^{\rho} \\ -|x|^{-\rho} & 0 \end{pmatrix} \begin{pmatrix} 1 & 0 \\ |x|^{-\rho} e^{-n(g_{n,+}(x) - g_{n,-}(x))} & 1 \end{pmatrix}. \nonumber
\end{align}
Remembering that $\pm 2\phi_{n}(z)$ provides an analytic continuation of $g_{n,+} - g_{n,-}$ to $\mathbb{C}_{\pm}$, and defining the analytic continuation of $|x|^{\rho}$ to be
\begin{equation}
    \omega(z) = \begin{cases}
        (-z)^{\rho}, \quad \text{if } \Re{z} < 0,\\
        z^{\rho}, \quad \text{if } \Re{z} > 0,
    \end{cases}
    \label{eq:omega_def}
\end{equation}
with principal branches of powers, we define the matrix-valued function $S$ as follows.
\begin{equation}
    S(z) = \begin{dcases}
        T(z), & \text{for } z \text{ outside the lens},\\
        T(z) \begin{pmatrix} 1 & 0 \\ -\omega(z)^{-1} e^{-2n \phi_{n}(z)} & 1 \end{pmatrix}, & \text{for } z \text{ in the upper parts of the lens},\\
        T(z) \begin{pmatrix} 1 & 0 \\ \omega(z)^{-1} e^{-2n \phi_{n}(z)} & 1 \end{pmatrix}, & \text{for } z \text{ in the lower parts of the lens}.
    \end{dcases}
    \label{eq:S_def}
\end{equation}
By the upper parts of the lens, we mean the regions enclosed by $\Sigma_{1} \cup \Sigma_{2}$ and $\Sigma_{4} \cup \Sigma_{5}$, and by the lower parts of the lens, we mean the regions enclosed by $\Sigma_{3} \cup \Sigma_{2}$ and $\Sigma_{6} \cup \Sigma_{5}$, with the contours $\Sigma_{j}$ shown schematically in \cref{fig:lens}. Because of the shape of these contours, this step is often referred to as `opening the lens'.

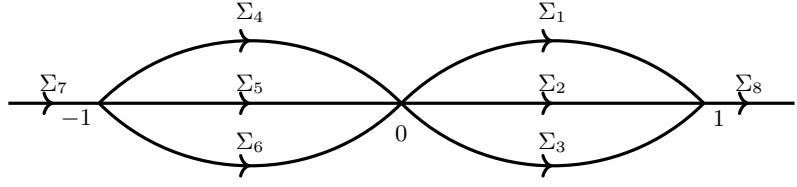
\begin{figure}[t]
\centering
\begin{tikzpicture}[scale=4]
    \begin{scope}[very thick,decoration={markings,mark=at position 0.5 with {\arrow{>}}}] 

        \draw[postaction={decorate}] (-1,0) -- (0,0) node[midway,above] {$\Sigma_{5}$};
        \draw[postaction={decorate}] (0,0) -- (1,0) node[midway,above] {$\Sigma_{2}$};

        \draw[postaction={decorate}] (-1.3,0) -- (-1,0) node[midway,above] {$\Sigma_{7}$};
        \draw[postaction={decorate}] (1,0) -- (1.3,0) node[midway,above] {$\Sigma_{8}$};

        \draw[postaction={decorate}] (0,0) to[out=45,in=135] (1,0);
        \draw[postaction={decorate}] (-1,0) to[out=45,in=135] (0,0);
        \draw[postaction={decorate}] (-1,0) to[out=-45,in=-135] (0,0);
        \draw[postaction={decorate}] (0,0) to[out=-45,in=-135] (1,0);
        
        \node[] at (0.5, 0.3) {$\Sigma_{1}$};
        \node[] at (0.5, -0.13) {$\Sigma_{3}$};
        \node[] at (-0.5, 0.3) {$\Sigma_{4}$};
        \node[] at (-0.5, -0.13) {$\Sigma_{6}$};
        
        \node[] at (-1.075,-0.05) {$-1$};
        \node[] at (1.05,-0.05) {$1$};
        \node[] at (0,-0.1) {$0$};
    \end{scope}
\end{tikzpicture}
\caption{The lens shaped contour $\Sigma = \bigcup_{j=1}^{6} \Sigma_{j}$.}
\label{fig:lens}
\end{figure}

Given the conditions for the Riemann-Hilbert problem for $T$, one can verify that $S$ is the unique solution of the following Riemann-Hilbert problem.
\begin{itemize}
    \item[$(S_{a})$] $S(z)$ is analytic in $\mathbb{C} \setminus \Sigma$.
    \item[$(S_{b})$] $S$ satisfies the following jump relations on $\Sigma$:
        \begin{align}
            S_{+}(z) &= S_{-}(z) \begin{pmatrix} 1 & 0 \\ \omega(z)^{-1} e^{-2n \phi_{n}(z)} & 1 \end{pmatrix}, \quad \text{for } z \in \Sigma \cap \mathbb{C}_{\pm},\label{eq:Sjump1}\\
            S_{+}(x) &= S_{-}(x) \begin{pmatrix} 0 & |x|^{\rho} \\ -|x|^{-\rho} & 0 \end{pmatrix}, \quad \text{for } x \in (-1,1) \setminus \{0\},\label{eq:Sjump2}\\
            S_{+}(x) &= S_{-}(x) \begin{pmatrix} 1 & |x|^{\rho} e^{n(g_{n,+}(x) + g_{n,-}(x) - V_{n}(x) - l_{n})} \\ 0 & 1 \end{pmatrix}, \quad \text{for } |x| > 1.\label{eq:Sjump3}
        \end{align}
    \item[$(S_{c})$] $S(z) = \mathds{1} + \mathcal{O}(1/z)$ as $z \to \infty$.
    \item[$(S_{d})$] For $\rho < 0$, the matrix function $S(z)$ has the following behavior as $z \to 0$:
        \begin{equation}
            S(z) = \mathcal{O} \begin{pmatrix} 1 & |z|^{\rho} \\ 1 & |z|^{\rho} \end{pmatrix}, \quad \text{as } z \to 0, z \in \mathbb{C} \setminus \Sigma. 
        \end{equation}
        For $\rho > 0$, the matrix function $S(z)$ has the following behavior as $z \to 0$:
        \begin{equation}
            S(z) = \begin{dcases}
                \mathcal{O} \begin{pmatrix} 1 & 1 \\ 1 & 1 \end{pmatrix}, & \text{as } z \to 0 \text{ from outside the lens},\\
                \mathcal{O} \begin{pmatrix} |z|^{-\rho} & 1 \\ |z|^{-\rho} & 1 \end{pmatrix}, & \text{as } z \to 0 \text{ from inside the lens}.
            \end{dcases}
        \end{equation}
\end{itemize}
We already know from \cref{lem:variational_cond} that the jump matrix for $|x| > 1$ is exponentially close in $n$ to the identity matrix. For the new jump on $\Sigma \cap \mathbb{C}_{\pm}$, we will show in \cref{lem:lens_boundaries} that $\sup_{z\in\Sigma} |e^{-2n \phi_{n}(z)}| \xrightarrow{n\to\infty} 0$, so this jump matrix will also tend to the identity as $n \to \infty$. Precisely how fast this happens depends on the scaling of $\Re[\phi_{n}(z)]$ with $n$. We will see later that in most cases we can get $\exp[-n \phi_{n}(z)]$ to decay superpolynomially in $n$. However, in the marginal case where the weight $w(x)$ decays only exponentially as $|x| \to \infty$, the decay of $\exp[-n \phi_{n}(z)]$ with $n$ will be slower, only super-polylogarithmic in the worst case (see \cref{lem:lens_boundaries} for details).

\subsection{Parametrix $N$ for the outside region}
\label{sec:outside_parametrix}
Since the jumps for $S(z)$ for $z$ away from $(-1,1)$ decay with $n$, it is useful to solve the Riemann-Hilbert problem for only the jump on $(-1,1)$, whose solution we will denote by $N(z)$.
\paragraph*{Riemann-Hilbert problem for $N$}
\begin{itemize}
    \item[$(N_{a})$] $N : \mathbb{C} \setminus [-1,1] \to \mathbb{C}^{2 \times 2}$ is analytic.
    \item[$(N_{b})$] $N_{+}(x) = N_{-}(x) \begin{pmatrix} 0 & |x|^{\rho} \\ -|x|^{-\rho} & 0 \end{pmatrix}$ for $x \in (-1,1) \setminus \{0\}$.
    \item[$(N_{c})$] $N(z) = \mathds{1} + \mathcal{O}\left(\dfrac{1}{|z|}\right)$ as $|z| \to \infty$.
\end{itemize}
We will construct a solution to this RH problem in terms of the Szeg\H{o} function $D$ associated with $|x|^{\rho}$ on $(-1,1)$. This is a scalar function which is the solution to a multiplicative scalar RH problem with the following conditions.

\paragraph*{Riemann-Hilbert problem for $D$}
\begin{itemize}
    \item[$(D_{a})$] $D : \mathbb{C} \setminus [-1,1] \to \mathbb{C}$ is analytic and nonzero.
    \item[$(D_{b})$] $D_{+}(x) D_{-}(x) = |x|^{\rho}$ for $x \in (-1,1) \setminus \{0\}$.
    \item[$(D_{c})$] $\lim_{z \to \infty} D(z) = D_{\infty}$ exists and is a positive real number.
\end{itemize}
In this case, one can check that the Szeg\H{o} function is given by
\begin{equation}
    D(z) = \dfrac{z^{\rho/2}}{\varphi(z)^{\rho/2}},
    \label{eq:D_def}
\end{equation}
where
\begin{equation}
    \varphi(z) \coloneqq z + (z + 1)^{1/2}(z - 1)^{1/2}
    \label{eq:Pinf_varphi_def}
\end{equation}
is the conformal map from $\mathbb{C} \setminus [-1,1]$ to the exterior of the unit circle; $\varphi(z)$ has a branch cut along $[-1,1]$ and behaves like $\varphi(z) \approx 2z$ as $z \to \infty$.

Having obtained $D$, it turns out that the solution to the RH problem for $N$ is~\cite{kuijlaarsRiemannHilbertAnalysisOrthogonal2003}
\begin{equation}
    N(z) = \begin{pmatrix}
        D_{\infty} & 0 \\
        0 & \frac{1}{D_{\infty}}
    \end{pmatrix}
    \begin{pmatrix}
        \frac{a(z) + a^{-1}(z)}{2} & \frac{a(z) - a^{-1}(z)}{2i} \\
        \frac{a(z) - a^{-1}(z)}{-2i} & \frac{a(z) + a^{-1}(z)}{2}
    \end{pmatrix}
    \begin{pmatrix}
        \frac{1}{D(z)} & 0 \\
        0 & D(z)
    \end{pmatrix},
    \label{eq:Pinf_sol}
\end{equation}
where
\begin{equation}
    a(z) \coloneqq \dfrac{(z-1)^{1/4}}{(z+1)^{1/4}} \quad \text{and} \quad D_{\infty} \coloneqq \lim_{z \to \infty} D(z) = 2^{-\rho/2}.
    \label{eq:a_Pinf_def}
\end{equation}
Since $N(z)$ is a product of three matrices each with $\det = 1$, we have $\det{N(z)} = 1$, so $N(z)$ is invertible. Furthermore, as $z \to \infty$, this matrix scales as
\begin{equation}
    N(z) = \mathds{1} + \dfrac{i}{2z} \begin{pmatrix}
        0 & 2^{-\rho} \\
        -2^{\rho} & 0
    \end{pmatrix} + \mathcal{O}\left(\dfrac{1}{z^{2}}\right),
    \label{eq:N_large_z_asymptotics}
\end{equation}
and as $z \to 0$ it scales as
\begin{equation}
    N(z) \sim \mathcal{O} \begin{pmatrix}
        z^{-\rho/2} & z^{\rho/2} \\
        z^{-\rho/2} & z^{\rho/2}
    \end{pmatrix},
\end{equation}
where the big-$\mathcal{O}$ notation is taken elementwise.

\subsection{Local parametrix $P$ near the endpoints $z = \pm 1$}
\label{sec:endpoints}
While it is true that the jumps of $S(z)$ decay with $n$ for $z$ away from the real axis, these jumps become nonnegligible at the points $z = 0$ and $z = \pm 1$ where the contour $\Sigma$ meet the real axis. We need to construct `local parametrices' which approximately solve the Riemann-Hilbert problem for $S$ near these points, and handle any potential singularities.

In this subsection we will construct the parametrices for the endpoints $z = \pm 1$, which we will build out of Airy functions (because the equilibrium measure is regular for large enough $n$). Our analysis here is standard and will closely follow that of Refs.~\cite{deiftStrongAsymptoticsOrthogonal1999,vanlessenStrongAsymptoticsLaguerreType2007}. First we focus on the parametrix near $z = 1$, since the construction near $z = -1$ will be closely related by symmetry. We will make use of the function $h_{n}(z)$ defined in \cref{eq:hn_def}, which is analytic in the region of analyticity of $V_{n}$, which by assumption contains a neighbourhood of $z = 1$. We will construct the parametrix $P(z)$ in a disk $U_{\delta_{2}} = \{z \in \mathbb{C}: |z-1| < \delta_{2}\}$ around $z = 1$, where the radius $\delta_{2} > 0$ is a small constant. $P(z)$ is the solution to the following Riemann-Hilbert problem.

\paragraph*{Riemann-Hilbert problem for $P$}
\begin{itemize}
    \item[$(P_{a})$] $P$ is analytic in $U_{\delta_{2}} \setminus \Sigma$.
    \item[$(P_{b})$] $P_{+}(z) = P_{-}(z) v_{S}(z)$ for $z \in \Sigma \cap U_{\delta_{2}}$, with $v_{S}$ the jump matrix for $S$.
    \item[$(P_{c})$] $P(z) N(z)^{-1} \sim \mathds{1} + \mathcal{O}(1/n)$ as $n \to \infty$, uniformly for $z$ on the boundary $\partial U_{\delta}$ of the disk $U_{\delta}$ for $\delta$ in compact subsets of $(0,\delta_{2})$.   
\end{itemize}
The construction of $P$ involves three steps. First we will construct a matrix-valued function that satisfies conditions $(P_{a})$ and $(P_{b})$. In order to do this we will transform this RHP into a RHP for $P^{(1)}$ with constant jump matrices and then construct a solution of the latter RHP. Finally, we will take the matching condition $(P_{c})$ into account.

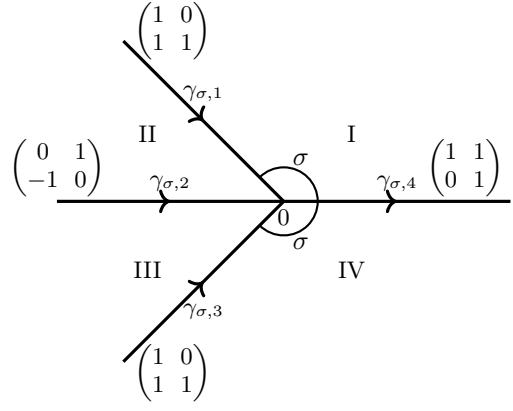
\begin{figure}[t]
\centering
\begin{tikzpicture}[scale=3]
    \begin{scope}[very thick,decoration={markings,mark=at position 0.5 with {\arrow{>}}}] 
        \draw[postaction={decorate}] (-1,0) -- (0,0) node[midway,above,label distance=0.1] {$\gamma_{\sigma,2}$};
        \draw[postaction={decorate}] (0,0) -- (1,0) node[midway,above] {$\gamma_{\sigma,4}$};
        \draw[postaction={decorate}] (135:1) -- (0,0) node[midway,label={[label distance=0.1]above:$\gamma_{\sigma,1}$}] {};
        \draw[postaction={decorate}] (225:1) -- (0,0) node[midway,label={[label distance=0.8]below:$\gamma_{\sigma,3}$}] {};

        \node[] at (0.3, 0.3) {I};
        \node[] at (0.3, -0.3) {IV};
        \node[] at (-0.6,0.3) {II};
        \node[] at (-0.6,-0.3) {III};
        
        \node[] at (0.8, 0.15) {$\begin{pmatrix} 1 & 1 \\ 0 & 1\end{pmatrix}$};
        \node[] at (-1, 0.15) {$\begin{pmatrix} 0 & 1 \\ -1 & 0\end{pmatrix}$};
        \node[] at (123:0.9) {$\begin{pmatrix} 1 & 0 \\ 1 & 1\end{pmatrix}$};
        \node[] at (237:0.9) {$\begin{pmatrix} 1 & 0 \\ 1 & 1\end{pmatrix}$};
    \end{scope}
    \coordinate[label=below:$0$] (O) at (0,0);
    \coordinate (R) at (0.05,0);
    \coordinate (Lu) at (135:0.05);
    \coordinate (Ld) at (225:0.05);
    \tkzMarkAngle[mark=none,size=0.15,thick](R,O,Lu)
    \tkzLabelAngle[pos = 0.2](R,O,Lu){$\sigma$}
    \tkzMarkAngle[mark=none,size=0.15,thick](Ld,O,R)
    \tkzLabelAngle[pos = 0.2](Ld,O,R){$\sigma$}
\end{tikzpicture}
\caption{The oriented contour $\gamma_{\sigma}$ and the jump matrix $v_{1}$ for $\Psi$ on $\gamma_{\sigma}$. The four straight rays $\gamma_{\sigma,1},\dots,\gamma_{\sigma,4}$ divide the complex plane into four regions I, II, III, and IV.}
\label{fig:unity_endpoint}
\end{figure}

\subsubsection{Transformation to constant jump matrices}
Following Ref.~\cite{vanlessenStrongAsymptoticsLaguerreType2007}, in order to transform to constant jump matrices, we seek the parametrix $P$ near $1$ in the following form
\begin{equation}
    P(z) = E(z) P^{(1)}(z) e^{-n \phi_{n}(z) \sigma_{3}} z^{-(1/2)\rho \sigma_{3}}, \quad \text{for } z \in U_{\delta_{2}} \setminus \Sigma,
    \label{eq:P_def}
\end{equation}
with $E$ an invertible analytic matrix-valued function in $U_{\delta_{2}}$, to be determined to ensure that the matching condition $(P_{c})$ is satisfied. Since $E$ is invertible and analytic, one can use \cref{eq:psi_cont_jump,eq:psi_cont_smooth,eq:phi_n_def} to show that the jump matrix for $P^{(1)}$ is piecewise constant:
\begin{equation}
    P^{(1)}_{+}(z) = \begin{dcases}
        P^{(1)}_{-}(z) \begin{pmatrix} 1 & 0 \\ 1 & 1 \end{pmatrix}, & z \in (\Sigma_{1} \cup \Sigma_{3}) \cap U_{\delta_{2}},\\
        P^{(1)}_{-}(z) \begin{pmatrix} 0 & 1 \\ -1 & 0 \end{pmatrix}, & z \in \Sigma_{2} \cap U_{\delta_{2}} = (1 - \delta_{2}, 1),\\
        P^{(1)}_{-}(z) \begin{pmatrix} 1 & 1 \\ 0 & 1 \end{pmatrix}, & z \in \Sigma_{8} \cap U_{\delta_{2}} = (1, 1 + \delta_{2}).
    \end{dcases}
\end{equation}
We will determine $P^{(1)}$ subject to these constant jump matrices by an explicit construction based on an auxiliary problem for $\Psi(\zeta)$ in the $\zeta$-plane with jumps on the oriented contour $\gamma_{\sigma} = \cup_{j=1}^{4} \gamma_{\sigma,j}$, shown in \cref{fig:unity_endpoint}, consisting of four straight rays
\begin{equation*}
    \gamma_{\sigma,1} : \arg{\zeta} = \sigma, \quad \gamma_{\sigma,2} : \arg{\zeta} = \pi, \quad \gamma_{\sigma,3} : \arg{\zeta} = -\sigma, \quad \gamma_{\sigma,4} : \arg{\zeta} = 0,
\end{equation*}
with $\sigma \in (\pi/3,\pi)$. These four rays divide the complex plane into four regions I, II, III, and IV, also shown in \cref{fig:unity_endpoint}.

\paragraph*{Riemann-Hilbert problem for $\Psi$}
\begin{itemize}
    \item[$(\Psi_{a})$] $\Psi$ is analytic in $\mathbb{C} \setminus \gamma_{\sigma}$. 
    \item[$(\Psi_{b})$] $\Psi_{+}(\zeta) = \Psi_{-}(\zeta) v_{1}(\zeta)$ for $\zeta \in \gamma_{\sigma}$, where $v_{1}$ is the piecewise constant matrix-valued function on $\gamma_{\sigma}$ defined as shown in \cref{fig:unity_endpoint}, i.e., $v_{1}(\zeta) = \begin{pmatrix} 1 & 0 \\ 1 & 1\end{pmatrix}$ for $\zeta \in \gamma_{\sigma,1}$, etc. This means that $\Psi$ has the same jumps on $\gamma_{\sigma}$ as $P^{(1)}$ does on $\Sigma_{S} \cap U_{\delta_{2}}$.
    \item[$(\Psi_{c})$] $\Psi$ has the following behaviour at infinity,
        \begin{align}
            \label{eq:Psi_asymptotic}
            \Psi(\zeta) \sim \, &\zeta^{-\sigma_{3}/4} \dfrac{1}{\sqrt{2}} \begin{pmatrix}
                1 & 1 \\ -1 & 1
            \end{pmatrix}\\
            &\times \left[\mathds{1} + \sum_{k=1}^{\infty} \frac{1}{2}\left(\frac{2}{3} \zeta^{3/2}\right)^{-k} \begin{pmatrix}
                (-1)^{k} (s_{k} + t_{k}) & s_{k} - t_{k} \\
                (-1)^{k} (s_{k} - t_{k}) & s_{k} + t_{k}
            \end{pmatrix}\right] e^{-(\pi i/4) \sigma_{3}} e^{-(2/3) \zeta^{3/2} \sigma_{3}},\nonumber
        \end{align}
    as $\zeta \to \infty$, uniformly for $\zeta \in \mathbb{C} \setminus \gamma_{\sigma}$ and $\sigma$ in compact subsets of $(\pi/3,\pi)$. Here,
    \begin{equation}
        s_{k} = \dfrac{\Gamma\left(3k + \frac{1}{2}\right)}{54^{k} k! \Gamma\left(k + \frac{1}{2}\right)}, \quad t_{k} = - \dfrac{6k+1}{6k-1} s_{k}, \quad \text{for } k \geq 1.
        \label{eq:sk_tk_def}
    \end{equation}
\end{itemize}
It is well known \cite{deiftOrthogonalPolynomialsRandom2000,deiftStrongAsymptoticsOrthogonal1999} that, with $\varpi \coloneqq e^{2\pi i/3}$ and $\Ai$ the Airy function, the solution to this RHP is
\begin{equation}
    \Psi(\zeta) = \sqrt{2\pi} e^{-\pi i/12} \times \begin{dcases}
        \begin{pmatrix} \Ai(\zeta) & \Ai(\varpi^{2} \zeta) \\ \Ai^{\prime}(\zeta) & \varpi^{2} \Ai^{\prime}(\varpi^{2}\zeta) \end{pmatrix} e^{-(\pi i/6) \sigma_{3}}, & \zeta \in \text{I},\\
        \begin{pmatrix} \Ai(\zeta) & \Ai(\varpi^{2} \zeta) \\ \Ai^{\prime}(\zeta) & \varpi^{2} \Ai^{\prime}(\varpi^{2}\zeta) \end{pmatrix} e^{-(\pi i/6) \sigma_{3}} \begin{pmatrix} 1 & 0 \\ -1 & 1 \end{pmatrix}, & \zeta \in \text{II},\\
        \begin{pmatrix} \Ai(\zeta) & -\varpi^{2} \Ai(\varpi \zeta) \\ \Ai^{\prime}(\zeta) & -\Ai^{\prime}(\varpi\zeta) \end{pmatrix} e^{-(\pi i/6) \sigma_{3}} \begin{pmatrix} 1 & 0 \\ 1 & 1 \end{pmatrix}, & \zeta \in \text{III},\\
        \begin{pmatrix} \Ai(\zeta) & -\varpi^{2} \Ai(\varpi \zeta) \\ \Ai^{\prime}(\zeta) & -\Ai^{\prime}(\varpi\zeta) \end{pmatrix} e^{-(\pi i/6) \sigma_{3}}, & \zeta \in \text{IV},\\
    \end{dcases}
    \label{eq:Psi_def}
\end{equation}
With this definition, we can now specify the precise form of the contour $\Sigma$ near $z=1$: $\Sigma$ is defined in $U_{\delta_{2}}$ as the inverse $f_{n}$-image of $\gamma_{\sigma} \cap f_{n}(U_{\delta_{2}})$, where $f_{n}(z)$ is a biholomorphic map between the $z$- and $\zeta$-planes to be constructed in the next section. We then define
\begin{equation}
    P^{(1)}(z) \coloneqq \Psi(f_{n}(z)), \quad \text{for } z \in U_{\delta_{2}} \setminus f_{n}^{-1}(\gamma_{\sigma}).
    \label{eq:P1_def}
\end{equation}

\subsubsection{Construction of biholomorphic map}
\begin{proposition}
    There exists a $\delta_{2} > 0$ such that for sufficiently large $n$ there exists a function $f_{n} : U_{\delta_{2}} \to f_{n}(U_{\delta_{2}}) \subset \mathbb{C}$ with the following properties.
    \begin{enumerate}[label=(\alph*)]
        \item $f_{n} : U_{\delta_{2}} \to f_{n}(U_{\delta_{2}})$ is biholomorphic.
        \item $f_{n}(U_{\delta_{2}} \cap \mathbb{R}) = f_{n}(U_{\delta_{2}}) \cap \mathbb{R}$, $f_{n}(U_{\delta_{2}} \cap \mathbb{C}_{\pm}) = f_{n}(U_{\delta_{2}}) \cap \mathbb{C}_{\pm}$.
        \item $\frac{2}{3}(f_{n}(z))^{3/2} = -n \phi_{n}(z)$ for all $z \in U_{\delta_{2}} \setminus (-\infty,1]$.
    \end{enumerate}
    \label{prop:fn_endpoint}
\end{proposition}
\begin{proof}[Proof (adapted from \cite{vanlessenStrongAsymptoticsLaguerreType2007}, Prop 3.19)]
Given the analyticity of $h_{n}(z)$ near $z = 1$, we will explicitly construct $f_{n}(z)$ as follows.
We define the auxiliary function $\hat{\varphi}_{n} : \mathbb{C} \setminus (-\infty,1] \to \mathbb{C}$ by
\begin{align}
    \hat{\varphi}_{n}(z) &= \dfrac{3}{4}(z-1)^{-3/2} \int_{1}^{z} h_{n}(s) (s+1)^{1/2} (s-1)^{1/2} \diff s,\label{eq:hatvarphi_def}\\
    &= \dfrac{h_{n}(1)}{\sqrt{2}} + \dfrac{3}{4}(z-1)^{-3/2} \int_{1}^{z} \left[h_{n}(s) (s+1)^{1/2} - h_{n}(1) \sqrt{2}\right] (s-1)^{1/2} \diff s.\label{eq:hatvarphi_int}
\end{align}
Note that $\hat{\varphi}_{n}(z)$ actually has no jumps across $(-1,1)$ because both $(z-1)^{-3/2}$ and $(s-1)^{1/2}$ flip sign, so we already know that $\hat{\varphi}_{n}(z)$ has an analytic continuation to $U_{\delta} \setminus \{1\}$ for some $\delta > 0$. To deal with the potential difficulty at $z=1$, note that, for $|s-1| < \delta$, Cauchy's theorem tells us that
\begin{align*}
    \left|h_{n}(s) (s+1)^{1/2} - h_{n}(1) \sqrt{2}\right| &= \left|(s-1) \dfrac{1}{2\pi i} \oint_{|w-1|=2\delta} \dfrac{h_{n}(w) (w+1)^{1/2} - h_{n}(1) \sqrt{2}}{(w-1)(w-s)} \diff w \right|,\\
    &\leq \dfrac{|s-1|}{\delta} \sup_{|w-1|=2\delta} \left|h_{n}(w) (w+1)^{1/2} - h_{n}(1) \sqrt{2}\right|.
\end{align*}
In \cref{lem:hn1_lower_bound} we prove that, for sufficiently large $n$ and sufficiently small $\delta$, there exists a constant $\epsilon>0$ (independent of $n$) such that $|h_{n}(w) - h_{n}(1)|<\epsilon$ for $|w-1|<2\delta$. Together with \cref{lem:hn1_scaling}, which tells us that $h_{n}(1) = 2p + o(1)$, it is easy to see that there is a constant $c > 0$ such that for $|s-1|<\delta$ and sufficiently large $n$ we have
\begin{equation}
    \left|h_{n}(s) (s+1)^{1/2} - h_{n}(1) \sqrt{2}\right| \leq c |s-1|.
\end{equation}
Inserting this into \cref{eq:hatvarphi_int}, we conclude there is a constant $C_{1} > 0$ such that for sufficiently large $n$ we have
\begin{equation}
    |\hat{\varphi}_{n}(z) - \hat{\varphi}_{n}(1)| \leq C_{1} |z-1|, \quad \text{for } |z-1|<\delta,
    \label{eq:hatvarphi_bound}
\end{equation}
where $\hat{\varphi}_{n}(1) = h_{n}(1) / \sqrt{2}$. Therefore the isolated singularity of $\hat{\varphi}_{n}(z)$ at $z=1$ is removable, so that $\hat{\varphi}_{n}$ is indeed analytic for $z \in U_{\delta}$. Having established these properties of $\hat{\varphi}_{n}(z)$, we then define the function
\begin{equation}
    \varphi_{n}(z) \coloneqq (z-1) \left(\hat{\varphi}_{n}(z)\right)^{2/3}, \quad \text{for } z \in U_{\delta}.
    \label{eq:varphi_def}
\end{equation}
Since $\hat{\varphi}_{n}(1) = h_{n}(1) / \sqrt{2}$ and $h_{n}(1) = 2p + o(1)$ by \cref{lem:hn1_scaling}, we conclude that $\Re[\hat{\varphi}_{n}(z)] > 0$ for $z \in U_{\delta}$ and sufficiently large $n$. This implies that $\varphi_{n}(z)$ is analytic for $z \in U_{\delta}$.

To establish that $\varphi_{n}(z)$ is injective, observe that, by \cref{eq:hatvarphi_bound,eq:varphi_def,lem:hn1_scaling,lem:hn1_lower_bound}, $\varphi_{n}(z)$ is uniformly (in $n$ and $z$) bounded in $U_{\delta}$. By Cauchy's theorem for derivatives, this implies that $\varphi_{n}^{\prime\prime}(z)$ is also uniformly (in $n$ and $z$) bounded in $U_{\delta}$ for a smaller $\delta$. Since $\hat{\varphi}_{n}(1) = h_{n}(1) / \sqrt{2} \neq 0$, we have $\varphi_{n}^{\prime}(1) = (h_{n}(1)/\sqrt{2})^{2/3}$, so that
\begin{equation*}
    \left|\varphi_{n}^{\prime}(z) - (h_{n}(1)/\sqrt{2})^{2/3}\right| = \left| \int_{1}^{z} \varphi_{n}^{\prime\prime}(s) \diff s \right| \leq C_{2} |z-1|, \quad \text{for } z \in U_{\delta},
\end{equation*}
for some constant $C_{2} > 0$ and sufficiently large $n$. By \cref{lem:hn1_scaling} we know there exists $h_{0} > 0$ such that $h_{n}(1) > h_{0}$ for sufficiently large $n$, and so by the inverse function theorem we conclude there exists $0 < \delta_{2} < \delta$ such that $\varphi_{n}$ is injective and hence biholomorphic in $U_{\delta_{2}}$.

Finally, we define the biholomorphic map $f_{n}(z)$ by
\begin{equation}
    f_{n}(z) \coloneqq n^{2/3} \varphi_{n}(z).
    \label{eq:biholomorphic_def}
\end{equation}
Property (b) follows from the fact that $\varphi_{n}$ and its inverse $\varphi_{n}^{-1}$ are real on the real axis, while property (c) follows from
\begin{equation}
    \phi_{n}(z) = -\dfrac{2}{3} (z-1)^{3/2} \hat{\varphi}_{n}(z),
\end{equation}
which follows from \cref{eq:phi_n_def,eq:hatvarphi_def}.

\end{proof}
To calculate residues we will need the values of $\hat{\varphi}_{n}(1)$ and $\hat{\varphi}_{n}^{\prime}(1)$, which can be calculated from the values of $h_{n}(1)$ and $h_{n}^{\prime}(1)$ via
\begin{align}
    \hat{\varphi}_{n}(1) &= \dfrac{1}{\sqrt{2}} h_{n}(1),\\
    \hat{\varphi}^{\prime}_{n}(1) &= \dfrac{3 \sqrt{2}}{40} \left[h_{n}(1) + 4 h_{n}^{\prime}(1)\right].
\end{align}

\subsubsection{Satisfying the matching condition}
In this final step, we construct the invertible analytic matrix-valued function $E(z)$ to ensure that the matching function $(P_{c})$ is satisfied. In the $n \to \infty$ limit, by comparing \cref{eq:P1_def,eq:Psi_asymptotic,eq:biholomorphic_def}, we see that we must define $E$ as
\begin{equation}
    E(z) \coloneqq N(z) z^{(1/2)\rho \sigma_{3}} e^{(\pi i/4) \sigma_{3}} \dfrac{1}{\sqrt{2}} \begin{pmatrix} 1 & -1 \\ 1 & 1 \end{pmatrix} f_{n}(z)^{\sigma_{3}/4}, \quad \text{for } z \in U_{\delta_{2}}.
    \label{eq:E_def}
\end{equation}
It is straightforward to check that $E(z)$ so defined is indeed analytic and invertible (see \cite{vanlessenStrongAsymptoticsLaguerreType2007}, Remark 3.21). This ends the construction of the parametrix $P$ near $z = 1$.

\subsubsection{Summary of the local analysis near $z=1$}
Fixing an angle $\sigma \in (\pi/3,\pi)$, we have defined the contour $\Sigma$ in a disk $U_{\delta_{2}}$ near $z=1$ as the inverse image of the contour $\gamma_{\sigma}$ under the map $f_{n}$ defined in \cref{eq:biholomorphic_def}. We have then found a solution
\begin{equation}
    P(z) = E(z) \Psi(f_{n}(z)) e^{-n \phi_{n}(z)\sigma_{3}} z^{-(1/2) \rho \sigma_{3}}, \quad \text{for } z \in \Sigma \cap U_{\delta_{2}}
    \label{eq:P_def_endpoint}
\end{equation}
to the RHP for $P$, where the matrix-valued function $E$ is defined in \cref{eq:E_def}, and the matrix-valued function $\Psi$ is defined in \cref{eq:Psi_def}. Furthermore, we have the asymptotic expansion
\begin{equation}
    P(z)N(z)^{-1} \sim \mathds{1} + \sum_{k=1}^{\infty} \Delta_{k}(z) \dfrac{1}{n^{k}}, \quad \text{as } n \to \infty,
    \label{eq:P_Pinf_expansion}
\end{equation}
uniformly for $z$ in compact subsets of $\{0 < |z-1| < \delta_{2}\}$ and $\sigma$ in compact subsets of $(\pi/3,\pi)$, where $\Delta_{k}$ is a meromorphic $2 \times 2$ matrix-valued function determined by \cref{eq:P_def,eq:P1_def,eq:E_def,eq:Psi_asymptotic,eq:biholomorphic_def}, and given explicitly by
\begin{equation}
    \Delta_{k}(z) = \dfrac{1}{2 \left(-\phi_{n}(z)\right)^{k}} N(z) z^{(1/2)\rho \sigma_{3}} \begin{pmatrix}
                (-1)^{k} (s_{k} + t_{k}) & i(s_{k} - t_{k}) \\
                -i(-1)^{k} (s_{k} - t_{k}) & s_{k} + t_{k}
            \end{pmatrix} z^{-(1/2) \rho \sigma_{3}} N(z)^{-1},
    \label{eq:Delta_k_def}
\end{equation}
where $\phi_{n}$ is defined in \cref{eq:phi_n_def}, $N(z)$ is given in \cref{eq:Pinf_sol}, and the scalars $s_{k}$ and $t_{k}$ are defined in \cref{eq:sk_tk_def}. One can verify that $\Delta_{k}(z)$ is meromorphic in a neighbourhood of $z=1$ by using the jump conditions of $\phi_{n}(z)$ and $N(z)$.

The first term in this asymptotic expansion is (suppressing $z$ arguments to reduce notational clutter)
\begin{align*}
    \Delta_{1} = \dfrac{-1}{144 \phi_{n}} \Big[ &\begin{pmatrix}
        a^{2} + a^{-2} & 2^{-\rho} i(a^{2} - a^{-2}) \\
        2^{\rho} i(a^{2} - a^{-2}) & -(a^{2} + a^{-2})
    \end{pmatrix}\\
    +3
    &\begin{pmatrix}
        (\varphi^{\rho} + \varphi^{-\rho})(a^{2} - a^{-2}) & 2^{-\rho}i\left[\varphi^{\rho}(a + a^{-1})^{2} + \varphi^{-\rho}(a - a^{-1})^{2} \right] \\
        2^{\rho}i\left[\varphi^{-\rho}(a + a^{-1})^{2} + \varphi^{\rho}(a - a^{-1})^{2} \right] & -(\varphi^{\rho} + \varphi^{-\rho})(a^{2} - a^{-2})
    \end{pmatrix}\Big].
\end{align*}
Note that $\Delta_{k}(z)$ has a pole of order $(3k+1)/2$ at $z = 1$.

\subsubsection{Parametrix near $z = -1$}
We can use symmetry to carry over most of the analysis of the parametrix near $z = 1$ to obtain a parametrix near $z = -1$. In particular, consider the transformation
\begin{equation}
    \wt{\Psi}(z) \coloneqq \sigma_{3} \Psi(-z) \sigma_{3}.
\end{equation}
Under the transformation $z \to -z$, the approach to a contour from the $(+)$-side gets mapped to the approach to the equivalent contour at $-z$ but from the $(-)$-side. This implies that the corresponding jump matrix is the \textit{inverse} of the jump matrix at $-z$. For the set of jump matrices considered here, the extra conjugation by $\sigma_{3}$ is sufficient to correct for this inverse, i.e.~$\sigma_{3} v_{j}^{-1} \sigma_{3} = v_{j}$, so that $\wt{\Psi}(z)$ has the correct jump behaviour near $z = -1$.

We construct the rest of the parametrix in a similar way to before, with the result
\begin{equation}
    P(z) = \wt{E}(z) \wt{\Psi}(-f_{n}(-z)) e^{-n \phi_{n}(-z) \sigma_{3}} (-z)^{-(\rho/2) \sigma_{3}},
    \label{eq:Pdef_left_endpoint}
\end{equation}
where $\wt{E}(z)$ is an analytic matrix-valued function given by
\begin{equation}
    \wt{E}(z) = N(z) (-z)^{(\rho/2) \sigma_{3}} e^{(\pi i/4) \sigma_{3}} \dfrac{1}{\sqrt{2}}\begin{pmatrix} 1 & 1 \\ -1 & 1 \end{pmatrix} f_{n}(-z)^{\sigma_{3}/4}.
\end{equation}
Notice the slight difference in the constant matrices in the definitions of $\wt{E}(z)$ and $E(z)$, which comes about from conjugating with $\sigma_{3}$. Note that for $\rho = 0$ we simply have $P(z) = \sigma_{3} P(-z) \sigma_{3}$, but this is no longer quite true for $\rho \neq 0$, the reason being that the identity $\sigma_{3} N(z) \sigma_{3} = N(-z)$ holds only for $\rho = 0$.

As before, we have an asymptotic expansion
\begin{equation}
    P(z) N^{-1}(z) \sim \mathds{1} + \sum_{k=1}^{\infty} \wt{\Delta}_{k}(z) \dfrac{1}{n^{k}}, \quad \text{as } n \to \infty,
    \label{eq:DeltaR_expansion_z=-1}
\end{equation}
uniformly for $z$ in compact subsets of $U_{\delta_{2}} \coloneqq \{0 < |z+1| < \delta_{2}\}$ and $\sigma$ in compact subsets of $(\pi/3,\pi)$, where $\wt{\Delta}_{k}$ is a meromorphic $2 \times 2$ matrix-valued function given by
\begin{equation}
    \wt{\Delta}_{k}(z) = \dfrac{1}{2 \left(-\phi_{n}(-z)\right)^{k}} N(z) (-z)^{(1/2)\rho \sigma_{3}} \begin{pmatrix}
                (-1)^{k} (s_{k} + t_{k}) & -i(s_{k} - t_{k}) \\
                i(-1)^{k} (s_{k} - t_{k}) & s_{k} + t_{k}
            \end{pmatrix} (-z)^{-(1/2) \rho \sigma_{3}} N(z)^{-1}.
\end{equation}
The first term in this expansion is
\begin{equation}
    \wt{\Delta}_{1}(z) = \dfrac{-1}{144 \phi_{n}(-z)}\left[
        \begin{pmatrix}
            a^{2} + a^{-2} & i 2^{-\rho}(a^{2} - a^{-2}) \\
            i 2^{\rho}(a^{2} - a^{-2}) & -(a^{2} + a^{-2})
        \end{pmatrix} - 3 \Omega_{1}(z)\right],
\end{equation}
where $\Omega_{1}(z)$ is defined for $U_{\delta_{2}} \cap \mathbb{C}_{+}$ as 
\begin{equation}
    \Omega_{1}(z) \coloneqq 
        \begin{pmatrix}
            (e^{-i \pi \rho} \varphi^{\rho} + e^{i \pi \rho} \varphi^{-\rho})(a^{2} - a^{-2}) & 2^{-\rho}i \left[e^{i \pi \rho} \varphi^{-\rho} (a - a^{-1})^{2} + e^{-i \pi \rho} \varphi^{\rho} (a + a^{-1})^{2}\right] \\
            2^{\rho}i \left[e^{i \pi \rho} \varphi^{-\rho} (a + a^{-1})^{2} + e^{-i \pi \rho} \varphi^{\rho} (a - a^{-1})^{2}\right] & -(e^{-i \pi \rho} \varphi^{\rho} + e^{i \pi \rho} \varphi^{-\rho})(a^{2} - a^{-2})
        \end{pmatrix},
\end{equation}
and for $U_{\delta_{2}} \cap \mathbb{C}_{-}$ as 
\begin{equation}
    \Omega_{1}(z) \coloneqq 
    \begin{pmatrix}
        (e^{i \pi \rho} \varphi^{\rho} + e^{-i \pi \rho} \varphi^{-\rho})(a^{2} - a^{-2}) & 2^{-\rho}i \left[e^{-i \pi \rho} \varphi^{-\rho} (a - a^{-1})^{2} + e^{i \pi \rho} \varphi^{\rho} (a + a^{-1})^{2}\right] \\
        2^{\rho}i \left[e^{-i \pi \rho} \varphi^{-\rho} (a + a^{-1})^{2} + e^{i \pi \rho} \varphi^{\rho} (a - a^{-1})^{2}\right] & -(e^{i \pi \rho} \varphi^{\rho} + e^{-i \pi \rho} \varphi^{-\rho})(a^{2} - a^{-2})
    \end{pmatrix}.
\end{equation}
In these expressions we have suppressed for brevity the arguments of the functions $a$ and $\varphi$, which are meant to be evaluated at $z$; the only function evaluated at $-z$ is $\phi_{n}$.

\subsection{Local parametrix $P$ near the origin}
\label{sec:origin}
In this section we construct a local parametrix $P$ which approximately solves the Riemann-Hilbert problem for $S$ in a disk $U_{\delta_{n}} = \{z \in \mathbb{C} : |z| < \delta_{n}\}$ of radius $\delta_{n} = \gamma / \beta_{n}$ centered at the origin. Our analysis will be inspired by that of \cite{kuijlaarsUniversalityEigenvalueCorrelations2003,vanlessenStrongAsymptoticsRecurrence2003,vanlessenStrongAsymptoticsLaguerreType2007}, but the main novelty will be that we construct the parametrix in a disk of radius shrinking like $\mathcal{O}(1/\beta_{n})$. We do this to avoid having to deal with any singularities in the equilibrium measure; it was for this purpose that we assumed analyticity of the weight function in a region $C_{\theta}$ which is invariant under rescaling, as well as in a disk of constant radius centered at the origin. The price we will pay for shrinking the disk is error bounds in the matching condition $P(z) N(z)^{-1} \sim \mathds{1} + o(1)$ which decay more slowly with $n$.
\paragraph*{Riemann-Hilbert problem for $P$}
\begin{itemize}
    \item[$(P_{a})$] $P$ is analytic in $U_{\delta_{n}} \setminus \Sigma$.
    \item[$(P_{b})$] $P_{+}(z) = P_{-}(z) v_{S}(z)$ for $z \in \Sigma \cap U_{\delta_{n}}$, with $v_{S}$ the jump matrix for $S$.
    \item[$(P_{c})$] $P(z) N(z)^{-1} \sim \mathds{1} + o(1)$ as $n \to \infty$, uniformly for $z$ on the boundary $\partial U_{\delta}$ of the disk $U_{\delta}$ for $\delta$ in compact subsets of $(0,\delta_{n})$.   
\end{itemize}
As in \cref{sec:endpoints}, we will first transform to a problem $P^{(1)}$ with constant jump matrices, and then construct an explicit solution to this problem using special functions.

\subsubsection{Transformation to constant jump matrices}
\label{sec:bessel_sol}
Following Ref.~\cite{kuijlaarsUniversalityEigenvalueCorrelations2003}, in order to transform to constant jump matrices, we seek the parametrix $P$ near $0$ in the following form
\begin{equation}
    P(z) = E_{n}(z) P^{(1)}(z) W(z)^{-\sigma_{3}} e^{-n \phi_{n}(z) \sigma_{3}},
\end{equation}
where $E_{n}(z)$ is an analytic matrix-valued function to be determined to ensure the matching condition $(P_{c})$, and $W(z)$ is defined by
\begin{equation}
    W(z) \coloneqq \begin{cases}
        z^{\rho/2}, & \text{if } \pi/2 < \left| \arg{f_{n}(z)} \right| < \pi,\\
        (-z)^{\rho/2}, & \text{if } 0 < \left| \arg{f_{n}(z)} \right| < \pi/2,
    \end{cases}
    \label{eq:W_def}
\end{equation}
where $f_{n} : U_{\delta_{n}} \to f_{n}(U_{\delta_{n}})$ is a biholomorphic map to be constructed in the next section. Note that $W(z)$ has a branch cut along the whole real axis. One can then verify that $P^{(1)}$ has a jump matrix which is piecewise constant~\cite{kuijlaarsUniversalityEigenvalueCorrelations2003}:
\begin{align}
    P^{(1)}_{+}(x) &= P^{(1)}_{-}(x) \begin{pmatrix} 0 & 1 \\ -1 & 0 \end{pmatrix}, \hspace{-8em} &\text{for } x \in \Sigma_{1}^{o} \cup \Sigma_{5}^{o},\\
    P^{(1)}_{+}(z) &= P^{(1)}_{-}(z) \begin{pmatrix} 1 & 0 \\ e^{- i \pi \rho} & 1 \end{pmatrix}, \hspace{-8em} &\text{for } z \in \Sigma_{2}^{o} \cup \Sigma_{6}^{o},\\
    P^{(1)}_{+}(z) &= P^{(1)}_{-}(z) e^{\frac{i \pi \rho}{2} \sigma_{3}}, \hspace{-8em} &\text{for } z \in \Sigma_{3}^{o} \cup \Sigma_{7}^{o},\\
    P^{(1)}_{z}(x) &= P^{(1)}_{-}(z) \begin{pmatrix} 1 & 0 \\ e^{i \pi \rho} & 1 \end{pmatrix}, \hspace{-8em} &\text{for } z \in \Sigma_{4}^{o} \cup \Sigma_{8}^{o},
\end{align}
where the contours $\Sigma_{j}^{o}$, $j=1,\dots,8$ are sketched in \cref{fig:bessel_parametrix}(a), with the understanding that the notation $\Sigma_{i}^{o}$ denotes $\Sigma_{i}$ without the origin. We construct an explicit solution for this RHP using a model solution $\Psi_{\rho/2}(\zeta)$ for an auxiliary RHP in the $\zeta$-plane, as given in \cite{kuijlaarsUniversalityEigenvalueCorrelations2003,vanlessenStrongAsymptoticsRecurrence2003} and illustrated in \cref{fig:bessel_parametrix}(b).

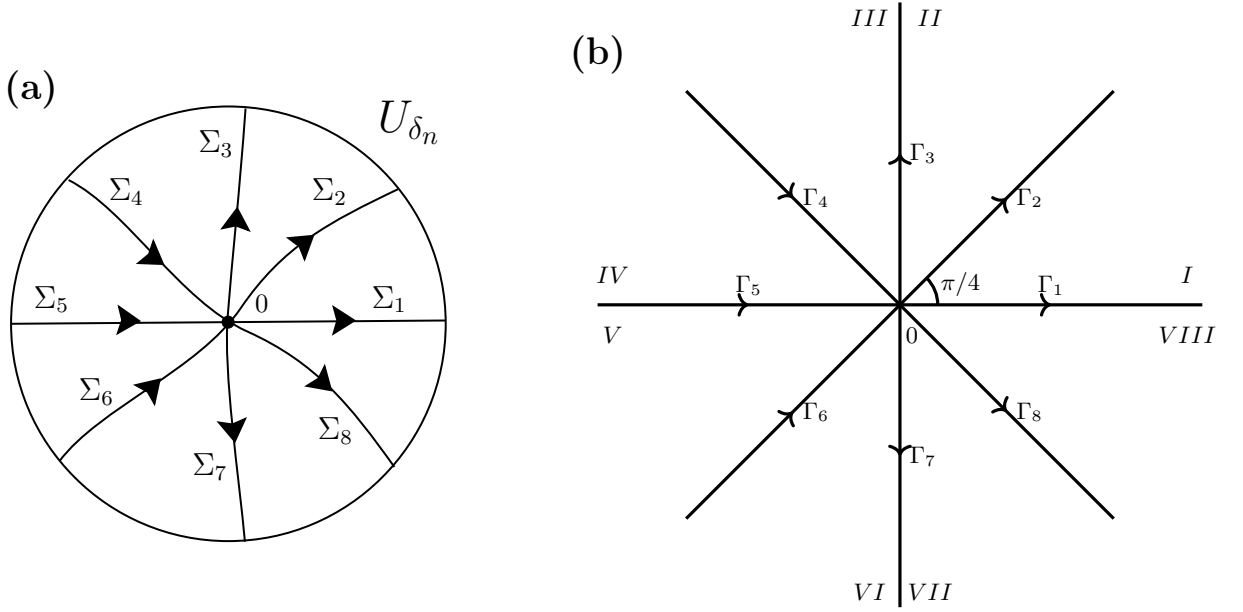
\begin{figure}[t]
\centering
\begin{minipage}{0.49\textwidth}
    \tikzset{every picture/.style={line width=0.75pt}} %

\begin{tikzpicture}[x=0.75pt,y=0.75pt,yscale=-1,xscale=1,scale=1.5]

\draw    (206.57,223.6) .. controls (204.57,201.6) and (199.3,167.3) .. (200.82,148) .. controls (202.33,128.7) and (204.78,106.67) .. (206.78,78.67) ;
\draw    (144.57,196.6) .. controls (157.8,178.91) and (190.99,166.26) .. (205.16,145.48) .. controls (219.33,124.7) and (234.78,116.67) .. (257.78,105.67) ;
\draw    (147.57,102.6) .. controls (167.28,112.57) and (184.45,142.92) .. (207.37,153.28) .. controls (230.29,163.64) and (241.75,177.85) .. (256.57,198.6) ;
\draw    (128.33,150.7) -- (273.57,149.6) ;
\draw  [fill={rgb, 255:red, 0; green, 0; blue, 0 }  ,fill opacity=1 ] (199.1,150.15) .. controls (199.1,149.13) and (199.93,148.3) .. (200.95,148.3) .. controls (201.97,148.3) and (202.8,149.13) .. (202.8,150.15) .. controls (202.8,151.17) and (201.97,152) .. (200.95,152) .. controls (199.93,152) and (199.1,151.17) .. (199.1,150.15) -- cycle ;
\draw  [fill={rgb, 255:red, 0; green, 0; blue, 0 }  ,fill opacity=1 ] (199.06,118.71) -- (203.99,112.12) -- (207.84,119.39) -- (203.55,117.8) -- cycle ;
\draw  [fill={rgb, 255:red, 0; green, 0; blue, 0 }  ,fill opacity=1 ] (206.61,182.55) -- (203.19,190.03) -- (197.89,183.74) -- (202.42,184.38) -- cycle ;
\draw  [fill={rgb, 255:red, 0; green, 0; blue, 0 }  ,fill opacity=1 ] (236.28,145.16) -- (243.19,149.61) -- (236.22,153.97) -- (237.5,149.57) -- cycle ;
\draw  [fill={rgb, 255:red, 0; green, 0; blue, 0 }  ,fill opacity=1 ] (164.46,145.03) -- (171.19,149.75) -- (164.05,153.82) -- (165.5,149.48) -- cycle ;
\draw  [fill={rgb, 255:red, 0; green, 0; blue, 0 }  ,fill opacity=1 ] (221.16,122.04) -- (229.35,121.29) -- (226.53,129.02) -- (224.83,124.77) -- cycle ;
\draw  [fill={rgb, 255:red, 0; green, 0; blue, 0 }  ,fill opacity=1 ] (170.03,171.22) -- (178.04,169.36) -- (176.3,177.4) -- (174.04,173.42) -- cycle ;
\draw  [fill={rgb, 255:red, 0; green, 0; blue, 0 }  ,fill opacity=1 ] (233.23,164.98) -- (235.38,172.92) -- (227.29,171.48) -- (231.18,169.07) -- cycle ;
\draw  [fill={rgb, 255:red, 0; green, 0; blue, 0 }  ,fill opacity=1 ] (175.94,123.13) -- (179.16,130.7) -- (170.94,130.38) -- (174.47,127.47) -- cycle ;
\draw   (128.33,150.7) .. controls (128.33,110.54) and (160.89,77.99) .. (201.04,77.99) .. controls (241.2,77.99) and (273.75,110.54) .. (273.75,150.7) .. controls (273.75,190.86) and (241.2,223.41) .. (201.04,223.41) .. controls (160.89,223.41) and (128.33,190.86) .. (128.33,150.7) -- cycle ;

\draw (209,140) node [anchor=north west][inner sep=0.75pt]  [font=\normalsize]  {$0$};
\draw (248,137) node [anchor=north west][inner sep=0.75pt]  [font=\large]  {$\Sigma _{1}$};
\draw (228,101) node [anchor=north west][inner sep=0.75pt]  [font=\large]  {$\Sigma _{2}$};
\draw (190,85) node [anchor=north west][inner sep=0.75pt]  [font=\large]  {$\Sigma _{3}$};
\draw (160,101) node [anchor=north west][inner sep=0.75pt]  [font=\large]  {$\Sigma _{4}$};
\draw (135,137) node [anchor=north west][inner sep=0.75pt]  [font=\large]  {$\Sigma _{5}$};
\draw (150,168) node [anchor=north west][inner sep=0.75pt]  [font=\large]  {$\Sigma _{6}$};
\draw (188,193) node [anchor=north west][inner sep=0.75pt]  [font=\large]  {$\Sigma _{7}$};
\draw (230,181) node [anchor=north west][inner sep=0.75pt]  [font=\large]  {$\Sigma _{8}$};
\draw (250,75) node [anchor=north west][inner sep=0.75pt]  [font=\LARGE]  {$U_{\delta _{n}}$};

\draw (125,65) node [anchor=north west][inner sep=0.75pt]  [font=\Large]  {\textbf{(a)}};

\end{tikzpicture}
\end{minipage}
\begin{minipage}{0.49\textwidth}
\begin{tikzpicture}[scale=4]
    \begin{scope}[very thick,decoration={markings,mark=at position 0.5 with {\arrow{>}}}] 

        \draw[postaction={decorate}] (-1,0) -- (0,0) node[midway,above] {$\Gamma_{5}$};
        \draw[postaction={decorate}] (0,0) -- (1,0) node[midway,above] {$\Gamma_{1}$};
        
        \draw[postaction={decorate}] (0,0) -- (0,1) node[midway,right] {$\Gamma_{3}$};
        \draw[postaction={decorate}] (0,0) -- (0,-1) node[midway,right] {$\Gamma_{7}$};

        \draw[postaction={decorate}] (0,0) -- (0.707,0.707) node[midway,right] {$\Gamma_{2}$};
        \draw[postaction={decorate}] (0,0) -- (0.707,-0.707) node[midway,right] {$\Gamma_{8}$};
        \draw[postaction={decorate}] (-0.707,0.707) -- (0,0) node[midway,right] {$\Gamma_{4}$};
        \draw[postaction={decorate}] (-0.707,-0.707) -- (0,0) node[midway,right] {$\Gamma_{6}$};
        
        \node[] at (0.04,-0.1) {$0$};

        \node[] at (0.95,0.1) {$I$};
        \node[] at (0.95,-0.1) {$VIII$};
        \node[] at (0.1,0.95) {$II$};
        \node[] at (-0.1,0.95) {$III$};

        \node[] at (-0.95,0.1) {$IV$};
        \node[] at (-0.95,-0.1) {$V$};
        \node[] at (0.1,-0.95) {$VII$};
        \node[] at (-0.1,-0.95) {$VI$};
        
        \coordinate (origin) at (0,0);
        \coordinate (r) at (0.707,0);
        \coordinate (ur) at (0.707,0.707);
        \coordinate (dr) at (0.707,-0.707);
        \coordinate (l) at (-0.707,0);
        \coordinate (ul) at (-0.707,0.707);
        \coordinate (dl) at (-0.707,-0.707);

        \pic [draw, "", angle eccentricity = 1.5] {angle = r--origin--ur};
        \node[] at (0.2,0.07) {$\pi/4$};

    \end{scope}
    \draw (-1.1,0.9) node [anchor=north west][inner sep=0.75pt]  [font=\Large]  {\textbf{(b)}};
\end{tikzpicture}
\end{minipage}
\caption{\textbf{(a)} The contour $\Sigma = \bigcup_{i=1}^{8} \Sigma_{i}$ in the $z$-plane on which $P^{(1)}(z)$ has jumps. We construct the parametrix in a disk $U_{\delta_{n}}$ of radius $\delta_{n}$ centered at the origin. \textbf{(b)} The contour $\Gamma_{\Psi}$ for the auxiliary RHP in the $\zeta$-plane used to obtain the parametrix near the origin.
}
\label{fig:bessel_parametrix}
\end{figure}

\paragraph*{Riemann-Hilbert problem for $\Psi_{\rho/2}$}
\begin{itemize}
    \item[$(\Psi_{\rho/2,a})$] $\Psi_{\rho/2}$ is analytic in $\mathbb{C} \setminus \Gamma_{\Psi}$.
    \item[$(\Psi_{\rho/2,b})$] $\Psi_{\rho/2}$ satisfies the following jump relations on $\Gamma_{\Psi}$:
        \begin{align}
            \Psi_{\rho/2,+}(\zeta) &= \Psi_{\rho/2,-}(\zeta) \begin{pmatrix} 0 & 1 \\ -1 & 0 \end{pmatrix}, \hspace{-8em} &\text{for } \zeta \in \Gamma_{1} \cup \Gamma_{5},\\
            \Psi_{\rho/2,+}(\zeta) &= \Psi_{\rho/2,-}(\zeta) \begin{pmatrix} 1 & 0 \\ e^{-i \pi \rho} & 1 \end{pmatrix},\hspace{-8em} &\text{for } \zeta \in \Gamma_{2} \cup \Gamma_{6},\\
            \Psi_{\rho/2,+}(\zeta) &= \Psi_{\rho/2,-}(\zeta) e^{\frac{i \pi \rho}{2}\sigma_{3}}, \hspace{-8em} &\text{for } \zeta \in \Gamma_{3} \cup \Gamma_{7},\\
            \Psi_{\rho/2,+}(\zeta) &= \Psi_{\rho/2,-}(\zeta) \begin{pmatrix} 1 & 0 \\ e^{i \pi \rho} & 1 \end{pmatrix},\hspace{-8em} &\text{for } \zeta \in \Gamma_{4} \cup \Gamma_{8}.
        \end{align} 
    \item[$(\Psi_{\rho/2,c})$] For $\rho<0$, $\Psi_{\rho/2}$ has the following behavior as $\zeta \to 0$:
        \begin{equation}
            \Psi_{\rho/2}(\zeta) = \mathcal{O}\begin{pmatrix}
                |\zeta|^{\rho/2} & |\zeta|^{\rho/2}\\
                |\zeta|^{\rho/2} & |\zeta|^{\rho/2}
            \end{pmatrix}, \quad \text{as } \zeta \to 0.
        \end{equation}
        For $\rho \geq 0$, $\Psi_{\rho/2}$ has the following behavior as $\zeta \to 0$:
        \begin{equation}
            \Psi_{\rho/2}(\zeta) = \begin{dcases}
                \mathcal{O}\begin{pmatrix} 
                |\zeta|^{\rho/2} & |\zeta|^{-\rho/2}\\
                |\zeta|^{\rho/2} & |\zeta|^{-\rho/2}
                \end{pmatrix}, & \text{as } \zeta \to 0 \text{ with } \zeta \in \text{II, III, VI, VII};\\
                \mathcal{O}\begin{pmatrix} 
                |\zeta|^{-\rho/2} & |\zeta|^{-\rho/2}\\
                |\zeta|^{-\rho/2} & |\zeta|^{-\rho/2}
                \end{pmatrix}, & \text{as } \zeta \to 0 \text{ with } \zeta \in \text{I, IV, V, VIII}.
            \end{dcases}
        \end{equation}
\end{itemize}
This RHP was solved in \cite[Eqs.~(4.26)-(4.33)]{vanlessenStrongAsymptoticsRecurrence2003}. It is built out of the modified Bessel functions $I_{\frac{1}{2}(\rho \pm 1)}$, $K_{\frac{1}{2}(\rho \pm 1)}$ and the modified Hankel functions $H^{(1)}_{\frac{1}{2}(\rho\pm 1)}$, $H^{(2)}_{\frac{1}{2}(\rho\pm 1)}$. The explicit formula for $\Psi_{\rho/2}$ in sector I is 
\begin{equation}
    \Psi_{\rho/2}(\zeta) = \dfrac{\sqrt{\pi}}{2} \zeta^{1/2} \begin{pmatrix}
        H^{(2)}_{\frac{1}{2}(\rho+1)}(\zeta) & -i H^{(1)}_{\frac{1}{2}(\rho+1)}(\zeta)\\[1em]
        H^{(2)}_{\frac{1}{2}(\rho-1)}(\zeta) & -i H^{(1)}_{\frac{1}{2}(\rho-1)}(\zeta)
    \end{pmatrix} e^{-(\frac{\rho}{2} + \frac{1}{4})i \pi \sigma_{3}}, \quad \text{for } 0 < \arg{\zeta} < \dfrac{\pi}{4},
    \label{eq:Psi_rho_def}
\end{equation}
and the expressions in the other sectors can be obtained from this one by following the jumps given in $(\Psi_{\rho/2,b})$. As $\zeta \to \infty$ in the first quadrant, we can use the asymptotic formulae \cite[9.2.7--9.2.10]{abramowitzHandbookMathematicalFunctions1965} for the Hankel functions to get the asymptotic expansion
\begin{equation}
    \Psi_{\rho/2}(\zeta) \sim \dfrac{1}{\sqrt{2}} \sum_{k=0}^{\infty} \dfrac{i^{k}}{(2 \zeta)^{k}} \begin{pmatrix} (-1)^{k} \left(\frac{1}{2}(\rho+1), k\right) & -i \left(\frac{1}{2}(\rho+1), k\right) \\[0.5em] -i (-1)^{k} \left(\frac{1}{2}(\rho-1), k\right) & \left(\frac{1}{2}(\rho-1), k\right) \end{pmatrix} e^{\frac{i \pi}{4} \sigma_{3}} e^{-\frac{i \pi \rho}{4} \sigma_{3}} e^{-i \zeta \sigma_{3}},
    \label{eq:Psi_rho_asymptotic}
\end{equation}
uniformly in $\zeta$, where for $k=0$ we define $(v,0) = 1$ and for $k \geq 1$
\begin{equation}
    (v,k) \coloneqq \dfrac{(4v^{2}-1)(4v^{2}-9) \cdots (4v^{2} - (2k-1)^{2})}{2^{2k} k!}.
    \label{eq:vk_def}
\end{equation}
See Ref.~\cite{vanlessenStrongAsymptoticsRecurrence2003} for expressions in the other quadrants.

We then define
\begin{equation}
    P^{(1)}(z) \coloneqq \Psi_{\rho/2}(n f_{n}(z)), \quad \text{for } z \in U_{\delta_{n}} \setminus f_{n}^{-1}(\Gamma),
    \label{eq:P1_def_origin}
\end{equation}
where $f_{n}(z)$ is a biholomorphic map to be constructed in the next section. 

\subsubsection{Construction of biholomorphic map}
Following Ref.~\cite{kuijlaarsUniversalityEigenvalueCorrelations2003}, we define the map $f_{n} : U_{\delta_{n}} \to f_{n}(U_{\delta_{n}})$ by
\begin{equation}
    f_{n}(z) \coloneqq \begin{cases}
        i \phi_{n}(z) - i\phi_{n,+}(0), & \text{if } \Im{z} > 0,\\
        -i \phi_{n}(z) - i\phi_{n,+}(0), & \text{if } \Im{z} < 0,
    \end{cases}
    \label{eq:fn_zero_def}
\end{equation}
where $\phi_{n,+}(0) = i \pi/2$ for an even weight function. From \cref{eq:psi_cont_jump,eq:phi_n_def} we know that $\phi_{n}(z)$ is analytic in $\mathbb{C}_{\pm}$ (within its domain) and flips sign across $(-1,1)$, which implies that $f_{n}(z)$ is analytic in a neighborhood of $z = 0$. On the real line we have
\begin{equation}
    f_{n}(x) = \pi \int_{0}^{x} \psi_{n}(s) \diff s, \quad \text{for } x \in (-\delta_{n}, \delta_{n}),
\end{equation}
which implies that $f_{n}^{\prime}(0) = \pi \psi_{n}(0) > 0$. The behavior of $f_{n}(z)$ near the origin is then given by
\begin{equation}
    f_{n}(z) = \pi \psi_{n}(0) z + \mathcal{O}(z^{3}), \quad \text{as } z \to 0,
\end{equation}
where $f_{n}^{\prime\prime}(0) = 0$ follows from the fact that $\psi_{n}(x)$ is even. The big-$\mathcal{O}$ hides potential $n$-dependence, which we will need to be more careful about in order to guarantee that $f_{n}(z)$ is actually biholomorphic in a neighborhood of the origin.

\begin{proposition}
    There exists a $\gamma > 0$ such that for $\delta_{n} = \gamma / \beta_{n}$ and sufficiently large $n$, $f_{n} : U_{\delta_{n}} \to f_{n}(U_{\delta_{n}})$ satisfies the following properties.
    \begin{enumerate}[label=(\alph*)]
        \item $f_{n} : U_{\delta_{n}} \to f_{n}(U_{\delta_{n}})$ is biholomorphic.
        \item $f_{n}(U_{\delta_{n}} \cap \mathbb{R}) = f_{n}(U_{\delta_{n}}) \cap \mathbb{R}$, $f_{n}(U_{\delta_{n}} \cap \mathbb{C}_{\pm}) = f_{n}(U_{\delta_{n}}) \cap \mathbb{C}_{\pm}$.
    \end{enumerate}
    \label{lem:fn_origin}
\end{proposition}
\begin{proof}
    Property (b) will follow as in \cref{prop:fn_endpoint}. Regarding property (a), we have already established that $f_{n}$ is analytic in a neighborhood of $z = 0$, so it remains to show that it is invertible. To that end, for $|z|\leq r_{0}$ we use Cauchy's theorem for derivatives on a circle of radius $r_{1} > r_{0}$ to give
    \begin{align}
        \left|f_{n}^{\prime\prime}(z)\right| &= \left|\dfrac{2}{2\pi i} \oint_{|u|=r_{1}} \dfrac{f_{n}(u)}{(u-z)^{3}} \diff u \right|,\\
        &\leq \dfrac{2 r_{1}}{(r_{1} - r_{0})^{3}} \sup_{|u|=r_{1}} |f_{n}(u)|.\label{eq:f2_deriv_bound}
    \end{align}
    We are implicitly taking $r_{0}$ and $r_{1}$ to scale like $\mathcal{O}(1/\beta_{n})$, but we will not explicitly denote their $n$ dependence. Now, from \cref{eq:fn_zero_def,eq:phi_n_def,eq:psi_hat_def}, we see that
    \begin{equation}
        \sup_{|u|=r_{1}} |f_{n}(u)| \leq \dfrac{r_{1}}{2} \sup_{|u|=r_{1}} \left| r(u) h_{n}(u) \right|.
    \end{equation}
    We show in \cref{lem:complex_hnz_uniform_lower} that, for sufficiently large $n$, we have $h_{n}(z) = h_{n}(0)[1 + o(1)]$ uniformly for $z$ in a sufficiently small disk of radius $\mathcal{O}(1/\beta_{n})$ centered at the origin, where the $o(1)$ refers to scaling with $n$. Combined with the fact that $r(u) = (u+1)^{1/2}(u-1)^{1/2}$ is bounded near zero, for sufficiently large $n$ we have
    \begin{equation}
        \sup_{|u|=r_{1}} |f_{n}(u)| \leq C r_{1} h_{n}(0).
    \end{equation}
    for some constant $C > 1/2$ (independent of $n$). Inserting this into \cref{eq:f2_deriv_bound} gives
    \begin{equation}
        \left|f_{n}^{\prime\prime}(z)\right| \leq 2 C \dfrac{r_{1}^{2}}{(r_{1} - r_{0})^{3}} h_{n}(0), \quad \text{for } |z| \leq r_{0} < r_{1}.
    \end{equation}
    Again for $|z| \leq r_{0}$, we then have
    \begin{equation}
        \left|f_{n}^{\prime}(z) - f_{n}^{\prime}(0)\right| = \left| \int_{0}^{z} f_{n}^{\prime\prime}(s) \diff s \right| \leq 2 C \dfrac{r_{0} r_{1}^{2}}{(r_{1} - r_{0})^{3}} h_{n}(0).
    \end{equation}
    We then set $r_{0} = r_{1} / N$ for some large $N > 1$, and use $f_{n}^{\prime}(0) = \pi \psi_{n}(0) = h_{n}(0) / 2$, so that
    \begin{equation}
        \left|f_{n}^{\prime}(z) - f_{n}^{\prime}(0)\right| \leq 4 C \dfrac{N^{2}}{(N-1)^{3}} f_{n}^{\prime}(0).
    \end{equation}
    By taking $N$ sufficiently large (but independent of $n$), we can make the RHS less than, say, $f_{n}^{\prime}(0) / 2$, so we have a uniform (in $z$) nonzero lower bound on $|f_{n}^{\prime}(z)|$ in the disk of radius $r_{0}$. We can also make this uniform in $n$ by using \cref{lem:hn0_scaling_sup} to establish that $f_{n}^{\prime}(0) = h_{n}(0)/2$ is $\Omega(1)$ as $n \to \infty$. This is enough to conclude that $f_{n}(z)$ is invertible and hence biholomorphic in a sufficiently small $\mathcal{O}(1/\beta_{n})$ neighborhood of the origin.
\end{proof}

\subsubsection{Satisfying the matching condition}
In this final step, we construct the invertible analytic matrix-valued function $E(z)$ to ensure that the matching function $(P_{c})$ is satisfied. Following Ref.~\cite{kuijlaarsUniversalityEigenvalueCorrelations2003}, we define $E_{n}(z)$ by
\begin{equation}
    E_{n}(z) \coloneqq E(z) e^{n \phi_{n,+}(0) \sigma_{3}} e^{-\frac{\pi i}{4} \sigma_{3}} \dfrac{1}{\sqrt{2}} \begin{pmatrix} 1 & i \\ i & 1 \end{pmatrix},
    \label{eq:En_def_origin}
\end{equation}
where the matrix valued function $E$ is defined in different regions of the complex plane by
\begin{equation}
    E(z) \coloneqq \begin{dcases} N(z) W(z)^{\sigma_{3}} e^{\frac{1}{4} \rho \pi i \sigma_{3}}, & \text{for } z \in f_{n}^{-1}(\text{I} \cup \text{II}),\\
    N(z) W(z)^{\sigma_{3}} e^{-\frac{1}{4} \rho \pi i \sigma_{3}}, & \text{for } z \in f_{n}^{-1}(\text{III} \cup \text{IV}),\\
    N(z) W(z)^{\sigma_{3}} \begin{pmatrix} 0 & 1 \\ -1 & 0 \end{pmatrix} e^{-\frac{1}{4} \rho \pi i \sigma_{3}}, & \text{for } z \in f_{n}^{-1}(\text{V} \cup \text{VI}),\\
    N(z) W(z)^{\sigma_{3}} \begin{pmatrix} 0 & 1 \\ -1 & 0 \end{pmatrix} e^{\frac{1}{4} \rho \pi i \sigma_{3}}, & \text{for } z \in f_{n}^{-1}(\text{VII} \cup \text{VIII}).
    \end{dcases}
\end{equation}
This ends the construction of the parametrix $P$ near $z = 0$.

\subsubsection{Summary of the local analysis near $z=0$}
We have defined the contour $\Sigma$ in a shrinking disk $U_{\delta_{n}}$ near $z=0$ as the inverse image of the contour $\Gamma_{\Psi}$ under the map $f_{n}$ defined in \cref{eq:fn_zero_def}. We have then found a solution
\begin{equation}
    P(z) = E_{n}(z) \Psi_{\rho/2}(n f_{n}(z)) W(z)^{-\sigma_{3}} e^{-n \phi_{n}(z) \sigma_{3}}, \quad \text{for } z \in \Sigma \cap U_{\delta_{n}}
    \label{eq:P_def_origin}
\end{equation}
to the RHP for $P$, where the matrix-valued function $E_{n}$ is defined in \cref{eq:En_def_origin}, the matrix-valued function $\Psi_{\rho/2}$ is defined in \cref{eq:Psi_rho_def}, and $W(z)$ is defined in \cref{eq:W_def}. Furthermore, following Ref.~\cite{vanlessenStrongAsymptoticsRecurrence2003}, on the boundary $\partial U_{\delta_{n}}$, we have the asymptotic expansion
\begin{equation}
    P(z) N(z)^{-1} \sim \mathds{1} + \sum_{k=1}^{\infty} \dfrac{\Delta(k,n)(z)}{[n f_{n}(z)]^{k}}, \quad \text{as } n \to \infty, \quad \text{for } z \in \partial U_{\delta_{n}},
    \label{eq:DeltaR_expansion_zero}
\end{equation}
where the coefficient matrix $\Delta(k,n)(z)$ is determined by \cref{eq:P_def_origin,eq:P1_def_origin,eq:Psi_rho_asymptotic,eq:fn_zero_def,eq:En_def_origin}, and given explicitly by
\begin{equation}
    \Delta(k,n)(z) = \dfrac{i^{k}}{2^{k+1}} E(z) e^{n \phi_{n,+}(0) \sigma_{3}} \begin{pmatrix} (-1)^{k} s_{\rho/2, k} & -t_{\rho/2,k} \\ -(-1)^{k} t_{\rho/2,k} & s_{\rho/2,k} \end{pmatrix} e^{-n \phi_{n,+}(0) \sigma_{3}} E(z)^{-1},
\end{equation}
and the constants $s_{\rho/2,k}$ and $t_{\rho/2,k}$ are defined as
\begin{equation}
    s_{\rho/2,k} \coloneqq \left(\dfrac{1}{2}(\rho+1), k\right) + \left(\dfrac{1}{2}(\rho-1),k\right), \quad \text{and} \quad t_{\rho/2,k} \coloneqq \left(\dfrac{1}{2}(\rho+1), k\right) - \left(\dfrac{1}{2}(\rho-1),k\right),
\end{equation}
where the expression $(v,k)$ was defined in \cref{eq:vk_def}. For later, we note that
\begin{equation}
    s_{\rho/2, 1} = \dfrac{\rho^{2}}{2} \quad \text{and} \quad t_{\rho/2, 1} = \rho.
\end{equation}

We will soon discuss this in more detail, but let us note two things about $\Delta(k,n)(z)$: first, that it is analytic in a neighbourhood of $z = 0$, and second, that its only $n$ dependence comes from the factors of $\exp[\pm n \phi_{n,+}(0)] = \exp[\pm n \pi i/2]$, where the value of $\phi_{n,+}(0)$ is fixed by the even symmetry of the weight function. Since this factor is $\mathcal{O}(1)$, the asymptotic $n$-decay of $P(z) N(z)^{-1} - \mathds{1}$ will be determined solely by the factors of $[n f_{n}(z)]^{k}$ in \cref{eq:DeltaR_expansion_zero}.

Finally, in the case $\rho=0$, this whole construction reduces to $P(z) = N(z)$ exactly. This will mean that $R(z)$ has no jump near the origin if $\rho=0$.

\subsection{Riemann-Hilbert problem for $R$}
Having constructed local parametrices $P(z)$ in disks centered at $z=0$ and $z=\pm 1$, we define
\begin{equation}
    S_{\mathrm{par}}(z) \coloneqq \begin{cases}
        P(z), & \text{if } |z-1| < \delta_{2} \text{ or } |z+1| < \delta_{2} \text{ or } |z| < \delta_{n},\\
        N(z), & \text{otherwise}.
    \end{cases}
    \label{eq:Spar_def}
\end{equation}
Now we define
\begin{equation}
    R(z) \coloneqq S(z) S_{\mathrm{par}}^{-1}(z).
    \label{eq:R_def}
\end{equation}
One can verify that $R(z)$ is the unique solution of the following RHP, with jumps on the contour $\wt{\Sigma}$ sketched in \cref{fig:R_jump_contour}.
\begin{figure}[t]
\centering
\begin{tikzpicture}[scale=4]
    \begin{scope}[very thick,decoration={markings,mark=at position 0.5 with {\arrow{>}}}] 

        \draw[postaction={decorate}] (-1.5,0) -- (-1.25,0) node[midway,above] {$\wt{\Sigma}_{8}$};
        \draw[postaction={decorate}] (1.25,0) -- (1.5,0) node[midway,above] {$\wt{\Sigma}_{9}$};
        
        \draw[postaction={decorate}] (-0.85,0.2) to[out=0,in=165] (-0.106,0.106);
        \draw[postaction={decorate}] (-0.85,-0.2) to[out=0,in=-165] (-0.106,-0.106);

        \draw[postaction={decorate}] (0.106,0.106) to[out=15,in=180] (0.85,0.2);
        \draw[postaction={decorate}] (0.106,-0.106) to[out=-15,in=-180] (0.85,-0.2);

        \node[] at (0.5, 0.3) {$\wt{\Sigma}_{1}$};
        \node[] at (0.5, -0.075) {$\wt{\Sigma}_{2}$};
        \node[] at (-0.5, 0.3) {$\wt{\Sigma}_{3}$};
        \node[] at (-0.5, -0.075) {$\wt{\Sigma}_{4}$};

        \node[] at (0, 0.25) {$\wt{\Sigma}_{5}$};
        \node[] at (1, 0.35) {$\wt{\Sigma}_{6}$};
        \node[] at (-1, 0.35) {$\wt{\Sigma}_{7}$};
        
        \node[] at (-1,0) {$-1$};
        \node[] at (1,0) {$1$};
        \node[] at (0,0) {$0$};
    \end{scope}
    \draw[very thick,decoration={markings, mark=at position 0.25 with {\arrow{<}}},postaction={decorate}] (-1,0) circle [radius=0.25];
    \draw[very thick,decoration={markings, mark=at position 0.25 with {\arrow{<}}},postaction={decorate}] (1,0) circle [radius=0.25];
    \draw[very thick,decoration={markings, mark=at position 0.25 with {\arrow{<}}},postaction={decorate}] (0,0) circle [radius=0.15];
\end{tikzpicture}
\caption{The contour $\wt{\Sigma}_{R} = \bigcup_{j=1}^{9} \wt{\Sigma}_{j}$ on which $R(z)$ has jumps. The circles $\wt{\Sigma}_{j=5,6,7}$ are determined in the construction of the local parametrices in \cref{sec:endpoints,sec:origin}; note that $\wt{\Sigma}_{5}$ has radius $\mathcal{O}(1/\beta_{n})$, so as $n\to\infty$ it is much smaller than $\wt{\Sigma}_{6,7}$ which have $\mathcal{O}(1)$ radius. The precise shape of the lens boundaries $\wt{\Sigma}_{j=1..4}$ is given in \cref{fig:lens_height_function}.}
\label{fig:R_jump_contour}
\end{figure}
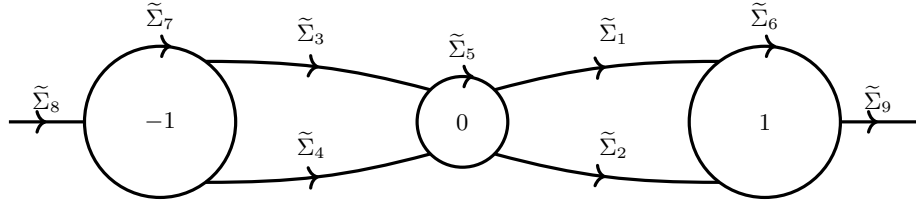
\paragraph*{Riemann-Hilbert problem for $R$}
\begin{itemize}
    \item[$(R_{a})$] $R : \mathbb{C} \setminus \wt{\Sigma}_{R} \to \mathbb{C}^{2 \times 2}$ is analytic. 
    \item[$(R_{b})$] $R$ has the following jump matrices.
        \begin{equation}
            R_{+}(s) = R_{-}(s) \begin{cases}
                N(s) v_{i}(s) N^{-1}(s), & \text{for } s \in \wt{\Sigma}_{i}, i = 1,2,3,4,8,9;\\
                P(s) N^{-1}(s), & \text{for } s \in \wt{\Sigma}_{5} \cup \wt{\Sigma}_{6} \cup \wt{\Sigma}_{7};\\
                \mathds{1}, & \text{otherwise.}
            \end{cases}
        \end{equation}
        Here $v_{i}(s)$ is the jump matrix for $S$ on the relevant contour (see \cref{eq:Sjump1,eq:Sjump2,eq:Sjump3}).
    \item[$(R_{c})$] $R(z) = \mathds{1} + \mathcal{O}(1/|z|)$ as $|z| \to \infty$.
\end{itemize}
In particular, the jump matrix $v_{R}(z)$ of $R$ is uniformly close to the identity, in the sense that, with $\Delta_{R}(z) \coloneqq v_{R}(z) - \mathds{1}$, as $n \to \infty$ we have $\norm{\Delta_{R}(z)}_{L_{\infty}(\wt{\Sigma})} \to 0$. In more detail, we have
\begin{equation}
    \begin{gathered}
    \norm{\Delta_{R}}_{L_{\infty}\left(\wt{\Sigma}_{j=1..4}\right)} = o\left(\dfrac{1}{\mathrm{poly}(\log{n})}\right),
    \quad
    \norm{\Delta_{R}}_{L_{2}\left(\wt{\Sigma}_{j=1..4}\right)} = o\left(\dfrac{1}{n^{1/2} \mathrm{poly}(\log{n})}\right),
    \quad
    \norm{\Delta_{R}}_{L_{1}\left(\wt{\Sigma}_{j=1..4}\right)} = o\left(\dfrac{1}{n \, \mathrm{poly}(\log{n})}\right),\\
    \begin{aligned}[c]
        &\norm{\Delta_{R}}_{L_{\infty}\left(\wt{\Sigma}_{5}\right)} = \mathcal{O}\left(\dfrac{\beta_{n}}{n h_{n}(0)}\right),\\
        &\norm{\Delta_{R}}_{L_{\infty}\left(\wt{\Sigma}_{6}\right)} = \mathcal{O}\left(\dfrac{1}{n}\right),\\
        &\norm{\Delta_{R}}_{L_{\infty}\left(\wt{\Sigma}_{7}\right)} = \mathcal{O}\left(\dfrac{1}{n}\right),\\
        &\norm{\Delta_{R}}_{L_{\infty}\left(\wt{\Sigma}_{8,9}\right)} = \mathcal{O}\left(e^{-c n}\right),
    \end{aligned}
    \qquad
    \begin{aligned}[c]
        &\norm{\Delta_{R}}_{L_{2}\left(\wt{\Sigma}_{5}\right)} = \mathcal{O}\left(\dfrac{\beta_{n}^{1/2}}{n h_{n}(0)}\right),\\
        &\norm{\Delta_{R}}_{L_{2}\left(\wt{\Sigma}_{6}\right)} = \mathcal{O}\left(\dfrac{1}{n}\right),\\
        &\norm{\Delta_{R}}_{L_{2}\left(\wt{\Sigma}_{7}\right)} = \mathcal{O}\left(\dfrac{1}{n}\right),\\
        &\norm{\Delta_{R}}_{L_{2}\left(\wt{\Sigma}_{8,9}\right)} = \mathcal{O}\left(e^{-c n}\right),
    \end{aligned}
    \end{gathered}
    \label{eq:DeltaR_bounds}
\end{equation}
where $c>0$ is a positive $\mathcal{O}(1)$ constant. When $\beta_{n}$ is sublinear in $n$, the bounds on $\wt{\Sigma}_{5}$ are polynomially small in $n$. In the marginal case $p=1$ where $\beta_{n} \sim \mathcal{O}(n / \log^{q}{n})$, the $L_{\infty}(\wt{\Sigma}_{5})$ bound decays polylogarithmically in $n$ provided $q > -1$, given the logarithmic divergence $h_{n}(0) = (\log{n})^{1 + o(1)}$ of the equilibrium measure proved in \cref{lem:hn0_scaling_sup}. The precise rate of convergence for $\wt{\Sigma}_{j=1..4}$ depends on $p$ and $q$; for $p>1$ the convergence is actually superpolynomial, while for $p=1$ the convergence may be superpolynomial or subpolynomial (but still super-polylogarithmic) depending on whether $q > 0$ or $q \leq 0$ (see the proof of \cref{lem:lens_boundaries} for details).

The bounds on $\wt{\Sigma}_{6,7,8,9}$ are standard; see \cite{deiftStrongAsymptoticsOrthogonal1999}. The $L_{\infty}$ bound on $\wt{\Sigma}_{5}$ follows from the fact that $\wt{\Sigma}_{5}$ has a radius of $\mathcal{O}(1/\beta_{n})$, combined with the asymptotic expansion \cref{eq:DeltaR_expansion_zero}, the definition \cref{eq:fn_zero_def} of $f_{n}(z)$, and the uniform lower bound for $h_{n}(z)$ given in \cref{lem:complex_hnz_uniform_lower}. The $L_{2}(\wt{\Sigma}_{5})$ bound follows from the $L_{\infty}(\wt{\Sigma}_{5})$ bound by Hölder's inequality. The $L_{\infty}$ bound on $\wt{\Sigma}_{j=1..4}$ follows from the fact that the matrix elements of $\Delta_{R}(z)$ are proportional to $e^{-2n \phi_{n}(z)} D(z)^{2}\omega(z)^{-1}$ on the lens boundaries, and $D(z)^{2}\omega(z)^{-1} = \mathcal{O}(|z|^{0})$ as $z\to 0$, while $|e^{-2n \phi_{n}(z)}| \xrightarrow{n\to\infty} 0$ super-polylogarithmically in $n$ by \cref{lem:lens_boundaries}. The $L_{1,2}(\wt{\Sigma}_{j=1..4})$ bounds follows from the second part of \cref{lem:lens_boundaries}.

\subsubsection{Solution for $R$}
We have constructed $R(z)$ such that its jumps are `small', in the sense that its jump matrix $v_{R}(z)$ is uniformly close to the identity as $n \to \infty$. This will allow us to write down its solution using standard techniques (see \cite[Appendix A]{deiftStrongAsymptoticsOrthogonal1999}). For $f \in L_{2}(\wt{\Sigma}_{R})$, we can define the Cauchy-Stieltjes transform
\begin{equation}
    (Cf)(z) \coloneqq \dfrac{1}{2\pi i} \int_{\wt{\Sigma}_{R}} \dfrac{f(s)}{s-z} \diff s, \quad z \in \mathbb{C} \setminus \wt{\Sigma}_{R}.
\end{equation}
We denote by $(C_{\pm}f)(z)$ the limiting functions as $z$ approaches $\wt{\Sigma}_{R}$ from the $\pm$ sides. Now with $\Delta_{R}(z) \coloneqq v_{R}(z) - \mathds{1}$, we define for matrix-valued functions $f : \wt{\Sigma}_{R} \to \mathbb{C}^{2 \times 2}$ the weighted Cauchy-Stieltjes transform
\begin{equation}
    C_{\Delta_{R}}f \coloneqq C_{-}(f \Delta_{R}).
\end{equation}

Now, by the discussion in the previous section, for $n \to \infty$ we have $\norm{\Delta_{R}}_{L_{\infty}} \to 0$, which by the boundedness of the Cauchy-Stieltjes transform implies that the $L_{2}$-operator norm $\norm{C_{\Delta_{R}}}_{L_{2}} \to 0$ as $n \to \infty$. Then for sufficiently large $n$, the operator $\mathrm{Id} - C_{\Delta_{R}}$ can be inverted by a Neumann series, and so we can employ the standard solution~\cite{deiftStrongAsymptoticsOrthogonal1999}
\begin{equation}
    R(z) = \mathds{1} + C\left(\Delta_{R} + \mu_{R} \Delta_{R}\right)(z), \quad \text{where } \mu_{R} \coloneqq \left(\mathrm{Id} - C_{\Delta_{R}}\right)^{-1} C_{-} \Delta_{R},
    \label{eq:R_exact_sol}
\end{equation}
Expanding $R(z) = \mathds{1} + R_{1}/z + \cdots$, this formula gives
\begin{equation}
    R_{1} = -\dfrac{1}{2\pi i} \int_{\wt{\Sigma}_{R}} \left[\Delta_{R}(y) + \mu_{R}(y) \Delta_{R}(y)\right] \diff y.
\end{equation}
From the definition of $\mu_{R}$, we have $\norm{\mu_{R}}_{L_{2}} \leq c \norm{\Delta_{R}}_{L_{2}}$ for some $c > 0$ by the boundedness of the Cauchy-Stieltjes transform, and hence by Cauchy-Schwarz we have
\begin{equation}
    R_{1} = -\dfrac{1}{2\pi i} \int_{\wt{\Sigma}_{R}} \Delta_{R}(y) \diff y + \mathcal{O}\left(\norm{\Delta_{R}}_{L_{2}(\wt{\Sigma}_{R})}^{2}\right).
    \label{eq:R1_integral}
\end{equation}

\subsubsection{Error bounds for $R$}
\label{sec:R_error}
By construction, we have $R(z) \to \mathds{1}$ as $n \to \infty$. In this section we give the rate of this convergence, which follows from combining \cref{eq:R_exact_sol} with \cref{eq:DeltaR_bounds}. In order to deal with points near the contour $\wt{\Sigma}$, one can employ a contour deformation argument as in the proof of \cite[Corollary 7.9]{deiftStrongAsymptoticsOrthogonal1999}. The sitation is most delicate for $z$ near the origin. For $|z| \leq \mathcal{O}(1/\beta_{n})$, we have the bound
\begin{equation}
    \norm{R(z) - \mathds{1}} \leq \mathcal{O}\left(\dfrac{\beta_{n}}{n h_{n}(0)}\right), \mathclap{\hspace{15em}\text{uniformly for } |z| \leq \mathcal{O}(1/\beta_{n}),}
    \label{eq:R_error_bound}
\end{equation}
with the dominant contribution coming from integrating over the contour $\wt{\Sigma}_{5}$. Again, this may decay only polylogarithmically in the worst case $p=1$, $q>-1$. For larger $|z|$, similar arguments using \cref{eq:DeltaR_bounds} lead to
\begin{equation}
    \norm{R(z) - \mathds{1}} \leq \mathcal{O}\left(\dfrac{1}{\beta_{n}|z|} \dfrac{\beta_{n}}{n h_{n}(0)}\right) + \mathcal{O}\left(\dfrac{1}{n}\right), \mathclap{\hspace{10em}\text{for } |z| \geq \Omega(1/\beta_{n}),}
    \label{eq:R_error_bound_bulk}
\end{equation}
where the factor of $\mathcal{O}(1/\beta_{n}|z|)$ gets cut off at $\mathcal{O}(1)$ for $|z| \sim \mathcal{O}(1/\beta_{n})$, and the factor of $\mathcal{O}(1/n)$ comes from the contours $\wt{\Sigma}_{6,7}$. The proof of \cref{lem:lens_boundaries} shows that the contribution from the lens boundaries $\wt{\Sigma}_{1..4}$ is $o(1/n \, \mathrm{poly}(\log{n}))$ and so is subleading. 

One can similarly establish bounds on the derivatives: near the origin we have
\begin{equation}
    \norm{R^{(k)}(z)} \leq \mathcal{O}\left(\beta_{n}^{k}\dfrac{\beta_{n}}{n h_{n}(0)}\right), \mathclap{\hspace{15em}\text{uniformly for } |z| \leq \mathcal{O}(1/\beta_{n}),}
    \label{eq:R_deriv_error_bound}
\end{equation}
and for larger $z$ this gets replaced with
\begin{equation}
    \norm{R^{(k)}(z)} \leq \mathcal{O}\left(\dfrac{1}{|z|^{k}}\dfrac{1}{\beta_{n}|z|} \dfrac{\beta_{n}}{n h_{n}(0)}\right) + \mathcal{O}\left(\dfrac{1}{n}\right), \mathclap{\hspace{10em}\text{for } |z| \geq \Omega(1/\beta_{n}).}
    \label{eq:R_deriv_error_bound_bulk}
\end{equation}
\begin{remark}
\label{rem:R_deriv}
Note that these bounds imply that the derivatives of $R(z)$ could in principle grow with $n$, in contrast to the more standard situation where all derivatives of $R(z)$ are bounded~\cite{deiftStrongAsymptoticsOrthogonal1999,kuijlaarsUniversalityEigenvalueCorrelations2003}. This is a consequence of shrinking the contour $\wt{\Sigma}_{5}$ like $\mathcal{O}(1/\beta_{n})$, which was necessitated by our fairly weak analyticity assumptions on $Q$. One might worry that these growing derivatives could dominate asymptotics for the Christoffel-Darboux (CD) kernel, which involves taking derivatives of $Y(z)$ (see \cref{sec:polynomial_asymptotics}). However, the contributions from $R$ derivatives turn out to still be subleading relative to the derivatives of leading terms, since those derivatives themselves bring factors which grow with $n$. To illustrate this, a comparison one has to make for CD asymptotics in the bulk is between $\norm{R^{\prime}(x)}$ and $n \phi_{n,+}^{\prime}(x)$, the latter coming from differentiating $e^{-n \phi_{n,+}(x) \sigma_{3}}$. From \cref{eq:phi_n_origin} we have 
\begin{equation}
    n \phi_{n,+}^{\prime}(x) = -i \pi \, n \psi_{n}(x) = \dfrac{-i\sqrt{1-x^{2}}}{2} n h_{n}(x).
\end{equation}
In the bulk $\sqrt{1-x^{2}}$ is bounded away from 0, so one can then use the uniform lower bounds on $h_{n}(x)$ available from \cref{lem:hn_positive_x_g1,lem:complex_hnz_uniform_lower,rem:hnz_lower_bound} to see that $n \phi_{n,+}^{\prime}(x)$ dominates over $\norm{R^{\prime}(x)}$ as estimated from \cref{eq:R_deriv_error_bound_bulk}. Similar conclusions can be made for asymptotics in the Bessel and Airy regions near $z=0$ and $z=\pm 1$ respectively.
\end{remark}

\clearpage
\section{Extracting the recurrence coefficients: proof of Theorem 1}
\label{sec:recurrence_proof}
\subsection{Reversing the transformations}
We have made a series of transformations from the original Riemann-Hilbert problem:
\begin{equation*}
    Y \mapsto U \mapsto T \mapsto S \mapsto R.
\end{equation*}
We will now reverse these transformations, with the goal of extracting the $z$-dependence of $Y(z)$ as $z \to \infty$.

Recall that $U(z) = \beta_{n}^{-(n+\rho/2)\sigma_{3}} Y(\beta_{n}z) \beta_{n}^{(\rho/2)\sigma_{3}}$. Making the expansions $U(z) = (\mathds{1} + U_{1}/z + \cdots) z^{n \sigma_{3}}$ and $Y(z) = (\mathds{1} + Y_{1}/z + \cdots) z^{n \sigma_{3}}$, linear independence gives
\begin{equation}
    Y_{1} = \beta_{n} \beta_{n}^{(n+\rho/2)\sigma_{3}} U_{1} \beta_{n}^{-(n+\rho/2)\sigma_{3}}.
\end{equation}
Next we recall that $U(z) = e^{\frac{1}{2} n l_{n} \sigma_{3}} T(z) e^{-\frac{1}{2} n l_{n}\sigma_{3}} e^{n g(z) \sigma_{3}}$. Expanding $T_{1} = \mathds{1} + T_{1}/z + \cdots$ and $e^{n g(z) \sigma_{3}} = (\mathds{1} + G_{1}/z + \cdots) z^{n \sigma_{3}}$, we get
\begin{equation}
    U_{1} = e^{\frac{1}{2} n l_{n} \sigma_{3}} T_{1} e^{-\frac{1}{2} n l_{n} \sigma_{3}} + G_{1},
\end{equation}
and hence
\begin{align}
    Y_{1} = &\beta_{n} \beta_{n}^{(n+\rho/2)\sigma_{3}} e^{\frac{1}{2}n l_{n} \sigma_{3}} T_{1} e^{-\frac{1}{2} n l_{n} \sigma_{3}} \beta_{n}^{-(n+\rho/2) \sigma_{3}}\\
    + &\beta_{n} \beta_{n}^{(n+\rho/2) \sigma_{3}} G_{1} \beta_{n}^{-(n+\rho/2) \sigma_{3}}.\nonumber
\end{align}
Note that, since $G_{1}$ came from the expansion of $e^{n g(z) \sigma_{3}}$, which is a diagonal matrix, $G_{1}$ itself is also diagonal, and hence will not contribute to the off-diagonal elements of $Y_{1}$, which are the ones that matter for the recurrence coefficients.

For the $T \to S$ transformation, we simply have $T_{1} = S_{1}$, since $S$ differs from $T$ only for small values of $z$. Finally, for the $S \to R$ transformation, we have $S(z) = R(z) N(z)$ for $z \to \infty$, and hence
\begin{equation}
    S_{1} = R_{1} + N_{1}.
\end{equation}
Substituting into the formula for $Y_{1}$ and using $b_{n}^{2} = (Y_{1})_{12} (Y_{1})_{21}$, we have
\begin{equation}
    b_{n}^{2} = \beta_{n}^{2} (R_{1} + N_{1})_{12} (R_{1} + N_{1})_{21}.
\end{equation}
Using the $z \to \infty$ asymptotics for $N(z)$, we get
\begin{equation}
    b_{n} = \dfrac{\beta_{n}}{2} \left(1 + 2i \left[2^{-\rho} (R_{1})_{21} - 2^{\rho} (R_{1})_{12} \right] + 4(R_{1})_{12} (R_{1})_{21} \right)^{1/2}.
\end{equation}
Note that the leading order behaviour is given by $b_{n} \simeq \beta_{n} / 2$, since $R_{1}$ is $o(1)$ as $n \to \infty$. From the solution for $R(z)$, we have
\begin{equation}
    2^{-\rho} (R_{1})_{21} - 2^{\rho} (R_{1})_{12} = \dfrac{1}{2\pi i} \int_{\wt{\Sigma}_{R}} \left[ 2^{\rho} (\Delta_{R})_{12}(y) - 2^{-\rho} (\Delta_{R})_{21}(y) + \cdots\right] \diff y.
    \label{eq:bn_contour_integral}
\end{equation}
The $(R_{1})_{12}(R_{1})_{21}$ term in the formula for $b_{n}$ will be subleading relative to this term because it is quadratic in $R_{1}$.

\subsection{Recurrence coefficients from residue calculus}
Next we use the asymptotic expansions of $\Delta_{R}$ constructed in the various local parametrices to evaluate the RHS of \cref{eq:bn_contour_integral} using the residue theorem.
\subsubsection{Integral over the circular contour $\wt{\Sigma}_{5}$ near $z = 0$}
Since $\Delta_{R}(z) = P(z)N(z)^{-1} - \mathds{1}$ on this contour, taking the first term of the asymptotic expansion \cref{eq:DeltaR_expansion_zero} and using $[a(z) + a(z)^{-1}]^{2} - [a(z) - a(z)^{-1}]^{2} = 4$, we get
\begin{equation}
    2^{\rho} (\Delta_{R})_{12}(y) - 2^{-\rho} (\Delta_{R})_{21}(y) = (-1)^{n+1} \dfrac{i \rho}{4n} \dfrac{1}{f_{n}(z)} \begin{dcases}
        \left(e^{\frac{i \pi \rho}{2}} \varphi(z)^{-\rho} + e^{-\frac{i \pi \rho}{2}} \varphi(z)^{\rho}\right), & z \in \mathbb{C}_{+},\\
        \left(e^{\frac{-i \pi \rho}{2}} \varphi(z)^{-\rho} + e^{\frac{i \pi \rho}{2}} \varphi(z)^{\rho}\right), & z \in \mathbb{C}_{-}.
    \end{dcases}
\end{equation}
Near $z = 0$, we have
\begin{equation}
    f_{n}(z) = \pi \psi_{n}(0)z + \mathcal{O}(z^{3}).
\end{equation}
However, $\varphi(z)$ has a branch cut along the real axis, with
\begin{equation}
    \varphi(z) = \begin{cases}
        i + \mathcal{O}(z), & z \in \mathbb{C}_{+},\\
        -i + \mathcal{O}(z), & z \in \mathbb{C}_{-}.
    \end{cases}
\end{equation}
This cancels out the differing phase factors in $\mathbb{C}_{\pm}$, so we get the Laurent expansion
\begin{equation}
    2^{\rho} (\Delta_{R})_{12}(y) - 2^{-\rho} (\Delta_{R})_{21}(y) = (-1)^{n+1} \dfrac{i\rho}{2n} \dfrac{1}{\pi \psi_{n}(0)} \dfrac{1}{z} + \mathcal{O}(z^{0}),
\end{equation}
valid in both $\mathbb{C}_{+}$ and $\mathbb{C}_{-}$. Then the contour integral gives
\begin{equation}
    \dfrac{1}{2\pi i} \int_{\wt{\Sigma}_{5}} \left[ 2^{\rho} (\Delta_{R})_{12}(y) - 2^{-\rho} (\Delta_{R})_{21}(y) + \cdots\right] \diff y = (-1)^{n} i\rho \dfrac{1}{n h_{n}(0)} + \cdots,
    \label{eq:bn_origin_contrib}
\end{equation}
where we note that $\wt{\Sigma}_{5}$ is defined with clockwise orientation, and we have used $\psi_{n}(0) = h_{n}(0) / 2\pi$. 

\subsubsection{Integral over the circular contour $\wt{\Sigma}_{6}$ near $z = 1$}
Since $\Delta_{R}(z) = P(z)N(z)^{-1} - \mathds{1}$ on this contour, taking the first term of the asymptotic expansion \cref{eq:P_Pinf_expansion}, we have
\begin{equation}
    2^{\rho} (\Delta_{R})_{12}(y) - 2^{-\rho} (\Delta_{R})_{21}(y) = -\dfrac{i}{12 n \phi_{n}(z)} \left[\varphi(z)^{\rho} - \varphi(z)^{-\rho} \right].
\end{equation}
Near $z = 1$, we have
\begin{equation}
    \varphi(z)^{\rho} - \varphi(z)^{-\rho} = 2 \sqrt{2} \rho (z-1)^{1/2} + \mathcal{O}[(z-1)^{3/2}].
\end{equation}
Since $\phi_{n}(z) \sim \mathcal{O}[(z-1)^{3/2}]$ near $z = 1$, the overall expression therefore has a simple pole at $z = 1$. Recalling that $\phi_{n}(z) = -(2/3) (z-1)^{3/2} \hat{\varphi}_{n}(z)$, with $\hat{\varphi}_{n}(z)$ analytic near $z = 1$, we get
\begin{equation}
    2^{\rho} (\Delta_{R})_{12}(y) - 2^{-\rho} (\Delta_{R})_{21}(y)  = \dfrac{i \rho}{2} \dfrac{1}{n h_{n}(1)} \dfrac{1}{z-1} + \mathcal{O}(z^{0}),
\end{equation}
where we have used $\hat{\varphi}_{n}(1) = h_{n}(1) / \sqrt{2}$. Therefore the contour integral gives
\begin{equation}
    \dfrac{1}{2\pi i} \int_{\wt{\Sigma}_{6}} \left[ 2^{\rho} (\Delta_{R})_{12}(y) - 2^{-\rho} (\Delta_{R})_{21}(y) + \cdots\right] \diff y = -\dfrac{i \rho}{2} \dfrac{1}{n h_{n}(1)} + \cdots,
    \label{eq:bn_z1_contrib}
\end{equation}
where again we take into account the clockwise orientation of $\wt{\Sigma}_{6}$. Note that \cref{lem:hn1_scaling} tells us that $h_{n}(1) \sim \mathcal{O}(1)$ as $n \to \infty$, so this contribution scales as $\mathcal{O}(1/n)$.

\subsubsection{Integral over the circular contour $\wt{\Sigma}_{7}$ near $z = -1$}
Similarly, taking into account the different expressions for $\wt{\Delta}_{1}(z)$ in $\mathbb{C}_{\pm}$ from \cref{eq:DeltaR_expansion_z=-1}, we get
\begin{equation}
    2^{\rho} (\Delta_{R})_{12}(y) - 2^{-\rho} (\Delta_{R})_{21}(y) = \dfrac{i}{12 n \phi_{n}(-z)} \left[e^{\mp i \pi \rho} \varphi(z)^{\rho} - e^{\pm i \pi \rho} \varphi(z)^{-\rho} \right].
\end{equation}
To find the residue, we exploit analyticity and evaluate this expression as $z \to x + i 0^{+}$ with $x < -1$. This gives
\begin{equation}
    \varphi(x)^{\rho} = e^{i \pi \rho} \left(1 + \sqrt{2} \rho (-x-1)^{1/2} + \cdots \right).
\end{equation}
Then we have
\begin{equation}
    e^{-i \pi \rho} \varphi(x)^{\rho} - e^{i \pi \rho} \varphi(x)^{-\rho} = 2 \sqrt{2} \rho (-x-1)^{1/2} + \cdots.
\end{equation}
Then again using $\phi_{n}(-z) = -(2/3)(-z-1)^{3/2} \hat{\varphi}_{n}(-z)$, we get
\begin{equation}
    2^{\rho} (\Delta_{R})_{12}(x) - 2^{-\rho} (\Delta_{R})_{21}(x)  = \dfrac{i \rho}{2} \dfrac{1}{n h_{n}(1)} \dfrac{1}{x+1} + \mathcal{O}(x^{0}),
\end{equation}
so the residue is $i \rho / 2 n h_{n}(1)$. Hence the contour integral gives
\begin{equation}
    \dfrac{1}{2\pi i} \int_{\wt{\Sigma}_{7}} \left[ 2^{\rho} (\Delta_{R})_{12}(y) - 2^{-\rho} (\Delta_{R})_{21}(y) + \cdots\right] \diff y = -\dfrac{i \rho}{2} \dfrac{1}{n h_{n}(1)} + \cdots,
    \label{eq:bn_zm1_contrib}
\end{equation}
so to leading order we get the same contribution as from the circle around $z = 1$.

\subsubsection{Contribution from the lens boundaries $\wt{\Sigma}_{j=1..4}$}
Let us focus on the lens boundary in the first quadrant, $\wt{\Sigma}_{1}$, with the analysis for the other lens boundaries following similarly. From \cref{eq:R_def,eq:Sjump1} we have
\begin{equation}
    \Delta_{R}(z) = e^{-2n \phi_{n}(z)} \omega(z)^{-1} N(z) \begin{pmatrix} 0 & 0 \\ 1 & 0 \end{pmatrix} N(z)^{-1}.
    \label{eq:DeltaR_lens_boundaries}
\end{equation}
By using the solution \cref{eq:Pinf_sol} for $N(z)$, one can check that the factor following $e^{-2n \phi_{n}(z)}$ is $\mathcal{O}(|z|^{0})$ as $z \to \infty$. Then we have
\begin{equation}
    \left| \int_{\wt{\Sigma}_{1}} \Delta_{R}(z) \diff z \right| \leq \mathcal{O}\left(\left| \int_{\wt{\Sigma}_{1}}  |e^{-2n \phi_{n}(z)}| \diff z\right|\right)
\end{equation}
In \cref{lem:lens_boundaries} we show that as $n\to\infty$ the RHS is $o(1/n \, \mathrm{poly}(\log{n}))$, so that the contribution from the lens boundary is subleading compared with the contributions from the local parametrices at $z=0$ and $z=\pm 1$ (c.f.~\cref{eq:bn_origin_contrib,eq:bn_z1_contrib,eq:bn_zm1_contrib}).

\subsubsection{Combining the contributions}
\label{sec:bn_combining}
With the dominant contributions from the circular contours from the local parametrices, we get
\begin{equation}
    b_{n} = \dfrac{\beta_{n}}{2} \left(1 + 2 \rho \dfrac{1}{n h_{n}(1)} - (-1)^{n} 2\rho \dfrac{1}{n h_{n}(0)} +  \cdots \right)^{1/2}.
\end{equation}
Expanding the square root gives
\begin{equation}
    \boxed{b_{n} = \dfrac{\beta_{n}}{2} \left(1 + \rho \left[\dfrac{1}{h_{n}(1)} - (-1)^{n} \dfrac{1}{h_{n}(0)}\right] \dfrac{1}{n} + \cdots \right),}
\end{equation}
which concludes the proof of \cref{thm:recurrence_theorem}. The leading error term encoded in the dots depends on the growth exponents $p$ and $q$ of the potential $Q(x)$. For $p>1$, the leading error is $\wt{\mathcal{O}}(1/n^{2-1/p})$, with equal order contributions from the $k=2$ term in the asymptotic expansion of the $z=0$ local parametrix, and from the $\mathcal{O}(\norm{\Delta_{R}}_{L_{2}(\Sigma_{5})}^{2})$ error in \cref{eq:R1_integral}. For $p=1$, this error reduces to $\mathcal{O}(1/n (\log{n})^{2+q+o(1)})$, which is subleading relative to the $1/n h_{n}(0) \sim 1/n (\log{n})^{1 + o(1)}$ term because we are assuming $q>-1$ if $p=1$.

As a check of the theorem, we can consider the generalized Hermite polynomials, which have weight function $\Phi(x)/2\pi = |x|^{\rho} \exp[-x^{2}]$. With $Q(x)=x^{2}$, we have $\beta_{n} = \sqrt{2n}$, and $h_{n}(0) = h_{n}(1) = 4$ (c.f.~\cref{lem:hn0_scaling,lem:hn1_scaling} with $p=2$). Substituting into \cref{thm:recurrence_theorem}, we get agreement to $\mathcal{O}(1/n)$ with the exact recurrence coefficients, which are known to be $b_{n} = (1/\sqrt{2}) \sqrt{n + \frac{1}{2}[1 - (-1)^{n}]\rho}$~\cite{mastroianniInterpolationProcessesBasic2008}. 

\subsection{Scaling of the leading coefficient $y_{n}$}
The leading coefficients $y_{n} > 0$ of the polynomial $p_{n}(x) = y_{n} x^{n} + \cdots$ can be extracted from the solution $Y(z)$ to the fundamental Riemann-Hilbert problem using \cref{eq:Y1_to_yn}, which requires the matrix element $(Y_{1})_{12}$. Reversing the transformations for $Y$ as before gives
\begin{equation}
    (Y_{1})_{12} = \beta_{n}^{2n + 1 + \rho} e^{n l_{n}} \left(i 2^{-(1+\rho)} + (R_{1})_{12}\right).
\end{equation}
The matrix element $(R_{1})_{12}$ can be derived from a residue calculation like that for $b_{n}$, and takes the value
\begin{equation}
(R_{1})_{12} = \dfrac{1}{n} \dfrac{i}{2^{1+\rho}}\left[\dfrac{\rho}{h_{n}(0)}\left(\dfrac{\rho}{2} - (-1)^{n} \right) + \dfrac{1}{12 h_{n}(1)}\left(2(2 + 6\rho + 3\rho^{2}) - 3 \dfrac{h_{n}^{\prime}(1)}{h_{n}(1)}\right)\right] + \cdots.
\end{equation}
The contribution to $(Y_{1})_{12}$ from $(R_{1})_{12}$ is therefore subleading by a factor of $1/n$ compared with the factor $i 2^{-(1+\rho)}$ coming from $(N_{1})_{12}$. Thus we have
\begin{equation}
    (Y_{1})_{12} = \dfrac{i}{2^{1+\rho}} \beta_{n}^{2n} \beta_{n}^{1+\rho} e^{n l_{n}} \left[1 + \mathcal{O}(1/n)\right].
    \label{eq:Y1_12_asymptotic}
\end{equation}
Then the leading coefficient $y_{n}$ scales like
\begin{equation}
    \boxed{y_{n} = \dfrac{2^{\rho/2}}{\sqrt{\pi}} \beta_{n}^{-n} \beta_{n}^{-(1+\rho)/2} e^{-\frac{1}{2} n l_{n}} \left[1 + \mathcal{O}(1/n)\right].}
    \label{eq:yn_asymptotic}
\end{equation}
By \cref{lem:lagrange_multiplier}, we know that the Lagrange multiplier $l_{n}$ will be $\mathcal{O}(1)$ and negative as $n \to \infty$. On the other hand, $\beta_{n} \sim n^{1/p}$, so the leading order scaling of $y_{n}$ will be dictated by the factor $\beta_{n}^{-n}$, such that
\begin{equation}
    y_{n} \sim e^{-\frac{1}{p} n \log{n} + \mathcal{O}(n)}.
\end{equation}
As a simple check of this scaling, we note that $y_{n}$ is related to the recurrence coefficients by
\begin{equation}
    y_{n} = (b_{1} b_{2} \cdots b_{n})^{-1}.
\end{equation}
Since the recurrence coefficients scale to leading order like $b_{n} \sim \beta_{n}$, this validates the factorial-like scaling $y_{n} \sim \beta_{n}^{-n}$. In Ref.~\cite{parkerUniversalOperatorGrowth2019}, by counting Dyck paths they give a lower bound for the moments
\begin{equation}
    \mu_{2n} \geq y_{n}^{-2},
\end{equation}
so our result gives
\begin{equation}
    \mu_{2n} \geq e^{\frac{2}{p} n \log{n} + \mathcal{O}(n)},
\end{equation}
so the moments grow at least as fast as a power of a factorial to leading order. Of course, this moment bound is immediate from just taking $b_{n} \sim n^{1/p}$, but our more precise estimate of $y_{n}$ will be relevant in the next section.

\clearpage
\section{Asymptotics of the orthogonal polynomials}
\label{sec:polynomial_asymptotics}
In this section we will use our constructed asymptotic solution of the Riemann-Hilbert problem for $Y(z)$ to obtain the $n \to \infty$ asymptotics of the orthogonal polynomials $p_{n}(z)$, as well as the associated Christoffel-Darboux (CD) kernel $K_{n}(x,y) = \sum_{m=0}^{n-1} p_{m}(x) p_{m}(y)$. From \cref{eq:fki_sol}, the polynomials $p_{n}(z)$ and $p_{n-1}(z)$ will be determined by the first column of $Y(z)$ via
\begin{align}
    p_{n}(z) &= y_{n} Y_{11}(z),\label{eq:pn_from_Y}\\
    p_{n-1}(z) &= \dfrac{-1}{2\pi i} \dfrac{1}{y_{n} b_{n}} Y_{21}(z),\label{eq:pnm1_from_Y}
\end{align}
where the recurrence coefficient $b_{n}$ and the leading coefficient $y_{n}>0$ of $p_{n}(x) = y_{n} x^{n} + \cdots$ are determined from $Y_{1}$ as in \cref{eq:Y1_to_bn,eq:Y1_to_yn}.

For the CD kernel, the Christoffel-Darboux formula gives
\begin{align}
    K_{n}(x,t) = b_{n} \dfrac{p_{n}(x) p_{n-1}(t) - p_{n-1}(x) p_{n}(t)}{x-t},
\end{align}
Again, we can evaluate this in terms of the matrix elements of the solution $Y$ to the Riemann-Hilbert problem:
\begin{equation}
    K_{n}(x,t) = \dfrac{-1}{2\pi i} \dfrac{Y_{11}(x)Y_{21}(t) - Y_{21}(x) Y_{11}(t)}{x-t}.
\end{equation}
We will mainly be interested in the diagonal case $t=x$, where we get
\begin{align}
    K_{n}(x,x) &= \dfrac{-1}{2\pi i} \left(Y_{11}^{\prime}(x) Y_{21}(x) - Y_{21}^{\prime}(x) Y_{11}(x)\right),\\
    &= \dfrac{1}{2\pi i} \det \begin{pmatrix} Y_{11}(x) & Y_{11}^{\prime}(x) \\ Y_{21}(x) & Y_{21}^{\prime}(x) \end{pmatrix}.\label{eq:cd_kernel_det}
\end{align}

\subsection{Behavior near the origin}
\label{sec:origin_unpacking}
Here we discuss the asymptotics near the origin $z=0$, where the scaling will be determined in terms of Bessel functions, by virtue of the local parametrix $\Psi_{\rho/2}$ constructed in \cref{sec:bessel_sol}. Following Ref.~\cite{kuijlaarsUniversalityEigenvalueCorrelations2003}, we reverse the transformations $Y \mapsto U \mapsto T \mapsto S \mapsto R$ for $z \to 0$ in the upper lens, in the region enclosed by the local parametrix near the origin. Combining \cref{eq:U_def,eq:T_def,eq:S_def,eq:P_def_origin,eq:Spar_def,eq:R_def}, this gives
\begin{align}
    Y(\beta_{n} z) &= \beta_{n}^{(n + \rho/2) \sigma_{3}} e^{(n l_{n}/2) \sigma_{3}} R(z) E_{n}(z) \Psi_{\rho/2}(n f_{n}(z)) z^{-(\rho/2) \sigma_{3}} e^{(i \pi \rho/2) \sigma_{3}} e^{-n \phi_{n}(z) \sigma_{3}} \label{eq:Y_unpacking_origin}\\
    &\times \begin{pmatrix} 1 & 0 \\ z^{-\rho} e^{-2n \phi_{n}(z)} & 1 \end{pmatrix} e^{-(n l_{n}/2) \sigma_{3}} e^{n g_{n}(z) \sigma_{3}} \beta_{n}^{-(\rho/2) \sigma_{3}}, \nonumber
\end{align}
where we are taking $z$ to be $\mathcal{O}(1/\beta_{n})$. The first column of $Y$ is then given by
\begin{equation}
    \begin{pmatrix} Y_{11}(\beta_{n} z) \\ Y_{21}(\beta_{n}z) \end{pmatrix} = (\beta_{n} z)^{-\rho/2} e^{n \left[g_{n}(z) - \phi_{n}(z) - (l_{n}/2)\right]} \beta_{n}^{(n+\rho/2)\sigma_{3}} e^{(n l_{n}/2) \sigma_{3}} R(z) E_{n}(z) \Psi_{\rho/2}(n f_{n}(z)) e^{(i \pi \rho/2) \sigma_{3}} \begin{pmatrix} 1 \\ 1 \end{pmatrix}.
\end{equation}
For $z$ in the sector $f_{n}^{-1}(I)$, we can combine \cref{eq:Psi_rho_def} with the Bessel function identities 9.1.3 and 9.1.4 of Ref.~\cite{abramowitzHandbookMathematicalFunctions1965} to get
\begin{equation}
    \Psi_{\rho/2}(n f_{n}(z)) e^{(i \pi \rho/2) \sigma_{3}} \begin{pmatrix} 1 \\ 1 \end{pmatrix} = e^{-i \pi /4} \sqrt{\pi} \left(n f_{n}(z)\right)^{1/2} \begin{pmatrix} J_{\frac{1}{2}(\rho+1)} \left(n f_{n}(z)\right) \\ J_{\frac{1}{2}(\rho-1)} \left(n f_{n}(z)\right) \end{pmatrix},
\end{equation}
where $J_{\alpha}$ is a Bessel function of the first kind. Sending $z \downarrow x \in (0, \delta_{n})$, \cref{lem:variational_cond,eq:g_diff} tell us that
\begin{equation}
    e^{n \left[g_{n,+}(x) - \phi_{n,+}(x) - (l_{n}/2)\right]} = e^{\frac{n}{2}\left[g_{n,+}(x) + g_{n,-}(x) - l_{n}\right]} = e^{\frac{n}{2} V_{n}(x)} = e^{\frac{1}{2} Q(\beta_{n}x)},
    \label{eq:potential_to_weight}
\end{equation}
where we used $2 \phi_{n,+}(x) = g_{n,+}(x) - g_{n,-}(x)$. We therefore get
\begin{equation}
    \begin{pmatrix} Y_{11}(\beta_{n} x) \\ Y_{21}(\beta_{n}x) \end{pmatrix} = e^{-i \pi /4} \sqrt{\pi} \beta_{n}^{-\rho/2} e^{\frac{1}{2}Q(\beta_{n}x)} \beta_{n}^{(n+\rho/2)\sigma_{3}} e^{(n l_{n}/2) \sigma_{3}} R(x) E_{n}(x) \left(n f_{n}(x)\right)^{1/2} x^{-\rho/2} \begin{pmatrix} J_{\frac{1}{2}(\rho+1)} \left(n f_{n}(x)\right) \\ J_{\frac{1}{2}(\rho-1)} \left(n f_{n}(x)\right) \end{pmatrix}.
    \label{eq:kernel_x_intermediate}
\end{equation}
Since by construction $R(x) = \mathds{1} + o(1)$, where the $o(1)$ refers to scaling with $n$, for large $n$ we can approximate $R(x) \approx \mathds{1}$. Then using the definition \cref{eq:En_def_origin} of $E_{n}(x)$, we can unpack this formula to get the explicit expressions (temporarily suppressing some function arguments to reduce notation)
\begin{align}
    Y_{11}(\beta_{n}x) \approx \dfrac{\pi e^{-i \pi/4}}{2^{3/2}} \sqrt{\dfrac{n f_{n}(x)}{w(\beta_{n}x)}} \beta_{n}^{n+\rho/2} e^{\frac{n l_{n}}{2}} 2^{-\rho/2} \Bigg[ &e^{\frac{i \pi}{4}(2n-\rho)} \varphi_{+}^{\rho/2}(a_{+}+a_{+}^{-1})\left(e^{-\frac{i \pi}{4}} J_{\frac{1}{2}(\rho+1)} + e^{\frac{i \pi}{4}} J_{\frac{1}{2}(\rho-1)}\right) \label{eq:Y12_origin} \\
    + &e^{-\frac{i \pi}{4}(2n-\rho)} \varphi_{+}^{-\rho/2}(a_{+}-a_{+}^{-1})\left(e^{\frac{i \pi}{4}} J_{\frac{1}{2}(\rho+1)} + e^{-\frac{i \pi}{4}} J_{\frac{1}{2}(\rho-1)}\right) \Bigg],\nonumber
\end{align}
\begin{align}
    Y_{21}(\beta_{n}x) \approx \dfrac{\pi e^{-i \pi/4}}{2^{3/2}} \sqrt{\dfrac{n f_{n}(x)}{w(\beta_{n}x)}} i \beta_{n}^{-(n+\rho/2)} e^{-\frac{n l_{n}}{2}} 2^{\rho/2} \Bigg[ &e^{\frac{-i \pi}{4}(2n-\rho)} \varphi_{+}^{-\rho/2}(a_{+}+a_{+}^{-1})\left(e^{\frac{i \pi}{4}} J_{\frac{1}{2}(\rho+1)} + e^{-\frac{i \pi}{4}} J_{\frac{1}{2}(\rho-1)}\right) \label{eq:Y21_origin} \\
    + &e^{\frac{i \pi}{4}(2n-\rho)} \varphi_{+}^{\rho/2}(a_{+}-a_{+}^{-1})\left(e^{-\frac{i \pi}{4}} J_{\frac{1}{2}(\rho+1)} + e^{\frac{i \pi}{4}} J_{\frac{1}{2}(\rho-1)}\right) \Bigg],\nonumber
\end{align}
with $w(x) = x^{\rho} e^{-Q(x)}$, and where the functions $a \equiv a(x)$ and $\varphi \equiv \varphi(x)$ are defined in \cref{eq:a_Pinf_def,eq:Pinf_varphi_def}. In this region their $+$ side expressions are $a_{+}(x) = e^{i \pi/4} (1-x)^{1/4} / (1+x)^{1/4}$ and $\varphi_{+}(x) = x + i \sqrt{1-x^{2}} = e^{i \arccos(x)}$. The Bessel functions $J_{\frac{1}{2}(\rho \pm 1)}$ should be evaluated with the argument $n f_{n}(x)$.

\subsubsection{Proof of \cref{lem:pn0_scaling}}
\label{sec:thm2_proof}
With these expressions for the first column of $Y$ near the origin, we can now get expressions for the orthogonal polynomials $p_{n}(x)$ and $p_{n-1}(x)$ using \cref{eq:pn_from_Y,eq:pnm1_from_Y}, as well as the expression \cref{eq:yn_asymptotic} for the leading coefficient $y_{n}$. For example, for $p_{n}$ we get
\begin{align}
    p_{n}(\beta_{n}x) \approx \dfrac{e^{-i \pi/4}}{2^{3/2}} \sqrt{\dfrac{n f_{n}(x)}{\beta_{n} w(\beta_{n}x)}}\Bigg[ &e^{\frac{i \pi}{4}(2n-\rho)} \varphi_{+}^{\rho/2}(a_{+}+a_{+}^{-1})\left(e^{-\frac{i \pi}{4}} J_{\frac{1}{2}(\rho+1)} + e^{\frac{i \pi}{4}} J_{\frac{1}{2}(\rho-1)}\right) \label{eq:pn_bessel} \\
    + &e^{-\frac{i \pi}{4}(2n-\rho)} \varphi_{+}^{-\rho/2}(a_{+}-a_{+}^{-1})\left(e^{\frac{i \pi}{4}} J_{\frac{1}{2}(\rho+1)} + e^{-\frac{i \pi}{4}} J_{\frac{1}{2}(\rho-1)}\right) \Bigg],\nonumber
\end{align}
where the arguments of $\varphi_{+} \equiv \varphi_{+}(x)$, $a_{+} \equiv a_{+}(x)$ and $J_{\frac{1}{2}(\rho \pm 1)} \equiv J_{\frac{1}{2}(\rho \pm 1)}(n f_{n}(x))$ are suppressed as before. For $\rho = 0$, this reduces to \cref{eq:pn_bulk_asymptotic} of the main text after using $n f_{n}(x) = \pi I_{n}(\beta_{n}x)$, with $I_{n}$ defined in \cref{eq:In_def}. By the discussion surrounding \cref{eq:R_error_bound}, this expression has a multiplicative error of $\mathcal{O}(\beta_{n} / n h_{n}(0))$.

Using $f_{n}(x) = \pi \int_{0}^{x} \psi_{n}(s) \diff s \approx \pi \psi_{n}(0) x$, we can take the limit $x \to 0^{+}$, giving
\begin{align}
    p_{n}(0) &\approx \cos\left(\dfrac{n \pi}{2}\right) \dfrac{2^{\frac{1}{2}-\frac{\rho}{2}}}{\Gamma\left[\frac{1+\rho}{2}\right]} e^{Q(0)/2} \dfrac{\left[\pi \sigma_{n}(0)\right]^{\rho/2}}{\beta_{n}^{1/2}},\\
    p_{n}^{\prime}(0) &\approx \sin\left(\dfrac{n \pi}{2}\right) \dfrac{2^{-\frac{1}{2}-\frac{\rho}{2}}}{\Gamma\left[\frac{3+\rho}{2}\right]} e^{Q(0)/2} \dfrac{\left[\pi \sigma_{n}(0)\right]^{1+\rho/2}}{\beta_{n}^{1/2}} \left(1 + \dfrac{(\rho+1)^{2}}{2\pi \beta_{n} \sigma_{n}(0)}\right),
\end{align}
where $\sigma_{n}(0) = n \psi_{n}(0) / \beta_{n}$ is the equilibrium density of the Coulomb gas at $\omega=0$. This proves \cref{lem:pn0_scaling} in the main text. As a sanity check, we can see that $p_{n}(0)$ vanishes for odd $n$ and $p_{n}^{\prime}(0)$ vanishes for even $n$. One can also check using \cref{eq:R_deriv_error_bound} that the contribution from the derivative $R^{\prime}(z)$ is subleading relative to the leading term by a factor of $\mathcal{O}(\beta_{n} / n h_{n}(0))$ (c.f.\ \cref{rem:R_deriv}), so the derivative $p_{n}^{\prime}(0)$ has the same multiplicative error as $p_{n}(0)$.

Similarly, using \cref{eq:pnm1_from_Y,eq:Y21_origin,eq:yn_asymptotic}, for $p_{n-1}$ we get
\begin{align}
    p_{n-1}(\beta_{n}x) \approx -\dfrac{e^{-i \pi/4}}{2^{3/2}} \dfrac{\beta_{n}}{2 b_{n}} \sqrt{\dfrac{n f_{n}(x)}{\beta_{n} w(\beta_{n}x)}}\Bigg[ &e^{-\frac{i \pi}{4}(2n-\rho)} \varphi_{+}^{-\rho/2}(a_{+}+a_{+}^{-1})\left(e^{\frac{i \pi}{4}} J_{\frac{1}{2}(\rho+1)} + e^{-\frac{i \pi}{4}} J_{\frac{1}{2}(\rho-1)}\right) \label{eq:pnm1_bessel} \\
    + &e^{\frac{i \pi}{4}(2n-\rho)} \varphi_{+}^{\rho/2}(a_{+}-a_{+}^{-1})\left(e^{-\frac{i \pi}{4}} J_{\frac{1}{2}(\rho+1)} + e^{\frac{i \pi}{4}} J_{\frac{1}{2}(\rho-1)}\right) \Bigg],\nonumber
\end{align}
with the same argument suppression as before. Note that \cref{thm:recurrence_theorem} tells us that the factor $\beta_{n} / 2 b_{n} = 1 + \mathcal{O}(1/n)$; for $\rho=0$, this equation reduces to \cref{eq:pnm1_bulk_asymptotic} of the main text after approximating $\beta_{n} / 2 b_{n} \approx 1$.

Keeping $\rho$ general, again taking the limit $x \to 0^{+}$ gives
\begin{align}
    p_{n-1}(0) &\approx \sin\left(\dfrac{n \pi}{2}\right)\dfrac{\beta_{n}}{2 b_{n}} \dfrac{2^{\frac{1}{2}-\frac{\rho}{2}}}{\Gamma\left[\frac{1+\rho}{2}\right]} e^{Q(0)/2} \dfrac{\left[\pi \sigma_{n}(0)\right]^{\rho/2}}{\beta_{n}^{1/2}},\\
    p_{n-1}^{\prime}(0) &\approx -\cos\left(\dfrac{n \pi}{2}\right)\dfrac{\beta_{n}}{2 b_{n}} \dfrac{2^{-\frac{1}{2}-\frac{\rho}{2}}}{\Gamma\left[\frac{3+\rho}{2}\right]} e^{Q(0)/2} \dfrac{\left[\pi \sigma_{n}(0)\right]^{1+\rho/2}}{\beta_{n}^{1/2}} \left(1 + \dfrac{\rho^{2}-1}{2\pi \beta_{n} \sigma_{n}(0)}\right),
\end{align}
For the CD kernel, the Christoffel-Darboux formula $K_{n}(0,0) = b_{n} \left(p_{n}^{\prime}(0) p_{n-1}(0) - p_{n-1}^{\prime}(0) p_{n}(0)\right)$ gives
\begin{equation}
    K_{n}(0,0) \approx \dfrac{1}{2^{1+\rho}\Gamma\left[\frac{1}{2}(1+\rho)\right]\Gamma\left[\frac{1}{2}(3+\rho)\right]} e^{Q(0)} \left[\pi \sigma_{n}(0)\right]^{1+\rho}. %
\end{equation}
This proves \cref{eq:exact_kn_formula} of the main text, with a multiplicative error term of $\mathcal{O}(\beta_{n} / n h_{n}(0))$ coming from neglecting the $R(0) - \mathds{1}$ and $R^{\prime}(0)$ contributions.

\subsection{Expressions for the spectral bootstrap}
\label{sec:spectral_bootstrap_derivation}
\subsubsection{Bessel bootstrap: behavior near the origin}
First let us derive \cref{eq:expmQ_finite} of the main text. Squaring \cref{eq:pn_bessel,eq:pnm1_bessel} and using the expressions $a_{+}(x) = e^{i \pi/4} (1-x)^{1/4} / (1+x)^{1/4}$ and $\varphi_{+}(x) = e^{i \arccos(x)}$, we get (after some algebraic simplifications)
\begin{align}
    p_{n-1}(\beta_{n} x)^{2} &+ p_{n}(\beta_{n}x)^{2} \approx \dfrac{n f_{n}(x)}{\beta_{n} w(\beta_{n}x)} \dfrac{1}{\sqrt{1-x^{2}}} \times \Bigg[ \left(J_{\frac{1}{2}(\rho-1)}^{2} + J_{\frac{1}{2}(\rho+1)}^{2}\right)(n f_{n}(x))\\
    &- (-1)^{n} x \left(\left(J_{\frac{1}{2}(\rho-1)}^{2} - J_{\frac{1}{2}(\rho+1)}^{2}\right)(n f_{n}(x)) \sin[\rho \arcsin{x}] + 2 \left(J_{\frac{1}{2}(\rho-1)} J_{\frac{1}{2}(\rho+1)}\right)(n f_{n}(x)) \cos[\rho \arcsin{x}]\right)\Bigg].\nonumber
\end{align}
Then we use the fact that $n f_{n}(x) = \pi I_{n}(\beta_{n} x)$ with $I_{n}$ defined in \cref{eq:In_def}. Rescaling $\beta_{n} x \mapsto \omega$, and using $w(x) = x^{\rho} \exp[-Q(x)]$, we get
\begin{align}
    e^{-Q(\omega)} &\approx \dfrac{1}{p_{n-1}(\omega)^{2} + p_{n}(\omega)^{2}} \dfrac{1}{\sqrt{\beta_{n}^{2} - \omega^{2}}} \dfrac{\pi I_{n}(\omega)}{\omega^{\rho}} \Bigg[\left(J_{\frac{1}{2}(\rho-1)}^{2} + J_{\frac{1}{2}(\rho+1)}^{2}\right)(\pi I_{n}(\omega)) \label{eq:expmQ_finite_full} \\
    - &(-1)^{n} \dfrac{\omega}{\beta_{n}} \left(\left(J_{\frac{1}{2}(\rho-1)}^{2} - J_{\frac{1}{2}(\rho+1)}^{2}\right)(\pi I_{n}(\omega)) \sin\left\{\rho \arcsin\left(\dfrac{\omega}{\beta_{n}}\right)\right\} + 2 \left(J_{\frac{1}{2}(\rho-1)} J_{\frac{1}{2}(\rho+1)}\right)(\pi I_{n}(\omega)) \cos\left\{\rho \arcsin\left(\dfrac{\omega}{\beta_{n}}\right)\right\}\right)\Bigg]. \nonumber
\end{align}
This simplifies to \cref{eq:expmQ_finite} upon dropping the term proportional to $\sin\{\rho \arcsin(\omega/\beta_{n})\}$, which for $\omega/\beta_{n} \ll 1$ is subleading by a factor of $\mathcal{O}(\beta_{n} \sigma_{n}(\omega))$ relative to the term proportional to $\cos\{\rho \arcsin(\omega/\beta_{n})\}$.

Now let us derive the expression \cref{eq:sigman_finite} involving the diagonal Christoffel-Darboux kernel. The general strategy will be to use the Christoffel-Darboux formula $K_{n}(\omega,\omega) = b_{n}\left(p_{n}^{\prime}(\omega) p_{n-1}(\omega) - p_{n-1}^{\prime}(\omega) p_{n}(\omega)\right)$, together with the expressions \cref{eq:pn_bessel,eq:pnm1_bessel} we derived for $p_{n}$ and $p_{n-1}$. After a lengthy but straightforward algebraic computation, we get the expression
\begin{align}
    \hspace{-3em} K_{n}(\omega,\omega) &\approx \dfrac{1}{4} \dfrac{\pi I_{n}(\omega)}{w(\omega)} \Bigg(\pi \sigma_{n}(\omega) \left[J_{\frac{1}{2}(\rho-1)}^{2} + J_{\frac{1}{2}(\rho+1)}^{2} - J_{\frac{1}{2}(\rho-3)} J_{\frac{1}{2}(\rho+1)} - J_{\frac{1}{2}(\rho-1)} J_{\frac{1}{2}(\rho+3)}\right](\pi I_{n}(\omega)) + \dfrac{\rho\left(J_{\frac{1}{2}(\rho-1)}^{2} + J_{\frac{1}{2}(\rho+1)}^{2}\right)(\pi I_{n}(\omega))}{\sqrt{\beta_{n}^{2} - \omega^{2}}} \nonumber  \\
    + &(-1)^{n} \dfrac{ \beta_{n}}{\beta_{n}^{2} - \omega^{2}} \left[\left(J_{\frac{1}{2}(\rho+1)}^{2} - J_{\frac{1}{2}(\rho-1)}^{2}\right)(\pi I_{n}(\omega)) \cos\left\{\rho \arcsin\left(\dfrac{\omega}{\beta_{n}}\right)\right\} + 2 \left(J_{\frac{1}{2}(\rho-1)} J_{\frac{1}{2}(\rho+1)}\right)(\pi I_{n}(\omega)) \sin\left\{\rho \arcsin\left(\dfrac{\omega}{\beta_{n}}\right)\right\}\right]\Bigg) \label{eq:sigman_finite_full},
\end{align}
where we used $\partial_{\omega}I_{n}(\omega) = \sigma_{n}(\omega)$. This expression reduces to \cref{eq:sigman_finite} after dropping all but the first term proportional to $\pi \sigma_{n}(\omega)$ in the large curly bracketed expression, since for $\omega /\beta_{n} \ll 1$ the other terms are subleading by a factor of either $\mathcal{O}(\beta_{n} \sigma_{n}(\omega))$ or $\mathcal{O}(\beta_{n}^{2} / \omega^{2})$.

\subsubsection{Bulk bootstrap}
Having already derived the spectral bootstrap equations near an algebraic divergence $\Phi(\omega) \sim |\omega|^{\rho}$, a quick way to arrive at the spectral bootstrap equations in the bulk is to simply send $\rho \to 0$ and then replace $\exp[-Q(\omega)] \mapsto \Phi(\omega)/2\pi$. One can check that properly doing the calculation by unpacking the expression for $Y(z)$ in the bulk gives the same expressions.

To derive \cref{eq:Phi_bulk} of the main text, we send $\rho \to 0$ in \cref{eq:expmQ_finite_full}; using $J_{\frac{-1}{2}}(x) = \sqrt{2/\pi x} \cos{x}$ and $J_{\frac{1}{2}}(x) = \sqrt{2/\pi x} \sin{x}$, we get
\begin{equation}
    e^{-Q(\omega)} \approx \dfrac{1}{p_{n-1}(\omega)^{2} + p_{n}(\omega)^{2}} \dfrac{1}{\sqrt{\beta_{n}^{2} - \omega^{2}}} \dfrac{2}{\pi}\left[1 - (-1)^{n} \dfrac{\omega}{\beta_{n}} \sin[2\pi I_{n}(\omega)]\right],
\end{equation}
which is equivalent to \cref{eq:Phi_bulk} once we set $\Phi(\omega) \equiv 2\pi \exp[-Q(\omega)]$.

To get \cref{eq:sigma_bulk} of the main text, we take $\rho \to 0$ of \cref{eq:sigman_finite_full} to get
\begin{equation}
    K_{n}(\omega,\omega) \approx \dfrac{2\pi\sigma_{n}(\omega)}{\Phi(\omega)}\left(1 - (-1)^{n} \dfrac{\beta_{n}}{\beta_{n}^{2}-\omega^{2}} \dfrac{\cos[2 \pi I_{n}(\omega)]}{2\pi \sigma_{n}(\omega)} \right).  
\end{equation}
This reduces to \cref{eq:sigma_bulk} upon dropping the second term in the brackets, which is subleading by a factor of $\mathcal{O}(\beta_{n} \sigma_{n}(\omega))$.
From \cref{eq:R_error_bound_bulk,eq:R_deriv_error_bound_bulk}, this asymptotic has a multiplicative error of $\mathcal{O}(|\omega|^{-1} \beta_{n} / n h_{n}(0)) + \mathcal{O}(1/n)$, where the factor of $|\omega|^{-1}$ gets cut off by $\mathcal{O}(1)$ for $|\omega| \leq \mathcal{O}(1)$.

\subsubsection{Airy bootstrap: behavior near the edge}
\label{sec:airy_bootstrap_derivation}
In this section we will consider frequencies near the edge of the spectrum, $\omega \approx \beta_{n}$. Similarly to \cref{sec:origin_unpacking}, we reverse the transformations $Y \mapsto U \mapsto T \mapsto S \mapsto R$, but now within the local parametrix centered at $z=1$. With the aim of approaching the real axis, we will take $z$ in the sector mapped to sector II of \cref{fig:unity_endpoint}. Combining \cref{eq:U_def,eq:T_def,eq:S_def,eq:P_def,eq:Spar_def,eq:R_def}, this gives
\begin{align}
    Y(\beta_{n} z) &= \beta_{n}^{(n + \rho/2) \sigma_{3}} e^{(n l_{n}/2) \sigma_{3}} R(z) E(z) \Psi(f_{n}(z)) e^{-n \phi_{n}(z) \sigma_{3}} z^{-(\rho/2)\sigma_{3}} \label{eq:Y_unpacking_edge}\\
    &\times \begin{pmatrix} 1 & 0 \\ z^{-\rho} e^{-2n \phi_{n}(z)} & 1 \end{pmatrix} e^{-(n l_{n}/2) \sigma_{3}} e^{n g_{n}(z) \sigma_{3}} \beta_{n}^{-(\rho/2) \sigma_{3}}. \nonumber
\end{align}
The first column of $Y$ is then given by
\begin{equation}
    \begin{pmatrix} Y_{11}(\beta_{n} z) \\ Y_{21}(\beta_{n}z) \end{pmatrix} = (\beta_{n} z)^{-\rho/2} e^{n \left[g_{n}(z) - \phi_{n}(z) - (l_{n}/2)\right]} \beta_{n}^{(n+\rho/2)\sigma_{3}} e^{(n l_{n}/2) \sigma_{3}} R(z) E(z) \Psi(f_{n}(z)) \begin{pmatrix} 1 \\ 1 \end{pmatrix}.
\end{equation}
Taking $z \downarrow x < 1$, using \cref{eq:potential_to_weight}, and replacing $\Psi$ with the sector II expression from \cref{eq:Psi_def}, we get
\begin{equation}
    \begin{pmatrix} Y_{11}(\beta_{n} x) \\ Y_{21}(\beta_{n}x) \end{pmatrix} = \sqrt{2\pi} e^{-i \pi/4} (\beta_{n} x)^{-\rho/2} e^{\frac{1}{2}Q(\beta_{n}x)} \beta_{n}^{(n+\rho/2)\sigma_{3}} e^{(n l_{n}/2) \sigma_{3}} R(x) E(x) \begin{pmatrix} \Ai(f_{n}(x)) \\ \Ai^{\prime}(f_{n}(x)) \end{pmatrix},
\end{equation}
where we recall that both $E(z)$ and $R(z)$ are analytic near $z=1$, and the biholomorphic map $f_{n}(z)$ is defined in \cref{prop:fn_endpoint}. Using the definition \cref{eq:E_def} of $E(x)$, after some algebra this becomes
\begin{align}
    &\begin{pmatrix} Y_{11}(\beta_{n} z) \\ Y_{21}(\beta_{n}z) \end{pmatrix} = \sqrt{\pi} e^{-i \pi/4} (\beta_{n}x)^{-\rho/2} e^{\frac{1}{2}Q(\beta_{n}x)} \beta_{n}^{(n + \rho/2)\sigma_{3}} e^{(n l_{n}/2)\sigma_{3}} R(x) 2^{-(\rho/2)\sigma_{3}} e^{(i \pi/4) \sigma_{3}} \\
    &\times \begin{pmatrix}
        \left[i a_{+} \sin\left(\frac{\rho}{2} \arccos{x}\right) + a_{+}^{-1}\cos\left(\frac{\rho}{2} \arccos{x}\right)\right] f_{n}^{1/4} \Ai(f_{n}) + \left[-a_{+}\cos\left(\frac{\rho}{2} \arccos{x}\right) - i a_{+}^{-1}\sin\left(\frac{\rho}{2} \arccos{x}\right)\right] f_{n}^{-1/4} \Ai^{\prime}(f_{n}) \\[1em]
        \left[-i a_{+} \sin\left(\frac{\rho}{2} \arccos{x}\right) + a_{+}^{-1}\cos\left(\frac{\rho}{2} \arccos{x}\right)\right] f_{n}^{1/4} \Ai(f_{n}) + \left[a_{+}\cos\left(\frac{\rho}{2} \arccos{x}\right) - i a_{+}^{-1}\sin\left(\frac{\rho}{2} \arccos{x}\right)\right] f_{n}^{-1/4} \Ai^{\prime}(f_{n}) \nonumber
    \end{pmatrix},
\end{align}
where for notational simplicity we have suppressed the arguments of $f_{n} \equiv f_{n}(x)$ and $a_{+} \equiv a_{+}(x) = e^{i \pi/4} (1-x)^{1/4} / (1+x)^{1/4}$, and used $\varphi_{+}(x) = e^{i \arccos{x}}$ for $0<x<1$. Approximating $R(x) \approx \mathds{1}$, then using \cref{eq:pn_from_Y,eq:pnm1_from_Y} to convert to $p_{n}(z)$ and $p_{n-1}(z)$, and the asymptotic \cref{eq:yn_asymptotic} for the leading coefficient $y_{n}$, we get
\small
\begin{align}
    \textstyle \hspace{-4.5em}p_{n}(\beta_{n}x) \approx \dfrac{1}{\sqrt{\beta_{n} w(\beta_{n}x)}}\left(\left[i a_{+} \sin\left(\frac{\rho}{2} \arccos{x}\right) + a_{+}^{-1}\cos\left(\frac{\rho}{2} \arccos{x}\right)\right] f_{n}^{1/4} \Ai(f_{n}) + \left[-a_{+}\cos\left(\frac{\rho}{2} \arccos{x}\right) - i a_{+}^{-1}\sin\left(\frac{\rho}{2} \arccos{x}\right)\right] f_{n}^{-1/4} \Ai^{\prime}(f_{n}) \right), \label{eq:airy_pn_asymptotic}\\
    \textstyle \hspace{-5.5em}p_{n-1}(\beta_{n}x) \approx \dfrac{\beta_{n}}{2 b_{n}} \dfrac{1}{\sqrt{\beta_{n} w(\beta_{n}x)}}\left(\left[-i a_{+} \sin\left(\frac{\rho}{2} \arccos{x}\right) + a_{+}^{-1}\cos\left(\frac{\rho}{2} \arccos{x}\right)\right] f_{n}^{1/4} \Ai(f_{n}) + \left[a_{+}\cos\left(\frac{\rho}{2} \arccos{x}\right) - i a_{+}^{-1}\sin\left(\frac{\rho}{2} \arccos{x}\right)\right] f_{n}^{-1/4} \Ai^{\prime}(f_{n}) \right).\label{eq:airy_pnm1_asymptotic}
\end{align}
\normalsize
Note that, by \cref{thm:recurrence_theorem}, the prefactor $\beta_{n}/2 b_{n} = 1 + \mathcal{O}(1/n)$, so we will henceforth approximate it as 1. Now, from the definition of $f_{n}(z)$ in \cref{prop:fn_endpoint}, we can deduce that the behavior of $f_{n}(x)$ near $x=1$ is approximately
\begin{equation}
    f_{n}(x) \approx (x-1) f_{n}^{\prime}(1), \quad x \to 1,
\end{equation}
where the derivative $f_{n}^{\prime}(1)$ is related to the equilibrium measure by
\begin{equation}
    f_{n}^{\prime}(1) = \left(n \hat{\varphi}_{n}(1)\right)^{2/3} = \left(\dfrac{n h_{n}(1)}{\sqrt{2}}\right)^{2/3}.
\end{equation}
This allows us to take the limit $x \to 1^{-}$ of \cref{eq:airy_pn_asymptotic,eq:airy_pnm1_asymptotic}, with the result
\begin{align}
    p_{n}(\beta_{n}) &\approx \dfrac{1}{\sqrt{\beta_{n}w(\beta_{n})}} \left[(2 n h_{n}(1))^{1/6} \Ai(0) - (\rho+1) (2n h_{n}(1))^{-1/6} \Ai^{\prime}(0)\right],\label{eq:pn_betan}\\
    p_{n-1}(\beta_{n}) &\approx \dfrac{1}{\sqrt{\beta_{n}w(\beta_{n})}} \left[(2 n h_{n}(1))^{1/6} \Ai(0) - (\rho-1) (2n h_{n}(1))^{-1/6} \Ai^{\prime}(0)\right],
\end{align}
where $\Ai(0) = (3^{\frac{2}{3}} \Gamma[\frac{2}{3}])^{-1}$ and $\Ai^{\prime}(0) = -(3^{\frac{1}{3}} \Gamma[\frac{1}{3}])^{-1}$. By taking the ratio of these expressions we eliminate the unknown weight $w(\beta_{n})$ and get a quadratic equation in $(2n h_{n}(1))^{1/6}$, either solution of which gives
\begin{equation}
    h_{n}(1) \approx \dfrac{1}{2n} \left(\dfrac{\Ai(0)}{\Ai^{\prime}(0)}\right)^{3} \left[\rho - \left(\dfrac{p_{n}(\beta_{n}) + p_{n-1}(\beta_{n})}{p_{n}(\beta_{n}) - p_{n-1}(\beta_{n})}\right)\right]^{3}.
\end{equation}
This gives a means of determining $h_{n}(1)$ solely in terms of the orthogonal polynomials, which can be computed in terms of the Lanczos coefficients using the three-term recurrence. One can then determine the weight $w(\beta_{n})$ at the endpoint $\omega=\beta_{n}$ by inverting \cref{eq:pn_betan}:
\begin{equation}
    w(\beta_{n}) \approx \dfrac{1}{\beta_{n} p_{n}(\beta_{n})^{2}}\left[(2 n h_{n}(1))^{1/6} \Ai(0) - (\rho+1) (2n h_{n}(1))^{-1/6} \Ai^{\prime}(0)\right].
\end{equation}
To get the first of the bootstrap equations, we sum the squares of \cref{eq:airy_pn_asymptotic,eq:airy_pnm1_asymptotic} and rearrange for $w(\beta_{n}x)$ to give
\begin{align}
    w(\beta_{n}x) &\approx \dfrac{2}{\beta_{n}} \dfrac{1}{p_{n}(\beta_{n}x)^{2} + p_{n-1}(\beta_{n}x)^{2}} \dfrac{1}{\sqrt{1-x^{2}}}\Big[- 2x \Ai\big(f_{n}(x)\big) \Ai^{\prime}\big(f_{n}(x)\big) \sin\left[\rho \arccos{x}\right] \\
    &+\big(-f_{n}(x)\big)^{1/2} \Ai\big(f_{n}(x)\big)^{2}\big(x \cos\left[\rho \arccos{x}\right]+1\big) - \big(-f_{n}(x)\big)^{-1/2} \Ai^{\prime}\big(f_{n}(x)\big)^{2}\big(x \cos\left[\rho \arccos{x}\right]-1\big)\Big]. \nonumber
\end{align}
To get the second bootstrap equation concerning the derivative $f_{n}^{\prime}(x)$, we need to compute the Christoffel-Darboux kernel $K_{n}(\beta_{n} x,\beta_{n} x)$, which involves differentiating the expressions \cref{eq:airy_pn_asymptotic,eq:airy_pnm1_asymptotic} for $p_{n}(\beta_{n}x)$ and $p_{n-1}(\beta_{n}x)$. Note that, because of the determinantal structure of $K_{n}$ (c.f.~\cref{eq:cd_kernel_det}), it is not necessary to differentiate the weight $w(\beta_{n}x)$, since all terms involving its derivative cancel exactly. It is then a matter of algebra to arrive at the expression
\begin{align}
    f_{n}^{\prime}(x) &\approx \dfrac{1}{2\Ai^{2} f_{n}^{2} - 2 (\Ai^{\prime})^{2} f_{n} - \Ai \Ai^{\prime}}\Big[-2 \beta_{n} w(\beta_{n}x) K_{n}(\beta_{n}x, \beta_{n}x) f_{n} + \frac{2}{1-x^{2}} \Ai \Ai^{\prime} \cos[\rho \arccos{x}] f_{n} \\
    &- \dfrac{\sin[\rho \arccos{x}]}{1-x^{2}} \left(\Ai^{2} \big(-f_{n}\big)^{3/2} - (\Ai^{\prime})^{2} (-f_{n})^{1/2}\right) - \dfrac{\rho}{\sqrt{1-x^{2}}}\left(\Ai^{2} \big(-f_{n}\big)^{3/2} + (\Ai^{\prime})^{2} (-f_{n})^{1/2}\right)\Big], \nonumber
\end{align}
where we have suppressed the arguments of $f_{n} \equiv f_{n}(x)$, $\Ai \equiv \Ai(f_{n}(x))$, and $\Ai^{\prime} \equiv \Ai^\prime(f_{n}(x))$. From \cref{eq:R_error_bound_bulk,eq:R_deriv_error_bound_bulk}, this expression has a multiplicative error of $\mathcal{O}(1/n)$ from neglecting the contributions from $R(z) - \mathds{1}$ and $R^{\prime}(z)$.

\subsection{Level-$n$ Green's function} 
\label{sec:greens_func_derivation}
In this section we derive the expressions in \cref{eq:n_gf_asymptotic,eq:greens_func_bulk} for the level-$n$ Green's function $G_{n}(z)$. From the definition in \cref{eq:greens_func_def}, we can deduce that
\begin{equation}
    G_{n}(z) = \int_{\mathbb{R}} \dfrac{\diff \mu^{(n)}(x)}{z-x}, \quad z \in \mathbb{C} \setminus \mathbb{R},
\end{equation}
where $\mu^{(n)}$ is the spectral measure which generates the `associated orthogonal polynomials' $\{p_{m}^{(n)}\}_{m=0}^{\infty}$, which are defined by the `$n$-shifted' recursion relation
\begin{equation}
    b_{m+1+n} p_{m+1}^{(n)}(x) = x p_{m}^{(n)}(x) - b_{m+n} p_{m-1}^{(n)}(x), \quad m \geq 0,
\end{equation}
with initial conditions $p_{-1}^{(n)}(x) = 0$, $p_{0}^{(n)}(x) = 1$. For $n=0$ we have $\diff \mu^{(0)}(x) = w(x) \diff x$, where $w$ is the original weight function.

The starting point for our analysis will be the relation
\begin{equation}
    G_{n}(z) = \dfrac{1}{b_{n}} \dfrac{C_{n}(z)}{C_{n-1}(z)},
    \label{eq:Gn_from_cauchy}
\end{equation}
where
\begin{equation}
    C_{n}(z) = \int_{\mathbb{R}} \dfrac{p_{n}(x)}{z-x} w(x) \diff x, \quad z \in \mathbb{C} \setminus \mathbb{R},
\end{equation}
is the weighted Cauchy-Stieltjes transform of $p_{n}(x)$ with respect to the original weight function $w$. This relation appears as Eq.~(3.7) in Ref.~\cite{vanasscheOrthogonalPolynomialsAssociated1991}. This reformulation is helpful because these weighted Cauchy-Stieltjes transforms appear (up to a normalization factor) in the second column of the solution $Y(z)$ to the fundamental Riemann-Hilbert problem discussed in \cref{sec:fundamental_rhp}. Indeed, a short computation gives
\begin{equation}
    G_{n}(z) = \dfrac{1}{(Y_{1})_{12}} \dfrac{Y_{12}(z)}{Y_{22}(z)},
    \label{eq:Gn_from_Y}
\end{equation}
where we used the fundamental solution \cref{eq:fki_sol}, together with the relations $b_{n} = y_{n-1} / y_{n}$ and \cref{eq:Y1_to_bn,eq:Y1_to_yn}. We have already determined the asymptotics of $(Y_{1})_{12}$ in \cref{eq:Y1_12_asymptotic}, so it remains to determine the scaling of $Y_{12}(z)$ and $Y_{22}(z)$. 

\subsubsection{Behavior away from $z=0$ and $z=\pm \beta_{n}$}
Now let us derive the expression \cref{eq:greens_func_bulk} for the level-$n$ Green's function away from the special points $z=0$ and $z = \pm \beta_{n}$. Unpacking the transformations \cref{eq:U_def,eq:T_def,eq:S_def,eq:R_def} for $Y(z)$, we have
\begin{equation}
    \begin{pmatrix} Y_{12}(\beta_{n}z) \\ Y_{22}(\beta_{n}z) \end{pmatrix} = \beta_{n}^{\rho/2} e^{\frac{n l_{n}}{2}} e^{-n g_{n}(z)} \beta_{n}^{(n+\rho/2)\sigma_{3}} e^{\frac{n l_{n}}{2} \sigma_{3}} R(z) N(z) \begin{pmatrix} 0 \\ 1 \end{pmatrix}.
\end{equation}
Note that, since we have focused only on the second column of $Y(z)$, this expression is accurate regardless of whether $z$ is inside or outside the lens, provided it is not near the centers $z=0$ and $z=\pm\beta_{n}$ of the local parametrices constructed in \cref{sec:endpoints,sec:origin}. Approximating $R(z) \approx \mathds{1}$, and using the form \cref{eq:Pinf_sol} for the outside solution $N(z)$, we get
\begin{equation}
    \dfrac{Y_{12}(\beta_{n} z)}{Y_{22}(\beta_{n} z)} \approx -i 2^{-\rho} \beta_{n}^{2n+\rho} e^{n l_{n}} \left(\dfrac{a(z) - a(z)^{-1}}{a(z) + a(z)^{-1}}\right).
\end{equation}
Using the definition $a(z) = (z-1)^{1/4} / (z+1)^{1/4}$ and the asymptotic form \cref{eq:Y1_12_asymptotic} for $(Y_{1})_{12}$, we get \cref{eq:greens_func_bulk} of the main text:
\begin{equation}
    \beta_{n} G_{n}(\beta_{n} z) \approx 2 \left(z - \sqrt{z + 1} \sqrt{z - 1}\right).
\end{equation}
There will be a multiplicative error from the $R(z) \approx \mathds{1}$ approximation; by the discussion surrounding \cref{eq:R_error_bound_bulk}, this error will be $\mathcal{O}(1/n)$ for $z = \mathcal{O}(1)$. For $\rho = 0$ it will remain of this order even as $z$ shrinks to 0, but for $\rho \neq 0$ the error will grow like $\mathcal{O}(1/(n h_{n}(0)|z|))$, eventually reaching $\mathcal{O}(\beta_{n}/n h_{n}(0))$ at $z \sim 1/\beta_{n}$.

\subsubsection{Behavior near the origin}
We will focus on the limiting behavior as $z$ approaches the positive real axis from $\mathbb{C}_{+}$; the other limits can be computed analogously. Starting from \cref{eq:Y_unpacking_origin}, for $z$ in the first quadrant and satisfying $|\beta_{n} z| \ll 1$, we have
\begin{equation}
    \begin{pmatrix} Y_{12}(\beta_{n}z) \\ Y_{22}(\beta_{n}z) \end{pmatrix} = (\beta_{n} z)^{\rho/2} e^{-\frac{i \pi \rho}{2}} e^{-n[g_{n}(z) - \phi_{n}(z) - l_{n}/2]} \beta_{n}^{(n + \rho/2)\sigma_{3}} e^{\frac{n l_{n}}{2} \sigma_{3}} R(z) E_{n}(z) \Psi_{\rho/2}(n f_{n}(z)) \begin{pmatrix} 0 \\ 1 \end{pmatrix}.
\end{equation}
Using the expression \cref{eq:Psi_rho_def} for $\Psi_{\rho/2}(n f_{n}(z))$ with $z$ in the first quadrant, and taking the limit $z \downarrow x \in (0,\delta_{n})$, this becomes
\begin{equation}
    \begin{pmatrix} Y_{12,+}(\beta_{n}x) \\ Y_{22,+}(\beta_{n}x) \end{pmatrix} = \dfrac{\sqrt{\pi}}{2} e^{-\frac{i \pi}{4}} \sqrt{w(\beta_{n} x)} \beta_{n}^{(n + \rho/2)\sigma_{3}} e^{\frac{n l_{n}}{2} \sigma_{3}} \sqrt{n f_{n}(x)} R(x) E_{n}(x) \begin{pmatrix} H_{\frac{1}{2}(\rho+1)}^{(1)}(n f_{n}(x)) \\[0.5em] H_{\frac{1}{2}(\rho-1)}^{(1)}(n f_{n}(x)) \end{pmatrix}.
\end{equation}
Using \cref{eq:En_def_origin} for $E_{n}(x)$, approximating $R(x) \approx \mathds{1}$, and using the asymptotic \cref{eq:Y1_12_asymptotic} for $(Y_{1})_{12}$, we get
\begin{equation}
    \setlength\arraycolsep{1pt}
    G_{n,+}(\beta_{n}x) \approx -\dfrac{2}{\beta_{n}} \dfrac{\begin{pmatrix} & H^{(1)}_{\frac{1}{2}(\rho+1)}(n f_{n}(x)) \left[e^{\frac{i \pi}{4}(2n-\rho-1)} \varphi_{+}^{\rho/2}(a_{+} + a_{+}^{-1}) + e^{-\frac{i \pi}{4}(2n-\rho-1)} \varphi_{+}^{-\rho/2}(a_{+} - a_{+}^{-1})\right] \\[0.5em] + & H^{(1)}_{\frac{1}{2}(\rho-1)}(n f_{n}(x)) \left[e^{\frac{i \pi}{4}(2n-\rho+1)} \varphi_{+}^{\rho/2}(a_{+} + a_{+}^{-1}) + e^{-\frac{i \pi}{4}(2n-\rho+1)} \varphi_{+}^{-\rho/2}(a_{+} - a_{+}^{-1})\right]\end{pmatrix}}{\begin{pmatrix} & H^{(1)}_{\frac{1}{2}(\rho+1)}(n f_{n}(x)) \left[e^{-\frac{i \pi}{4}(2n-\rho-1)} \varphi_{+}^{-\rho/2}(a_{+} + a_{+}^{-1}) + e^{\frac{i \pi}{4}(2n-\rho-1)} \varphi_{+}^{\rho/2}(a_{+} - a_{+}^{-1})\right] \\[0.5em] + & H^{(1)}_{\frac{1}{2}(\rho-1)}(n f_{n}(x)) \left[e^{-\frac{i \pi}{4}(2n-\rho+1)} \varphi_{+}^{-\rho/2}(a_{+} + a_{+}^{-1}) + e^{\frac{i \pi}{4}(2n-\rho+1)} \varphi_{+}^{\rho/2}(a_{+} - a_{+}^{-1})\right]\end{pmatrix}},
\end{equation}
where we have suppressed the arguments of $\varphi_{+} \equiv \varphi_{+}(x)$ and $a_{+} \equiv a_{+}(x)$. Let us now rescale $\beta_{n} x \mapsto \omega$. As $\omega \to 0$, both $\varphi_{+}(\omega/\beta_{n})$ and $a_{+}(\omega/\beta_{n})$ tend to $\mathcal{O}(1)$ constants, while in general the Hankel functions will have an algebraic divergence. Thus, to leading order in $\omega/\beta_{n} \ll 1$, we can approximate $\varphi_{+}(\omega/\beta_{n}) \approx \varphi_{+}(0) = i$ and $a_{+}(\omega/\beta_{n}) \approx a_{+}(0) = e^{i \pi/4}$. With these approximations, we get the simpler expression
\begin{equation}
    G_{n}(\omega + i 0^{+}) \approx -\dfrac{2}{\beta_{n}} \dfrac{\cos\left(\frac{n \pi}{2}\right) H_{\frac{1}{2}(\rho-1)}^{(1)}(\pi I_{n}(\omega)) + \sin\left(\frac{n \pi}{2}\right) H_{\frac{1}{2}(\rho+1)}^{(1)}(\pi I_{n}(\omega))}{\cos\left(\frac{n \pi}{2}\right) H_{\frac{1}{2}(\rho+1)}^{(1)}(\pi I_{n}(\omega)) - \sin\left(\frac{n \pi}{2}\right) H_{\frac{1}{2}(\rho-1)}^{(1)}(\pi I_{n}(\omega))},
\end{equation}
where we have replaced $n f_{n}(\omega/\beta_{n}) = \pi I_{n}(\omega)$. Note that the difference between even and odd $n$ is now manifest: for $-1 < \rho < 1$ we have $G_{2n}(\omega + i 0^{+}) \sim \omega^{\rho}$ as $\omega \to 0 $ while $G_{2n+1}(\omega + i 0^{+}) \sim 1/\omega^{\rho}$, consistent with the recursion relation \cref{eq:greens_func_recursion}. Taking $n$ even and approximating $I_{n}(\omega) = \int_{0}^{\omega} \sigma_{n}(\omega^{\prime}) \diff \omega^{\prime} \approx \sigma_{n}(0) \omega$, we recover \cref{eq:n_gf_asymptotic} of the main text. By the discussion surrounding \cref{eq:R_error_bound}, this expression has a multiplicative error of $\mathcal{O}(\beta_{n} / n h_{n}(0))$.

\subsubsection{Behavior near $z=\pm \beta_{n}$}
Finally let us derive an expression for the level-$n$ Green's function near the endpoints $z=\pm \beta_{n}$. We will focus on the behavior near $z=\beta_{n}$, with similar results holding near $z=-\beta_{n}$. Starting from \cref{eq:Y_unpacking_edge}, the second column of $Y$ in a neighborhood of $\beta_{n}$ is given by
\begin{equation}
    \begin{pmatrix}
        Y_{12}(\beta_{n}z) \\ Y_{22}(\beta_{n}z)
    \end{pmatrix}
    = (\beta_{n}z)^{\rho/2} e^{-n[g_{n}(z) - \phi_{n}(z) - l_{n}/2]} \beta_{n}^{(n + \rho/2)\sigma_{3}} e^{\frac{n l_{n}}{2} \sigma_{3}} R(z) E(z) \Psi(f_{n}(z)) \begin{pmatrix} 0 \\ 1 \end{pmatrix},
\end{equation}
where the Airy parametrix $\Psi(\zeta)$ was defined in \cref{eq:Psi_def}, and the map $f_{n}(z)$ was defined in \cref{eq:biholomorphic_def}. Focusing on $z$ in quadrant $I$ of \cref{fig:unity_endpoint}, this becomes
\begin{equation}
    \begin{pmatrix}
        Y_{12}(\beta_{n}z) \\ Y_{22}(\beta_{n}z)
    \end{pmatrix}
    = \sqrt{2\pi} e^{\frac{i \pi}{12}} (\beta_{n}z)^{\rho/2} e^{-n[g_{n}(z) - \phi_{n}(z) - l_{n}/2]} \beta_{n}^{(n + \rho/2)\sigma_{3}} e^{\frac{n l_{n}}{2} \sigma_{3}} R(z) E(z) \begin{pmatrix} \Ai(\xi_{n}(z)) \\ e^{\frac{4\pi i}{3}} \Ai^{\prime}(\xi_{n}(z)) \end{pmatrix},
\end{equation}
where $\xi_{n}(z) \equiv e^{4 \pi i/3} f_{n}(z)$. Using the expression \cref{eq:E_def} for $E(z)$, approximating $R(z) \approx \mathds{1}$, and using the asymptotic formula \cref{eq:Y1_12_asymptotic}, substituting into \cref{eq:Gn_from_Y} gives
\begin{equation}
    \setlength\arraycolsep{0.5pt}
    G_{n}(\beta_{n}z) \approx \dfrac{-2}{\beta_{n}} \dfrac{\begin{pmatrix}& \left[a(z)(\varphi^{\rho/2} - \varphi^{-\rho/2})(z) + a(z)^{-1}(\varphi^{\rho/2} + \varphi^{-\rho/2})(z)\right] \xi_{n}(z)^{1/4} \Ai(\xi_{n}(z)) \\ - &  \left[a(z)(\varphi^{\rho/2} + \varphi^{-\rho/2})(z) + a(z)^{-1}(\varphi^{\rho/2} - \varphi^{-\rho/2})(z)\right] \xi_{n}(z)^{-1/4} \Ai^{\prime}(\xi_{n}(z))\end{pmatrix}}{\begin{pmatrix}& \left[a(z)(\varphi^{\rho/2} - \varphi^{-\rho/2})(z) - a(z)^{-1}(\varphi^{\rho/2} + \varphi^{-\rho/2})(z)\right] \xi_{n}(z)^{1/4} \Ai(\xi_{n}(z)) \\ - & \left[a(z)(\varphi^{\rho/2} + \varphi^{-\rho/2})(z) - a(z)^{-1}(\varphi^{\rho/2} - \varphi^{-\rho/2})(z)\right] \xi_{n}(z)^{-1/4} \Ai^{\prime}(\xi_{n}(z))\end{pmatrix}},
\end{equation}
where $a(z) = (z-1)^{1/4} / (z+1)^{1/4}$ and $\varphi(z) = z + (z-1)^{1/2}(z+1)^{1/2}$. Near $z=1$, we have $f_{n}(z) \approx f_{n}^{\prime}(1) (z-1)$, where $f_{n}^{\prime}(1) = (n h_{n}(1)/\sqrt{2})^{2/3} \sim \mathcal{O}(n^{2/3})$. Then for $z = 1 + \mathcal{O}(n^{-2/3})$, we have $G_{n}(\beta_{n}z) \approx G_{n}(\beta_{n}) + \mathcal{O}(z-1)$, with
\begin{equation}
    G_{n}(\beta_{n}) \approx \dfrac{2}{\beta_{n}} \left(\dfrac{2 f_{n}^{\prime}(1)^{1/2} \Gamma\left[\frac{1}{3}\right] - \sqrt{2} (-3)^{1/3} (\rho+1) \Gamma\left[\frac{2}{3}\right]}{2 f_{n}^{\prime}(1)^{1/2} \Gamma\left[\frac{1}{3}\right] - \sqrt{2} (-3)^{1/3} (\rho-1) \Gamma\left[\frac{2}{3}\right]}\right).
\end{equation}
Since $f_{n}^{\prime}(1) \sim \mathcal{O}(n^{2/3}) \to \infty$ as $n \to \infty$, we conclude that $G_{n}(\beta_{n}) \approx 2/\beta_{n}$ to leading order as $n\to\infty$. By the discussion surrounding \cref{eq:R_error_bound_bulk}, this expression has a multiplicative error of $\mathcal{O}(1/n)$.

\clearpage
\section{Detailed analysis of the equilibrium measure}
\label{sec:endpoint_analysis}
\subsection{Behavior of the equilibrium measure near the origin}
\subsubsection{Behavior at the origin}
\begin{lemma}
    Consider $W = \exp(-Q(x))$ with $Q \in \mathrm{VSLF}(p,q)$ for $p \geq 1$.

    If $p > 1$, then for all $q \in \mathbb{R}$ we have
    \begin{equation}
        \lim_{n \to \infty} h_{n}(0) = \dfrac{2p}{p-1}.
    \end{equation}
    If $p=1$ and $q>-1$, then as $n\to \infty$ we have
    \begin{equation}
        h_{n}(0) = (\log{n})^{1+o(1)}.
    \end{equation}
    \label{lem:hn0_scaling_sup}
\end{lemma}
\begin{proof}
    \textbf{(Case: $p>1$)}
    From \cref{eq:hn_even_integral} we have
    \begin{equation}
        h_{n}(0) = \dfrac{2}{\pi} \int_{0}^{1} \dfrac{V_{n}^{\prime}(s)}{s} \dfrac{\diff s}{\sqrt{1 - s^{2}}}.
    \end{equation}
    From \cref{lem:vslf_characterization}(v) we know that $\lim_{n \to \infty} V_{n}^{\prime}(s) = (p/\lambda_{p}) s^{p-1}$, uniformly for $s$ in any compact subinterval of $(0,\infty)$, where the constant $\lambda_{p}$ is defined in \cref{eq:lambda_def}. 
This implies pointwise convergence of the integrand for $s \in (0, 1)$. In order to apply the dominated convergence theorem we appeal to \cref{cor:useful_estimate_T} 
choosing $\epsilon$ so that $p-1-\epsilon >0$. However, this upper bound is only available for $s \in [A/\beta_{n},1]$. Observe that the remaining part of the integral can be bounded by
\begin{equation}
\int_{0}^{A/\beta_{n}} \dfrac{|V_{n}^{\prime}(s)|}{s} \dfrac{\diff s}{\sqrt{1 - s^{2}}} \leq \dfrac{1}{\sqrt{1 - (A/\beta_{n})^{2}}}  \dfrac{\beta_{n}}{n} \int_{0}^{A} \dfrac{|Q^{\prime}(u)|}{u} \diff u 
\leq \mathcal{O} \left( \dfrac{\beta_{n}}{n} \right),
\end{equation}
because $Q^{\prime}(0)=0$ due to the assumed evenness of $Q$. As $\lim_{n \to \infty} \beta_{n}/n =0$ for $p > 1$ by \cref{eq:vslf_characterization_prop5} we have shown
\begin{align}
    \lim_{n \to \infty} h_{n}(0) &= \dfrac{2}{\pi} \dfrac{p}{\lambda_{p}} \int_{0}^{1} \dfrac{s^{p-2}}{\sqrt{1 - s^{2}}} \diff s,\\
     &= \dfrac{2p}{p-1},
\end{align}
where the last step follows from evaluating the integral and using the definition of $\lambda_{p}$.
    
\textbf{(Case: $p=1$)}

The proof for $p>1$ does not work for $p=1$ because the limiting integrand is no longer integrable, having a logarithmic divergence at the origin. We will therefore need to be more careful about handling this divergence. From \cref{eq:hn_even_integral}, we have the integral expression for $h_{n}(0)$,
\begin{align}
    h_{n}(0) = \dfrac{2}{\pi} \dfrac{\beta_{n}}{n} \int_{0}^{\beta_{n}} \dfrac{Q^{\prime}(u)}{u \sqrt{1-(u  / \beta_{n})^{2}}} \diff u.
    \label{eq:helper_one_T}
\end{align}
where we have changed variables for convenience. We split the integral into three parts, $I_{1}$, $I_{2}$, and $I_{3}$,
\begin{equation}
    I_1 + I_2 + I_3 \equiv \left( \int_{0}^{A_n} + \int_{A_{n}}^{\beta_{n}/2} + \int_{\beta_{n}/2}^{\beta_{n}} \right) \dfrac{Q^{\prime}(u)}{u \sqrt{1-(u / \beta_{n})^{2}}} \diff u,
    \label{eq:helper_three_T}
\end{equation}
where 
we define $A_{n}$ as
\begin{equation}
    A_{n} \coloneqq (\log{n})^{\alpha}, \quad \text{for some }  \alpha > 0 .
\end{equation}
We show that the integrals $I_3$ and $I_1$ are dominated by $I_2$. Estimating 
\begin{equation}
1 \leq \left[1-(u  / \beta_{n})^{2}\right]^{-1/2} \leq 2/\sqrt{3}, \quad \text{for all } A_n \leq u \leq \beta_{n}/2,
\end{equation}
and using \cref{lem:vslf_characterization}(iii) we obtain for $I_2$ lower and upper bounds
\begin{equation}
\Omega\left((\log{\beta_{n})^{q+1-\epsilon/2}}\right) \leq I_2 \leq \mathcal{O}\left((\log{\beta_{n})^{q+1+\epsilon/2}}\right),
\label{eq:helper_two_T}
\end{equation}
where we have chosen $\epsilon$ so small that $q+1-\epsilon/2 > 0$. The dominance of $I_2$ now follows from
\begin{equation}
| I_3 | \le 2 (\log{\beta_{n}})^{q+\epsilon/2} \int_{\beta_{n}/2}^{\beta_{n}} \dfrac{1}{ \beta_{n} \sqrt{1-(u / \beta_{n})^{2}}} \diff u = \mathcal{O}\left((\log{\beta_n})^{q+\epsilon/2}\right),
\label{eq:helper_four_T}
\end{equation}
\begin{equation}
| I_1 | = \mathcal{O}\left((\log{A_{n})^{q+1+\epsilon/2}}\right),
\label{eq:helper_five_T}
\end{equation}
where the latter estimate uses again \cref{lem:vslf_characterization}(iii) and that $Q^{\prime}(u)/u$ is bounded on any compact subset of $\mathbb{R}$ due to $Q^{\prime}(0)=0$. Note further that $\beta_{n}/n = \Theta\left(1/Q^{\prime}(\beta_{n})\right)$; see \cref{lem:vslf_characterization}(iv). Then it follows from \cref{eq:helper_one_T,eq:helper_two_T} and \cref{lem:vslf_characterization}(iii)
that
\begin{equation}
\Omega\left((\log{\beta_{n})^{1-\epsilon}}\right) \leq h_n(0) \leq \mathcal{O}\left((\log{\beta_{n})^{1+\epsilon}}\right),
\end{equation}
and the claim follows from  \cref{lem:vslf_characterization}(vi) as $\epsilon$ can be chosen arbitrarily small.
\end{proof}

\subsubsection{Uniform lower bounds near the origin}
In this section we derive uniform lower bounds on the real part of $h_{n}(z)$ for $z$ near the origin. The general strategy will be to reduce to the analysis we performed for $z = 0$.
\begin{lemma}
    Consider $Q \in \mathrm{CVSLF}(p,q,\theta,\gamma)$. Let $0 < \gamma^{\prime} < \gamma$, $0 < \theta^{\prime} < \theta$, and $\alpha^{\prime} \geq 0$ arbitrary for $p>1$ but constrained to $\alpha^{\prime} < 1+q$ for $p=1$. Define the set $\Lambda_{n}$ by
    \begin{equation}
    \Lambda_{n} \coloneqq \{ \lambda \in C_{\theta^{\prime}} : | \lambda | <  (\log{n})^{\alpha^{\prime}} \} \,  \cup \, \left\{\lambda : | \lambda | < \gamma^{\prime}\right\},
    \label{eq:set_T}
    \end{equation}
    Then, as $n\to\infty$, we have uniformly for $z \in \Lambda_{n} /\beta_{n}$ that
    \begin{equation}
        h_{n}(z) = h_{n}(0) \left[1 + o(1)\right].
    \end{equation}
    \label{lem:complex_hnz_uniform_lower}
\end{lemma}
Together with the scaling of $h_{n}(0)$ derived in \cref{lem:hn0_scaling_sup}, we conclude that, for sufficiently large $n$, one can obtain a uniform lower bound, say $h_{n}(0) / 2$, for the real part of $h_{n}(z)$ if $z \in \Lambda_{n}/\beta_{n}$.

\begin{proof}
Choose $\alpha > \alpha^{\prime}$  arbitrary for $p>1$ but still constrained to $\alpha < 1+q$ for $p=1$. Set $A_n := \left(\log{n}\right)^{\alpha}$ and denote $r(z) = (z+1)^{1/2}(z-1)^{1/2}$. By the discussion in \cref{sec:equilibrium_measure}, we may represent $h_{n}(z)$ as
\begin{equation}
    h_{n}(z) = -2 \left[\dfrac{1}{2\pi i}\left(\int_{-1}^{-A_{n}/\beta_{n}} + \int_{A_{n}/\beta_{n}}^{1}\right) \dfrac{V_{n}^{\prime}(s)}{r_{+}(s) (s-z)} \diff s + \dfrac{1/2}{2\pi i} \oint_{\Gamma_{n}/\beta_{n}} \dfrac{V_{n}^{\prime}(s)}{r(s) (s-z)} \diff s \right],
\end{equation}
where $r_{+}(s) = i \sqrt{1 - s^{2}}$, and $\Gamma_n$ are simply closed clockwise oriented curves as sketched in \cref{fig:bowtie_sketch}. Important features of the contours are the following. i) $\Gamma_n$ is contained in $C_{\theta} \cup \left\{z : |z| < \gamma\right\}$ and has a positive $n$-independent lower bound on its distance from the set $\Lambda_n$ (see \eqref{eq:set_T}) that is contained in the interior of 
$\Gamma_n$. ii) Moreover, the rightmost and leftmost points of $\Gamma_n$ are $A_n$ and $-A_n$ respectively. iii) The length of $\Gamma_n$ is bounded by $\mathcal{O}(A_n)$.
\begin{figure}[t]
    \begin{center}
        \tikzset{every picture/.style={line width=0.75pt}} %

\begin{tikzpicture}[x=0.75pt,y=0.75pt,yscale=-1,xscale=1]

\draw  [draw opacity=0][fill={rgb, 255:red, 222; green, 222; blue, 222 }  ,fill opacity=1 ] (342.49,212.95) .. controls (338.19,220.74) and (330.02,226) .. (320.65,226) .. controls (311.29,226) and (303.12,220.74) .. (298.82,212.95) -- (320.65,200.34) -- cycle ; \draw   (342.49,212.95) .. controls (338.19,220.74) and (330.02,226) .. (320.65,226) .. controls (311.29,226) and (303.12,220.74) .. (298.82,212.95) ;  
\draw  [draw opacity=0][fill={rgb, 255:red, 222; green, 222; blue, 222 }  ,fill opacity=1 ] (298.81,188.02) .. controls (303.11,180.22) and (311.27,174.97) .. (320.64,174.97) .. controls (330.01,174.97) and (338.18,180.22) .. (342.48,188.02) -- (320.64,200.62) -- cycle ; \draw   (298.81,188.02) .. controls (303.11,180.22) and (311.27,174.97) .. (320.64,174.97) .. controls (330.01,174.97) and (338.18,180.22) .. (342.48,188.02) ;  
\draw  [draw opacity=0][fill={rgb, 255:red, 222; green, 222; blue, 222 }  ,fill opacity=1 ] (320.65,200.34) -- (385.61,162.61) -- (385.61,238.07) -- cycle ;
\draw  [draw opacity=0][fill={rgb, 255:red, 222; green, 222; blue, 222 }  ,fill opacity=1 ] (385.6,163.32) .. controls (404.02,169.67) and (415.77,188.18) .. (412.58,207.62) .. controls (410.2,222.14) and (400.08,233.5) .. (387,238.38) -- (372.23,201) -- cycle ; \draw   (385.6,163.32) .. controls (404.02,169.67) and (415.77,188.18) .. (412.58,207.62) .. controls (410.2,222.14) and (400.08,233.5) .. (387,238.38) ;  
\draw  [draw opacity=0][fill={rgb, 255:red, 222; green, 222; blue, 222 }  ,fill opacity=1 ] (255.54,238.2) .. controls (237.18,231.7) and (225.59,213.08) .. (228.94,193.67) .. controls (231.52,178.76) and (242.25,167.26) .. (255.91,162.74) -- (269.23,200.63) -- cycle ; \draw   (255.54,238.2) .. controls (237.18,231.7) and (225.59,213.08) .. (228.94,193.67) .. controls (231.52,178.76) and (242.25,167.26) .. (255.91,162.74) ;  
\draw  [draw opacity=0][fill={rgb, 255:red, 222; green, 222; blue, 222 }  ,fill opacity=1 ] (320.5,200.47) -- (255.54,238.2) -- (255.54,162.74) -- cycle ;
\draw  (150.75,200.04) -- (491.75,200.04)(321.25,96) -- (321.25,300) (484.75,195.04) -- (491.75,200.04) -- (484.75,205.04) (316.25,103) -- (321.25,96) -- (326.25,103)  ;
\draw [fill={rgb, 255:red, 222; green, 222; blue, 222 }  ,fill opacity=1 ]   (342.04,188.12) -- (385.6,163.32) ;
\draw [fill={rgb, 255:red, 222; green, 222; blue, 222 }  ,fill opacity=1 ]   (298.77,187.62) -- (255.91,162.74) ;
\draw [fill={rgb, 255:red, 222; green, 222; blue, 222 }  ,fill opacity=1 ]   (299.1,213.4) -- (255.54,238.2) ;
\draw [fill={rgb, 255:red, 222; green, 222; blue, 222 }  ,fill opacity=1 ]   (342.31,212.69) -- (387,238.38) ;
\draw  [draw opacity=0] (341.99,236.95) .. controls (337.69,244.74) and (329.52,250) .. (320.15,250) .. controls (310.79,250) and (302.62,244.74) .. (298.32,236.95) -- (320.15,224.34) -- cycle ; \draw  [color={rgb, 255:red, 57; green, 114; blue, 249 }  ,draw opacity=1 ] (341.99,236.95) .. controls (337.69,244.74) and (329.52,250) .. (320.15,250) .. controls (310.79,250) and (302.62,244.74) .. (298.32,236.95) ;  
\draw  [draw opacity=0] (235.07,275.26) .. controls (198.88,259.79) and (177.56,221) .. (186.22,182.04) .. controls (192.14,155.44) and (210.68,134.82) .. (234.37,124.85) -- (268.03,200.22) -- cycle ; \draw  [color={rgb, 255:red, 57; green, 114; blue, 249 }  ,draw opacity=1 ] (235.07,275.26) .. controls (198.88,259.79) and (177.56,221) .. (186.22,182.04) .. controls (192.14,155.44) and (210.68,134.82) .. (234.37,124.85) ;  
\draw [color={rgb, 255:red, 57; green, 114; blue, 249 }  ,draw opacity=1 ]   (298.32,236.95) -- (235.25,274.5) ;
\draw  [draw opacity=0] (298.98,163.35) .. controls (303.26,155.55) and (311.42,150.28) .. (320.79,150.25) .. controls (330.15,150.23) and (338.33,155.47) .. (342.65,163.25) -- (320.84,175.91) -- cycle ; \draw  [color={rgb, 255:red, 57; green, 114; blue, 249 }  ,draw opacity=1 ] (298.98,163.35) .. controls (303.26,155.55) and (311.42,150.28) .. (320.79,150.25) .. controls (330.15,150.23) and (338.33,155.47) .. (342.65,163.25) ;  
\draw [color={rgb, 255:red, 57; green, 114; blue, 249 }  ,draw opacity=1 ]   (298.98,163.35) -- (234.25,125) ;
\draw  [draw opacity=0] (405.71,124.56) .. controls (441.95,139.9) and (463.42,178.61) .. (454.9,217.6) .. controls (449.08,244.23) and (430.62,264.91) .. (406.96,274.98) -- (373.03,199.72) -- cycle ; \draw  [color={rgb, 255:red, 57; green, 114; blue, 249 }  ,draw opacity=1 ] (405.71,124.56) .. controls (441.95,139.9) and (463.42,178.61) .. (454.9,217.6) .. controls (449.08,244.23) and (430.62,264.91) .. (406.96,274.98) ;  
\draw [color={rgb, 255:red, 57; green, 114; blue, 249 }  ,draw opacity=1 ]   (342.6,163.11) -- (405.53,125.33) ;
\draw [color={rgb, 255:red, 57; green, 114; blue, 249 }  ,draw opacity=1 ]   (342.21,236.71) -- (407.08,274.82) ;
\draw  [draw opacity=0][dash pattern={on 4.5pt off 4.5pt}] (298.81,152.02) .. controls (303.11,144.22) and (311.27,138.97) .. (320.64,138.97) .. controls (330.01,138.97) and (338.18,144.22) .. (342.48,152.02) -- (320.64,164.62) -- cycle ; \draw  [dash pattern={on 4.5pt off 4.5pt}] (298.81,152.02) .. controls (303.11,144.22) and (311.27,138.97) .. (320.64,138.97) .. controls (330.01,138.97) and (338.18,144.22) .. (342.48,152.02) ;  
\draw  [dash pattern={on 4.5pt off 4.5pt}]  (342.48,152.02) -- (434,93) ;
\draw  [dash pattern={on 4.5pt off 4.5pt}]  (298.81,152.02) -- (206.5,95) ;

\draw  [draw opacity=0][dash pattern={on 4.5pt off 4.5pt}] (342.3,250.99) .. controls (337.92,258.73) and (329.7,263.91) .. (320.33,263.81) .. controls (310.96,263.71) and (302.85,258.37) .. (298.63,250.53) -- (320.6,238.15) -- cycle ; \draw  [dash pattern={on 4.5pt off 4.5pt}] (342.3,250.99) .. controls (337.92,258.73) and (329.7,263.91) .. (320.33,263.81) .. controls (310.96,263.71) and (302.85,258.37) .. (298.63,250.53) ;  
\draw  [dash pattern={on 4.5pt off 4.5pt}]  (298.63,250.53) -- (206.5,308.59) ;
\draw  [dash pattern={on 4.5pt off 4.5pt}]  (342.3,250.99) -- (434.01,308.97) ;

\draw  [color={rgb, 255:red, 57; green, 114; blue, 249 }  ,draw opacity=1 ] (441.95,145.74) -- (442.78,155.84) -- (434.2,150.44) ;
\draw   (452.95,196.55) -- (461.05,204.66)(461.05,196.55) -- (452.95,204.66) ;
\draw [color={rgb, 255:red, 0; green, 0; blue, 0 }  ,draw opacity=0.39 ]   (320.65,200.34) -- (329.5,178.36) ;
\draw [shift={(330.25,176.5)}, rotate = 111.92] [color={rgb, 255:red, 0; green, 0; blue, 0 }  ,draw opacity=0.39 ][line width=0.75]    (6.56,-1.97) .. controls (4.17,-0.84) and (1.99,-0.18) .. (0,0) .. controls (1.99,0.18) and (4.17,0.84) .. (6.56,1.97)   ;
\draw   (408.95,195.55) -- (417.05,203.66)(417.05,195.55) -- (408.95,203.66) ;
\draw   (179.95,195.55) -- (188.05,203.66)(188.05,195.55) -- (179.95,203.66) ;
\draw  [color={rgb, 255:red, 57; green, 114; blue, 249 }  ,draw opacity=1 ] (192.74,244.15) -- (192.18,234.03) -- (200.61,239.65) ;

\draw (446.5,135.9) node [anchor=north west][inner sep=0.75pt]    {$\Gamma _{n}$};
\draw (457.5,180) node [anchor=north west][inner sep=0.75pt]  [font=\footnotesize]  {$A_{n} =(\log n)^{\alpha }$};
\draw (327.25,130) node [anchor=north west][inner sep=0.75pt]  [font=\scriptsize]  {$\gamma $};
\draw (326.34,183.31) node [anchor=north west][inner sep=0.75pt]  [font=\scriptsize]  {$\gamma ^{\prime }$};
\draw (412.75,206.9) node [anchor=north west][inner sep=0.75pt]  [font=\scriptsize]  {$(\log n)^{\alpha ^{\prime }}$};
\draw (373.75,177) node [anchor=north west][inner sep=0.75pt]    {$\Lambda _{n}$};
\draw (158.25,200.9) node [anchor=north west][inner sep=0.75pt]  [font=\footnotesize]  {$-A_{n}$};

\end{tikzpicture}
    \end{center}
    \caption{Sketch of the contour $\Gamma_n$, and the sets $\Lambda_n$ and $C_{\theta^{\prime}} \cup \left\{z : |z| < \gamma\right\}$.}
    \label{fig:bowtie_sketch} 
\end{figure}

First we employ the symmetries of $Q$ and $r$ and change variables in the integrals to $u = \beta_{n} s$. Denoting $\lambda \equiv  \beta_{n} z$ we obtain
\begin{equation}
    h_{n}(z) = \dfrac{2}{\pi} \dfrac{\beta_{n}}{n} \left[ \int_{A_n}^{\beta_{n}} \dfrac{u Q^{\prime}(u)}{(u^{2} - \lambda^{2}) \sqrt{1 - (u/{\beta_{n})^{2}}}} \diff u +
     \dfrac{i}{4} \oint_{\Gamma_{n}} \dfrac{Q^{\prime}(u)}{r(u/\beta_{n}) (u-\lambda)} \diff u \right] .
    \label{eq:hnz_double_integral}
\end{equation}
In order to reduce to the $\lambda = 0$ case, we want to examine the effect of replacing $1/(u^{2} - \lambda^{2}) \to 1/u^{2}$. To that end, we consider the difference
\begin{equation}
    \dfrac{1}{u^{2} - \lambda^{2}} - \dfrac{1}{u^{2}} = \dfrac{\lambda^{2}}{u^{2}(u^{2} - \lambda^{2})},
\end{equation}
so that
\begin{equation}
  \dfrac{\pi}{2} \dfrac{n}{\beta_{n}} h_{n}(z) =  \int_{A_n}^{\beta_{n}} \dfrac{Q^{\prime}(u)}{u \sqrt{1 - (u/{\beta_{n})^{2}}}} \left(1 +\dfrac{\lambda^{2}}{u^{2} - \lambda^{2}}\right)\diff u +
     \dfrac{i}{4} \oint_{\Gamma_{n}} \dfrac{Q^{\prime}(u)}{r(u/\beta_{n}) (u-\lambda)} \diff u \equiv J_1 + J_2.
\end{equation}
Recall from \cref{{eq:helper_one_T},{eq:helper_three_T}} that 
\begin{equation}
    \dfrac{\pi}{2} \dfrac{n}{\beta_{n}} h_n(0) = I_1 + I_2 + I_3 , 
 \end{equation}
 a representation that also holds in the case $p >1$. We therefore need to establish $(J_1 + J_2)/(I_1+I_2+I_3) =1 + o(1)$: We prove this by verifying that the following limits for $n \to \infty$ hold uniformly in $z$:
 \begin{equation}
 \dfrac{I_1}{I_1+I_2+I_3} \to 0, \quad  \dfrac{J_1}{I_2+I_3} \to 1, \quad  \dfrac{J_2}{I_2+I_3} \to 0 .
 \end{equation}
The second claim follows from the observation that for $(\log{n})^{\alpha} = A_{n} \leq u \leq \beta_{n}$ and $| \lambda | \leq (\log{n})^{\alpha^{\prime}}$ one has
\begin{equation}
\dfrac{\lambda^{2}}{u^{2} - \lambda^{2}} \leq \mathcal{O}\left( (\log{n})^{-2(\alpha-\alpha^{\prime})}\right) .
\end{equation}
The first claim has already been established in the case $p=1$, see  \cref{eq:helper_two_T,eq:helper_four_T,eq:helper_five_T}. In the case $p>1$, it follows from \cref{lem:vslf_characterization}(iii) and from the boundedness of $Q^{\prime}(u)/u$ on bounded sets that $I_1=\mathcal{O}\left( (\log{n})^{\delta_{1}}\right)$, whereas by  \cref{lem:vslf_characterization}(vi) and \cref{lem:hn0_scaling_sup} we have
$I_1+I_2+I_3=\Omega\left(n/\beta_{n}\right)=\Omega\left(n^{\delta_{2}}\right)$, for some positive numbers $\delta_{1}, \delta_{2}$. Consequently $I_1/(I_1+I_2+I_3) \to 0$ as $n \to \infty$.

Finally, we turn to the third claim. Observe that by the choice of the path of integration $\Gamma_n$ the denominator of the integrand of the term $J_2$ is uniformly bounded away from $0$.
Using \cref{eq:complex_bound_T} we obtain the sup-norm bound $\mathcal{O}\left(A_{n}^{p-1} (\log{A_{n}})^{q+\epsilon/2}\right)$ on the integrand yielding
$J_2=\mathcal{O}\left(A_{n}^{p} (\log{A_{n}})^{q+\epsilon/2}\right)$. For $p > 1$ we have $I_2+I_3=\Omega(n/\beta_{n})$ which establishes $J_2/(I_2+I_3) \to 0$. In the case $p=1$
we learn from \cref{eq:helper_two_T,eq:helper_four_T} that $I_2+I_3 = \Omega\left((\log{\beta_{n})^{q+1-\epsilon/2}}\right)$ and the third claim follows from $\alpha -(q+1) < 0$ by choosing $\epsilon$ small enough. This completes the proof.
\end{proof}

\begin{remark}
    By using a similar contour to that depicted in \cref{fig:bowtie_sketch}, one can combine the definition \cref{eq:hn_def} and \cref{lem:vslf_characterization,lem:Qprime_complex_bound}, in order to establish the following upper bounds on $|h_{n}(z)|$:
    \begin{equation}
        |h_{n}(z)| \leq \begin{cases}
            \mathcal{O}\big(\log{n}\big), & p>1; \\
            \mathcal{O}\big((\log{n})^{1+\epsilon}\big), & p=1,
        \end{cases}
    \end{equation}
    where $\epsilon > 0$ can be made arbitrarily small. These bounds are uniformly valid for $z = x + i y$ with $|x| < 1 - \delta$ for arbitrary but fixed $\delta > 0$, and $z$ in a suitable $\Theta(|x|)$-sized neighborhood of $x$.
    \label{rem:hnz_upper_bound}
\end{remark}

\subsection{Uniform lower bounds for the equilibrium measure on the real line}
\label{sec:poseqmeasureproof}
Let us now prove \cref{lem:hn_positive_x_g1} from \cref{sec:equilibrium_measure}, which we restate here.
\firstposeqmeasure*
\begin{proof}
    Since $h_{n}(x)$ is even, it suffices to consider $x \geq 0$. The complexity of the proof depends on the magnitude of $x$, and whether it depends on $n$.
    Let $M>0$ be given. It is our goal to construct a corresponding constant $C>0$ uniformly for sufficiently large values of $n$ as stated in the lemma.
    
    \vspace{\baselineskip}
    \textbf{Case 1 (near origin): $x  \in [0, (\log{n})^{\alpha} / \beta_{n}]$, with $\alpha > 0$ for $p > 1$ and $\alpha < 1+q$ for $p=1$}
    \vspace{\baselineskip}

    In this near-zero regime, we can apply \cref{lem:complex_hnz_uniform_lower} to conclude that $h_{n}(x) = h_{n}(0) [1 + o(1)]$ uniformly in $x$, where the $o(1)$ refers to scaling with $n$. Then the asymptotics of $h_{n}(0)$ from \cref{lem:hn0_scaling_sup} suffice to give a uniform positive lower bound on $h_{n}(x)$ in this region. 
    Thus, any choice, say $0 <C \leq1$ would work for this region.

\vspace{\baselineskip}
   \textbf{Case 2 (bulk): $x  \in [(\log{n})^{\alpha} / \beta_{n} ,M]$, with $\alpha$ as in case 1}
    \vspace{\baselineskip}

    By \cref{lem:vslf_characterization}(i), there exists $A > 0$ such that $u Q^{\prime}(u)$ is increasing for $u > A$. This implies $s V_{n}^{\prime}(s)$ is increasing for $s > A/\beta_{n}$, and therefore in this region of increase $x V_{n}^{\prime}(x) - s V_{n}^{\prime}(s)$ has the same sign as $x - s$. If we split the integral \cref{eq:hn_even_integral} for $h_{n}(x)$ into two regions,
    \begin{align}
        h_{n}(x) &= \dfrac{2}{\pi}\left( \int_{0}^{A/\beta_{n}} + \int_{A/\beta_{n}}^{1} \right) \dfrac{x V_{n}^{\prime}(x) - s V_{n}^{\prime}(s)}{x^{2} - s^{2}} \dfrac{\diff s}{\sqrt{1-s^{2}}},\\
        &\equiv I_{1} + I_{2}, \nonumber
    \end{align}
    the second integral $I_{2}$ is an integral of a manifestly positive integrand, so is positive. Let us first show that $I_{2}$ can be lower bounded by a positive constant. Since $I_{2}$ has a positive integrand, we can lower bound it by truncating the region of integration. In addition, we learn from \cref{lem:vslf_characterization}(v) 
    that for all $\epsilon > 0$ there exists $n_0 \in \mathbb N$ such that for all $n \geq n_0$ and $s \in [1/4, 1]$  we have
    \begin{equation}
    \left| s V_{n}^{\prime}(s) - \dfrac{p}{\lambda_{p}} s^{p} \right| \leq \dfrac{\epsilon}{2} .
    \end{equation}
   Using in addition the above mentioned monotonicity of $s V_{n}^{\prime}(s)$ we obtain for $(\log{n})^{\alpha} / \beta_{n} \leq x \leq1/2$ that
   \begin{equation}
   \dfrac{\pi}{2} I_2 \geq 
   \int_{3/4}^{1} \dfrac{s V_{n}^{\prime}(s) - x V_{n}^{\prime}(x)}{s^{2} - x^{2}} \dfrac{\diff s}{\sqrt{1-s^{2}}}
   \geq \int_{3/4}^{1} \left[\dfrac{3}{4}V_{n}^{\prime}(3/4) - \dfrac{1}{2}V_{n}^{\prime}(1/2) \right] \diff s \geq  \dfrac{1}{4} \left[ \dfrac{p}{\lambda_{p}}\left( \left( \dfrac{3}{4}\right)^p - 
   \left( \dfrac{1}{2}\right)^p\right) -\epsilon \right],
   \end{equation}
   and for $x \in [1/2, M]$ that
     \begin{equation}
   \dfrac{\pi}{2} I_2 \geq 
   \int_{1/8}^{1/4} \dfrac{x V_{n}^{\prime}(x) - s V_{n}^{\prime}(s)}{x^{2} - s^{2}} \dfrac{\diff s}{\sqrt{1-s^{2}}}
   \geq \dfrac{1}{x^{2}} \int_{1/8}^{1/4} \left[\dfrac{1}{2}V_{n}^{\prime}(1/2) - \dfrac{1}{4}V_{n}^{\prime}(1/4) \right] \diff s \geq  \dfrac{1}{8 x^{2}} \left[ \dfrac{p}{\lambda_{p}}\left( \left( \dfrac{1}{2}\right)^p - 
   \left( \dfrac{1}{4}\right)^p\right) -\epsilon \right].
   \end{equation}
  From these estimates it follows that for any $p \geq 1$ one can choose $\epsilon > 0$ and an $M$-dependent constant $C>0$ such that for all $n \geq n_0$ we have $I_2 \geq 2C$. In order to complete the proof we need to show that 
 \begin{equation}
 | I_1 | \leq \dfrac{I_2}{2} ,
 \label{eq:final_claim_T}
 \end{equation} 
  for all sufficiently large values of $n$. Let us start upper-bounding the modulus of $I_1$ by observing
    \begin{align}
        I_{1} &\equiv \dfrac{2}{\pi} \int_{0}^{A/\beta_{n}} \dfrac{x V_{n}^{\prime}(x) - s V_{n}^{\prime}(s)}{x^{2} - s^{2}} \dfrac{\diff s}{\sqrt{1-s^{2}}},\\
        &= \dfrac{2}{\pi} \dfrac{\beta_{n}}{n} \int_{0}^{A} \dfrac{v Q^{\prime}(v) - u Q^{\prime}(u)}{v^{2} - u^{2}} \dfrac{\diff u}{\sqrt{1 - (u/\beta_{n})^{2}}},\nonumber
    \end{align}
    where $v \equiv \beta_{n} x \in [(\log{n})^{\alpha}, \beta_{n} M]$, and we rescaled the integration variable for convenience. With $f_{Q}(y) \coloneqq \frac{\mathrm{d}}{\mathrm{d} y}[y Q^{\prime}(y)] = Q^{\prime}(y) + y Q^{\prime \prime}(y)$, we have
    \begin{equation}
        \dfrac{v Q^{\prime}(v) - u Q^{\prime}(u)}{v^{2} - u^{2}} = \dfrac{1}{v+u} \int_{0}^{1} f_{Q}[v - t(v-u)] \diff t.
        \label{eq:fQ_rewriting}
    \end{equation}
    The argument of $f_{Q}$ is always between $u$ and $v$, and since $u \in [0,A]$ while $v \to \infty$, this argument is always in $[0,v]$. For bounded arguments, we know $f_{Q}(u)$ is bounded due to our analyticity assumption on $Q$. For growing arguments, note that the condition \cref{eq:vsf_prop3} on $Q$ implies that
    \begin{align}
        f_{Q}(y) &= Q^{\prime}(y) \left[1 + \dfrac{y Q^{\prime\prime}(y)}{Q^{\prime}(y)}\right], \\
        &\leq Q^{\prime}(y) \times 2p \mathclap{\hspace{10em} \text{for large enough } y,}\label{eq:fQ_rewriting2_T}
    \end{align}
    with $Q^{\prime}(y) > 0$ also for large enough $y$, given condition \cref{eq:vsf_prop1}. Again for sufficiently large $y$, from \cref{eq:vslf_characterization_prop3} we have the bound $Q^{\prime}(y) \leq y^{p-1} (\log{y})^{q + \epsilon}$ for $\epsilon > 0$. For $p > 1$ this bound increases with $y$, while for $p=1$ it can increase or decrease depending on the sign of $q$. Taken together, these considerations of both bounded and growing arguments imply a bound, valid for $t \in [0,1]$ and $u \in [0,A]$, of the form
    \begin{equation}
        \Big|f_{Q}[v - t(v-u)]\Big| \leq c \, v^{p-1} (\log{v})^{\eta_{q}}
    \end{equation}
    for some constant $c > 0$, where we have defined 
    \begin{equation}
        \eta_{q} \coloneqq \begin{cases}
            q+\epsilon, & q \geq 0; \\
            0, & q < 0,
        \end{cases}
    \end{equation}
    where $\epsilon > 0$ can be made small by taking $n$ large. Then we have the bound
    \begin{align}
        |I_{1}| &\leq \dfrac{2}{\pi} \dfrac{\beta_{n}}{n} c \, v^{p-1} (\log{v})^{\eta_{q}}  \times \int_{0}^{A} \dfrac{1}{v+u} \dfrac{\diff u}{\sqrt{1 - (u/\beta_{n})^{2}}}.
    \end{align}
    The integral can be upper bounded simply as
    \begin{equation}
        \int_{0}^{A} \dfrac{1}{v+u} \dfrac{\diff u}{\sqrt{1 - (u/\beta_{n})^{2}}} \leq \dfrac{1}{\sqrt{1 - (A/\beta_{n})^{2}}} \int_{0}^{A} \dfrac{1}{v+u} \diff u = \dfrac{1}{\sqrt{1 - (A/\beta_{n})^{2}}}\log\left(1 + \dfrac{A}{v}\right) = \mathcal{O}\left(\dfrac{1}{v}\right),
    \end{equation}
    since $A/v \to 0$ and $A/\beta_{n} \to 0$ as $n \to \infty$. Thus we conclude that
    \begin{equation}
        |I_{1}| \leq \mathcal{O}\left(\dfrac{\beta_{n}}{n} v^{p-2} (\log{v})^{\eta_{q}} \right).
        \label{eq:uniform_lower_I1_conditional}
    \end{equation}
    We now analyze this on a case by case basis.

    \textbf{Case $\mathcal{A}$: $p>1$}

    Using $v = \mathcal{O}(\beta_{n})$ and the asymptotics for $\beta_{n}$ in \cref{eq:vslf_characterization_prop5}, for $p \geq 2$ we have $v^{p-2} =  \mathcal{O}(\beta_{n}^{p-2})$, and so
    \begin{equation}
        |I_{1}| \leq \wt{\mathcal{O}}\left(\dfrac{\beta_{n}^{p-1}}{n}\right) = \wt{\mathcal{O}}\left(n^{-1/p}\right),
    \end{equation}
    where the $\wt{\mathcal{O}}$ notation hides polylogarithmic factors which are unimportant for $p > 1$. Equally, for $1 < p < 2$, $v^{p-2}$ is $o(1)$, and hence in this case
    \begin{equation}
        |I_{1}| \leq \wt{\mathcal{O}}(\beta_{n} / n) = \wt{\mathcal{O}}(n^{-(p-1)/p}).
    \end{equation}
    Overall we conclude for all $p>1$ that $|I_{1}|$ vanishes as $n \to \infty$ and condition~\eqref{eq:final_claim_T} is satisfied for $n$ large enough.
    
    \textbf{Case $\mathcal{B}$: $p=1$}

    In this marginal case, we need to refine our lower bound on $I_{2}$. First we rewrite the integrand of $I_2$ in terms of $f_{Q}$ as in \cref{eq:fQ_rewriting}, but this time we use the bound
    \begin{align}
        f_{Q}(y) &= Q^{\prime}(y) \left[1 + \dfrac{y Q^{\prime\prime}(y)}{Q^{\prime}(y)}\right],\\
        &\geq Q^{\prime}(y) \times (p/2) \mathclap{\hspace{10em} \text{for large enough } y,}
        \label{eq:fQ_lower_bound}
    \end{align}
    which again follows from the condition \cref{eq:vsf_prop3} on $Q$. Let $A^{\prime}\geq A$ be so large that we can apply \cref{eq:fQ_lower_bound} and the lower bound in \cref{eq:vslf_characterization_prop3} on $Q^{\prime}$ throughout the slightly reduced domain of integration $ [A^{\prime}, \beta_{n}]$. Then, as $p=1$, we have the lower bound
    \begin{equation}
        I_{2} \geq \dfrac{1}{\pi} \dfrac{\beta_{n}}{n} \int_{A^{\prime}}^{\beta_{n}} \int_{0}^{1} \dfrac{\left(\log\left[v - t(v-u)\right]\right)^{q - \epsilon}}{v+u} \diff t \, \diff u.
        \label{eq:uniform_lower_I2_log}
    \end{equation}
    We can take $A^{\prime}$ large enough that the integrand is positive, such that we can lower bound $I_{2}$ by restricting the $t$ integral domain. The argument $v-t(v-u)$ of the logarithm linearly interpolates between $u$ and $v$, and if $q-\epsilon$ is positive (negative), we restrict to the half of the $t$-integral where $v-t(v-u)$ is lower (upper) bounded by its value $(v+u)/2$ at the midpoint. Evaluating the integrals then gives the lower bound
    \begin{equation}
        I_{2} \geq \Omega\left(\dfrac{\beta_{n}}{n} \left[\big(\log(v+u)\big)^{1+q-\epsilon}\right]_{u=A^{\prime}}^{u=\beta_{n}}\right).
        \label{eq:I2_lower_bound}
    \end{equation}
    If $v \leq \mathcal{O}(\beta_{n}^{\delta})$ for some $\delta<1$, then this bound is of the order $\Omega((\beta_{n}/n) (\log{\beta_{n}})^{1+q-\epsilon})$ since the upper and lower limits do not cancel to leading order. In that case, using \cref{eq:uniform_lower_I1_conditional} to bound $|I_{1}|$, we have
    \begin{equation}
        \dfrac{|I_{1}|}{I_{2}} \leq \mathcal{O}\left(\dfrac{1}{v} \dfrac{1}{(\log{\beta_{n}})^{1+q-\epsilon-\eta_{q}}}\right).
    \end{equation}
    In the region $v \in [(\log{n})^{\alpha}, \beta_{n}^{2/3}]$, we can bound $1/v$ by $1/(\log{n})^{\alpha}$, so then
    \begin{equation}
        \dfrac{|I_{1}|}{I_{2}} \leq \mathcal{O}\left(\dfrac{1}{(\log{\beta_{n}})^{1+q-\epsilon-\eta_{q}+\alpha}}\right).
    \end{equation}
    For $q \geq 0$, the exponent is bounded by $1 + \alpha - 2\epsilon$, while for $q < 0$, it is bounded by $1 + q + \alpha - \epsilon$, which shows that in both cases the ratio goes to zero as $n \to \infty$. Equally, in the region $v \in [\beta_{n}^{2/3}, M\beta_{n}]$, the $I_{2}$ lower bound from \cref{eq:I2_lower_bound} becomes $\Omega((\beta_{n}/n) (\log{\beta_{n}})^{q-\epsilon})$ because the leading term cancels, while the $1/v$ from the $I_{1}$ bound is at most $1/\beta_{n}^{2/3}$, so then 
    \begin{equation}
        \dfrac{|I_{1}|}{I_{2}} \leq \mathcal{O}\left(\dfrac{1}{\beta_{n}^{2/3}} \dfrac{1}{(\log{\beta_{n}})^{q-\epsilon-\eta_{q}}}\right),
    \end{equation}
    which also goes to zero as $n\to\infty$. Thus \cref{eq:final_claim_T} holds in all cases considered.

   \end{proof}
    
   \begin{remark}
        Combining \cref{lem:complex_hnz_uniform_lower,eq:I2_lower_bound}, it follows that for $p=1$ and $|x| < \beta_{n}^{-1/3}$, we have the uniform lower bound
        \begin{equation}
            h_{n}(x) \geq \Omega\big((\log{n})^{1-\epsilon}\big),
        \end{equation}
        where $\epsilon > 0$ can be made arbitrarily small by taking $n$ large. (This bound remains true for $|x|< \beta_{n}^{-\delta}$ with $\delta > 0$.)
        \label{rem:hnz_lower_bound}
   \end{remark}

\subsection{Behavior of the equilibrium measure near the endpoints}
\subsubsection{Behavior at the endpoint}
\begin{lemma}
    For $Q \in \mathrm{VSLF}(p,q)$ with $p>0$, we have
    \begin{equation}
        \lim_{n\to\infty} h_{n}(1) = 2p.
    \end{equation}
    \label{lem:hn1_scaling_T}
\end{lemma}
\begin{proof}
    From the integral form \cref{eq:hn_even_integral} we have
    \begin{equation}
        h_{n}(1) = \dfrac{2}{\pi} \int_{0}^{1} \dfrac{V_{n}^{\prime}(1) - s V_{n}^{\prime}(s)}{1 - s^{2}} \dfrac{\diff s}{\sqrt{1 - s^{2}}}.
    \end{equation}
    Similar to the proof of \cref{lem:hn0_scaling_sup} for $p>1$, we will use the result from \cref{lem:vslf_characterization}(v) that $\lim_{n\to\infty} s V_{n}^{\prime}(s) = (p/\lambda_{p}) s^{p}$, uniformly for $s$ in compact subsets of $(0,\infty)$, which proves pointwise convergence of the integrand for $s \in (0,1)$. In order to apply the dominated convergence theorem we argue differently for the integration domains $[0, 1/2]$ and $[1/2, 1]$. Choosing $A>0$ according to \cref{lem:vslf_characterization}(i) so that $sV_{n}^{\prime}(s)$ is increasing on $[A/\beta_{n}, 1/2]$, the integrand is upper bounded by $(4/3)^{3/2}V_{n}^{\prime}(1) \leq \mathcal{O}(1)$ on this interval. Since
    \begin{equation}
     \int_{0}^{A/\beta_{n}} s V_{n}^{\prime}(s) \diff s = \dfrac{1}{n \beta_{n}} \int_{0}^{A} u Q^{\prime}(u) \diff u \to 0,
    \end{equation}
    as $n\to \infty$, we may interchange the $n$-limit and integration on the domain $[0, 1/2]$. For the remaining part of the integral it suffices to show that 
     \begin{equation}
      \dfrac{V_{n}^{\prime}(1) - s V_{n}^{\prime}(s)}{1 - s} \leq \mathcal{O}(1),
      \end{equation}
    uniformly for $s \in [1/2, 1]$. Here we use the function $f_{Q}$ introduced in the proof of \cref{lem:hn_positive_x_g1} and the relations presented in \cref{eq:fQ_rewriting,eq:fQ_rewriting2_T}. We obtain
    \begin{equation}
    \dfrac{V_{n}^{\prime}(1) - s V_{n}^{\prime}(s)}{1 - s} = \dfrac{\beta_{n}}{n} \dfrac{\beta_{n}Q^{\prime}(\beta_{n}) - \beta_{n}s Q^{\prime}(\beta_{n}s)}{\beta_{n} - \beta_{n}s}
    \leq 2p \dfrac{\beta_{n}}{n} \int_{0}^{1} Q^{\prime} (\beta_{n}[1-t(1-s)]) \diff s.
    \end{equation}
    Applying \cref{cor:useful_estimate_T} with $\epsilon=1$ and keeping  in mind that $1-t(1-s) \in [s, 1] \subset [1/2, 1]$ for $0 \leq t \leq 1$ we conclude for sufficiently large $n$ that
     \begin{equation}
    \dfrac{V_{n}^{\prime}(1) - s V_{n}^{\prime}(s)}{1 - s} 
    \leq 2p \dfrac{2p}{\lambda_{p}} \int_{0}^{1}  [1-t(1-s)]^{p-2} \diff s \leq \mathcal{O}(1).
    \end{equation}
   In summary we may conclude that
    \begin{equation}
        \lim_{n \to \infty} h_{n}(1) = \dfrac{2}{\pi} \dfrac{p}{\lambda_{p}} \int_{0}^{1} \dfrac{1 - s^{p}}{1-s^{2}}\dfrac{\diff s}{\sqrt{1 - s^{2}}} = 2p,
    \end{equation}
    where the last step follows from evaluating the integral and using the definition of $\lambda_{p}$ in \cref{eq:lambda_def}.
\end{proof}

\subsubsection{Uniform lower bounds near $z=\pm 1$}
Let us now prove a uniform lower bound on $h_{n}(z)$ for $z$ in a neighbourhood of $z = 1$, with the case $z=-1$ following by symmetry. This will follow by combining the following lemma about uniform continuity with \cref{lem:hn1_scaling_T}, which shows that $h_{n}(1) = 2p + o(1)$ as $n \to \infty$.
\begin{lemma}
    Let $Q \in \mathrm{CVSLF}(p,q,\theta,\gamma)$ and $p > 0$. Then there exist constants $C, \varrho > 0$ and $n_{0} \in \mathbb{N}$ such that, for all sufficiently large $n\geq n_{0}$, we have
    \begin{equation}
        \left| h_{n}(z) - h_{n}(1) \right| \leq C |z-1|, \quad \text{for all } |z-1| \leq \varrho.
    \end{equation}
    \label{lem:hn1_lower_bound}
\end{lemma}
\begin{proof}
Fix $\delta >0$ such that the circle $\{z : |z-1|=\delta \}$ is contained in the cone $C_{\theta}$. Furthermore choose $0 < \varrho < \wt{\varrho} < \delta$. Note that the claim follows if we can establish $|h_{n}^{\prime} (z)| \leq C$ for all $|z-1| \leq \varrho$. Since $h_{n}$ is analytic in $C_{\theta}$, we may obtain such a bound from the Cauchy integral formula for derivatives
 \begin{equation}
    h_{n}^{\prime}(z) = \frac{1}{2 \pi i} \oint_{|s-1| = \wt{\varrho}} \frac{h_{n}(s)}{(s-z)^{2}} \diff s,
\end{equation}
and from an upper bound $C_{1}$ for $|h_{n}(z)|$ on $\{z : |z-1|=\wt{\varrho}\}$ that then allows us to choose $C=\wt{\varrho} C_{1}/(\wt{\varrho}- \varrho)^2$.  In order to obtain such a bound $C_1$ we use the representation of \cref{eq:hn_representation_useful_T} for $h_n(z)$ and $ |z-1|=\wt{\varrho}$. We can deform the contour of integration to obtain
\begin{equation}
    h_{n}(z) = \dfrac{1}{\pi} \int_{-1}^{1-\delta} \dfrac{V_{n}^{\prime}(s) - V_{n}^{\prime}(z)}{s-z} \dfrac{\diff s}{\sqrt{1-s^{2}}} + \dfrac{1}{2\pi i} \oint_{|s-1|=\delta} \dfrac{V_{n}^{\prime}(s) - V_{n}^{\prime}(z)}{s-z} \dfrac{\diff s}{r(s)} .
     \label{eq:hnz_decomposition}
\end{equation}
We have the following bounds on the integrals. First, for $|s-1| = \delta$ and $|z-1| = \wt{\varrho}$ we obtain from \cref{lem:complex_VSF_to_freud} for $n$ sufficiently large that
\begin{equation}
    \left| 
    \dfrac{V_{n}^{\prime}(s) - V_{n}^{\prime}(z)}{s-z} \frac{1}{r(s)}
    \right| \le \dfrac{|V_{n}^{\prime}(s)| + |V_{n}^{\prime}(z)|}{\left(\delta - \wt{\varrho} \right)\sqrt{\delta}} \leq \dfrac{4p}{\lambda_{p}} \dfrac{(1+\delta)^{p-1}}{\left(\delta - \wt{\varrho} \right)\sqrt{\delta}}.
\end{equation}
Second, we turn to the first summand on the right hand side of \cref{eq:hnz_decomposition}. For $|z-1| = \wt{\varrho}$ we obtain
\begin{equation}
    \left| \int_{-1}^{1-\delta} \dfrac{V_{n}^{\prime}(s) - V_{n}^{\prime}(z)}{s-z} \dfrac{\diff s}{\sqrt{1-s^{2}}} 
    \right| \leq \dfrac{1}{(\delta - \wt{\varrho}) \sqrt{\delta}} \left[ 2  \int_{0}^{1} |V_{n}^{\prime}(s)| \diff s
    + \dfrac{4p}{\lambda_{p}}  (1+\wt{\varrho})^{p-1}\right] ,
\end{equation}
for large $n$. We use \cref{cor:useful_estimate_T} with $0 <\epsilon <p$ to bound the integral in the previous upper bound and obtain
\begin{equation}
    \int_{0}^{1} |V_{n}^{\prime}(s)| \diff s \leq \dfrac{\beta_{n}}{n} \int_{0}^{A/\beta_{n}}  |Q^{\prime}(\beta_{n}s)| \diff s +  \int_{A/\beta_{n}}^{1} s^{p-1-\epsilon} \diff s \leq
    \dfrac{1}{n}  \int_{0}^{A}  |Q^{\prime}(u)| \diff u + \dfrac{1}{p-\epsilon},
\end{equation}
that finally yields the desired uniform bound $C_1$ on $|h_n(z)|$ for $|z-1| = \wt{\varrho}$ and large enough values of $n$.
\end{proof}

\subsection{Decay of contributions from the lens boundaries}
\begin{figure}[t]
    \centering
    \tikzset{every picture/.style={line width=0.75pt}} %

\begin{tikzpicture}[x=0.75pt,y=0.75pt,yscale=-1,xscale=1]

\draw  (60,220) -- (410,220)(60,100) -- (60,220) -- cycle (403,215) -- (410,220) -- (403,225) (55,107) -- (60,100) -- (65,107)  ;
\draw    (80,220) -- (80,230) ;
\draw    (140,220) -- (140,230) ;
\draw    (220,220) -- (220,230) ;
\draw    (300,220) -- (300,230) ;
\draw    (380,220) -- (380,230) ;
\draw   (99.5,206) .. controls (99.5,200.48) and (103.98,196) .. (109.5,196) .. controls (115.02,196) and (119.5,200.48) .. (119.5,206) .. controls (119.5,211.52) and (115.02,216) .. (109.5,216) .. controls (103.98,216) and (99.5,211.52) .. (99.5,206) -- cycle ;

\draw   (170,205.5) .. controls (170,199.98) and (174.48,195.5) .. (180,195.5) .. controls (185.52,195.5) and (190,199.98) .. (190,205.5) .. controls (190,211.02) and (185.52,215.5) .. (180,215.5) .. controls (174.48,215.5) and (170,211.02) .. (170,205.5) -- cycle ;
\draw   (250.5,206) .. controls (250.5,200.48) and (254.98,196) .. (260.5,196) .. controls (266.02,196) and (270.5,200.48) .. (270.5,206) .. controls (270.5,211.52) and (266.02,216) .. (260.5,216) .. controls (254.98,216) and (250.5,211.52) .. (250.5,206) -- cycle ;
\draw   (331.5,206.5) .. controls (331.5,200.98) and (335.98,196.5) .. (341.5,196.5) .. controls (347.02,196.5) and (351.5,200.98) .. (351.5,206.5) .. controls (351.5,212.02) and (347.02,216.5) .. (341.5,216.5) .. controls (335.98,216.5) and (331.5,212.02) .. (331.5,206.5) -- cycle ;
\draw  [fill={rgb, 255:red, 0; green, 0; blue, 0 }  ,fill opacity=1 ] (77,190.13) .. controls (77,188.4) and (78.4,187) .. (80.13,187) .. controls (81.85,187) and (83.25,188.4) .. (83.25,190.13) .. controls (83.25,191.85) and (81.85,193.25) .. (80.13,193.25) .. controls (78.4,193.25) and (77,191.85) .. (77,190.13) -- cycle ;
\draw  [fill={rgb, 255:red, 0; green, 0; blue, 0 }  ,fill opacity=1 ] (217,139.88) .. controls (217,138.15) and (218.4,136.75) .. (220.13,136.75) .. controls (221.85,136.75) and (223.25,138.15) .. (223.25,139.88) .. controls (223.25,141.6) and (221.85,143) .. (220.13,143) .. controls (218.4,143) and (217,141.6) .. (217,139.88) -- cycle ;
\draw  [fill={rgb, 255:red, 0; green, 0; blue, 0 }  ,fill opacity=1 ] (297,119.88) .. controls (297,118.15) and (298.4,116.75) .. (300.13,116.75) .. controls (301.85,116.75) and (303.25,118.15) .. (303.25,119.88) .. controls (303.25,121.6) and (301.85,123) .. (300.13,123) .. controls (298.4,123) and (297,121.6) .. (297,119.88) -- cycle ;
\draw  [fill={rgb, 255:red, 0; green, 0; blue, 0 }  ,fill opacity=1 ] (377,110.13) .. controls (377,108.4) and (378.4,107) .. (380.13,107) .. controls (381.85,107) and (383.25,108.4) .. (383.25,110.13) .. controls (383.25,111.85) and (381.85,113.25) .. (380.13,113.25) .. controls (378.4,113.25) and (377,111.85) .. (377,110.13) -- cycle ;
\draw    (80.13,190.13) -- (220.13,139.88) ;
\draw    (220.13,139.88) -- (300.13,119.88) ;
\draw    (300.13,119.88) -- (380.13,110.13) ;

\draw (43,78) node [anchor=north west][inner sep=0.75pt] [font=\large]   {$\Delta _{n}( x)$};
\draw (416,214) node [anchor=north west][inner sep=0.75pt] [font=\large]   {$x$};
\draw (70,235.4) node [anchor=north west][inner sep=0.75pt]  [font=\footnotesize]  {$\dfrac{\gamma}{\beta _{n}}$};
\draw (115,233.4) node [anchor=north west][inner sep=0.75pt]  [font=\footnotesize]  {$\dfrac{(\log n)^{\alpha }}{\beta _{n}}$};
\draw (205.5,234.4) node [anchor=north west][inner sep=0.75pt]  [font=\footnotesize]  {$\dfrac{1}{\beta _{n}^{2/3}}$};
\draw (286,234.4) node [anchor=north west][inner sep=0.75pt]  [font=\footnotesize]  {$\dfrac{1}{\beta _{n}^{1/3}}$};
\draw (366,234.4) node [anchor=north west][inner sep=0.75pt]  [font=\footnotesize]  {$1-\delta_{2} $};
\draw (102.5,199.5) node [anchor=north west][inner sep=0.75pt]   [align=left] {$A$};
\draw (173.25,199.5) node [anchor=north west][inner sep=0.75pt]   [align=left] {$B$};
\draw (254,199.5) node [anchor=north west][inner sep=0.75pt]   [align=left] {$C$};
\draw (335,200.5) node [anchor=north west][inner sep=0.75pt]   [align=left] {$D$};
\draw (121,142.5) node [anchor=north west][inner sep=0.75pt]  [font=\normalsize]  {$\frac{x}{(\log n)^{\alpha }}$};
\draw (314,94) node [anchor=north west][inner sep=0.75pt]  [font=\normalsize]  {$\frac{x}{(\log n)^{1+\alpha }}$};
\draw (240.04,121.27) node [anchor=north west][inner sep=0.75pt]  [font=\scriptsize,rotate=-346.33] [align=left] {Linear};
\draw (231.44,138.1) node [anchor=north west][inner sep=0.75pt]  [font=\scriptsize,rotate=-346.33] [align=left] {interpolation};

\end{tikzpicture}
    \caption{Height function $\Delta_{n}(x)$ setting the imaginary part of the lens boundary $\wt{\Sigma}_{1}$ (see \cref{fig:R_jump_contour}), as a function of $x=\Re{z}$ in the different regions $A,B,C,D$. The contour is piecewise linear. The exponent $\alpha$ is chosen as in \cref{lem:lens_boundaries}, and $\delta_{2}$ and $\gamma$ are as in \cref{prop:fn_endpoint,lem:fn_origin}. The specific exponents of $\beta_{n}^{2/3}$ and $\beta_{n}^{1/3}$ are not important; any decreasing pair of powers in $(0,1)$ would suffice.}
    \label{fig:lens_height_function}
\end{figure}
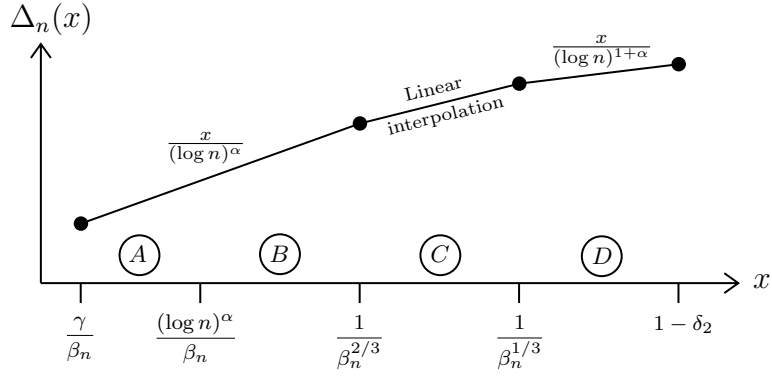
\begin{lemma}
    Consider $Q \in \mathrm{CVSLF}(p,q,\theta,\gamma)$, and set $\alpha > 0$ arbitrarily for $p>1$ but constrained to $\alpha < 1+q$ for $p=1$. Let the regions $A,B,C,D$ and the height function $\Delta_{n}(x)$ be defined as indicated in \cref{fig:lens_height_function}. Then, for $\mathbb{X} \in \{A,B,C,D\}$, the following properties hold.
    \begin{equation}
        \sup_{x \in \mathbb{X}} \left| e^{-2n \phi_{n}(x + i \Delta_{n}(x))} \right| \xrightarrow{n \to \infty} 0.
        \label{eq:prop_a}
    \end{equation}
    \begin{equation}
        n \, \mathrm{poly}(\log{n}) \int_{\mathbb{X}} \left| e^{-2n \phi_{n}(x + i \Delta_{n}(x))}\right| \diff x \xrightarrow{n \to \infty} 0.
        \label{eq:prop_b}
    \end{equation}
    \label{lem:lens_boundaries}
\end{lemma}
\begin{proof}
    The proof relies on the following basic equality which follows from \cref{eq:psi_hat_def,eq:phi_n_origin}:
    \begin{align}
        \left|e^{-2n \phi_{n}(x + i \Delta_{n}(x))}\right| &= \exp\left[n \Re\left(\int_{0}^{\Delta_{n}(x)} i r(x+i t) h_{n}(x+i t) \diff t\right)\right],
        \label{eq:basic_equality}
    \end{align}
    where $r(z) = (z-1)^{1/2} (z+1)^{1/2}$. For $x \in \mathbb{X}$ and $t \in [0,\Delta_{n}(x)]$, one can check that $\arg{r(x+ i t)} \in [\pi/4, 3\pi/4]$.

    We will define a function $\gamma_{n}(x)$ such that we can prove the bounds
    \begin{align}
        |h_{n}(x+i t)| &\geq \gamma_{n}(x),\label{eq:hn_gamma_bound}\\
        |\arg{h_{n}(x+it)}| &\leq \pi/6,\label{eq:hn_arg_bound}       
    \end{align}
    for $x \in \mathbb{X}$ and $t \in [0,\Delta_{n}(x)]$. Then from \cref{eq:basic_equality} we get
    \begin{equation}
        \left|e^{-2n \phi_{n}(x + i \Delta_{n}(x))}\right| \leq \exp\left[-\sin\left(\frac{\pi}{12}\right) n \gamma_{n}(x) \Delta_{n}(x)\right].
    \end{equation}
    We begin with the case $p=1$, and will indicate how $p>1$ can be treated at the end of the proof.
    
    \textbf{Region $A$}
    
    From \cref{lem:complex_hnz_uniform_lower}, in this region we can choose $\gamma_{n}(x) = h_{n}(0)/2$ to satisfy \cref{eq:hn_gamma_bound}. With this choice, \cref{eq:hn_arg_bound} follows from $|h_{n}(x+it) - h_{n}(0)| \leq h_{n}(0)/2$. From \cref{fig:lens_height_function} we recall that $\Delta_{n}(x) = x / (\log{n})^{\alpha}$. \cref{eq:prop_a} then follows from \cref{lem:hn0_scaling_sup} and
    \begin{equation}
        n \Delta_{n}(x) \gamma_{n}(x) \geq \Omega\left(\dfrac{n}{\beta_{n}} \dfrac{1}{(\log{n})^{\alpha}} (\log{n})^{1-\epsilon}\right) \geq \Omega\left((\log{n})^{1+q-\alpha-2\epsilon}\right),
    \end{equation}
    where the second inequality uses \cref{eq:vslf_characterization_prop5}; recall that $\alpha < 1+q$ for $p=1$, so the exponent can always be made positive by choosing $\epsilon$ small enough. \cref{eq:prop_b} then follows from a simple $L_{\infty}$ bound.

    \textbf{Regions $B$--$D$}

    In all these regions, we choose $\gamma_{n}(x) = h_{n}(x)/2$, so that \cref{eq:hn_gamma_bound,eq:hn_arg_bound} follow from $|h_{n}(x+ it) - h_{n}(x)| \leq h_{n}(x)/2$. Let us now verify that this is a valid choice for $\gamma_{n}(x)$. From \cref{rem:hnz_upper_bound} we have a uniform upper bound on $|h_{n}(z)|$ in a suitable $\Theta(|x|)$-sized neighborhood of $x$. One can then use the Cauchy integral formula for derivatives to upper bound $|h_{n}^{\prime}(x+it)| \leq \mathcal{O}((\log{n})^{1+\epsilon}/|x|)$. Then we can bound the difference as $|h_{n}(x+ it) - h_{n}(x)| \leq \mathcal{O}(\Delta_{n}(x) (\log{n})^{1+\epsilon}/|x|)$. To show that this is less than $h_{n}(x)/2$, in regions $B$\&$C$ we use the lower bound $h_{n}(x) \geq \Omega((\log{n})^{1-\epsilon})$ from \cref{rem:hnz_lower_bound}, while in region $D$ we use $h_{n}(x) \geq \Omega(1)$ from \cref{lem:hn_positive_x_g1}.
    
    With this choice of $\gamma_{n}(x)$, one can show that in regions $C$\&$D$, \cref{eq:prop_a,eq:prop_b} are satisfied because $n \gamma_{n}(x) \Delta_{n}(x) \geq \Omega(n^{1/3 - \epsilon})$. In region $B$, \cref{eq:prop_a} follows from
    \begin{equation}
        n \Delta_{n}(x) \gamma_{n}(x) \geq \Omega\left(\dfrac{n}{\beta_{n}} (\log{n})^{1-\epsilon}\right) \geq \Omega\left((\log{n})^{1+q-2\epsilon}\right),
    \end{equation}
    where we used $x \geq (\log{n})^{\alpha} /\beta_{n}$, the lower bound from \cref{rem:hnz_lower_bound}, and \cref{eq:vslf_characterization_prop5}. Since $1+q>0$, the exponent can be made positive by choosing $\epsilon$ sufficiently small. \cref{eq:prop_b} follows from the calculation
    \begin{align}
        n \, \mathrm{poly}(\log{n}) \int_{B} \exp\left[-\Omega\left(n \, (\log{n})^{1-\epsilon} \dfrac{x}{(\log{n})^{\alpha}}\right)\right] \diff x &\leq \mathcal{O}\left(\dfrac{n \, \mathrm{poly}(\log{n})}{n (\log{n})^{1-\epsilon-\alpha}} \exp\left[-\Omega\left(\dfrac{n}{\beta_{n}} (\log{n})^{1-\epsilon}\right)\right]\right),\\
        &\leq \mathcal{O}\Big(\mathrm{poly}(\log{n}) \exp\left[-\Omega\left((\log{n})^{1+q-2\epsilon}\right)\right]\Big),
    \end{align}
    where we used \cref{eq:vslf_characterization_prop5} for $\beta_{n}$. For $q>0$ the upper bound converges to zero superpolynomially in $n$, while for $q \leq 0$ the convergence may be subpolynomial (but still super-polylogarithmic). This concludes the proof for $p=1$.

    For $p>1$ the proof is much simpler. It suffices to take $\Delta_{n}(x) = x / (\log{n})^{2}$ and $\gamma_{n}(x) = (1/2) \min_{x\in[0,1]} h_{n}(x) \geq \Omega(1)$. The bound $|h_{n}(x+it) - h_{n}(x)| \leq h_{n}(x)/2$ follows from \cref{lem:complex_hnz_uniform_lower} for $x \in A$, and for $x \in B,C,D$ by an upper bound on the derivative $h_{n}^{\prime}(x+it)$ using \cref{rem:hnz_upper_bound} as before. Overall one finds that
    \begin{equation}
        n \gamma_{n}(x) \Delta_{n}(x) \geq \Omega\left(\dfrac{n}{\beta_{n}} (\log{n})^{-2}\right) \geq \Omega\left(n^{1 - 1/p -\epsilon}\right),
    \end{equation}
    which suffices to prove \cref{eq:prop_a,eq:prop_b} for $p>1$.
\end{proof}

\clearpage
\section{Airy bootstrap}
\label{sec:airy_bootstrap}
Since it will be useful for \cref{sec:universality}, we include here a version of the spectral bootstrap suited for recovering the spectral function near the `edge' $\omega \approx \beta_{n}$. A modification is necessary because the polynomial asymptotics in this region are qualitatively different from those in the bulk of the spectrum: they are governed by the `Airy' universality class rather than the `sine' universality class. These Airy asymptotics hold in the `regular' case where the density $\sigma_{n}(\omega)$ vanishes like a square root as $\omega \to \beta_{n}$ from below. For our class of spectral functions, we prove this is the case for large enough $n$ (see \cref{sec:equilibrium_measure}).

The Airy asymptotics are most easily stated in terms of a function $f_{n}(x)$ which is related to $I_{n}$ by
\begin{equation}
    f_{n}(x) = -\left(\dfrac{3\pi}{2}\left[\dfrac{n}{2} - I_{n}(\beta_{n} x)\right]\right)^{2/3}.
    \label{eq:fn_In_relation}
\end{equation}
Since $I_{n}(\beta_{n}) = \int_{0}^{\beta_{n}} \sigma_{n}(\omega) \diff \omega = n/2$, we have $f_{n}(1) = 0$. Moreover, given the square root vanishing of the derivative $\partial_{\omega}I_{n}(\omega) = \sigma_{n}(\omega)$, using \cref{eq:psin_rescaling,eq:psin_hn} we have
\begin{equation}
    f_{n}^{\prime}(1) = \left(\dfrac{n h_{n}(1)}{\sqrt{2}}\right)^{2/3}.
    \label{eq:fnp1}
\end{equation}
The asymptotics of $h_{n}(1)$ are given in \cref{lem:hn1_scaling}, where it is shown that $h_{n}(1)$ is $\mathcal{O}(1)$ as $n \to \infty$, so $f_{n}^{\prime}(1) \sim \mathcal{O}(n^{2/3})$.

For the Airy bootstrap, we start at the endpoint $\omega=\beta_{n}$, and then work backwards to lower frequencies. The derivation of the following equations is given in \cref{sec:airy_bootstrap_derivation}.

First we approximate $h_{n}(1)$ using 
\begin{equation}
    h_{n}(1) \approx \dfrac{1}{2n} \left(\dfrac{\Ai(0)}{\Ai^{\prime}(0)}\right)^{3} \left[\rho - \left(\dfrac{p_{n}(\beta_{n}) + p_{n-1}(\beta_{n})}{p_{n}(\beta_{n}) - p_{n-1}(\beta_{n})}\right)\right]^{3},
    \label{eq:hn1_asymptotic}
\end{equation}
where $\Ai$ is the Airy function. Then we can estimate the spectral function at $\omega=\beta_{n}$ using \cref{eq:fnp1} and
\begin{equation}
    \dfrac{\Phi(\beta_{n})}{2\pi} \approx \dfrac{1}{\beta_{n} p_{n}(\beta_{n})^{2}}\left[\big(2 f_{n}^{\prime}(1)\big)^{1/4} \Ai(0) - (1+\rho) \big(2 f_{n}^{\prime}(1)\big)^{-1/4} \Ai^{\prime}(0)\right].
    \label{eq:spectral_function_endpoint}
\end{equation}
We can now proceed with the main bootstrap equations: with $x \equiv \omega/\beta_{n}$, the update of the spectral function is given by
\begin{align}
    \dfrac{\Phi(\omega)}{2\pi} &\approx \dfrac{2}{\beta_{n}} \dfrac{1}{p_{n}(\omega)^{2} + p_{n-1}(\omega)^{2}} \dfrac{1}{\sqrt{1-x^{2}}}\Big[- 2x \Ai\big(f_{n}(x)\big) \Ai^{\prime}\big(f_{n}(x)\big) \sin\left[\rho \arccos{x}\right] \label{eq:phi_airy_bootstrap} \\
    &+\big(-f_{n}(x)\big)^{1/2} \Ai\big(f_{n}(x)\big)^{2}\big(x \cos\left[\rho \arccos{x}\right]+1\big) - \big(-f_{n}(x)\big)^{-1/2} \Ai^{\prime}\big(f_{n}(x)\big)^{2}\big(x \cos\left[\rho \arccos{x}\right]-1\big)\Big], \nonumber
\end{align}
while the update of the phase function $f_{n}(x)$ is given by
\begin{align}
    f_{n}^{\prime}(x) &\approx \dfrac{1}{2\Ai^{2} f_{n}^{2} - 2 (\Ai^{\prime})^{2} f_{n} - \Ai \Ai^{\prime}}\Big[-2 \beta_{n} \dfrac{\Phi(\omega)}{2\pi} K_{n}(\omega, \omega) f_{n} + \frac{2}{1-x^{2}} \Ai \Ai^{\prime} \cos[\rho \arccos{x}] f_{n} \label{eq:fnp_airy_bootstrap} \\
    &- \dfrac{\sin[\rho \arccos{x}]}{1-x^{2}} \left(\Ai^{2} \big(-f_{n}\big)^{3/2} - (\Ai^{\prime})^{2} (-f_{n})^{1/2}\right) - \dfrac{\rho}{\sqrt{1-x^{2}}}\left(\Ai^{2} \big(-f_{n}\big)^{3/2} + (\Ai^{\prime})^{2} (-f_{n})^{1/2}\right)\Big], \nonumber
\end{align}
where we have suppressed the arguments of $f_{n} \equiv f_{n}(x)$, $\Ai \equiv \Ai(f_{n}(x))$, and $\Ai^{\prime} \equiv \Ai^\prime(f_{n}(x))$.

Given a choice of frequency spacing $0 < \delta \omega \ll 1$ and a minimum frequency $\omega_{\mathrm{min}}$, the main routine then works as follows.
\begin{figure}[t]
    \centering
    \includegraphics[height=18em]{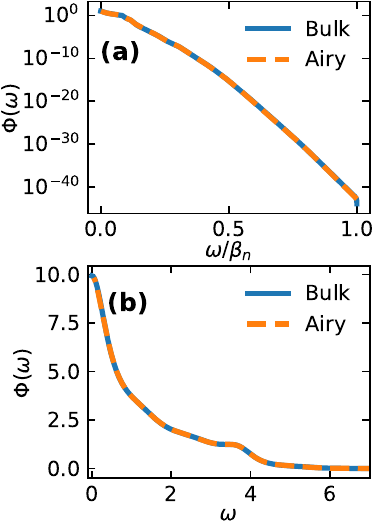}
    \includegraphics[height=18em]{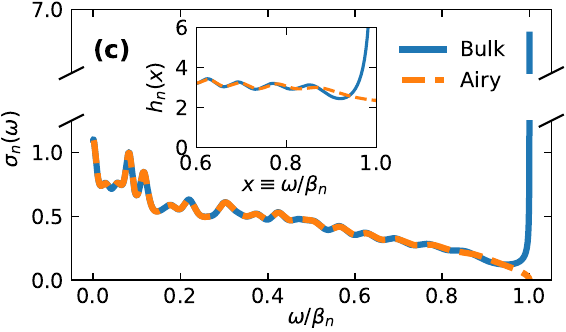}
    \caption{\textbf{(a,b)} Log- and linear-scale spectral function $\Phi(\omega)$ of the energy current operator of the mixed field Ising model, as computed using the `bulk' spectral bootstrap of \cref{sec:bulk_bootstrap} and the `Airy' spectral bootstrap of \cref{sec:airy_bootstrap}. While they are intended to work in different frequency regimes, the estimates of $\Phi(\omega)$ agree remarkably well throughout the bulk. \textbf{(c)} Comparison between the extracted equilibrium measure $\sigma_{n}(\omega)$ for the same model. They agree throughout the bulk, but the bulk bootstrap breaks down near $\omega = \beta_{n}$, whereas the Airy bootstrap produces the expected square root vanishing of $\sigma_{n}(\omega)$ as $\omega \to \beta_{n}$. Indeed, the inset shows $h_{n}(x \equiv \omega/\beta_{n}) = 2\pi \frac{\beta_{n}}{n} \frac{\sigma_{n}(\omega)}{\sqrt{1-(\omega/\beta_{n})^{2}}}$; for the Airy bootstrap this is finite and as $x \to 1$ approaches the value $h_{n}(1) \approx 2$ expected from \cref{lem:hn1_scaling} with $p=1$. Here we set $n=40$, with $\beta_{n} \approx 2 b_{n} \approx 46.1$.}
    \label{fig:airy_bootstrap_benchmark}
\end{figure}
\begin{enumerate}
    \setcounter{enumi}{-1}
    \item Set $\omega = \beta_{n}$ and $f_{n}(1) = 0$. Compute $f_{n}^{\prime}(1)$ using \cref{eq:fnp1,eq:hn1_asymptotic}, and $\Phi(\beta_{n})$ using \cref{eq:spectral_function_endpoint}.
    \item Increment $\omega \mapsto \omega - \delta \omega$.
    \item Set $f_{n}\left(\dfrac{\omega}{\beta_{n}}\right) = f_{n}\left(\dfrac{\omega+\delta\omega}{\beta_{n}}\right) - f_{n}^{\prime}\left(\dfrac{\omega+\delta\omega}{\beta_{n}}\right) \times \dfrac{\delta\omega}{\beta_{n}}$.
    \item Compute $\Phi(\omega)$ using \cref{eq:phi_airy_bootstrap}.
    \item Compute $f_{n}^{\prime}(\omega/\beta_{n})$ using \cref{eq:fnp_airy_bootstrap}.
    \item Repeat steps 1--4 until $\omega=\omega_{\mathrm{min}}$, then terminate.
\end{enumerate}
Note that in step 2 we divide the derivative by $\beta_{n}$ because $f_{n}(x)$ is defined in terms of $x \equiv \omega / \beta_{n}$.

In \cref{fig:airy_bootstrap_benchmark} we show a benchmark of this `Airy bootstrap' on the energy current operator of the mixed field Ising model (MFIM), as we did in \cref{sec:bulk_bootstrap}, and compare to the results from the `bulk bootstrap' described in that section. These two methods are designed to work in different frequency regimes: the bulk bootstrap starts at $\omega=0$ and iterates up to some $\omega_{\mathrm{max}} \ll \beta_{n}$, while the Airy bootstrap starts at $\omega=\beta_{n}$ and iterates down to some $\omega_{\mathrm{min}}$. They should agree in frequency regimes where both methods are approximately valid, but \textit{a priori} they could disagree elsewhere. In fact, in \cref{fig:airy_bootstrap_benchmark}(a,b), we find that the extracted spectral functions $\Phi(\omega)$ agree remarkably well throughout almost the whole frequency range $\omega \in [0,\beta_{n}]$, with only some small deviation at $\omega\approx\beta_{n}$. This good agreement lends credence to our investigation of the high-frequency tail of $\Phi(\omega)$ in \cref{sec:bulk_bootstrap}, where we found $\Phi(\omega \to \infty) \sim \exp[-\mathcal{O}(\omega \log{\omega})]$, consistent with locality bounds on operator growth in 1D~\cite{parkerUniversalOperatorGrowth2019}. It is also perhaps surprising that the Airy bootstrap can produce accurate estimates of the spectral function at moderate frequencies, since one might initially have worried about the numerical stability of a procedure which begins with the spectral function as small as $\Phi(\omega=\beta_{n}) \sim 10^{-43}$.

More significant deviation between the bulk and Airy bootstraps can be found by looking at their estimates for the equilibrium measure $\sigma_{n}(\omega)$, as shown in \cref{fig:airy_bootstrap_benchmark}(c). This is computed directly as part of the bulk bootstrap, while for the Airy bootstrap we can recover it from $f_{n}(\omega/\beta_{n})$ using \cref{eq:fn_In_relation} and $\sigma_{n}(\omega) = \partial_{\omega} I_{n}(\omega)$. We find that the two methods give closely matching estimates of $\sigma_{n}(\omega)$ for $0\leq\omega/\beta_{n} \lesssim 0.9$, but as $\omega$ approaches the spectral edge at $\omega = \beta_{n}$, the bulk bootstrap estimate blows up, while the Airy bootstrap produces the square root vanishing expected for `regular' equilibrium measures (we remind the reader that we prove that this regularity property holds at large-$n$ for our class of spectral function---see \cref{sec:equilibrium_measure}). To corroborate this, in the inset to \cref{fig:airy_bootstrap_benchmark}(c) we plot the function (c.f.~\cref{eq:psin_rescaling,eq:psin_hn})
\begin{equation}
    h_{n}(x \equiv \omega/\beta_{n}) = 2\pi \dfrac{\beta_{n}}{n} \dfrac{\sigma_{n}(\omega)}{\sqrt{1-(\omega/\beta_{n})^{2}}},
\end{equation}
which should be finite at $x=1$ if $\sigma_{n}(\omega)$ vanishes like a square root at $\omega=\beta_{n}$. For the Airy bootstrap estimate of $h_{n}(x)$, this is indeed the case, and we can also see that the limiting value of $h_{n}(x)$ as $x \to 1$ is close to 2, which is expected from \cref{lem:hn1_scaling} with $p=1$, given the $\Phi(\omega \to \infty) \sim \exp[-\mathcal{O}(\omega \log{\omega})]$ high frequency decay of the spectral function. On the other hand, the bulk bootstrap estimate does not successfully capture this behavior. Although the results of \cref{fig:airy_bootstrap_benchmark}(a,b) show that this disagreement in estimating $\sigma_{n}(\omega)$ is not fatal for approximating the spectral function, it will be more relevant that we have an accurate estimate of $\sigma_{n}(\omega)$ in \cref{sec:universality}, where we check for the emergence of universality.

\clearpage
\section*{Glossary of Symbols}
\label{sec:glossary}
We use standard asymptotic notation. Let $f,g : \mathbb{R}^{+} \to \mathbb{C}$ be two functions. We write $f(x) = \mathcal{O}(g(x))$ iff there exist constants $C,x_{0}>0$ such that $|f(x)| < C |g(x)|$ for all $x>x_{0}$. We write $f(x) = o(g(x))$ iff, for every constant $C>0$, there exists $x_{0}>0$ such that $|f(x)| < C |g(x)|$ for all $x>x_{0}$. We write $f(x) = \Omega(g(x))$ iff $g(x) = \mathcal{O}(f(x))$, and $f(x) = \omega(g(x))$ iff $g(x) = o(f(x))$. We write $f(x) = \Theta(g(x))$ iff $f(x) = \mathcal{O}(g(x))$ and $f(x) = \Omega(g(x))$. We write $f(x) = \wt{\mathcal{O}}(g(x))$ iff $f(x) = \mathcal{O}(g(x) (\log|g(x)|)^{k})$ for some finite $k$, and $f(x) = \wt{\Omega}(g(x))$ iff $g(x) = \wt{\mathcal{O}}(f(x))$. We write $f(x) = \mathrm{poly}(g(x))$ iff $f(x) = \mathcal{O}((g(x))^{k})$ for some finite $k$.
\bgroup
\def\arraystretch{2}
{\makeatletter
\renewcommand\table@hook{\footnotesize}
\makeatother
\begin{longtable}{ccc} 
    \toprule
    \textbf{Symbol} & \textbf{Description} & \textbf{Definition} \\
    \midrule
    \endfirsthead
    \toprule
    \textbf{Symbol} & \textbf{Description} & \textbf{Definition} \\
    \midrule
    \endhead
    $\mathcal{L}$ & Liouvillian superoperator & $\mathcal{L}(\cdot) = [H, \cdot \,]$ \\
    $A$ & Initial Lanczos operator (not necessarily normalized) & System dependent \\
    $C(t)$ & Autocorrelation function & $C(t) = (A|A(t)) = (A | e^{i \mathcal{L} t} A)$ \\
    $\Phi(\omega)$ & Spectral function & $\Phi(\omega) = \int_{\mathbb{R}} e^{-i\omega t} C(t) \diff t$ \\
    $\rho$ & Low-frequency power-law exponent & $\Phi(\omega\to 0) \sim |\omega|^{\rho}$ \\
    $G(z)$ & Green's function (resolvent) & $G(z) = \left(A \left| \frac{1}{z - \mathcal{L}}\right| A\right)$ \\
    $b_{n}$ & $n$th Lanczos coefficient & \cref{eq:lanczos_coefficient_def} \\
    $O_{n}$ & $n$th Lanczos operator & \cref{eq:lanczos_operator_def} (see also \cref{eq:lanczos_vector}) \\
    $\mathcal{L}_{n}$ & Level-$n$ Liouvillian & $\mathcal{L}_{n} = \mathcal{Q}_{n} \mathcal{L} \mathcal{Q}_{n}$ with $\mathcal{Q}_{n} = \mathds{1} - \sum_{m=0}^{n-1} |O_{m})(O_{m}|$ \\
    $G_{n}(z)$ & Level-$n$ Green's function & $G_{n}(z) = \left(O_{n} \left| \frac{1}{z - \mathcal{L}_{n}}\right| O_{n}\right)$ \\
    $w(\omega)$ & Weight function for orthogonal polynomials & $w(\omega) = \Phi(\omega) / 2\pi$ \\
    $\Phi_{n}(\omega)$ & $n$th spectral function (w.r.t.\ $\mathcal{L}_{n}$) & $\Phi_{n}(\omega) = \int_{\mathbb{R}} e^{-i\omega t} (O_{n}| e^{i \mathcal{L}_{n}t} O_{n}) \diff t = 2 \Im[G_{n}(\omega -i 0^{+})]$ \\
    $\wt{\Phi}_{n}(\omega)$ & $n$th spectral function (w.r.t.\ $\mathcal{L}$) & $\Phi_{n}(\omega) = \int_{\mathbb{R}} e^{-i\omega t} (O_{n}| e^{i \mathcal{L}t} O_{n}) \diff t$ \\
    $p_{n}(\omega)$ & Degree $n$ orthonormal polynomial w.r.t.\ $\Phi(\omega)/2\pi$ & \cref{eq:poly_orthogonality} (see also \cref{eq:three_term_recursion}) \\
    $y_{n}$ & Positive leading coefficient of $p_{n}(\omega)$ & $p_{n}(\omega) = y_{n} \omega^{n} + \mathcal{O}(\omega^{n-1}), \quad y_{n} > 0$ \\
    $P_{n}(\omega)$ & Monic orthogonal polynomial & $P_{n}(\omega) = (1/y_{n}) p_{n}(\omega)$ \\
    $Q(\omega)$ & Potential (for Coulomb gas) & Defined implicitly by $\Phi(\omega)/2\pi \equiv |\omega|^{\rho} e^{-Q(\omega)}$ \\
    $p,q$ & Polynomial growth exponents for potential & $Q(\omega \to \infty) \sim \omega^{p} (\log{\omega})^{q + o(1)}$ (see \cref{sec:potential_definitions}) \\
    $\sigma_{n}(\omega)$ & Equilibrium density (with charge $n$) & \cref{eq:coulomb_gas_def} \\
    $\beta_{n}$ & $n$th Mhaskar-Rakhmanov-Saff (MRS) number & \cref{eq:MRS_def} (asymptotics in \cref{eq:beta_n_asymptotics}; satisfies $\beta_{n} = 2 b_{n}[1 + \mathcal{O}(1/n)]$) \\
    $V_{n}(x)$ & Rescaled potential & $V_{n}(x) = Q(\beta_{n}x)/n$ \\
    $\psi_{n}(x)$ & Rescaled equilibrium density & $\psi_{n}(x) = (\beta_{n}/n) \sigma_{n}(\beta_{n}x)$ \\
    $\psi^{(p)}(x)$ & Ullman distribution & \cref{eq:freud_eq_measure} \\
    $h_{n}(x)$ & Non-semicircle part of eq.~measure & $h_{n}(x) = 2\pi \psi_{n}(x)/\sqrt{1 - x^{2}}$ (see also \cref{eq:hn_integral}) \\
    $K_{n}(x,y)$ & Christoffel-Darboux kernel & $K_{n}(x,y) = \sum_{m=0}^{n-1} p_{m}(x) p_{n}(y)$ \\
    $\hat{K}_{n}(x,y)$ & Weighted Christoffel-Darboux kernel & $\hat{K}_{n}(x,y) = \sqrt{w(x)w(y)} K_{n}(x,y)$ \\
    $\mathbb{S}(u,v)$ & Sine kernel & $\mathbb{S}(u,v) = \frac{\sin[\pi(u-v)]}{\pi(u-v)}$ \\
    $\mathbb{J}_{\rho/2}(u,v)$ & Bessel kernel & $\mathbb{J}_{\rho/2}(u,v) = \pi \sqrt{u} \sqrt{v} \frac{J_{\frac{\rho+1}{2}}(\pi u) J_{\frac{\rho-1}{2}}(\pi v) - J_{\frac{\rho-1}{2}}(\pi u) J_{\frac{\rho+1}{2}}(\pi v)}{2(u-v)}$ \\
    $\mathbb{A}(u,v)$ & Airy kernel & $\mathbb{A}(u,v) = \frac{\Ai(u) \Ai^{\prime}(v) - \Ai^{\prime}(u) \Ai(v)}{u-v}$ \\
    $\theta_{n}(\omega)$ & (WKB) Phase factor & $\theta_{n}(\omega) = -\pi \int_{\omega}^{\beta_{n}} \sigma_{n}(\omega^{\prime}) \diff \omega^{\prime} - \frac{\pi}{4}$ \\
    $I_{n}(\omega)$ & Cumulative equilibrium measure & $I_{n}(\omega) = \int_{0}^{\omega} \sigma_{n}(\omega^{\prime}) \diff \omega^{\prime}$ \\ 
    $\mathcal{J}$ & Heat/spin/charge current operator & Hamiltonian-dependent (see \cref{sec:transport_coeffs}) \\
    $g_{n}(z)$ & Logarithmic transform of equilibrium measure & $g_{n}(z) = \int_{-1}^{1} \log(z-s) \psi_{n}(s) \diff s, \quad z \in \mathbb{C} \setminus \mathbb{R}$ \\
    $l_{n}$ & Lagrange multiplier for Coulomb gas energy minimization & \cref{eq:lagrange_at_zero} \\
    $Y(z)$ & Solution to fundamental Riemann-Hilbert problem & \cref{eq:fki_sol} \\
    $\sigma_{3}$ & 3rd Pauli matrix (not to be confused with $\sigma_{n}(\omega)$) & $\sigma_{3} = \mathrm{diag}(1,-1)$ \\
    $Y_{1}$ & $\mathcal{O}(1/z)$ correction to leading $Y(z\to\infty)$ scaling & $z^{-n \sigma_{3}} Y(z) = \mathds{1} + Y_{1} / z + \mathcal{O}(1/z^{2})$ as $z \to \infty$ \\
    $U(z)$ & $Y(z)$ rescaled by $\beta_{n}$ & $U(z) = \beta_{n}^{-(n+\rho/2)\sigma_{3}} Y(\beta_{n}z) \beta_{n}^{(\rho/2)\sigma_{3}}$ \\
    $T(z)$ & $U(z)$ rescaled by the log-transformed equilibrium measure & $T(z) = e^{-(n l_{n}/2)\sigma_{3}} U(z) e^{(n l_{n}/2)\sigma_{3}} e^{-n g_{n}(z) \sigma_{3}}$ \\
    $\hat{\psi}_{n}(z)$ & Analytic continuation of equilibrium measure & \cref{eq:psi_hat_def} \\
    $\phi_{n}(z)$ & Integral of equilibrium measure & $\phi_{n}(z) = -\pi i \int_{1}^{z} \hat{\psi}_{n}(s) \diff s, \quad z \in \mathbb{C} \setminus \mathbb{R}$ \\
    $S(z)$ & Solution of intermediate RHP with contour deformation & \cref{eq:S_def} \\
    $N(z)$ & Solution of RHP for outside region & \cref{eq:Pinf_sol} \\
    $D(z)$ & Szeg\H{o} function for $|x|^{\rho}$ on $[-1,1]$ & $D(z) = z^{\rho/2} / \varphi(z)^{\rho/2}$ with $\varphi(z) = z + (z+1)^{1/2} (z-1)^{1/2}$ \\
    $a(z)$ & Generates solution of outside RHP for $\rho=0$ & $a(z) = (z-1)^{1/4} / (z+1)^{1/4}$ \\
    $r(z)$ & Appears in solutions for RHPs on $[-1,1]$ & $r(z) = (z-1)^{1/2} (z+1)^{1/2}$ \\
    $P(z)$ & Solution of RHP in local parametrices near $z=\pm 1$ and $z=0$ & \cref{eq:P_def_endpoint,eq:Pdef_left_endpoint} for $z=\pm 1$, and \cref{eq:P_def_origin} for $z=0$ \\
    $f_{n}(z)$ & Biholomorphic map to auxiliary $\zeta$-plane & \hyperref[prop:fn_endpoint]{Prop.~\ref*{prop:fn_endpoint}} for $z$ near $z=1$, and \cref{eq:fn_zero_def} for $z$ near $z=0$ \\
    $P^{(1)}(z)$ & Solution of constant-jump RHP for local parametrices & \cref{eq:P1_def} for $z=1$, and \cref{eq:P1_def_origin} for $z=0$ \\
    $\Psi(\zeta)$ & Solution of auxiliary RHP for $z=\pm 1$ local parametrices & \cref{eq:Psi_def} \\
    $\Psi_{\rho/2}(\zeta)$ & Solution of auxiliary RHP for $z=0$ local parametrix & \cref{eq:Psi_rho_def} \\
    $\omega(z)$ & Analytic continuation of $|x|^{\rho}$, continuous across $\mathbb{R}$ & \cref{eq:omega_def} \\
    $W(z)$ & Function related to $|x|^{\rho}$, discontinuous across $\mathbb{R}$ & \cref{eq:W_def} \\
    $S_{\mathrm{par}}(z)$ & Approximate solution of RHP for $S(z)$ using local parametrices & \cref{eq:Spar_def} \\
    $R(z)$ & Solution of residual RHP & $R(z) = S(z) S_{\mathrm{par}}(z)^{-1}$, solved in \cref{eq:R_exact_sol} \\
    $\wt{\Sigma}_{j}$ & Contour on which $R(z)$ has a jump, $j=1,\dots,9$ & See \cref{fig:R_jump_contour} \\
    \bottomrule
\end{longtable}}
\egroup

\end{document}